\documentclass[a4paper,11pt,bibliography=totoc%,BCOR12mm%,twoside,openright,%draft%DIV=calc%,fleqn
		]{scrreprt}

\usepackage[english,ngerman]{babel}
\usepackage{scrhack}

\usepackage[utf8]{inputenc}

\usepackage{scrpage2}
\automark[section]{chapter}
\setheadsepline{0.5pt}

\usepackage{verbatim}

\usepackage{enumerate}

\usepackage{mathrsfs}
\usepackage{amssymb}
\usepackage{bbm}
\usepackage{color}
\usepackage{graphicx}
\usepackage{float}
\usepackage{mathtools}
\usepackage{amsthm}
\usepackage{amsmath}

\newcommand{\norm}[1]{\left\|#1\right\| }

\newcommand{\epsi}[0]{\varepsilon}
\newcommand{\mrm}[1]{{\mathrm{#1}}}
\newcommand{\field}[1]{\mathbb{#1}}
\newcommand{\C}{\field{C}}
\newcommand{\R}{\field{R}}
\newcommand{\N}{\field{N}}
\newcommand{\Z}{\field{Z}}
\newcommand{\id}[0]{\mathbf{1}}

\newcommand{\E}{{\mathrm{e}}}

\newcommand{\D}{{\mathrm{d}}}

\DeclareMathOperator{\Tr}{Tr}

\newcommand{\im}{\mathrm{i}}

\newcommand{\bra}[1]{\langle #1 |}
\newcommand{\ket}[1]{| #1 \rangle}
\newcommand{\LzO}{{L^2(\Omega)}}
\newcommand{\LiO}{{L^\infty(\Omega)}}
\newcommand{\LiOt}{{L^\infty(\tilde \Omega)}}
\newcommand{\LzOt}{{L^2(\tilde \Omega)}}
\newcommand{\LeOt}{{L^1(\tilde \Omega)}}
\newcommand{\LzOf}{{L^2(\Omega_\mathrm{f})}}
\newcommand{\LiOf}{{L^\infty(\Omega_\mathrm{f})}}
\newcommand{\LeOf}{{L^1(\Omega_\mathrm{f})}}
\newcommand{\LzOc}{{L^2(\Omega_\mathrm{c})}}
\newcommand{\LiOc}{{L^\infty(\Omega_\mathrm{c})}}

\newcommand{\Op}{{\mathrm{Op}}}
\DeclareMathOperator{\supp}{supp}

\newcommand{\weq}[2]{\stackrel{\mathclap{#1}}{#2}}

\newcommand\numberthis{\addtocounter{equation}{1}\tag{\theequation}}

\theoremstyle{plain}% default
\newtheorem{thm}{Theorem}%[chapter]
\theoremstyle{plain}
\newtheorem{lem}{Lemma}[chapter]

\newtheorem{cor}[lem]{Corollary}
\newtheorem*{thm*}{Theorem}%
\newtheorem*{lem*}{Lemma}
\newtheorem*{prop*}{Proposition}
 
\newtheorem*{cor*}{Corollary}

\theoremstyle{definition}
\newtheorem{defn}{Definition}%[section]
\newtheorem{exmp}{Example}%[section]
\newtheorem*{defn*}{Definition}%[section]
\theoremstyle{remark}
\newtheorem{rem}{Remark}

\makeatletter
\DeclareFontFamily{OMX}{MnSymbolE}{}
\DeclareSymbolFont{MnLargeSymbols}{OMX}{MnSymbolE}{m}{n}
\SetSymbolFont{MnLargeSymbols}{bold}{OMX}{MnSymbolE}{b}{n}
\DeclareFontShape{OMX}{MnSymbolE}{m}{n}{
    <-6>  MnSymbolE5
   <6-7>  MnSymbolE6
   <7-8>  MnSymbolE7
   <8-9>  MnSymbolE8
   <9-10> MnSymbolE9
  <10-12> MnSymbolE10
  <12->   MnSymbolE12
}{}
\DeclareFontShape{OMX}{MnSymbolE}{b}{n}{
    <-6>  MnSymbolE-Bold5
   <6-7>  MnSymbolE-Bold6
   <7-8>  MnSymbolE-Bold7
   <8-9>  MnSymbolE-Bold8
   <9-10> MnSymbolE-Bold9
  <10-12> MnSymbolE-Bold10
  <12->   MnSymbolE-Bold12
}{}

\let\llangle\@undefined
\let\rrangle\@undefined
\DeclareMathDelimiter{\llangle}{\mathopen}%
                     {MnLargeSymbols}{'164}{MnLargeSymbols}{'164}
\DeclareMathDelimiter{\rrangle}{\mathclose}%
                     {MnLargeSymbols}{'171}{MnLargeSymbols}{'171}
\makeatother

\usepackage[toc]{glossaries}
\usepackage{glossary-longragged}
\newglossary[slg]{is}{syi}{syg}{List of symbols}
\newglossary[slg2]{is2}{syi2}{syg2}{List of symbols2}
\makeglossaries

\newglossaryentry{symb:H^N}{
name={\ensuremath{\mathscr{H}^N}},
description={$N$-particle Hilbert space },
sort=symbolHilbertspace, type=is
}

\newglossaryentry{symb:HNepsi}{
name={\ensuremath{H^\epsi_N}},
description={Rescaled $N$-particle Hamiltonian},
sort=symbolHamiltonian, type=is
}
\newglossaryentry{symb:psiNepsi}{
name={\ensuremath{\psi^\epsi_N}},
description={$N$-particle wave function element of $\mathscr{H}^N$},
sort=symbolpsi, type=is
}
\newglossaryentry{symb:epsi}{
name={\ensuremath{\epsi}},
description={Parameter controlling the strength of the confinement},
sort=symbolepsi, type=is
}

\newglossaryentry{symb:r}{
name={\ensuremath{r}},
description={Element of $\Omega$},
sort=symbolr, type=is
}

\newglossaryentry{symb:y}{
name={\ensuremath{y}},
description={Element of $\Omega_\mathrm{c}$},
sort=symboly, type=is
}

\newglossaryentry{symb:x}{
name={\ensuremath{x}},
description={Element of $\Omega_\mathrm{f}$},
sort=symbolx, type=is
}

\newglossaryentry{symb:omega}{
name={\ensuremath{\Omega}},
description={One-particle configuration space, subset of $\R^3$},
sort=symbolomega, type=is
}

\newglossaryentry{symb:omegac}{
name={\ensuremath{\Omega_\mathrm{c}}},
description={Confined part of the one-particle configuration space},
sort=symbolomegac, type=is
}

\newglossaryentry{symb:omegaf}{
name={\ensuremath{\Omega_\mathrm{f}}},
description={Free part of the one-particle configuration space},
sort=symbolomegaf, type=is
}

\newglossaryentry{symb:N}{
name={\ensuremath{N}},
description={Number of particles},
sort=symbolN, type=is
}

\newglossaryentry{symb:w}{
name={\ensuremath{w}},
description={Two-particle interaction potential },
sort=symbolw, type=is
}

\newglossaryentry{symb:theta}{
name={\ensuremath{\theta}},
description={Scaling parameter element of $[0,1]$ },
sort=symboltheta, type=is
}

\newglossaryentry{symb:h}{
name={\ensuremath{h}},
description={One-particle Hamiltonian},
sort=symbolhamiltonian, type=is
}
\newglossaryentry{symb:V}{
name={\ensuremath{V}},
description={External potential},
sort=symbolV, type=is
}
\newglossaryentry{symb:WetN}{
name={\ensuremath{W^{\epsi, \theta,N}}},
description={Scaled two-particle interaction potential},
sort=symbolWetN, type=is
}
\newglossaryentry{symb:varphi}{
name={\ensuremath{\varphi}},
description={One-particle wave function element of $L^2(\Omega)$},
sort=symbolphi, type=is
}
\newglossaryentry{symb:Phi}{
name={\ensuremath{\Phi}},
description={Free part of the one-particle wave function governed by a nonlinear PDE},
sort=symbolPhi, type=is
}
\newglossaryentry{symb:chi}{
name={\ensuremath{\chi}},
description={Confined, stationary part of the one-particle wave function},
sort=symbolchi, type=is
}
\newglossaryentry{symb:a}{
name={\ensuremath{a}},
description={Scaling parameter depending on the number of confined directions},
sort=symbola, type=is
}
\newglossaryentry{symb:alpha}{
name={\ensuremath{\alpha}},
description={Counting functional},
sort=symbolalpha, type=is
}
\newglossaryentry{symb:gamma}{
name={\ensuremath{\gamma}},
description={One-particle density matrix},
sort=symbolgamma, type=is
}
\newglossaryentry{symb:p}{
name={\ensuremath{p}},
description={Projection onto $\varphi$},
sort=symbolp, type=is
}
\newglossaryentry{symb:q}{
name={\ensuremath{q}},
description={Projection onto the orthogonal complement of $\varphi$},
sort=symbolq, type=is
}
\newglossaryentry{symb:PkN}{
name={\ensuremath{P_{k,N}}},
description={Projection onto $k$ "bad" particles},
sort=symbolPkN, type=is
}
\newglossaryentry{symb:beta}{
name={\ensuremath{\beta}},
description={Counting functional with weight function n},
sort=symbolbeta, type=is
}
\newglossaryentry{symb:n}{
name={\ensuremath{n}},
description={Weight function $\sqrt{\frac{k}{N}}$},
sort=symboln, type=is
}
\newglossaryentry{symb:w0}{
name={\ensuremath{w^0}},
description={Hartree approximation to the two-particle interaction potential $w$ },
sort=symbolw0, type=is
}
\newglossaryentry{symb:ws}{
name={\ensuremath{w_s}},
description={Singular part of the two-particle interaction potential $w$ },
sort=symbolws, type=is
}
\newglossaryentry{symb:winfty}{
name={\ensuremath{w_\infty}},
description={Bounded part of the two-particle interaction potential $w$ },
sort=symbolwinfty, type=is
}
\newglossaryentry{symb:DHepsiN}{
name={\ensuremath{D(H_N^\epsi)}},
description={Domain of the operator $H^\epsi_N$ },
sort=symboldomain, type=is
}
\newglossaryentry{symb:b}{
name={\ensuremath{b}},
description={Coupling parameter in the NLS},
sort=symbola, type=is
}
\newglossaryentry{symb:fepsi}{
name={\ensuremath{f(\epsi)}},
description={Function controlling the convergence of $w^\epsi $ to $w^0$ },
sort=symbolfepsi, type=is
}
\newglossaryentry{symb:eta}{
name={\ensuremath{\eta}},
description={Parameter controlling the rate of convergence},
sort=symboleta, type=is
}
\newglossaryentry{symb:nu}{
name={\ensuremath{\nu}},
description={Parameter controlling the dependence of $\epsi$ on $N$ },
sort=symbolnu, type=is
} 
\newglossaryentry{symb:Epsi}{
name={\ensuremath{E^\psi}},
description={Energy per particle of the wave function $\psi$ },
sort=symbolEnergypsi, type=is
}
\newglossaryentry{symb:Evarphi}{
name={\ensuremath{E^\varphi}},
description={Energy of the wave function $\varphi$ },
sort=symbolEnergypsi, type=is
}

\begin{document}

\selectlanguage{english}
\pagenumbering{roman}

\begin{titlepage}
    \begin{center}
        %\vspace*{3cm}
 
      % \begin{LARGE}    \textbf{Mean Field Limits in Strongly Confined Systems}  \end{LARGE} 
      \begin{huge}       \textbf{Mean Field Limits in \\ \vspace{0.3cm} Strongly Confined Systems}                           \end{huge}

        \vspace{0.5cm}
       % Thesis Subtitle
        
        \vspace{6.0cm}
        
        \textbf{DISSERTATION}
        
        der Mathematisch-Naturwissenschaftlichen Fakultät\\
	der Eberhard Karls Universität Tübingen\\
	zur Erlangung des Grades eines\\
	Doktors der Naturwissenschaften\\
	(Dr. rer. nat.)

        \vspace{5.5cm}
        
        vorgelegt von\\
        Johannes Jakob von Keler\\
	aus Schwäbisch Hall
        
        \vfill
        
       % \includegraphics[width=0.4\textwidth]{university}
        
        %Department Name\\
        %University Name\\
        %Country\\
        Tübingen 2014
        
    \end{center}
\end{titlepage}

\chapter*{Abstract}

We consider the dynamics of $N$ interacting Bosons in three dimensions which are strongly confined in one or two directions. 
We analyze the two cases where the interaction potential $w$ is rescaled by either $N^{-1}w(\cdot)$ or $a^{3\theta-1}w(a^\theta \cdot)$    
and choose the initial wavefunction to be close to a product wavefunction.   
For both scalings we prove that in the mean field limit  $N\rightarrow \infty $ the dynamics of the $N$-particle system is described by a nonlinear equation in two or one dimensions.
In the case of the scaling $N^{-1}w(\cdot)$ this equation is the Hartree equation and for the scaling $a^{3\theta-1}w(a^\theta \cdot) $ 
the nonlinear Schrödinger equation. In both cases we obtain explicit bounds for the rate of convergence of the $N$-particle dynamics to the one-particle dynamics.

\selectlanguage{ngerman}

%In dieser Arbeite werden effective 1D/2D einteilchen Bewegungsgleichungen für bosonisches Vielteilchen System in stark anisotropen Fallen hergeleitet.
\chapter*{Zusammenfassung}

In dieser Arbeit werden bosonische Vielteilchensysteme in drei Raumdimensionen untersucht, die durch ein äußeres Potential in einer bzw. zwei Raumdimensionen stark eingeschränkt sind.
Das Ziel dieser Arbeit ist es, solche $N$-Teilchensysteme durch eine effektive Einteilchengleichung zu approximieren. Im Gegensatz zu den bestehenden Arbeiten in diesem Gebiet
ist diese effektive Gleichung aufgrund des starken äußeren Potentials zwei- bzw. eindimensional. Es wird bewiesen, dass diese Approximation im thermodynamischen Limes
$N \rightarrow \infty$ exakt wird. Darüber hinaus werden für diese Approximation explizite Konvergenzgeschwindigkeiten angegeben.
Diese sind im Besonderen für die Anwendbarkeit der Ergebnisse auf physikalische
Experimente von Bedeutung. % Die Möglichkeit Fehlerabschätzung zu erhalten ist auch der wichtigste Unterschied zwischen dieser Arbeit und den Ergebnissen in den zwei kürzlich erschienenen
% Artikeln von Chen und Holmer \cite{CheHol13,CheHol14}.
Im Folgenden werden die Inhalte der jeweiligen Kapitel kurz zusammengefasst.
%\\

Kapitel 2 gibt einen Überblick über die mathematische Beschreibung bosonischer Vielteilchensysteme. Die dazu verwendete Schrödingergleichung mit Paarwechselwirkung wird eingeführt und die mathematischen
Konzepte für die Beschreibung von Bose-Einstein-Kondensation %, die eine Grundvoraussetzung für die Existenz effektiver Gleichungen für bosonische Vielteilchensysteme ist, 
werden definiert. Dabei wird erklärt, warum die Existenz eines Bose-Einstein-Kondensates essentiell für die Beschreibung boson-ischer Vielteilchensysteme durch eine effektive Einteilchengleichung ist.
Des Weiteren werden
die Mean-Field-, die Nichtlineare Schrödingergleichungs- und die Gross-Piteavski Skalierung der Vielteilchen-Schrödingergleichung anhand von physikalischen Experimenten und den bestehenden 
mathematischen Ergebnissen beschrieben.    
%\\

In Kapitel 3 wird zuerst die mathematische Notation, in der die Ergebnisse formuliert und die Beweise dargestellt werden, festgelegt. Danach werden die zwei positiven Funktionale $\alpha$ und 
$\beta$ definiert,
die von Pickl in \cite{Pic08} eingeführt wurden. Mithilfe von $\alpha $ oder $\beta$ kann die Dynamik eines Vielteilchensystems mit der Dynamik eines Einteilchensystems verglichen werden.
Dabei folgt aus der Konvergenz von $\alpha \rightarrow 0 $ oder $\beta \rightarrow 0$ im thermodynamischen Limes 
%gegen 0 
eine gute Approximation der Vielteilchendynamik durch die Einteilchendynamik. Dieses Kapitel schließt mit der Präsentation und Diskussion der Hauptresultate der Arbeit. 
Im Mean-Field-Fall sind diese im Wesentlichen von der Form 
\begin{align*}
 \alpha(t) \leq C(t) N^{-1},  
\end{align*}
wobei $C(t)$ eine monoton steigende Funktion mit $C(0)=0 $ ist. Für den Fall einer Skalierung, die zu einer nichtlinearen Schrödingergleichung führt und die durch den
Parameter $\theta $ kontrolliert wird, erhalten wir das Ergebnis
\begin{align*}
  \beta(t)  \leq C(t) N^{-\eta(\theta)}. 
\end{align*}
Hier bestimmt der Parameter $\eta(\theta)> 0$, dessen genaues Verhalten aus dem später geführten Beweis folgt, die Konvergenzgeschwindigkeit.
%Zum Schluss diese Kapitels werden diese Resultate ausführlich diskutiert. 
%\\

Kapitel 4 stellt für einen einfachen Fall der Mean-Field-Skalierung eines Vielteilchensystems einen sehr ausführlichen Beweis dar. Dieser dient zum einen dazu, die Methode von Pickl \cite{Pic08, KnoPic09,Pic11} für stark eingeschränkte Systeme zu veranschaulichen, 
wobei diese Methode in diesem Fall nur geringfügig geändert werden muss. Zum anderen liefert dieser Beweis eine Vorlage für die folgenden, technisch aufwändigeren Beweise. %\\

In Kapitel 5 werden die beiden Funktionale $\alpha$ und $\beta$ ausführlich diskutiert. Diese Diskussion ist angelehnt an \cite{Pic11,KnoPic09,PetPic14}. 
Es wird der Zusammenhang der beiden Funktionale mit dem für Mean-Field-Limiten gebräulicheren Konvergenzbegriff, der durch die Spurnorm
gegeben ist, aufgezeigt. Danach werden grundlegende Eigenschaften von $\alpha$ und $\beta$ und der in ihnen enthaltenen Projektionen $p,q$ und $P_{k,N}$ dargestellt. Diese Eigenschaften 
werden für die in Kapitel 6 und 7 folgenden Beweise 
benötigt. Zuletzt wird der Nutzen des Funktionals $\beta$ im Vergleich zu $\alpha$ thematisiert.
%\\

In Kapitel 6 wird der Beweis aus Kapitel 4 so erweitert, dass nun Paarwechselwirkungen mit stärkeren Singularitäten zugelassen werden können. 
Dazu werden im Vergleich zu Kapitel 4 zusätzliche Abschätzungen
benötigt, die mit Hilfe von Energieerhaltung hergeleitet werden können. Die dazu verwendeten Techniken werden im Detail dargestellt, da sie in den folgenden Beweisen wiederverwendet werden.
Abschließend  wird der Beweis analog zu Kapitel 4 durchgeführt. 
%\\

In Kapitel 7 wird der Fall einer Skalierung, die zu einer nichtlinearen Schrödinger-gleichung führt, bewiesen. Dabei 
wird der Fall eines stark einschränkenden Potential in zwei Richtungen betrachtet. Die Grundidee des Beweises bleibt die gleiche wie in Kapitel 4 und 6.
Es wird aber eine weitere Energieabschätzung benötigt, um die Wechselwirkung des Vielteichensystems mit der Wechselwirkung des effektiven Systems vergleichen zu können. Darüber hinaus entsteht die Schwierigkeit,
dass nun die Konvergenzgeschwindigkeit von mehreren Termen der Form $N^{f(\theta)} \epsi^{g(\theta)} $ abhängt, die miteinander in Konkurrenz stehen. Hier gibt $\epsi$ die Stärke des einschränkenden Potentials an.  
Die verschiedenen Terme der Form $N^{f(\theta)} \epsi^{g(\theta)} $ führen dazu, dass die Abschätzungen der vorigen Kapitel zusätzlich verfeinert werden müssen und nur noch bestimmte Kombinationen der beiden Parameter
 $N$ und $\epsi$ möglich sind. 

%In dem einleitenden Kapitel 1 führt in die in dieser Arbeit untersuchte Thematik eine und gibt einen.

\chapter*{Danksagung}

An dieser Stelle möchte ich mich bei allen bedanken, die mich beim Erstellen dieser Doktorarbeit unterstützt haben.
\\

Zuerst möchte ich mich bei meinem Doktorvater Stefan Teufel herzlich für die sehr menschliche und fachliche hervorragende Betreuung bedanken. 
Er verstand es, in kritischen Phasen der Doktorarbeit die entscheidenden Impulse zu geben. Darüber hinaus war es sehr bereichernd, von seinem klaren physikalisch-mathematischen Verständnis % teilhaben zu können und
 lernen zu können.
%durch sein klares physikalische-mathematisches Verständnis lernen zu können. 
Ich danke meinem Zweitbetreuer Christian Hainzl für zahlreiche wertvolle Diskussionen. 
Des Weiteren danke ich Peter Pickl für die Bereitschaft Zweitgutachter dieser Arbeit zu sein, für die Einführung in die Mean-Field-Skalierungen und für viele hilfreiche Erklärungen zur p-q Akrobatik.
\\

%Sehr fördernt für das gelingen der Arbeit war dtie freundliche und produktive Atmosphäre in der Arbeitsgruppe Mathematisch Physik  vielen gemeinsamen Unternehmungen  außerhalb der Uni   
Die freundliche, aufgeschlossene und produktive Atmosphäre in der Arbeitsgruppe Mathematische Physik trug viel zum Gelingen dieser Arbeit bei. %  sowie die weit über die Uni hinausgehen Unternehmungen danken.
Dafür Dank an Stefan, Julian, Stefan, Jonas, Sebastian, Andreas, Gerhardt, Andreas, Silvia, Wolfgang, Pascal, Tim, Carla und Frank.
\\ 
      
Meinen Eltern danke ich für den großen Rückhalt den sie für mich darstellen, und dass sie mir die Freiheit gaben, meine Wünsche entwickeln zu können und meinen eigenen Weg zu finden.
\\

Zuletzt möchte ich besonders meiner Frau Vera für die große Ausdauer, mit der sie mich immer wieder nach Fehlschlägen ermutigt hat, danken.
Die Kraft und Sicherheit, die ich aus unserem gemeinsamen Leben zusammen mit unserem Sohn schöpfen konnte, half mir die beschwerlichen Zeiten
der Promotion zu durchstehen.

%war eine entscheidende Grundlage für das gelingen dieser Arbeit. 
%Sie und unseren Sohn Ruben sind für mich ein Gegengewicht zu den physikalischen Theorien und den mathematischen Beweisen  
%Die große Bereicherung die ich durch Sie und unseren Sohn Ruben erfahren habe, ließ mich immer wieder neue Motiavation schöpfen!  

\selectlanguage{english}

%  \thispagestyle{empty} 
% \cleardoublepage 

\tableofcontents 

\clearpage
\pagenumbering{arabic}

\chapter{Introduction}

%An important aspect of physics is the approximation of 
In physics it is important to be able to approximate complex systems and general theories by effective theories or equations which are simpler to analyze and easier to solve.
Effective equations are used in every area of physics starting from the description of gases to the description of gravitation in our solar system. 
%For example the behavior of gases, fluids and even gravitation as experienced in every day live are described %nearly all physical aspects one encounters on a daily basis are described 
%by effective equations to very high accuracy.
It is impossible to obtain quantitative or even just qualitative results directly from the underlying 
microscopic or general theories without any insight on how to simplify them. %is 
%close to 
%impossible. %For the description of the 
For example, in order to describe the behavior of a gas at room temperature, one will use the thermodynamic variables pressure, temperature
and volume rather than the positions of the molecules which the gas is made of.

Mathematically the derivation of an effective equation implies proving that a solution of the effective equation is close to a solution of the equation of the complex system for suitable initial data.
%From a mathematical point of view finding effective equations amounts to proofing that the solutions of the full equations for the system are close to solutions 
%of the effective equations for suitable initial data. 
The sense in which these solutions are close %is given by a suitable norm or functional which %has to incorporate
%where the exact form depends on both the 
%description of the full system and the description of the effective system. 
 % 
depends on the respective descriptions of the system and is determined by a norm or in general by a suitable functional.  

There are many different approaches to derivation of such an effective equation.
One important mathematical approach %to proof the existence of effective equations is % such effective models rigorously is 
is to use
%One way to obtain such effective dynamics is to 
 the large number of microscopic objects -- as in the example of the gas -- %that are the basic building blocks of the system as a starting point for 
as a starting point for a statistical analysis from which one obtains effective equations.
Prominent examples of such effective equations are the Navier-Stokes and Boltzmann equations for classical systems and the Hartree and Hartree-Fock equations for quantum mechanical systems.
A different approach %to obtain effective equations
 is to identify the vastly different length scales inherent in a system and to use separation of scales to reduce the number of physically relevant degrees of freedom.
The mathematical techniques used in this context come form adiabatic theory. The most prominent example for such an effective equation is the Born-Oppenheimer approximation,
where the different masses of the nucleons and the electrons lead %leads to the introduction of two different time scales which is then exploited to find effective equations. 
to a separation of scales that can be exploited to derive effective equations.

%In this thesis we study a system where both the adiabatic and the statistical aspect are present at the same time and are used to derive a effective dynamic for the system.
%The system we describe is that of a cold Bose gas confined in a trap which is in strongly confining in one or tow dimension. 

In this thesis we study the dynamics of cold Bose gases confined in a trap that is strongly confining in one or two dimensions.
Such a system is described by $N$ interacting particles, where $N\sim 10^{3}-  10^{ 7} $ and is thus amenable to a statistical analysis.
At the same time, the strongly confining potential introduces a separation of scales.    
%In the equation describing this system both the adiabatic and the statistical
These two aspects can be combined to derive effective dynamics for the system.
This system is physically interesting
since it has become accessible by experiments in the last years \cite{GoeVogKet01,SchSal01}. From a mathematical point of view this system is of interest because 
% challenge lies in 
one has to adapt the methods used to derive 
effective equations for Bose gases %to describe the adiabatic methods and    
in a way allowing exploitation of %one is able exploit 
the adiabatic structure of the problem. % can be exploited.% at the same time. %given through the strong confinement, as well.

In the last decade there has been much progress in obtaining rigorous results for effective dynamics for cold Bose gases \cite{ErdYau01,ElgErdSchYau06,ErdSchYau07b,RodSch07,KnoPic09,Pic10,BenOliSch12} and
the references therein.
%In this papers the results 
In general, these results state that
%Theses results are generally of the form that 
the time evolution of the $N$-particle wave function $\psi_t$ can be approximated by a product $\varphi_t^{\otimes N}$, where $\varphi_t$ %is a one-particle
%wave function %$\varphi$.%, if $\psi_0 \approx \varphi_0^{\otimes N}$.
 %Here $\varphi $ 
%solves a nonlinear Schrödinger equation. 
is a solution of a nonlinear one-particle Schrödinger equation.   

In the case of an additional strong confinement one expects $\psi_t \approx \varphi_t^{\otimes N}$ still to be true. However, the particles
should be in a stationary state in the confined directions if the constraining potential is strong enough. %The heuristic idea is \textcolor{red}{??}. 
Mathematically this implies %/ can be states as that the function $\varphi $ is a product/ 
$\varphi_t$ has a product structure  $\varphi_t = \Phi_t \chi $, where $\chi$ is a time independent function in the confined directions 
and the function $\Phi_t$ is expected to solve a nonlinear Schrödinger equation in the unconfined directions. 
%&and $\chi$ is an eigenfunction in the confined direction. 
%In this thesis we show that this heuristic ideas indeed holds.

The proof of this heuristic idea has recently been given in two papers by Chen and Holmer \cite{CheHol13,CheHol14}. However, they used techniques 
that make it impossible to determine the rate of convergence of the approximation $\psi_t \approx \varphi_t^{\otimes N}$ which is particularly important for the physical interpretation. % a big draw back.
In this thesis we offer a %rigorous 
derivation of the approximation $\psi_t \approx \varphi_t^{\otimes N}$ that allows us to give explicit error bounds for the convergence rates
in terms of powers of the particle number $N$ and the confinement strength $\epsi^{-1}$ of the external potential. In the following we explain the considered problem in more detail.
%(We come back to this issue in Chapter\,\ref{chap:pp} )     
\\ 

%\subsubsection*{In Greater Detail}

The dynamics of %of a system of
a Bose gas of $N$ particles in $\R^3$ is described by the Schrödinger equation  
\begin{align}\label{equ:schroedinger}
 \im \partial_t \psi_t = H \psi_t
\end{align}
for a symmetric complex-valued wave function $\psi_t(x_1,\cdots ,x_N) \in L^2(\R^{3N}) $.
The Hamiltonian $H$ of such a system is of the form
\begin{align*}
 H:= \sum_{i=1}^N h_i+ \sum_{i < j}^N w_N(x_i-x_j),
\end{align*}
where  $w: \R^3\rightarrow \R$ is a radial symmetric pair interaction.
The subscript $N$ denotes a scaling which will be discussed in detail in Chapter 2. Each operator $h_i$ is a one-particle operator acting only on the coordinate $x_i$ defined by
\begin{align*}
 h= -\Delta_x+ \frac{1}{\epsi^2} V(\epsi^{-1} x^\perp  ).
\end{align*}    
Here the external potential $\epsi^{-2} V^\perp(\epsi^{-1} x^\perp  )$ describes the strong confinement in the direction $x^\perp$, where $(x^\parallel ,x^\perp)=x$, and the parameter  $\epsi \ll 1 $ controls
the strength of the confinement. 

The effective dynamics that %will approximate the dynamics of $\psi_t= \E^{\im H t}\psi_0$ 
we are looking for
are described by the time evolution of a one-particle wave function $\varphi_t$. 
The function $\varphi_t$ has a product structure $\varphi_t(x)= \Phi_t(x^\parallel) \chi(x^\perp) $, where $\chi$ is the eigenfunction to the smallest eigenvalue of the operator 
\begin{align}\label{equ:ineigen}
 -\Delta_{x^\perp}+ \epsi^{-2}  V(\epsi^{-1} x^\perp).
\end{align}
The function $\Phi_t$ solves 
\begin{align}\label{equ:innonlin}
 \im \partial_t \Phi_t=  (-\Delta_{x^\parallel} + w^{\Phi_t}(x^\parallel) ) \Phi_t,  
\end{align}
where $ w^{\Phi_t}$ is a nonlinear potential. The exact form of $ w^{\Phi_t}$ depends on the scaling of $w_N$ and will be explained in Chapter 2.

The goal of this thesis is to justify for suitable initial data  $\psi_0 \approx  \varphi_0^{\otimes N} $  the approximation 
\begin{align*}
  \E^{\im H t} \psi_0 \approx  \varphi_t^{\otimes N},
\end{align*}
%if $\psi_0 \approx  \varphi_0^{\otimes N} $ and 
where the components of $\varphi_t$ are solutions of \eqref{equ:ineigen} and \eqref{equ:innonlin}. Hereby one important %we an important
aspect is to obtain results for the deviation from this approximation for large but finite $N$ and small but nonzero $\epsi$.

To illustrate in which sense this approximation can be expected to hold, let us consider the case $\psi_0= \varphi_0^{\otimes N}$ and $w_N=0$. 
In this case one directly obtains $\psi_t= \varphi_t^{\otimes N}$. However, in the presence of an interaction potential this will in general be false, since the interaction will
lead to correlations between the particles. Note that although there are correlations in the wave function $\psi_t$, a symmetric $\psi$ will stay symmetric under the time evolution generated by $H$.
As a result of the correlations the statement $\psi_t \approx  \varphi_t^{\otimes N} $ can only hold as an approximation. For systems without a strongly confining potential the regime 
and the sense in which this approximation 
holds are well understood and are explained in the next chapter.
Therefore the first step of this thesis is to give precise mathematical meaning to the symbol $\approx $ for the case of a strongly confined system. For this we will use a method first introduced by Pickl in
 \cite{Pic08} which focuses on measuring how many correlations have developed and thus gives quantitative results on how much $\psi_t$ deviates from $ \varphi_t^{\otimes N}$.

\subsubsection*{Overview}
In Chapter\,\ref{chap:pp} we explain the physical models and give a summary of the mathematical results for cold Bose gases. We begin with some historical remarks and then
continue with the definition and results for Bose-Einstein condensation. This serves as a physical justification for the choice of the special initial state $\psi_0 \approx  \varphi_0^{\otimes N} $.
At the same time this motivates the mathematical models and objects considered in this thesis.
They are defined in the first part of Chapter\,\ref{chap:mathres}. In the second part of Chapter\,\ref{chap:mathres} we state our main results. In the next chapter %\,\ref{chap:thm1}
we give a short proof
for a toy model which will provide a blueprint for the more technical proofs that will follow. In Chapter\,\ref{chap:meaofconv} we introduce some notation associated with the method of Pickl.
Finally we prove the two main theorems of this thesis in Chapter\,\ref{chp:proofthm2} and\,\ref{chap:thm3}.     
\chapter{Physical Motivation and Overview of Mathematical Results} \label{chap:pp}
In this chapter we explain the physical origin of the examined equations by summarizing known mathematical results for the Bose gas and its dynamics. 
This discussion is based on the book of Lieb, Seiringer, Solovej and Yngvason \cite{LieSeiSolYng05} and we refer to this book for more details.  
\section{Historical Overview of the Study of the Bose Gas}
The analysis of the Bose gas goes back to S.N. Bose and A. Einstein. In 1924 Einstein predicted, based on a paper by Bose, 
that a homogeneous, noninteracting Bose gas at low temperature would form a new state of matter today known as Bose-Einstein condensate.
This theory was first applied to explain the properties of liquid helium, which had first been liquefied by Omnes in 1908.
However, the atoms in liquid helium are strongly interacting and it is still a mathematically open problem to prove Bose-Einstein condensation 
in a weakly interacting system let alone in a strongly interacting system.\\
The first steps to answer this question were taken by Bogoliubov in 1947 in a semi-rigorous mathematical analysis of Bose-Einstein condensation. 
In the 1950's and 1960's a renewed interest in the question gave rise to new theoretical insights. However, there were no substantial advances in the mathematical understanding of the problem. \\
Up to the beginning of the 1990's there was neither significant experimental nor theoretical nor  mathematical progress in this field. This, however, suddenly changed as 
experiments with ultracold gases became feasible and the first Bose-Einstein condensate was obtained in 1995 \cite{CorWie95, Ket95} for which Cornell, Wieman and Ketterle
received the Nobel Price in 2001. In the subsequent years this discovery had a strong impact on the physics community and a huge number of articles were published.
\\   
%
% Around the same time some progress was made in obtaining mathematical results for Bose-Einstein condensation. In the years from 95 to 05 the most results where about
% ground state energies of various models, with some condensation results as well and starting from 03 there where also results about the dynamics of condensates.
%
Since the publication of the paper \cite{LieYng98} by Lieb and Yngvason at the end of the 90's there has been steady progress in the mathematical understanding of Bose-Einstein condensation
and in closely related fields as well.   \\
Until today Bose-Einstein condensates have stayed a very active research area in the branches of experimental, theoretical and mathematical physics.

\section{The Mathematical Description of Interacting Bose Gases}

%In the following we discuss the existing mathematical results. 
In the following we discuss the mathematical description of an interacting Bose gas and its condensation. %for interacting particles.  
% This discussion is based on the book of Lieb, Seihringer, Solovej and Yngvason \cite{LieSeiSolYng05} and we refer to this book for more details. 
The starting point for the description of $N$ interacting Bosons in a large box $\Lambda \subset \R^3$ with volume $V=L^3$ is the Hamiltonian

\begin{align}\label{equ:ham}
 H_N=  \sum_{i=1}^N -\underbrace{ \frac{\hbar^2}{2m}}_{=: \mu} \Delta_i + \sum_{i \leq j} w(x_i-x_j),
\end{align}
 where $w$ is a radial symmetric interaction potential. For an ideal Bose gas we have $w=0$, so the eigenfunctions of $H_N$ are product functions. 
The system is said to be in the state of Bose-Einstein condensation if a macroscopic part of the particles has the same eigenfunction. 
For an ideal Bose gas in three dimensions Einstein proved that beyond a critical temperature $T_c$ such a behavior indeed occurs. %such a behavior is theoretically expected. 
However, in the case of nonzero $w$ we have to introduce a new notion for Bose-Einstein condensation since the eigenfunctions of $H_N$ are no longer products of single particle states.
This was first done by Penrose and Onsager in \cite{PenOns56}. 
\begin{defn*}
 A system described by a wave function $\psi \in L^2(\R^{3 N}) $ is in the state of Bose-Einstein condensation if 
\begin{align}\label{equ:bec}
 \norm{\gamma^\psi}_{\mathcal{L}(L^2(\R^3))} \geq c
\end{align}
in the limit $N\rightarrow \infty$, $L \rightarrow \infty$ with $ N/L^3 $ fixed for a $c> 0$. 

\end{defn*}
 Here the operator $\gamma^\psi$ is the one-particle density matrix associated with $\psi$ and it is defined by its kernel

\begin{align*}
\gamma^\psi(x,x') := \int \psi(x,x_2,\cdots ,x_N) \bar \psi(x,x_2,\cdots ,x_N) \D x_2 \cdots \D x_N.
\end{align*}

It turns out that proving \eqref{equ:bec} for a system with a Hamiltonian of the form \eqref{equ:ham} with a genuine interaction $w$ 
is a complicated problem and only few results exist.
 %\textcolor{red}{is a very challenging problem which is still largely unsolved.
Quoting page 5 of \cite{LieSeiSolYng05}: "In fact, BEC has, so far, never been proved for many-body Hamiltonians with genuine interactions -- except for one special case: hard core bosons on a lattice at half-filling \cite{DysLieSim78, KenLieSha88}."
The only results that exist for a general Hamiltonian of the form \eqref{equ:ham} prove that the Hamiltonian's ground state energy has in leading order the structure expected from a Bose-Einstein condensate. 
These results are obtained for gases at low density and in the thermodynamic limit. The proofs can be found in \cite{LieSeiSolYng05} and in references therein. 
%
%
%There are however results that the ground state energy per particle for Bose gases in the thermodynamic limit combined with a low density limit has the structure expected by a Bose-Einstein condensate
%see \cite{LieSeiSolYng05} and reference there in. 
Here the thermodynamic limit means to consider $N$ Bosons in a box of length $L$ and to let $N$ and $L$ tend to infinity with fixed density $\rho= N/L^3$. The low density limit is defined by  
\begin{align}
\rho^{1/3}a \ll 1,  
\end{align}
where $a$ is the scattering length of the potential $w$. Roughly speaking the scattering length captures how the interaction behaves in low-energy interaction processes. For
a detailed explanation see the appendix of \cite{LieYng01}.

\subsection{The Gross-Pitaevskii Scaling}

%Although there are only two results about Bose-Eistein condensation 
In the experimental relevant situation of trapped, dilute Bose gases, however, there
exist proofs of Bose-Einstein condensation in an asymptotic limit. In this setting the Hamiltonian of the system is complemented by the trap potential $V$
\begin{align}\label{equ:hamtrapp}
 H_N=  \sum_{i=1}^N -\mu \Delta_i +V(x_i) + \sum_{i \leq j} w(x_i-x_j).
\end{align}
In addition to the scattering length of the interaction potential, the length scale associated with the ground state energy $\hbar \omega$ of the one-particle operator $-\mu \Delta+ V  $ can be introduced.
It is standard to define the so-called oscillator length by
\begin{align*}
a_0:=\sqrt{\frac{\hbar}{m \omega}}.
\end{align*} 
%
%In addition to the scattering length of the interaction potential there exists length. we  have the length scale associated with the ground state energy $\hbar \omega$ of the one particle operator $-\mu \Delta+ V  $
%which defines the so-called oscillator length by
%\begin{align*}
% a_0:=\sqrt{\frac{\hbar}{m \omega}}.
%\end{align*}
In experiments the number of trapped particles $N$ is of order $10^{3}-10^{7}$ %and can be up to $10^{8}$
 and for a positive scattering length $a$ the ratio $a/ a_0 $ is typically of order $10^{-3}$.
Hence it is mathematically reasonable to consider, in addition to the limit $N\rightarrow \infty $, the asymptotic of $a/a_0\rightarrow 0 $ for a Hamiltonian of the form \eqref{equ:hamtrapp}. 
If we keep the potentials $V$ and $w$ fixed, this asymptotic can be implemented in two mathematically equivalent ways.
Either we set $\tilde V(x)= a_0^{-2} V(x/a_0)$ or $\tilde  w(x)= a^{-2} w( x/a)$. %The latter turns out to be more convenient for calculation and thus we will stick to it.
It is standard to use the latter and to set the scattering length of $w$ equal to $1$ so that 
\begin{align*}
 \mathrm{scat}\big( \tilde  w \big)=  a.
\end{align*}
In this limit the system is described by 
\begin{align}\label{equ:hamscaled}
 H_N= \sum_{i=1}^N  -\mu \Delta_i + V(x_i) + \sum_{i \leq j} a^{-2} w(a^{-1}( x_i-x_j))
\end{align}
for $N \rightarrow \infty$ and $a \rightarrow 0$. 
However, this asymptotic turns out to describe the behavior of a Bose-Einstein condensate only if $N a$ stays fixed. This fact can be motivated by the scaling 
properties of the  % $E^{\mathrm{GP}}(N,a)$ of the  
%It turns out that the ground state energy denoted by $E^{\mathrm{QM}}$ of \eqref{equ:hamscaled} is associated with the ground state of the 
Gross-Pitaevskii (GP) energy functional. Following from experimental evidence and
 theoretical prediction \cite{Pit61, Gro61,Gro63}, this functional should describe the ground state energy $E^{\mathrm{QM}}$ of the Hamiltonian \eqref{equ:hamscaled}. 
 The Gross-Pitaevskii energy $E^{\mathrm{GP}}(N,a)$ is defined by
\begin{align}\label{equ:gpenergy}
E^{\mathrm{GP}}(N,a):= \inf_{\varphi  }  \int \mu |\nabla \varphi(x) |^2+ V(x)|\varphi(x)|^2+ 4 \pi \mu a |\varphi(x)|^4  \D x
\end{align}
with the normalization constraint
\begin{align*}
 \int |\varphi|^2 \D x =N.
\end{align*}
The Gross-Pitaevskii functional has the following scaling property
\begin{align*}
E^{\mathrm{GP}}(N,a)= N E^{\mathrm{GP}}(1,Na).
\end{align*}
%
%Thus it is not surprising that the limit of $N$ and $a$ has to be taken such that $N a= const$ for all terms to in \eqref{equ:gpenergy}.
Since all terms of \eqref{equ:gpenergy} are expected to contribute in the limit $N\rightarrow \infty $ and $a \rightarrow 0$, the scaling property of $E^{\mathrm{GP}}$ implies $Na= \mathrm{const.}$
%Thus if we  for all the terms in \eqref{equ:gpenergy} to be of the same order we have fix $Na$.  

In the article \cite{LieSeiYng00} Lieb, Seiringer and Yngvason gave the mathematically precise relation between $E^{\mathrm{QM}}$ and $E^{\mathrm{GP}}(N,a)$. 
They prove that for $N \rightarrow \infty$ and  fixed $g=4\pi N a$ 
\begin{align}\label{equ:gpenergycon}
 \lim_{N \rightarrow \infty} \frac{1}{N} E^{\mathrm{QM}}(N,a)= E^{\mathrm{GP}}(g).
\end{align}
In the same limit Lieb and Seiringer \cite{LieSei02} proved Bose-Einstein condensation 
\begin{align}\label{equ:spurcon}
 \Tr \big |\gamma^{\psi_N}- |\varphi^{\mathrm {GP}} \rangle \langle  \varphi^{\mathrm {GP}}| \big | \stackrel{N \rightarrow \infty }{\longrightarrow }0,
\end{align}
where $\varphi^{\mathrm {GP}}$ is the minimizer of  \eqref{equ:gpenergy}.
Note that this result is stronger than \eqref{equ:bec} and implies 100\% condensation. There is a variety of mathematical ways to describe 100\% condensation. We will discuss two of them in
detail in Chapter\,\ref{chap:meaofconv} and refer to \cite{Mic07} for a more detailed presentation. %of this question.% we refer to .

% Note that the the results describes a system in the dilute limit since the condition $g$ fixed implies 
% %that the system is in a dilute limit since the mean density $\bar \rho$ for such a system is of order $N$ and dilute here means
% \begin{align*}
%  \bar \rho a^3 \ll 1
% \end{align*}
% since the mean density for trapped gases $\bar \rho$ is of order $N$.

\section{Connection of GP-Scaling with Mean Field Scaling}
%So far we discussed the properties of the ground state and the occurrence of Bose-Einstein condensation in the GP-scaling. 
%In this section we will discuss the time evolution generated by the Hamiltonians as well.
% It is natural to consider the time evolution associated Hamiltonian \eqref{equ:hamscaled} in addition to the analysis of its ground state.
% This time evolution is defined through the Schrödinger equation 

In this section 
%Before we discuss the dynamic we explain how the %in equation \eqref{equ:hamscaled} 
we discuss how the GP-scaling is connected with the mean field scaling.  %can be interpreted as a singular mean field limit.% into the mean field picture and discuss the time evolution of the considered Hamiltonians.
The mean field scaling of a system of $N$ particles is defined by
\begin{align}\label{equ:hamscaledmean}
 H_N= \sum_{i=1}^N  \Delta_i  + \frac{1}{N} \sum_{i \leq j} w(x_i-x_j).
\end{align}
%where the factor $N^{-1}$ is 
The reason for the name mean field is best explained by a heuristic argument. 
Let all $N$ particles be in the same state $\varphi$ which implies that they are all distributed like $|\varphi|^2$. Therefore the interaction potential $w$ at the point $x$ can be approximated by 
the mean contribution coming from each particle
\begin{align}\label{equ:heuristic1}
 \frac{1}{N} \sum_{j=1}^N w(x-x_j) \approx  \frac{1}{N}  \sum_{j=1}^N \int_{\R^3}  w(x-x_j) |\varphi(x_j)|^2 \D x_j \nonumber \\ 
= \int_{\R^3}  w(x-x_1) |\varphi(x_1)|^2 \D x_1
=(w*|\varphi|^2)(x).
\end{align}
Hence the interaction which one particle feels % generated by all other particles,
 can in this situation be approximated by the mean value of one particle. 

%To fit in the above mean filed picture 
Now we can explain how the GP-scaling of the Hamiltonian \eqref{equ:hamscaled} can be interpreted as a "singular mean field limit".
For mathematical convenience we neglect the trap potential, set all physical constants equal to one and set $a=N^{-1}$. %which is consistent with the conditions for \eqref{equ:gpenergycon}. 
Now the Hamiltonian \eqref{equ:hamscaled} can be rewritten  
\begin{align}\label{equ:hamscaledmeans}
 H_N= \sum_{i=1}^N  \Delta_i  + \frac{1}{N} \sum_{i \leq j} w_N(x_i-x_j),
\end{align}
where $w_N(x) :=N^{3} w(Nx)$ converges for $N \rightarrow \infty$ in the weak sense of measures to a delta function.
By formally inserting this in the calculation \eqref{equ:heuristic1} we obtain 
\begin{align*}
 (w_N*|\varphi|^2)(x) \stackrel{N \rightarrow \infty}{\longrightarrow} |\varphi(x)|^2 
\end{align*}
which is except for the wrong constant the appropriate %interaction in the Hamiltonian associated with
energy given in \eqref{equ:gpenergy}. %$E^{\mathrm{GP}}$.  

Now we introduce the parameter $\theta \in [0,1]$ to be able to describe these two scaling limits in the same framework. 
We define 
\begin{align}\label{equ:wpotscal}
 w_N^\theta(x)= N^{3 \theta} w(N^\theta x)
\end{align}
and subsequently the corresponding Hamiltonian  
\begin{align}\label{equ:hamscaledtheta}
 H_N^\theta= \sum_{i=1}^N  \Delta_i  + \frac{1}{N} \sum_{i \leq j}  w_N^\theta(x_i-x_j).
\end{align}
%This Hamiltonian covers different interaction scenarii which received a lot of attention in the mathematical physics community.
In addition to the mean field regime $\theta=0$ and the GP-scaling regime $\theta=1$ we obtain a third regime for $\theta \in (0,1)$.
These regimes are characterized by the different one-particle Hamiltonians $h$ that describe the ground state energy and the dynamics of the $N$-particle system in the limit $N \rightarrow \infty$.   
In all three regimes $h$ has the form
\begin{align}\label{equ:hamnonlin}
 h= -\Delta \varphi + w_\varphi \varphi.
\end{align}
Note that due to the nonlinearity the ground state energy associated with $h$ is defined by
\begin{align*}
E^\varphi = \inf_{\norm{\varphi}=1} \langle \varphi,(-\Delta+\frac{1}{2}w_\varphi)\varphi \rangle.
\end{align*}

The question of whether
\begin{align}\label{equ:encon}
 \lim_{N \rightarrow \infty} \frac{1}{N} E^{\mathrm{QM}}= E^{\varphi}
\end{align}
 has been answered in all three regimes.

\begin{itemize}
 \item

For $\theta=0$ the interaction $w_\varphi$ is equal to $w*|\varphi|^2$ as expected due to the heuristic argument \eqref{equ:heuristic1}.
%The connection between the ground state of $E^\varphi$ and the ground state of $E^{H_N^\theta}$ in the sense of \eqref{equ:gpenergycon} was proven 
The question posed by equation \eqref{equ:encon} was studied %in many different cases 
with various assumptions on $w$ in \cite{FanSpoVer80,BenLie83,LieYau87,Wer92,Sei11,GreSei13} 
and recently in great generality in \cite{LewNamRou13}. In the last years the question of excitations close to the 
the ground state $E^\varphi$ was considered as well. For this question we refer to \cite{LewNamSch13} and the references therein.

\item 
In the case $\theta\in (0,1)$ the nonlinearity is $w_\varphi=  |\varphi|^2 \int_{\R^3} w $. This regime is referred to as the nonlinear Schrödinger(NLS) limit. 
The question of \eqref{equ:encon} %of the characterization of the ground state 
has not been considered often in the literature but the results for the case $\theta=1 $ apply a fortiori. Recently the authors of \cite{LewNamRou14}
proved error bounds for the rate of convergence of \eqref{equ:encon} depending on the value of $\theta$.

\item
For $\theta=1$ we have $w_\varphi= 8 \pi b |\varphi|^2 $ with $b= \mathrm{scatt}(w)$  %which agrees with 
in accordance with \ref{equ:gpenergy}. %, where $b=1$.
For completeness' sake we restate the references for the proof of \eqref{equ:encon} %given in the last sections 
\cite{LieYng98,LieSeiYng00,LieSeiSolYng05} and for a review \cite{LieSeiSolYng05}.
\end{itemize}

\section{Dynamics of Bose Gases}

For experiments with Bose gases the time evolution plays an important role. Thus it is natural to consider the evolution which is generated by the Hamiltonian \eqref{equ:hamscaled} through the Schrödinger equation 
\begin{align*}
 \im \partial_t \psi_t = H_N^\theta \psi_t.
\end{align*}
One expects that for all $\theta \in [0,1]$ and under the assumption, that the initial state is a condensate, the system stays close to this condensate under the time evolution in the sense of  
\begin{align}\label{equ:spurcont}
 \Tr \big |\gamma^{\psi_N(t)}- |\varphi(t) \rangle \langle  \varphi(t)| \big | \stackrel{N \rightarrow \infty }{\longrightarrow }0.
\end{align}
Here the time evolution of $\varphi$ is generated by the appropriate form of the Hamiltonian $h$ defined in \eqref{equ:hamnonlin}. 
%in addition to the analysis of its ground state.
%This time evolution is defined through the Schrödinger equation 
These dynamics can be subdivided in the same three regimes as above. % we obtain the same three regimes for the dynamic as for the question of the ground state energy. 
\begin{itemize}
 \item For $\theta=0 $ the evolution equation is given by
\begin{align*}
 \im \partial_t \varphi= (-\Delta +w*|\varphi|^2) \varphi
\end{align*}
 and is called the Hartree equation. As in the case of the ground state energy many different people contributed to the answer of \eqref{equ:spurcont}. 
The following list makes no claim to completeness \cite{Spo80,ErdYau01, RodSch07,KnoPic09,Pic11}.

\item
For $\theta \in (0,1) $ the evolution equation is given by
\begin{align*}
 \im \partial_t \varphi= (-\Delta +\int_{\R^3} w\, \D x \; |\varphi|^2) \varphi.
\end{align*}
The study of this case is often motivated by the desire to gain insights on how to solve the case $\theta=1$. We refer to \cite{ElgErdSchYau06,ErdSchYau07,Pic08,Pic10B} for various results for these dynamics. 

\item
For $\theta =1 $ and $\mathrm{scatt}(w)=b$ the evolution equation is given by
\begin{align*}
 \im \partial_t \varphi= (-\Delta + 8 \pi b |\varphi|^2) \varphi.
\end{align*}
This problem was solved under various assumptions in \cite{ErdSchYau09,ErdSchYau09,Pic10,BenOliSch12}.
\end{itemize}

\section{Bose Gases and Strong Confinement}\label{sec:intconfinement}

In recent years it has become possible \cite{GoeVogKet01,SchSal01} to do experiments on cold, trapped Bose gases that are confined strongly in one or two directions
such that the behavior of the gas can be described by an effective equation in two or one dimension. % or one-dimensional equation.% This raises the question of how these systems should be modeled by an genuine two or one dimensional model
%or rather by a asymptotic limit of a three dimensional system. Since there is no BEC in 2D and \textcolor{red}{1D}
%unless the temperature is absolute zero confer \cite{LieSeiSolYng05} p.69, 
%it is clear that the answer is the latter.

These experiments can be described by the Hamiltonian of equation \eqref{equ:hamscaledtheta} if we add a strongly confining potential $V^\perp$
\begin{align}\label{equ:hamscaledthetaconf}
 H_N^\theta= \sum_{i=1}^N  \Delta_i + \epsi^{-2} V^\perp( \epsi^{-1} x^\perp_i )  + \frac{1}{N} \sum_{i \leq j}  w_N^{\epsi, \theta}(x_i-x_j).
\end{align}
Here the parameter $\epsi \ll 1$ describes the strength of the confinement and $x^\perp $ are the coordinates of the strongly confined direction. 
We use the notation $x_i=(x_i^\parallel,x^\perp_i)$. Note that now the scaling of the two-particle interaction $w$ depends on $\epsi $ as well. For the moment we will neglect this dependence
and explain its origin later. 

If we were to take $\epsi$ fixed, %an would not consider the asymptotics for $\epsi \rightarrow 0$
the results presented in the last two sections for the ground state, the condensation and the dynamics 
of \eqref{equ:hamscaledtheta} would hold. However, the effective theory would still be three-dimensional. This is of course not what we intend and what the experiments suggest.
Mathematically this reflects the fact that the estimates used to obtain the results of the last sections are not uniform in $\epsi$ and hence can not hold for %the asymptotic of 
$\epsi \rightarrow 0$.

From a physical point of view the most interesting case of \eqref{equ:hamscaledthetaconf} is $\theta=1$. However, to prove the existence
of an effective equation for the dynamics generated by \eqref{equ:hamscaledthetaconf}, 
%in this case, 
in the case $\theta=1$, 
is a challenging problem which is still open. Thus we will first discuss the
relatively simple case $\theta=0$ and then come to the case $\theta > 0$.  
 
\subsection{Strong Confinement for the Mean Field Scaling }

In the mean field regime $\theta=0$ the Hamiltonian \eqref{equ:hamscaledthetaconf} is given by
\begin{align}\label{equ:hamscaledconfmean}
 H_N= \sum_{i=1}^N-  \Delta_i+ \epsi^{-2} V (\epsi^{-1} x^\perp_i  )  + \frac{1}{N} \sum_{i \leq j} w(x_i-x_j),
\end{align}     
where there is no dependence of the two-particle interaction $w$ on $\epsi$ in this regime.

The analysis of the dynamics generated by \eqref{equ:hamscaledconfmean} for different classes of interactions $w$ and their approximation by effective dynamics is the first part of this thesis. 
We will measure the errors of this approximation with a functional defined by Pickl which is equivalent to using the norm \eqref{equ:spurcont}. 
To the knowledge of the author this problem has not been considered before.

% A work related to this subject is \cite{AbdMehPin05} for the interaction potential $w=\frac{1}{|x|}$. They show that the Hartree equation %which is in this context also called the Schrödinger Poisson equation
% with a strong confining potential is described by a $2D/1D$ Hartree equation. Phrased in the setting of this work this amounts to taking
% the limited $N \rightarrow \infty $ first and afterwards do an asymptotic expansion in $\epsi$. This does not describe the physical situation explained above, where the asymptotics of $N$ and $\epsi$ have to be considered simultaneous.

The results of this thesis for the mean field scaling are phrased for the Hamiltonian
\begin{align}\label{equ:hamscaledconfmeanttran}
 H_N= \sum_{i=1}^N  -\Delta_{x^\parallel_i}% +V(t,x^\parallel_i) 
+ \epsi^{-2}(-\Delta_{\tilde x^\perp_i}+  V(\tilde x^\perp_i ))  + \frac{1}{N} \sum_{i \leq j} w(x^\parallel_i-x^\parallel_j ,
 \epsi(\tilde x^\perp_i-\tilde x^\perp_j) )
\end{align}     
which originates from \eqref{equ:hamscaledconfmean} by a coordinate transformation $\tilde x^\perp= \epsi x^\perp$.
This is done since the analysis of \eqref{equ:hamscaledconfmeanttran} is mathematically more convenient than \eqref{equ:hamscaledconfmean}.     
%The reason for this is that for the mathematical analysis it is
%In this thesis we consider the cases where the 
For our analysis we make the assumption that the confining potential $V^\perp$ is a hard wall
potential outside a bounded set $ x^\perp \in \Omega$ i.e.
\begin{align}\label{equ:hardwalls}
 V^\perp(x)=\infty \quad \forall x \in \Omega^c, 
\end{align}
where $\Omega^c$ is the complement of $\Omega$ in the direction of the confinement.
%This results into the wave functions being defined on $\R   $   
 This is only a technical assumption to avoid 
the use of additional energy estimates for the strongly excited modes in the confined direction. % From a physical standpoint this is no limitation since the experiment takes place in a well delimited environment.  

We obtain our results with the help of a method developed by Pickl in \cite{Pic08,Pic10,KnoPic09,Pic11}. These results are phrased for two functionals $\alpha$ and $\beta$ that are explained in detail in Chapter\,\ref{chap:meaofconv}.
Translated to the trace norm setting of \eqref{equ:spurcont} our results are
\begin{align}\label{equ:convhour}
 \Tr\big|\gamma^\psi(t)- |\varphi(t) \rangle \langle \varphi(t)| \big| \leq C(t)N^{-\eta},
\end{align}
where $\eta=1/2$ if the interaction $w$ has at most $L^2$-singularities and $\eta=\frac{5s-6}{8s}$ for interactions  $ w \in L^s$ with $s \in (6/5,2)$. %The function $C(t)$ is monotone increasing in $t$ with $C(0)=0$.

A paper related to this subject is \cite{AbdMehPin05} in which the interaction potential $w=\frac{1}{|x|}$. The authors show that the Hartree equation %which is in this context also called the Schrödinger Poisson equation
with a strong confining potential is described by a $2D/1D$ Hartree equation. Phrased in the setting of this work this amounts to taking
the limit $N \rightarrow \infty $ first and afterwards doing an asymptotic expansion in $\epsi$. This does not describe the physical situation explained above, where the asymptotics of $N$ and $\epsi$ must be considered simultaneously.

\subsection{Strong Confinement for NLS-scaling and GP-scaling}

Now we come to the case $\theta >0 $.
Here one must be careful in defining a sensible equivalent to $\eqref{equ:hamscaledtheta}$ in the presence of a strongly confining potential, since now we %have to
consider 
two asymptotic limits at the same time. The first one comes from the strongly confining potential which is expressed by the parameter $\epsi$ and the second one from the derivation 
of the GP-scaling which was defined with the help of the parameter $a$ \eqref{equ:hamscaled}. 
To be able to identify the appropriate scaling we write the Hamiltonian with the parameters $\epsi,a,\theta$ in the way they were introduced in
 \eqref{equ:hamscaled},\eqref{equ:wpotscal} and \eqref{equ:hamscaledthetaconf}        
% However since the GP-scaling is already by itself a asymptotic approximation one has to be careful how to define the correct Hamiltonian.
% We take the Hamiltonian \eqref{equ:hamscaled} where we and set the $\mu=1$ and write the confining potential in two parts: the strongly confining part $V^\perp$ and a just confining part $V$.
% The Hamiltonian then reads  
%
\begin{align}\label{equ:hamscaledconf}
 H_N^{\theta,\epsi}= \sum_{i=1}^N - \Delta_i +  \epsi^{-2} V^\perp( \epsi^{-1} x^\perp_i )  +  \sum_{i \leq j} a^{1-3 \theta}  w(a^{-\theta}(x_i-x_j)).
\end{align}
%where $\epsi \ll 1$ describes the strength of the confinement, $x^\perp $ are the coordinates of the strongly confined direction and $x^\parallel$ are the coordinates of the confined direction.
%If we would fixed a value of $\epsi$ the ground state energy and the state of condensation could be described by equation \eqref{equ:gpenergycon} and \eqref{equ:spurcon}.
%However the bounds in the proofs for these two equations are not uniform with $\epsi$, hence one has to take a closer look at this problem.
To determine a sensible scaling behavior of this Hamiltonian we use the existing results for its ground state energy.  
%For the ground state energy this was
In the case $\theta=1 $ and strong confinement in one direction this result was obtained by Schnee and Yngvason in \cite{SchYng07} and for the case of a strong confinement in two directions
 by Lieb, Seiringer and Yngvason in \cite{LieSeiYng03}.
They showed that the Gross-Pitaevskii regime is given by $N \rightarrow \infty$ and $a,\epsi \rightarrow 0$ with $ Na/\epsi$ fixed in the former case and $ Na/\epsi^2$ fixed in the latter case.
In this regime they both proved 
\begin{align}
\lim  \frac{E^{\mathrm{QM}}(N,h,a)-N \epsi^{-2} E^\perp}{E^{\mathrm{GP}}_{\mathrm{ 2D/1D}}(N,g)} = 1,
\end{align}
where $E^\perp$ is the ground state energy of the operator $\Delta_{x^\perp}+ V^\perp(x^\perp)$.
The parameter g is a modified coupling parameter defined by
\begin{align*}
 g_{\mathrm{2D}}= \int \chi(x^\perp)^4 \D x^\perp a/\epsi  \qquad  g_{\mathrm{1D}}= \int \chi(x^\perp)^4 \D x^\perp a/\epsi^2,
\end{align*} 
where  $ \chi(x^\perp)$ is the  eigenfunction associated with $E^\perp$. %  the ground state function of the operator $\Delta_{x^\perp}+ V^\perp(x^\perp)$.

% The question of condensation was also addressed in the articles \cite{SchYng07} and \cite{LieSeiYng03}. In the $2D$ case
% condensation was proven in the sense of \eqref{equ:spurcon} to a function $\varphi= \Phi \chi $ where $\Phi $ is the minimizer of the to $E^{\mathrm{GP}}_{\mathrm{ 2D}}$
% corresponding Gross-Pitaevskii energy functional. In the $1D$ case the convergence of the density, which is the diagonal part of the $\gamma^\psi$, to $|\varphi|^2$ was proven in a weak 
% $L^1$ sense.

Following from %\cite{LieSeiYng03}
the above, the appropriate Hamiltonian which must be considered in the case of strong confinement in two directions is   
\begin{align*}
 H_N^\theta= \sum_{i=1}^N  \Delta_i + \epsi^{-2} V^\perp(x^\perp_i \epsi^{-1} )  + \frac{\epsi^2}{N} \sum_{i \leq j} (N \epsi^{-2})^{3 \theta}  w\big( (N\epsi^{-2})^\theta (x_i-x_j)\big),
\end{align*}
where we set $a= N^{-1} \epsi^2$ for mathematical convenience.

The study of the dynamics generated by this Hamiltonian is the second part of this thesis.  
%The second part of this thesis is to study the dynamics generated by this Hamiltonian.
Our results are again, as in the case $\theta=0$, phrased for a rescaled Hamiltonian,   %where we will work again as in the   derived from \eqref{equ:hamscaledconf} 
where we set $\tilde  x^\perp =\epsi x^\perp$ and thus obtain
 \begin{align*}%\label{equ:hamscaledconftheta}
 H_N= \sum_{i=1}^N-\Delta_{x^\parallel_i} %+V(t,x^\parallel_i) 
&+ \epsi^{-2}(-\Delta_{\tilde x^\perp_i}+  V(\tilde x^\perp_i )) \\  
& +   \frac{\epsi^2}{N}  \sum_{i \leq j}(N \epsi^{-2})^{3 \theta}   w\Big ( (N \epsi^{-2})^{\theta}\big (x^\parallel_i-x^\parallel_j,
 \epsi(\tilde x^\perp_i-\tilde x^\perp_j) \big) \Big).
\end{align*}
As before we assume the condition \eqref{equ:hardwalls} for the confining potential $V^\perp$ in our proofs.   

As in the Mean Field case we use the method of Pickl to obtain our results. Translated to the trace norm our results imply
\begin{align}\label{equ:convernlspur}
 \Tr\big|\gamma^\psi(t)- |\varphi(t) \rangle \langle \varphi(t)| \big| \leq C(t) N^{-\eta},
\end{align}
where 
    \begin{align*}
		\eta(\theta) = \begin{cases}
		\frac{4\theta-1}{6-8\theta} \qquad &\mathrm{for}\; \theta \in (\frac{1}{4}, \frac{7}{24}]\\
		\frac{1-3\theta}{8-18\theta}  \qquad &\mathrm{for}\; \theta \in (\frac{7}{24} ,\frac{1}{3}).
		\end{cases}
	  \end{align*}
The rate of convergence is at best of order $1/20$. However, it should be possible to improve this rate
 by combining ideas introduced in this work with methods used in \cite{Pic10}. 
For the proof of equation \eqref{equ:convernlspur} we assume the interaction potential $w$ to be an element of $  L^\infty$  with compact support. 
%The potentials $w$ we can allow for \eqref{equ:convernlspur} to hold are $w\in L^\infty$ with compact support.

As already mentioned, very recently the same problem was considered by Chen and Holmer in the case of confinement in one direction \cite{CheHol13} and for a confinement in two 
directions \cite{CheHol14}. In these articles the authors used the techniques of the BBGKY hierarchy to derive their results. For $0 <\theta < c $ and under the assumption that the interaction potential $w$ is a Schwartz function they showed
\begin{align}\label{equ:converconf}
 \Tr\big|\gamma^\psi(t)- |\varphi(t) \rangle \langle \varphi(t)| \big| \stackrel{N,\epsi}{ \longrightarrow} 0,
\end{align}
where $c=2/5$ for confinement in one direction and $c=3/7$ for the confinement in two directions.
However, a disadvantage of using the BBGKY hierarchy is that it only provides convergence of the left hand-side of \eqref{equ:converconf} but no rate of convergence. 

\chapter{Mathematical Results} \label{chap:mathres}

\section{A Concise Definition of the Mathematical Model} 
As motivated in the last section we state the mathematical description of the model analyzed in this thesis. 
The \gls{symb:N}-particle system is described by a wave function $\gls{symb:psiNepsi} \in \gls{symb:H^N} $. Here 
\begin{align*}
 \mathscr{H}^N:= L^2_+(\Omega^{N},\D r_1 \cdots \D r_N)
\end{align*}
is the subspace of $L^2(\Omega^{N},\D r_1 \cdots \D r_N) $ consisting of wave functions $\psi_N(r_1, \dots, r_N)$ which are symmetric under permutation of their arguments $\gls{symb:r}_1, \dots ,r_N \in \gls{symb:omega}$. 
The parameter $ \gls{symb:epsi} \ll 1 $ controls the strength of the confinement and the set $\Omega \subset \R^3 $ encodes the shape of the confinement.

We consider the two %equally interesting
cases of confinement in one and two directions. In the former case $\Omega:= \R^2 \times [c,d]$ with $c,d \in \R$, $c < d$ and  $0 \in (c,d) $. In the latter
 case $\Omega:= \R \times  \Omega'$ with $  \Omega' $ a compact subset of $\R^2$ with $0 \in \mathring{ \Omega}'$ and smooth boundary $\partial \Omega'$. %with $0 $ contained in the interior of ${\tilde \Omega} $.
To be able to treat both cases at the same time we introduce 
the notation $\Omega= \Omega_\mathrm{f} \times \Omega_{\mathrm{c}} $ and $r=(x,y)$, where $\gls{symb:y} \in \gls{symb:omegac} $ are the coordinates of the "confined"
 direction and $\gls{symb:x} \in \gls{symb:omegaf} $ are the coordinates of the "free" direction.

% 
% If the confining potential acts only in one dimension
% then $\Omega:= \R^2 \times [c,d]$ with $c,d \in \R$, $c < d$
% and $0 \in (c,d) $. In the case in which the confining potential acts in two dimensions $\Omega:= \R \times \tilde  \Omega$ with $ \tilde \Omega_1 $ 
% a compact subset of $\R^2$ with $0 \in \mathring{\tilde \Omega}_1 $.
% In both cases we write $(\gls{symb:x}_i,\gls{symb:y}_i)=r_i \in \Omega$ where $x_i$ is an element of what we call "free" direction which we denote it by $\gls{symb:omegaf}$. In the former case this is $\R^2$ and $\R$ in the latter case.
%  Correspondingly $y_i$ is an element of the "confined" direction denoted by $\gls{symb:omegac}$ which is the complement in $\Omega$ of the "free" direction. 
%the confined direction  $[c,d]$ in the former case and $\Omega_1$ in the latter case.       

The equation which governs the behavior of $\psi_N^\epsi$ is the $N$-particle Schrödinger equation

\begin{align}\label{equ:schhartree}
 \im \partial_t \psi_N^\epsi(t)=\gls{symb:HNepsi} \psi_N^\epsi(t) \qquad \Psi_N^\epsi(0)= \Psi_{N,0}^\epsi,
\end{align}
where the Hamiltonian has the form 
\begin{align*}
  H_N^{\epsi}= \sum_{i=1}^N h^\epsi_i+ \sum_{i \leq j}^N W^{\epsi, \theta,N}(r_i-r_j).
\end{align*} 
Here $\gls{symb:h}^\epsi_i$ is a one-particle Hamiltonian $h^\epsi$ acting on the coordinate $r_i$ defined by
\begin{align*}
 h^\epsi= -\Delta_x- \frac{1}{\epsi^2} \Delta_y+ \gls{symb:V}(t,x,\epsi y),
\end{align*}
where $V$ is a time dependent external potential, $\Delta_x $ is the Laplacian on $\Omega_{\mathrm{f}}$ and $\Delta_y$ is the Dirichlet Laplacian on $\Omega_{\mathrm{c}}$.  
The parameter $\gls{symb:theta} \in [0,1] $ controls the range of the pair interaction $\gls{symb:WetN}(r_i-r_j)$ which consists of a spherical symmetric function $\gls{symb:w}: \R^3 \rightarrow  \R$ 
combined with a scaling depending on the parameters. % defined by the value of $\theta$. 
In the case $\theta=0$ the interaction is scaled as

\begin{align}\label{equ:pothartree}
  W^{\epsi, 0, N}(r_i-r_j):= \frac{1}{N} w \big((x_i-x_j),\epsi(y_i-y_j) \big ).
\end{align}
%which we write as $ w^\epsi(r_i-r_j) $ for shorter notation.
%
In the case $\theta \in (0,1]$ we have
 \begin{align}\label{equ:pothgp}
W^{\epsi, \theta, N}(r_i-r_j):= \gls{symb:a}^{1-3 \theta} w\Big( a^{-\theta} \big(x_i-x_j,\epsi(y_i-y_j)\big)\Big).%=b^{1-3 \theta}  w^\epsi(b^{-\theta} (r_i-r_j))
\end{align} 
The value of $a$ depends on the number of the confined directions. For a confinement in one direction $a=\epsi N^{-1}$ and in the case of confinement in two directions $a=\epsi^2 N^{-1}$. 

We denote the one-particle wave function that will approximate $\Psi_N^\epsi$ by $\gls{symb:varphi} \in L^2(\Omega) $. It has always a product structure and consists of 
the two functions $\gls{symb:chi}(y)$ and $\gls{symb:Phi}(x)$. For all values of $\theta$ the function $\chi $ %lives in the confined direction and
is an eigenfunction of the $\epsi$-dependent Dirichlet Laplacian $\epsi^{-2}\Delta_y $ on $\Omega_\mathrm{c}$. The function 
$\Phi(x)$ lives on $\Omega_\mathrm{f}$ and is a solution of a nonlinear equation. One expects $\Phi(x)$ to solve the Hartree equation for $\theta=0$, for $\theta \in (0,1)$ the nonlinear Schrödinger equation (NLS)
and for $\theta = 1$ the Gross-Pitaevskii (GP) equation. Here the use of the names NLS and GP has physical and historical reasons since the only difference between
both equations is the value of the constant in front of the nonlinearity.         
 
\subsection{A Concise Definition of the Functional Comparing $\psi_N^\epsi$ with $\varphi$ }%Counting Functional }
%Functional measuring the distance of $\psi_N^\epsi$ and $\varphi$
The functional we will use to determine convergence of $\Psi_N^\epsi$ to $\varphi$ was introduced by Pickl in \cite{KnoPic09,Pic11}. %we gave a glimpse of the heuristic idea in \ref{} and 
We give a thoroughly account of them in Chapter \ref{chap:meaofconv}.  
Here we limit ourselves to the mathematical definitions. We will use two different functionals denoted by $\alpha$ and $\beta$. The functional $\alpha $ is given by
\begin{align}\label{equ:alpha1}
 \gls{symb:alpha}\big(\varphi(t),\psi_N^\epsi(t)\big):= 1- \langle \varphi(t), \gls{symb:gamma}^{\psi_N^\epsi(t)} \varphi(t) \rangle_{L^2(\Omega)},
\end{align}
where $ \gamma^{\psi_N^\epsi(t)} $ is the one-particle density matrix of $\psi_N^\epsi(t)$.
To introduce $\beta$ we first define the projection operators 
\begin{align}\label{equ:pq}
 \gls{symb:p}_i(t):=  \varphi(t,r_i) \langle \varphi(t,r_i), \cdot \rangle_{L^2(\Omega, \D r_i)} \qquad \gls{symb:q}_i(t):=\id-p_i(t)  
\end{align}
and
\begin{align*}
 \gls{symb:PkN}(t):=(q_1(t) \cdots q_k(t) p_{k+1}(t) \cdots p_N(t))_\mathrm{sym}.
\end{align*}
Now we can define
\begin{align*}
 \gls{symb:beta}\big(\varphi(t),\psi_N^\epsi(t)\big):=  \sum_{k=0}^N \sqrt \frac{k}{N}  \langle \psi^{\epsi}_N(t),P_{k,N}(t)  \psi^{\epsi}_N(t) \rangle_{L^2(\Omega^{N})} 
\end{align*}
which can be viewed as a generalization of $ \alpha^{\epsi,N}(t)$ since written with the projections $P_{k,N}(t)$
\begin{align*}
%  \alpha^{\epsi,N}(t)
\alpha \big(\varphi(t),\psi_N^\epsi(t)\big) =  \langle  \psi^{\epsi}_N(t), q_1(t)  \psi^{\epsi}_N(t) \rangle_{L^2(\Omega^{N})}  = \sum_{k=0}^N  \frac{k}{N}  \langle \psi^{\epsi}_N(t),P_{k,N}(t)  \psi^{\epsi}_N(t) \rangle_{L^2(\Omega^{N})}  .
\end{align*}
% In section \ref{} we also discuss the advantages and disadvantages of $\alpha^{\epsi,N}(t)$ and $ \beta^{\epsi,N}(t)$  and their relationship with 
% the more often used measure
For both $\alpha$ and $\beta$ we define the shorthands
\begin{align*}
 \alpha \big(\varphi(t),\psi_N^\epsi(t)\big):=  \alpha^{\epsi,N}(t) \qquad \beta \big(\varphi(t),\psi_N^\epsi(t)\big):=  \beta^{\epsi,N}(t).
\end{align*}
In Section \ref{sec:relalphatr} we discuss the relationship between $\alpha^{\epsi,N}(t)$ and %the measure  
%how they are related to the more often used 
\begin{align*}
 \Tr \big| \gamma^{\psi_N^\epsi(t)} - |\varphi^\epsi(t) \rangle \langle \varphi^\epsi(t)| \big | 
\end{align*}
through the inequality
\begin{align}\label{equ:tralpha}
 \Tr \big| \gamma^{\psi_N^\epsi(t)} - |\varphi^\epsi(t) \rangle \langle \varphi^\epsi(t)| \big | \leq \sqrt{ 8 \alpha^{\epsi,N}(t) }.
\end{align}
This inequality holds for $ \beta^{\epsi,N}(t) $ as well since $\alpha^{\epsi,N}(t) \leq \beta^{\epsi,N}(t) $.  

\section{Main Results}
\subsection{The Hartree Case: $\theta=0$}

In the case $\theta=0$ the Hamiltonian which governs $\psi_N^\epsi$ takes the form %solves the Schrödinger equation with the Hamiltonian  
\begin{align}\label{equ:hhartree}
  H_N^{\epsi}= \sum_{i=1}^N h^\epsi_i + \frac{1}{N} \sum_{i \leq j}^N w^\epsi(r_i-r_j),
\end{align}
where 
\begin{align}\label{equ:wepsi}
w^\epsi(r_i-r_j):=w \big((x_i-x_j),\epsi(y_i-y_j) \big  ).
\end{align}
For the ease of the presentation we work without an external potential in the one-particle Hamiltonian
\begin{align*}
 h^\epsi= -\Delta_x- \frac{1}{\epsi^2} \Delta_y.
\end{align*}
The nonlinear Hartree equation that governs $\Phi(t)$ is   
\begin{align*}
 \im \partial_t \Phi(t) = (-\Delta_x + \gls{symb:w0}* |\Phi(t)|^2) \Phi(t) \qquad \Phi(0)=\Phi_0,
\end{align*}
where $\gls{symb:w0}$ will be defined in the assumptions below. For an interaction with enough regularity it can be defined by the restriction of $w$ on $(\Omega_\mathrm{f} \times 0)$.

% For $m \in \N$ we denote the in $L^2$ normed eigenfunctions of the Dirichlet Laplacian $-\epsi^2 \Delta_y $ on $\Omega_c$ by the functions $\chi_m$ with
% ascending order Eigenvalues $E^\epsi_m$. The time evolved $\chi_m(t)$ are given by
% \begin{align*}
%  \chi_m(t)=  \exp({-\im E_m^\epsi t } )\chi_m
% \end{align*}

Let the set $\{\chi_m\}_{m=0}^\infty$ be an orthonormal basis of $L^2(\Omega_\mathrm{c})$ such that for all $m$ $\chi_m$ is an eigenfunction of the Dirichlet Laplacian $\epsi^{-2}\Delta$ on $\Omega_\mathrm{c}$.
Furthermore let the corresponding eigenvalues $E_m^\epsi$ fulfill
\begin{align*}
 0 < E^\epsi_0 < E^\epsi_1 \leq E^\epsi_2 \leq \cdots.
\end{align*}
The eigenvalues $E_m^\epsi$ satisfy the relation $E_m^\epsi= \epsi^{-2}E_m $, where $E_m$ are the eigenvalues of $\Delta$ on $\Omega_\mathrm{c}$. 
We define the one-particle function $\varphi$ by
\begin{align*}
\varphi(t):= \Phi(t) \chi,
\end{align*}
where $\chi= \chi_m$ for a $m \in \{0,1,2,\dots \}$.
We define the set $\tilde \Omega:= \Omega_\mathrm f \times \tilde \Omega_\mathrm c$, where $\tilde \Omega_\mathrm c:=\{ y \,| \, \exists y_1,y_2  \in \Omega_\mathrm c : y=y_1-y_2  \} $.
This set is introduced since we will have to control the norm of the interaction $w$ on $L^p(\tilde \Omega)$.
%The definition is not surprising since in the equations the interaction enters with $w(r_i-r_j)$ for $r_i,r_j \in \Omega$.  

Now we state the assumptions on the interaction potential $w$.

\begin{description}
  \item[A1]\label{ass:A1} Let $w=\gls{symb:ws}+\gls{symb:winfty}$ such that %  \begin{itemize}
 %   \item 
	    for all $\epsi \in (0,1]$ there exists a $C \in \R^+$ such that
	    \begin{align*}
	      \norm{w_s^\epsi}_{L^2(\tilde \Omega)} \leq C \quad \norm{w_\infty^\epsi}_{L^\infty(\tilde \Omega) }\leq C.
	    \end{align*}
  %  \item 
	    There exists $w^0_s,w^0_\infty :\Omega_f \rightarrow \R $ and a function $\gls{symb:fepsi}:(0,1]\rightarrow \R^+ $ with $f(\epsi) \stackrel{\epsi \rightarrow 0}{\rightarrow} 0$ such that
	    \begin{align*}
	      \norm{w_s^\epsi- w_s^0}_{L^1(\tilde \Omega)} \leq f(\epsi) \quad \norm{w^\epsi_\infty-w^0_\infty}_{L^\infty(\tilde \Omega) } \leq f(\epsi),
	    \end{align*}
	    where  $w^0_s(x,y):=w^0_s(x)$, $w^0_\infty(x,y):=w^0_\infty(x)$ for $(x,y)\in \tilde \Omega$ and let $w^0_s \in  L^1(\Omega_f) $, $w^0_\infty \in  L^\infty(\Omega_f)$. For short notation we define 
	    \begin{align*}
	     w^0:=w^0_s+w^0_\infty.
	    \end{align*}

%   \end{itemize}
%   \item[A2] Let $ A \in C^2(\R^4,\R ) $ and let there be a $C \in \R$  
% such that 
% \begin{align*}
%  \forall t,x,y  \quad |\partial_t A(t,x,0) |< C \qquad |\partial_t \partial_y A(t,x,y) |< C
% \end{align*}
\end{description}

\begin{thm}\label{thm:thm1}Let the assumption A1 hold, $t\in [0,\infty)$,  $ \psi^\epsi_N(0) \in \gls{symb:DHepsiN} $ with $\norm{\psi^\epsi_N(0)}_{L^2(\Omega^N)} =1$, 
$\Phi_0 \in H^2(\Omega_f)$ with $\norm{\Phi_0}_{L^2(\Omega_f)}=1$. Then %and $\norm{\chi(0)}=1$ then

\begin{align}\label{equ:thm1}
\alpha^{\epsi,N}(t) \leq   \alpha^{\epsi,N}(0) \exp(C(t)) + (f(\epsi)+ \frac{1}{N}) (\exp(C(t))-1),
\end{align}
where
\begin{align*}
 C(t):=4 \big(\norm{ w^0_s}_{L^1(\Omega_\mathrm{f})} +\norm{ w^0_\infty}_{L^\infty(\Omega_\mathrm{f})}+  \norm{w^\epsi_s}_{L^2(\tilde \Omega)} +\norm{w^\epsi_\infty}_{L^\infty(\tilde \Omega)} \big)  \\ 
\times  \int_0^t (1+ \norm{\varphi(s)}_{L^\infty(\Omega)}+ \norm{\Phi(s)}_{L^\infty(\Omega_\mathrm f)})^2  \D s.
\end{align*}
\end{thm}

\begin{rem}
 \begin{enumerate}
  \item The inequality \eqref{equ:tralpha} together with \eqref{equ:thm1} implies for the one-particle density matrix of $\psi_N^\epsi$ the bound
	\begin{align*}
	 \Tr | \gamma^{\psi_N^\epsi(t)} - p(t)| \leq \sqrt{8} \exp(\frac{1}{2} C(t))  \bigg( \big( \Tr | \gamma^{\psi_N^\epsi(0)} - p(0)| \big)^{\frac{1}{2}} + f(\epsi)^{\frac{1}{2}} + N^{-\frac{1}{2}} \bigg).
	\end{align*}

  \item The appearance of $\alpha(0)$ in equation \eqref{equ:thm1} is not surprising. If the functional $\alpha$ is large for the initial states $\psi(0)$ and $\varphi(0)$  
	we can not expect $\alpha$ to be small for later times. From a mathematical standpoint we can take any sequence $ \psi_N^\epsi(0)$ such that 
	$\alpha(\psi_N^\epsi(0),\varphi(0)) \stackrel{N,\epsi}{ \longrightarrow } 0 $ as an initial condition e.g. $\psi^\epsi_N(0)= \varphi(0)^{\otimes N } $. 
	From a physical standpoint one should take the state $\psi_N^\epsi(0)$ to be the minimizer of the energy, where one adds a suitable trap potential in the $x$-direction to the Hamiltonian.
	The question of $\alpha(\psi_N^\epsi(0),\varphi(0)) \stackrel{N,\epsi}{ \longrightarrow } 0 $ for this state is exactly the question of condensation discussed in Chapter\,\ref{chap:pp}. 
	Without a strongly confining potential this convergence is well known
	%\textcolor{red}{ see for example }
	cf. \cite{LewNamRou13} and references therein. The same question with a strongly confining potential is
      % there are 
	to the authors knowledge still open. However, there is no reason to believe $\alpha(\psi_N^\epsi(0),\varphi(0)) \stackrel{N,\epsi}{ \longrightarrow } 0 $ should not hold for the ground state.

	%Although there are to the authors knowledge no results of condensation for this system 
	%the theorem is meaningful for the initial condition $\psi(0)=\varphi^{\otimes N}$.
	% for the convergence of 
	%\begin{align*}
	% \Tr | \gamma^{\psi_N^\epsi} - p |
	%\end{align*}
	%for this system 

  \item If we disregard the convergence rate of $\alpha(0)$ to $0$, equation \eqref{equ:thm1} does not put any constraint on $\epsi$ and $N$.
        % shows that the rate of convergence depends on $\epsi$ and $N $
	%, where $\epsi$ and $ N$ are independent of each other. % 
	%\textcolor{red}{independently of each other}. 
	%The actual rate of convergence in \eqref{equ:thm1} depends on how fast the potential convergence with $\epsi$.
	Hence, regardless of the convergence rate of $w^\epsi \rightarrow w^0$, $\epsi$ can be chosen as a function of $N$ such that the rate of convergence is $N^{-1}$ in \eqref{equ:thm1}.

  \item In addition to the hard wall confinement we can add $\epsi^{-2} V^\perp(y)$ for any bounded potential $V^\perp$ in the $N$-particle Hamiltonian. The only difference in this situation is that then $\chi$ is an eigenfunction of 
        the operator $\epsi^{-2}(-\Delta_y+V^\perp(y))$ on $\Omega_\mathrm{c}$. 
 
  \item  Beeing able to allow exited states in the confined direction seems quite unphysical since one expects the excited states in the confined direction to decay under the time evolution due to the high
	 energy. This seems to be an artifact of this toy model together with the condition A1 on the interaction potential. Since this artifact vanishes if we relax the condition A1 for the next theorem. 
    
  \item Other than in the indirect way in condition A1 the dimension of the confinement does not play any roll in the theorem. 
	For example in the case of a confinement in one direction the potential $w=|r|^{-q}$ with $q<1$ fulfills A1. In the case of a confinement in two directions the potential $w=|r|^{-q}$ with $q< 1/2$ fulfills A1.  
% 
%   \item If we combine the estimate form Lemma\,\ref{lem:3termeg}.\ref{lem:3.4g} with Theorem \ref{} we can allow external potentials $V$ which obey the condition:\\    
% 	$ A \in C^2(\R^4,\R ) $ and let there be a $C \in \R$ such that 
% 	\begin{align*}
% 	  \forall t,x,y  \quad |\partial_t A(t,x,0) |< C \qquad |\partial_t \partial_y A(t,x,y) |< C
% 	\end{align*}

  \item The set $\tilde \Omega$ is only essential for the convergence of $w^\epsi$ to $w^0$. To assume $\norm{w_s^\epsi}_{L^2(\tilde \Omega)} \leq C$ is due to the rotational symmetry equivalent to
	 $\norm{w_s^\epsi}_{L^2( \Omega)} \leq C$.
	%, but nevertheless  
	%using only the norm on $\tilde \Omega $ rather than on $\R^3$ keeps the constant $C(t)$ as small as possible.   

  \item The boundedness of $\norm{\varphi(t)}_{L^\infty(\Omega)}$ and $\norm{\Phi(t)}_{L^\infty(\Omega_\mathrm f)}$ follows from the condition on $\varphi(0)$. This is well known and is
	discussed in Appendix\,\ref{app:regsol}. 

\end{enumerate}

\end{rem}

To be able to formulate a theorem similar to the last one, however with weaker assumptions on the interaction potential, we introduce the one-particle energies  $E^{\psi}(t)$ and $ E^{\varphi}(t)$
defined by
\begin{align}\label{equ:engpsi}
  E^{\psi}(t):= \frac{1}{N}\langle \psi^\epsi_N(t),H^\epsi_N \psi_N^\epsi(t) \rangle_{L^2(\Omega^{N})}
 \end{align}
and 
\begin{align}\label{equ:enghart}
  E^{\varphi}(t):&= \langle \varphi(t) ,\big(-\Delta_x-\frac{1}{\epsi^2} \Delta_y+ \frac{1}{2}w^0*|\Phi(t)|^2 \big)  \varphi(t) \rangle_{L^2(\Omega)}. 
\end{align}
%
%For shorter notation we will denote them by $\gls{symb:Epsi}$ and $\gls{symb:Evarphi}$. 
By direct calculation one finds that they are both independent of time, cf. Lemma\,\ref{lem:constE}. 
Now we state the assumptions which allow stronger singularities in the pair interaction.
\begin{description}

  \item[A1'] Let $w=\gls{symb:ws}+\gls{symb:winfty}$ such that  %  \begin{itemize}
 %   \item 
	    for all $\epsi \in (0,1]$ there exists a $C \in \R$ such that
	    \begin{align*}
	      \norm{w_s^\epsi}_{L^s(\tilde \Omega)} \leq C \quad \norm{w_\infty^\epsi}_{L^\infty( \tilde \Omega) }\leq C
	    \end{align*}
	    for a $s\in (s_0,2)$ with $s_0=\frac{6}{5}$.\\
  %  \item 
	    There exists $w^0_s,w^0_\infty :\Omega_f \rightarrow \R $ and a function $\gls{symb:fepsi}:(0,1]\rightarrow \R^+ $ with $f(\epsi) \stackrel{\epsi \rightarrow 0}{\rightarrow} 0$ such that
	    \begin{align*}
	      \norm{w_s^\epsi- w_s^0}_{L^1(\tilde \Omega)} \leq f(\epsi) \quad \norm{w^\epsi_\infty-w^0_\infty}_{L^\infty(\tilde \Omega) } \leq f(\epsi),
	    \end{align*}
 	    where  $w^0_s(x,y):=w^0_s(x)$, $w^0_\infty(x,y):=w^0_\infty(x)$ for $(x,y)\in \tilde \Omega$ and let $w^0_s \in  L^1(\Omega_f) $, $w^0_\infty \in  L^\infty(\Omega_f)$.% For short notation we define 
% 	    \begin{align}
% 	     w^0:=w^0_s+w^0_\infty
% 	    \end{align}

%

%
 \item[A2']
  Let $H_N^\epsi$ be self-adjoint with $D(H_N^\epsi) \subset D(\sum_{i=1}^N h_i^\epsi)$.

  \item[A3']
  Let the two-particle interaction $w$ be nonnegative.
\end{description}

% \begin{rem}
%  
% The assumption A1' is up to the stronger singularities of $w_s$ the same as assumption A1. 
%    
% %     \begin{align*}
% %       \norm{w_s^\epsi}_{L^2(\Omega)} \leq C 
% %     \end{align*}
% % %
% %   replaced by
% % %
% %     \begin{align*}
% %       \norm{w_s^\epsi}_{L^r(\Omega)} \leq C 
% %     \end{align*}
% %    for a $s\in (s_0,2)$ with $s_0=\frac{6}{5}$
% \end{rem}

\begin{thm}\label{thm:thm2}
  Let the assumptions A1'-A3' hold, $t\in[0,\infty)$, $\Phi_0 \in H^2(\Omega_f)$ with $\norm{\Phi_0}_{L^2(\Omega_\mathrm{f})}=1$, $ \psi^\epsi_N(0) \in D(H_N) $
 with $\norm{\psi^\epsi_N(0)}_{L^2(\Omega^N)}=1$ and $\chi=\chi_0$,
  then there exists a $C \in \R^+ $ depending only on $w,w^0$ such that  
    \begin{align}\label{equ:thm2}
      \beta^{\epsi,N}(t) \leq  \beta^{\epsi,N}(0) \exp(C g(t)) + ( E^\psi-E^\varphi+f(\epsi)+N^{-\gls{symb:eta}})(\exp(C g(t))-1),
    \end{align}
  where $\eta=\frac{5s-6}{4s} $ %-\frac{s/s_0-1}{2s/s_0-s}  $
 and
  \begin{align*}
    g(t)=  \int_0^t \big( \norm{\varphi(t')}_{H^2(\Omega)} +\norm{\varphi(t')}_{ L^\infty(\Omega)}\big )^3  \D  t'.
   \end{align*}
%
%where $C$ depends only on the norms of the interaction.
%
\end{thm}

\begin{rem}
 \begin{enumerate}

    \item Similar to Remark 1.1 and as a result of $\alpha \leq \beta$ equation \eqref{equ:thm2} implies a bound for the one-particle density matrix of $\psi^\epsi_N$ with the rate of convergence given by
	  the square root of the right-hand side of \eqref{equ:thm2}. 

   \item  (See Remark 1.2) In addition to the condition $\alpha(0) \rightarrow 0 $ we now also need $ E^\psi(0) \rightarrow E^\varphi(0)$ for $\epsi \rightarrow 0$ and $N \rightarrow \infty$ 
	  to hold for the theorem to have any predictive power.  % In the Gross-Pitaevskii case such estimates are known \cite{LieSeiYng03,SchYng07}. In this case this result are not known.
	  This is for example true if the initial wave function is a product state given by $ \varphi(0)^{\otimes N}$,
	  %$\psi_N^\epsi(0)=  (\varphi^\epsi)^{\otimes N}$ 
	  then $E^\psi \rightarrow E^\varphi$ 
	  for $\epsi \rightarrow 0$ with the rate $f(\epsi)$. % \textcolor{red}{ In general we know that if  $\alpha(0) \rightarrow 0 $ then $ E^\psi(0) \rightarrow E^\varphi(0)$ as well.}
          % See Lemma\ref{}. 

%     \item The question of the %commutation of the two limits $\epsi \rightarrow 0 $ and $N \rightarrow \infty$ is determined by how $ E^\psi_N \rightarrow E^\varphi$ depends on the order of the limits.
% 	  rate of convergence is determined by how $ E^\psi \rightarrow E^\varphi$. For $\psi_N^\epsi(0)=  (\varphi^\epsi)^{\otimes N}$ we can choose $\epsi$ as function of $N$ such that
% 	  the rate is given by $N^\eta $.
% 	  %It follows from \eqref{equ:thm1} that the rate of convergence depence on  $\epsi$ and $N $ seperatly. 
% 	  %The actual rate of convergence in \eqref{equ:thm1} depends on how fast the potential convergence with $\epsi$.
% 	  %Hence regardless of the convergence rate of $w^\epsi \rightarrow w^0$, $\epsi$ can be chosen as a function of $N$ such that the rate of convergence is $N^{-1}$ in \eqref{equ:thm1}.

    \item In addition to the hard wall confinement we can add $\epsi^{-2} V^\perp(y)$ in the $N$-particle Hamiltonian for any bounded potential $V^\perp$.
	  The only difference in this situation is that then $\chi$ is an eigenfunction of 
          the operator $\epsi^{-2}(-\Delta_y+V^\perp(y))$ on $\Omega_\mathrm{c}$. 

    \item If we combine Lemma\,\ref{lem:3termeg}.\ref{lem:3.4g} with Theorem \ref{thm:thm2} we can allow external potentials $V \in C^2(\R^4,\R ) $,
	  where $ \partial_t  V(t,x,y), \partial_y  V(t,x,y)  \in C_c(\R^4)$ and  
	  \begin{align*}
	   \norm{ V(t)}_{L^\infty(R^3)} \leq C
	%  |\langle \psi(t), V(t) \psi(t) \rangle |_{L^2(\Omega^N)} \leq C
	  \end{align*}
	  for all $t \in [0,\infty).$    
% 	  $ A \in C^2(\R^4,\R ) $ and let there be a $C \in \R$  
% 	  such that 
% 	  \begin{align*}
% 	  \forall t,x,y  \quad |\partial_t A(t,x,0) |< C \qquad |\partial_t \partial_y A(t,x,y) |< C
% 	  \end{align*}

    \item The condition on $\chi$ to be the ground state function in the confined direction is, as physically expected, now necessary for our proof of \eqref{equ:thm2}.   

    \item The boundedness of $ \norm{\varphi(t)}_{H^2(\Omega)}$ and $\norm{\varphi(t)}_{ L^\infty(\Omega)}$ follows from the condition on $\varphi(0)$. This is well known and discussed in Appendix\,\ref{app:regsol}. 

	  %The existents of the used norms of $\varphi$ in $g(t)$ follows from the condition on $\varphi(0)$ and is discussed in Appendix \ref{app:regsol}. 

\item We can allow the potential to be negative if there exists a  %  \textcolor{red}{Remark \ref{rem:negpot}}
%The condition $A3'$  can be replaced by a weaker condition. Let $h$ be such the potential energy can be bounded by not all of the kinetic energy:\
%There exists a
 constant $\kappa \in (0,1) $, such that
\begin{align*}
 0 \leq (1-\kappa)( h_1+ h_2)+w^\epsi_{12}.
\end{align*}
%in the sense of operators. 

 \end{enumerate}

\end{rem}
 
\begin{exmp}[Coulomb Potential]\label{exp:coul}
 
Let the full pair interaction be the Coulomb potential 
\begin{align*}
 w:=\frac{1}{|r|}
\end{align*}
and $w^0$ the restriction of $w$ on $\Omega_\mathrm f \times 0$ 
\begin{align*}
 w^0:=\frac{1}{|x|}.
\end{align*}
For a confinement in one direction the condition $A1'$ holds with $f(\epsi)=\epsi$ and $s=2-\delta$ $\forall \delta >0 $ thus 
\begin{align*}
\beta(t) \leq \exp(C h(t)) (\beta_0 + E^\psi-E^\varphi+\epsi+N^{\frac{1}{2}-\delta}).
\end{align*}
This is a generalization of the results in \cite{AbdMehPin05}. In this paper the authors considered the limit $\epsi \rightarrow 0$ of the Hartree equation with $\frac{1}{|r^\epsi|}$ as the interaction potential.
%The Hartree equation with the Coulomb potential is called the Schrödinger-Poission system.\\
In the case of a confinement in two directions the condition $A1'$ does not hold for the Coulomb potential and it is 
%an interesting
an open questions if the simultaneous limit $N \rightarrow \infty$
and $\epsi \rightarrow 0$ is well defined in this case. 

\end{exmp}

\subsection{The NLS Case with a Confinement in Two Directions}

In the case $\theta=(0,1)$ and a confinement of the system in two directions the wave function $\psi_N^\epsi$ solves the Schrödinger equation with the Hamiltonian  
\begin{align*}
  H_N^{\epsi}= \sum_{i=1}^N h^\epsi_i +  \sum_{i \leq j}^N W^{\epsi,\theta,N}(r_i-r_j),
\end{align*}
 where 
\begin{align*}%\label{equ:wepsitheta}
  W^{\epsi,\theta,N}(r_i-r_j):=  (N^{-1} \epsi^2)^{1-3\theta} w \Big ( (N^{-1} \epsi^2)^{-\theta}  \big( (x_i-x_j),\epsi(y_i-y_j) \big  ) \Big)
\end{align*}
and
\begin{align*}
 h^\epsi= -\Delta_x- \frac{1}{\epsi^2} \Delta_y+V(t,x,\epsi y).
\end{align*}
The one-particle wave function $\varphi$ is as before defined by
\begin{align*}
 \varphi(t):=\Phi(t)\chi_0.
\end{align*}
The function $\chi_0$ was defined as the ground state of $-\epsi^{-2} \Delta_y$ on $\Omega_\mathrm{c}$.    
%
% \begin{align*}
%  \chi(t)=\exp(-\im E_0^\epsi t)\chi_0
% \end{align*}
%
The function $\Phi(t)$ is governed by the NLS equation with external potential    
\begin{align*}
 \im \partial_t \Phi(t) = (-\Delta_x + V(t,x,0)+ b |\Phi(t)|^2) \Phi(t) \qquad \Phi(0)=\Phi_0, 
\end{align*}
where 
\begin{align*}
 \gls{symb:b}= \int_{\R^3} w   \int |\chi_0|^4(y) \, \D y.
\end{align*}
$
$
To account for the external field $V$ we modify the functional $\beta$ slightly. The one-particle energy $ E^\psi(t)$ is defined as before in equation \eqref{equ:engpsi}
and the Gross-Pitaevskii energy $E^\varphi(t)$ is defined in analogy to the Hartree case by

\begin{align}\label{equ:enggross}
  E^{\varphi}(t):&= \langle \varphi(t) ,\big(-\Delta_x-\frac{1}{\epsi^2} \Delta_y+ V(t,x,0)+ \frac{1}{2}b |\Phi(t)|^2 \big)  \varphi(t) \rangle_{L^2(\Omega)}.
\end{align}
Now we can define

\begin{align*}
 \tilde \beta^{\epsi,N}(t):= \beta^{\epsi,N}(t)+ |E^\psi(t)-E^\varphi(t)|
\end{align*}
and state our assumptions.
\begin{description}
  \item[B1] Let the interaction potential $w$ be a positive, radial symmetric function with compact support and $w\in L^\infty(\R^3)$.

  \item[B2] Let $ V \in C^2(\R^4,\R ) $  such that $ \partial_t  V(t,x,y),\partial_y  V(t,x,y)  \in C_c(\R^4)$ and  
\begin{align*}
   \norm{ V(t)}_{L^\infty(R^3)} \leq C
\end{align*}
for all $t \in [0,\infty).$  

%     \begin{align*}
%       \partial_t  A(t,x,y) \in C_c(\R_t)
%     \end{align*}
%     $E^\psi(t)\leq 0$ for all $t \in [0,\infty)$ and $|\langle \psi,A \psi \rangle | \leq C $
%      and 
%     \begin{align*}
%         \partial_t \partial_y A(t,x,y) \in C_b(\R^4)
%     \end{align*}

   \item[B3]
    Let the energy per particle away from the ground state in the $y$-direction be bounded for $t=0$:
    \begin{align*}
     \sup_{N,\epsi}  N^{-1}  \langle \psi^\epsi_N(0),( H^\epsi_N(0)- N \frac{E_0}{\epsi^2})\psi^\epsi_N(0) \rangle_{L^2(\Omega^N)} \leq C
    \end{align*}
    for a $C\in \R^+$.
\end{description}

\begin{thm}\label{thm:thm3}
Let the assumptions B1-B3 hold, $t\in [0,\infty)$ let $\Phi_0 \in H^2(\Omega_\mathrm{f})$ with $\norm{\Phi_0}_{L^2(\Omega_\mathrm f)}=1$, $ \psi^\epsi_N(0) \in D(H_N) $ with $\norm{\psi^\epsi_N(0)}_{L^2(\Omega^N)}=1$ and $\chi=\chi_0$.
Let $\theta \in (\frac{1}{4},\frac{1}{3})$ and $\epsi(N)=N^{-\gls{symb:nu}}$ with $\frac{1}{2} < \nu < \frac{\theta}{1-2\theta}$ 
then for all such $\epsi(N)$ there exists a $\eta > 0$ and a $C \in \R^+$, which only depends on $w$, such that 
    \begin{align}\label{equ:thm3}
      \tilde \beta^{\epsi,N}(t) \leq  \tilde \beta^{\epsi,N}(0)\exp(C h(t)) +N^{-\eta}(\exp(C g(t))-1)
    \end{align}
with
 \begin{align*}
    g(t)=   \norm{\chi}^2_\LiOc \int_0^t \Big( \norm{\varphi(s)}_{H^2(\Omega)\cap L^\infty(\Omega)}+ \norm{\Delta  |\varphi(s)|^2}_{L^2(\Omega)}\norm{\varphi(s)}_{L^\infty(\Omega)}\\
		+ \norm{\dot V(s)}_{L^\infty(\Omega)} + \norm{ V(s)}^{1/2}_{L^\infty(\Omega)}  \Big)  \D  s .
   \end{align*}
% \norm{\chi}_{\LiOc}^2
\end{thm}

\begin{rem}
 \begin{enumerate}
	\item The optimal value of $\eta$ is given by %can be calculated for fixed $\theta$ from the explicit bounds given in the proof of \eqref{equ:thm3}.   
	    \begin{align*}
		\eta(\theta) = \begin{cases}
		\frac{4\theta-1}{3-4\theta} \qquad &\mathrm{for}\; \theta \in (\frac{1}{4}, \frac{7}{24}]\\
		\frac{1-3\theta}{4-9\theta}  \qquad &\mathrm{for}\; \theta \in (\frac{7}{24} ,\frac{1}{3}).
		\end{cases}
	  \end{align*}
	However, $\eta$ is at best of order $1/10$. Using the same methods as Pickl in \cite{Pic10} it should be possible to improve this rate.  % leading to $\eta$  . 
	  
    \item Similar to Remark 1.1 and as a result of $\alpha \leq \tilde \beta$ equation \eqref{equ:thm3} implies a bound for the one-particle density matrix of $\psi^\epsi_N$ with the rate of convergence given by
	  the square root of the right side of \eqref{equ:thm3}. 
  
    \item (See Remark 1.2) 	 
	  The  theorem is only meaningful if 
	  \begin{align}\label{equ:remtr}
	  \tilde \beta(0) \rightarrow 0 \quad \mathrm{for} \quad \epsi \rightarrow 0 \quad \mathrm{and} \quad N \rightarrow \infty.
	  \end{align}
	  From a mathematical standpoint we can take $\psi(0)= \varphi(0)^{\otimes N}$, then \eqref{equ:remtr} holds. 
	  Physically $\tilde \beta(0) \rightarrow 0 $ represents the question of condensation and was shown for $\theta =1$ for a confinement in two directions in
	  \cite{LieSeiYng03} and for a confinement in  one direction in \cite{SchYng07}. A fortiori these results hold for $\theta \in (0,1) $ as well, cf. \cite{LewNamRou14}.

% 	  \begin{align}\label{equ:remtr}
% 	    \\rightarrow 0, \qquad  E^\psi \rightarrow E^\varphi 
% 	  \end{align}
% 	  for $\epsi \rightarrow 0$ and $N \rightarrow \infty$. 
% 	  For  $\theta=1$, $\psi$ the ground state of a trapped Bose gas and $\varphi$ the ground state of the respective Gross-Pitaevskii energy functional the energies indeed converge.
% 	  This was shown in \cite{LieSeiYng03}. The convergence of the operators in the trace norm	 
% 	%  which is also important for the significance of \eqref{equ:thm3}
% 	  is in the case of a confinement in two directions to the authors knowledge not known. Only the convergence of the respective densities was shown in a weak $L^1$ sense in  \cite{LieSeiYng03}.
% 	  However in the case %of no confinement and for 
% 	  of a confinement in one direction this convergence was shown
% 	  %in the former case in \cite{LieSei02} and 
% 	  in \cite{SchYng07}. A fortiori all these results hold for $\theta \in (0,1) $ as well confer \cite{LewNamRou14}. 
% 
% 	  Hence the assumption of the convergence of the first term in \eqref{equ:remtr} seems mathematical reasonable and is further backed by the physical experiments in \cite{GoeVogKet01,SchSal01}. 
% 	  If we take $\psi(0)=\varphi(0)^{\otimes N}$ of course \eqref{equ:remtr} holds.

    \item The assumptions of Theorem \ref{thm:thm3} show that the two limits do not commute in general but have to be taken in the subset defined in the assumptions.
	  The condition $\nu < \frac{\theta}{1-2\theta}$ is necessary for the support of the interaction to scale in the NLS way   %that dynamics are governed by a reduced 3-dim NLS theory instead of a 2-dim equation. 
	   and the condition $\nu > 1/2$  ensures that due to energy conservation there are no exited states in the confined direction.

%    \item The optimal value of $\eta$ can be calculated for fixed $\theta$ from the explicit bounds given in the proof of \eqref{equ:thm3}.   

    \item In addition to the hard wall confinement we can add $\epsi^{-2} V^\perp(y)$ for any bounded potential $V^\perp$ in the $N$-particle Hamiltonian. 
	  The only difference in this situation is that then $\chi$ is an eigenfunction of 
          the operator $\epsi^{-2}(-\Delta_y+V^\perp(y))$ on $\Omega_\mathrm{c}$. 

    %\item The possibility of the inclusion of a time-dependent external potential is not only from a mathematical point interesting but even more so from a physical. The reason for this is that in an experiment
%	  one first confines the atoms in an external field such that they form a condensate and then changes the external fields to be able to observe their dynamic. 
 %At the cost of even more technicalities the condition on $A$ can be relaxed further such that the typical harmonic trapping potentials are included.

    \item With the help of the methods developed in \cite{Pic10} it should be possible to extend this result up to $\theta < 2/3$ maybe at a cost of $\sqrt{ \log N}$ in the exponential.
	  As Chen and Holmer conjectured in \cite{CheHol13} for a confinement in one dimension we expect the above theorem to hold for $\theta \in (0,1]$ with only the condition
	  $ \nu < \frac{\theta}{1-2\theta}$ for $\theta \in (0,1/2)$ and no condition on $\nu$ for $\theta \in [1/2,1]$. 
	  %on $\nu$: $1/2< \nu < \frac{\theta}{1-2\theta}$ for $\theta \in (1/4,1/2)$ and $1/2< \nu $ for  $\theta \in [1/2,1]$.     
    
    \item %The existents of the used norms of $\varphi$ in the function $g(t)$ follows from the condition on $\varphi_0$. This is discussed in Appendix\,\ref{app:regsol}. 
	   The boundedness of $\norm{\varphi(s)}_{H^2(\Omega)\cap L^\infty(\Omega)}$ and $ \norm{\Delta  |\varphi(s)|^2}_{L^2(\Omega)}\norm{\varphi(s)}_{L^\infty(\Omega)}$
	   follows from the condition on $\varphi(0)$. This is well known and discussed in Appendix\,\ref{app:regsol}. 
  % \item \textcolor{red}{ Einschränkendes potential dazu}
 \end{enumerate}

\end{rem}

\section{Outline of the Proofs}
The proofs of the main results %Theorem \ref{thm:thm1},\ref{thm:thm2} and \ref{thm:thm3}
are given in Chapter \ref{chap:thm1}-\ref{chap:thm3}. 
In Chapter \ref{chap:thm1} we prove Theorem \ref{thm:thm1}. This proof can be understood as a nontechnical blueprint for the method used in the following ones.
In Chapter \ref{chap:meaofconv} we develop the notation associated with the measure $\beta$, explain this measure in more detail and state inequalities we often use in proofs. 
In the two remaining chapters we prove Theorem \ref{thm:thm2} and Theorem \ref{thm:thm3}.
  
The general idea of all proofs is straight forward: First we calculate the derivative of the measure and afterwards we try to bound this derivative by the measure itself and by terms which turn to zero in the limit.
Then the application of the Grönwall lemma leads to the desired results. This process is depicted in great detail in the nontechnical case of Theorem \ref{thm:thm1} in Chapter \ref{chap:thm1}.     

\section{Outlook}
There are several interesting questions beyond the scope of this thesis.
The most obvious questions are to prove results for the rate of convergence for $1/3 \leq  \theta \leq 1 $ in
the case of strong confinement in two directions and for $\theta \in (0,1]$ in the case of strong confinement in one direction.  
Another point is to enlarge the class of allowed two-particle interactions for the above questions.
Furthermore, one could try to improve the rates of convergence, possibly with the help of the methods used in \cite{BenOliSch12} if they are applicable. 
%Furthermore, the improvement of the rates of convergence seems possible. Here one could maybe use 

Apart from these questions there are more questions coming from the adiabatic structure of the problem.
Is it possible to obtain higher orders corrections in $\epsi$ like in adiabatic theory? Can one allow a strongly confining potential which depends on the coordinates of the free directions?

% 
% There are several interesting questions beyond the ones answered in this thesis.
% For the confined mean field systems it would be interesting to have a result for the coulomb potential with a confinement in two directions either positive or negative.
% A whole array of questions comes form the adiabatic side of this question: Can one go two higher orders in $\epsi$, is it possible to allow a strongly confining potential which vary with $x$ and how does the right
% scaling for this situation look like? \\
% For the NLS part the first goal would to get results for larger $\theta$ and at the end be able to prove the $GP$ case $\theta=1$ for both confinement in one and two directions.
% After that one could ask the same adiabatic questions as for mean field case?.
% \textcolor{red}{ Also one could look at the different regimes given in \cite{LieSeiSolYng05} and try describe the dynamics given in this regimes.}

% In ein dim eingeschränkende\\
% Für größere $\theta$ \\
% Ander eingeschränkende Potential\\
% Höhere Ordnungen Möglich?\\
%  Einschränkende Potential verändern in x?
% 
% \begin{rem}[Notation]
% Wohin??\\
% Wie die potential benamsen?
% Norms $L^2(\R^3)$ $L^2(\R^{3N})$,$L^2(\Omega)$, $L^2(\Omega_\mathrm{f})$,$L^2(\Omega_\mathrm{c})$?\\
% 
% operator Norms  ?\\
% $w \in L^2 + L^p$\\
% Parameter Zeit $t, \epsi, \beta N$
% Welche Schritte angeben
% 
%\end{rem}

\section{Notation Used for the Proofs}
We will drop the dependencies on $t$, $\epsi$ and $N$ for better representation whenever this does not lead to confusion. 
%We needed we use the notation $A(t)$ and $A|_t$ interchangeably to denote the value of $A$ at time $t$. 
We abbreviate $A \leq C B $ by $A \lesssim B$, where the constant $C$ depends only on $L^p$-norms of $w$ and the number
of confined directions but never on $t,\epsi$ and $N$. 
For a function defined as a sum $f=f_1+f_2$ we define the shorthand
\begin{align*}
 \norm{f}_{L^p+L^q}:=\norm{f_1}_{L^p}+\norm{f_2}_{L^q} 
\end{align*}
and for any function $f$
\begin{align*}
 \norm{f}_{L^p \cap L^q}:=\norm{f}_{L^p}+\norm{f}_{L^q}.
\end{align*}
For the scaler product in $L^2(\Omega^N)$ we define the shorthand
\begin{align*}
 \llangle \cdot,\cdot \rrangle := \langle \cdot , \cdot \rangle_{L^2(\Omega^N)}
\end{align*}
and for the $L^2$-norm on $\Omega^N$ we use
 \begin{align*}
\norm{ \cdot} := \norm { \cdot}_{L^2(\Omega^N)}.
\end{align*}
We write $w_{ij}$ for $w(r_i-r_j)$ and hence 
we write $w^s_{12}$ for $w_{s}(r_i-r_j)$ and $w^\infty_{12}$ for $w_{\infty}(r_i-r_j)$.
%The operator norm of a multiplication operator $f$ is equal to and will always be denoted by $\norm{f}_\infty$.
%
In the case where $\theta=0$ we set for all calculations $w(r_i-r_j)=0 \;\; \forall r_i,r_j  \notin \tilde \Omega  $. This has no impact on the estimated terms since 
the terms are always of the form
\begin{align*}
 \llangle \psi, w(r_i-r_j) \psi \rrangle
\end{align*} 
which only depends on the values of $w$ on the set ${\tilde \Omega}$. 
%
%If we consider a function $f$ defined on a domain $A$ then if $\dim A= \dim B$ we set $f|_{B \setminus A}=0$ for $A \subset B $. 
We sometimes regard $\varphi$ as a function on $\R^3$ where we set $\varphi(r)=0$ for $r \notin \Omega $.
Where it is convenient we use the Dirac notation for scalar products in $L^2(\Omega)$ 
% \begin{align*}
%  \bra{\cdot} \ket{cdot}:=\varphi 
% \end{align*}
and for projections on a function
\begin{align*}
\ket{\varphi(r)} \bra{\varphi(r)}:= \varphi(r) \langle \varphi(r), \cdot \rangle_{L^2(\Omega,\D r)}.
\end{align*}
We denote the Sobolev spaces %on a set $\Omega$
 by $W^{k,p}$ and use $H^k$ for $W^{k,2}$. % and respectively $W^{k,p}_0(\Omega)$ and $H^k_0(\Omega)$ for
%suitable $\Omega$ the Sobolev spaces of functions $f$ with $Tf=0$ where the is the trace operator $T$.
The space of the weak $L^p$-functions is denoted by $L^p_w$.

\chapter{Proof of Theorem\,\ref{thm:thm1}  }  \label{chap:thm1}

This following proof can be seen as an illustration of Pickl's method \cite{KnoPic09,Pic11} for a model with a strongly confining potential.

The idea is to use a Grönwall argument for $\alpha $, so the first step is to check that $\alpha \in   C^1(\R) $ and then to control the derivative by terms that either become
negligible in the limit $N \rightarrow \infty, \epsi \rightarrow 0$ or are bounded by $C \alpha$. It turns out that it is best to calculate the time derivative of $\alpha $ in the form
\begin{align*}
 \alpha= \llangle \psi, q_1 \psi \rrangle
\end{align*}
 and then to decompose the derivative of $\alpha$ in terms that can be estimated one by one.
The decomposition is such that the part for which the mean field cancels the full interaction is separated from the rest. 
This decomposition will recur in the proofs of all theorems of this thesis and is essential to the method of Pickl.

\begin{rem}
To make the representation of the following calculation as clear as possible we replace the prefactor $N^{-1}$ in front of the interaction in equation \eqref{equ:schhartree}  by $(N-1)^{-1}$.
Thus the considered $N$-particle Hamiltonian is
\begin{align*}
  H_N^{\epsi}= \sum_{i=1}^N h^\epsi_i + \frac{1}{N-1} \sum_{i \leq j}^N w^\epsi(r_i-r_j).
\end{align*}
This change clarifies the calculations significantly since no extra terms of order $\mathcal{O}(N^{-1})$ appear in the calculations and
 at the same time this does not change the dynamics generated by this Hamiltonian for large $N$.%, since only terms of order $\mathcal{O}(N^{-1})$ are neglected,
\end{rem}

% \begin{rem}[Notation]
% We will drop the dependencies on $t$, $\epsi$ and $N$ for better representation whenever this does not lead to confusion. 
% For a function defined as a sum $f=f_1+f_2$ we define the short hand
% \begin{align*}
%  \norm{f}_{L^p+L^q}:=\norm{f_1}_{L^p}+\norm{f_2}_{L^q}.  
% \end{align*}
% For the scaler product in $L^2(\Omega^N)$ we define the shorthand
% \begin{align*}
%  \llangle \cdot,\cdot \rrangle := \langle \cdot , \cdot \rangle_{L^2(\Omega^N)}
% \end{align*}
% %
% %
% We write $w_{ij}$ for $w(r_i-r_j)$ and hence 
% we write $w^s_{12}$ for $w_{s}(r_i-r_j)$ and $w^\infty_{12}$ for $w_{\infty}(r_i-r_j)$.
% 
% \end{rem}

% \begin{lem}\label{diffbar}
% Let $\Psi(t) \in C^1 (\R, L^2(\Omega))$, $ \varphi \in C^1(\R, L^2(\Omega)) $ then $\alpha \in C^1(\R)$.
% 
% \begin{rem}\textcolor{red}{für mich}\\
% Das $\beta \in C^1(\R, \C)  $  für vorgegebenes $\psi_0, \phi_0$ und $w$ ist, folgt aus: Eingeschaften der Lösungen $C^1(\R,L^2)$ und Nachrechnen mit Hilfe der Eigenschaften des Skalaprodukts! \\
% $n \leq m \Rightarrow \langle \psi, \widehat n \psi \rangle \leq \langle \psi, \widehat m \psi \rangle $
% \end{rem}
%\end{lem}

We begin with the decomposition of the derivative of $\alpha$.

\begin{lem}\label{lem:ausrechnen}
Control of the derivative of $\alpha$
\begin{align*}
 \partial_t \alpha \leq  \mathrm{I}+\mathrm{II}+\mathrm{III},
\end{align*}
where 
\begin{align*}
 \mathrm {I} := 2|\llangle \psi , p_1 p_2  W^\epsi_{12} q_1 p_2 \psi \rrangle|\\ \mathrm{ II} := 2|\llangle \psi , p_1 p_2  W^\epsi_{12} q_1 q_2 \psi \rrangle|\\  \mathrm{ III} := 2|\llangle \psi , p_1 q_2  W^\epsi_{12} q_1 q_2 \psi \rrangle|
\end{align*}
and
\begin{align}\label{equ:Wepsi12}
 W^\epsi_{12}:&=   w^\epsi_{12}- (w^0 \ast |\Phi|^2)(x_1).
\end{align}
\end{lem}

In the first term the mean field cancels the full interaction and the term will thus be small. The second and the third term will be controlled by $\alpha$. The physically intuition is
that both of these terms are small for a $\psi$ close to a product state, since in this case $q_1q_2 \psi$ is small. However, making this idea rigorous via mathematical estimates is the main work of the proof.
The estimation results are summed up by the next lemma.  

\begin{lem}\label{lem:aabschaetzen}
\begin{enumerate}
 \item  \label{lem:aabschaetzen1}
\begin{align}\label{aabschaetzen1} 
  \mathrm{I}\leq 2 f(\epsi)(1+ \norm{ \varphi}_{L^\infty(\Omega)}^2)
 \end{align}
\item \label{lem:aabschaetzen2}
\begin{align}\label{aabschaetzen2}
 \mathrm{II}\leq 2 \norm{w^\epsi}_{L^2(\tilde \Omega)+L^\infty(\tilde \Omega)}(1+\norm{\varphi}_{L^\infty(\Omega)})(\alpha+\frac{1}{N})
\end{align}

\item \label{lem:aabschaetzen3}
\begin{align}\label{aabschaetzen3}
 \mathrm{III} \leq  2(\norm{w^0}_{L^1(\Omega_\mathrm{f})+L^\infty(\Omega_\mathrm{f})} +\norm{w^\epsi}_{L^2(\tilde \Omega)+L^\infty(\tilde \Omega)} )(1+\norm{\varphi}_{L^\infty(\Omega)} +\norm{\Phi}^2_\LzOf)\alpha
\end{align}

\end{enumerate}
 \end{lem}
 
Finally we state a version of the Grönwall Lemma. Its application is the final step in the proof of Theorem\,\ref{thm:thm1}.

 \begin{lem}[Grönwall]\label{lem:gron}
Let the function $f: \R \rightarrow \R $ for $t\in [0,\infty)$ satisfy the inequality
\begin{align*}
 \dot f(t) \leq C(t)(f(t)+\delta),
\end{align*}
where $C: \R \rightarrow \R$ and $\delta $ is a real constant. Then for $t\in [0,\infty)$
\begin{align*}
 f(t)\leq \E^{\int_0^t C(s) \D s}f(0)+\big(\E^{\int_0^t C(s) \D s}-1\big)\delta.
\end{align*}

\end{lem}

\begin{proof}[Proof of Theorem \ref{thm:thm1}]
Lemma \ref{lem:ausrechnen} and Lemma \ref{lem:aabschaetzen} lead to the following bound on $\dot  \alpha$ 
\begin{align*}
 \dot \alpha \leq   C(t) \big( \alpha +\frac{1}{N}+f(\epsi)     \big),
\end{align*}
where 
\begin{align*}
 C(t):%&=\max(C_1,C_2,C_3)\\
&=4 \big(\norm{ w^0}_{\LeOf+ \LiOf}+  \norm{w^\epsi}_{\LzOt+\LiOt} \big)   \\
&\quad\times \int_0^t (1+ \norm{\varphi(s)}_\LiO+ \norm{\Phi(s)}_\LiOf)^2  \D s.
\end{align*}
Now the claim follows with Lemma \ref{lem:gron}.

\end{proof}

\section*{Proof of the Lemmas}

\begin{proof}[Proof of Lemma\,\ref{lem:ausrechnen}]
Recall the definition of $\alpha$ 
\begin{align*}
\alpha: \,&\R \rightarrow [0,1], \\ &t \mapsto \llangle \Psi(t), q_1(t) \Psi(t) \rrangle.   
 \end{align*}
% We calculated the difference between the difference quotient and the derivative  \textcolor{red}{!schöner!}
% \begin{align*}
%  \frac{1}{h} &\big(\langle \Psi(t+h), q_1(t+h) \Psi(t+h) \rangle -\langle \Psi(t), q_1(t) \Psi(t) \rangle \big)\\
% &- \langle \Psi(t), q_1'(t) \Psi(t) \rangle -\langle \Psi'(t), q_1(t) \Psi(t) \rangle -\langle \Psi(t), q_1(t) \Psi'(t) \rangle   \\
% &= \langle \frac{1}{h}(\Psi(t+h)-\Psi(t))-\Psi'(t), q_1(t+h) \Psi(t+h) \rangle \\
% &+\langle \Psi(t),\frac{1}{h}( q_1(t+h)  \Psi(t+h)-q_1(t) \Psi(t))-q_1(t) \Psi'(t)- q_1'(t) \Psi(t)  \rangle\\
% &+\langle \Psi'(t),q_1(t+h)\Psi(t+h)-q_1(t)\Psi(t)\rangle\\    
% \end{align*}
The image of $\alpha$ is $[0,1]$ since $\norm{\psi}=1$ and $q(t)$ is a orthonormal projection.
The functional $\alpha $ is an element of $ C^1(\R)$ since the scalar product is linear, $\psi(t) \in C^1(\R,\mathcal{H}^N) $ and  $q_1(t) \in C^1(\R,\mathcal{L}(\mathcal{H}^N)) $
which follows from $\varphi(t) \langle \varphi(t), \cdot \rangle \in C^1(\R, \mathcal{L}(\mathcal{H}))$.
%
%Now using Cauchyschwarz and $\norm{q_1 \Psi} \leq \norm {q_1} \norm{\psi}$ together with $q_1(t) \in C^1(\R,\mathcal{L}(L^2(\R^{3N}))) $ finishes the proof.
%This is true, since $q_1(t):= \varphi(t) \langle \varphi(t), \cdot \rangle_{L^2(\R^3)} \otimes \id \otimes \cdots \otimes \id  $ is in $C^1(\R,\mathcal{L}(L^2(\R^{3N})))$
%
For the next calculation we note
\begin{align*}
 \partial_t \bigg( \varphi(t) \langle \varphi(t), \cdot \rangle_{L^2(\Omega)}\bigg)&= (\partial_t \varphi(t)) \langle \varphi(t), \cdot \rangle_{L^2(\Omega)}+
\varphi(t) \langle \partial_t \varphi(t), \cdot \rangle_{L^2(\Omega)}\\
&=-\im h^\Phi \varphi(t) \langle \varphi(t), \cdot \rangle_{L^2(\Omega)} + \im \varphi(t) \langle \varphi(t), h^\Phi \cdot \rangle_{L^2(\Omega)},
\end{align*}
where $h^\Phi=-\Delta_x + w^0* |\Phi(t)|^2 $. This equation can be written in a more compact form for the operator $q(t)$  
\begin{align}\label{ablq1}
 \im \partial_t q(t)= [h^\Phi,q(t)].
\end{align}
With the above remarks we can calculate
\begin{align}
 \partial_t \alpha&=\partial_t \big \llangle  \psi , q_1  \psi \big \rrangle \notag \\
&= \big \llangle  \dot \psi , q_1  \psi \big \rrangle  +\big \llangle  \psi , q_1  \dot \psi \big \rrangle +  \big \llangle  \psi ,(\partial_t q_1)   \psi \big \rrangle
 \notag\\
&= \im \big \llangle  \psi ,H_N q_1  \psi \big \rrangle -\im \big \llangle  \psi , q_1  H_N  \psi \big \rrangle -
\im \big \llangle  \psi ,[H_{x_1} {+H_{y_1}}, q_1]  \psi \big \rrangle \notag \\
&=\im \big \llangle  \psi ,[H_N, q_1]  \psi \big \rrangle-\im \big \llangle  \psi ,[H_{x_1} {+H_{y_1}} , q_1]  \psi \big \rrangle \notag \\
&=\im( \big \llangle  \psi ,[H_N-h_{x_1}^\Phi , q_1]  \psi \big \rrangle \label{Ia},
\end{align}
where we used equation \eqref{ablq1}.
Since only the parts of $H_N$ which act on the first particle do not commute with $q_1$ we find
\begin{align}\label{comHNq1}
\llangle \psi,[H_N,q_1] \psi \rrangle &= \llangle \psi, [-\Delta_{x_1}-\frac{1}{\epsi^2} \Delta_{y_1}+  w^\epsi_{12} ,q_1]\psi \rrangle,
\end{align}
where we used the symmetry of $\psi$ to write 
\begin{align*}
 \frac{1}{N-1} \sum_{j=2}^N w_{1j}^\epsi =  w^\epsi_{12}.
\end{align*}

Inserting \eqref{comHNq1} in equation \eqref{Ia} all one-particle operators vanish since $-\Delta_x$ cancels and $-\frac{1}{\epsi^2} \Delta_{y_1} $ commutes with a projection onto
one of its eigenfunction and hence with $q_1$.
We are left with

\begin{align}\label{equ:ableitung}
 \partial_t \alpha= \im \big \llangle  \psi ,[  W^\epsi_{12}, q_1]  \psi \big \rrangle,
\end{align}
where we recall that $W^\epsi_{12}$ is a shorthand for $w^\epsi_{12}- (w^0 \ast |\Phi|^2)(x_1).$
% \begin{align*}
%  W^\epsi_{12}&=  w^\epsi_{12}- w^0 \ast |\Phi|^2.
% \end{align*}
%
The next step is the decomposition of \eqref{equ:ableitung} to this end we insert $ \id =p_1+q_1$ on both sides of the commutator of \eqref{equ:ableitung} leading to  

 \begin{align*}
 \partial_t \alpha= \im \big \llangle  \psi,  p_1 W^\epsi_{12} q_1  \psi \big \rrangle- \im \big \llangle  \psi,  q_1  W^\epsi_{12} p_1  \psi \big \rrangle= 
-2 \Im \big \llangle  \psi,  p_1  W^\epsi_{12} q_1  \psi \big \rrangle.
\end{align*}
Last we insert $\id=(p_2+q_2)$ on each side of $W_{12}^\epsi$
 \begin{align*}
 \partial_t \alpha&= -2 \Im \big \llangle  \psi,  p_1  p_2 W^\epsi_{12} q_1 p_2 \psi \big \rrangle
-2 \Im \big \llangle  \psi,   p_1p_2 W^\epsi_{12} q_1 q_2 \psi \big \rrangle\\
&\quad-2 \Im \big \llangle  \psi,   p_1 q_2  W^\epsi_{12} q_1 q_2 \psi \big \rrangle
-2 \Im \big \llangle  \psi,   p_1 q_2 W^\epsi_{12} q_1 p_2 \psi \big \rrangle,
\end{align*}
where $\Im \big \llangle  \psi,   p_1 q_2 W^\epsi_{12} q_1 p_2 \psi \big \rrangle=0$ since it is the imaginary part of a self-adjoint operator $p_1 q_2 W^\epsi_{12} q_1 p_2$ under exchange of particle $1$ and $2$.
Taking the absolute value of the right side proves the lemma. 
\end{proof}

\begin{proof}[Proof of Lemma \ref{lem:aabschaetzen}.\ref{lem:aabschaetzen1} ]
Here we show that the mean field interaction cancels the full interaction. If we examine $p_2 W^\epsi_{12} p_2$ we find
\begin{align*}
 p_2 W^\epsi_{12} p_2 &= p_2  \Bigg(  w^\epsi_{12}- w^0 \ast |\Phi|^2)  \Bigg) p_2 \\
&= | \varphi(r_2) \rangle  \langle  \varphi(r_2)| w^\epsi(r_1-r_2) - (w^0*|\Phi|^2)(r_1)      |\varphi(r_2) \rangle \langle \varphi(r_2)|\\
&= p_2  \Big( \int_\Omega w^\epsi(r_1-r_2) |\varphi(r_2)|^2  \D r_2 - \big( w^0 \ast |\Phi|^2|\chi|^2) \big) (r_1)\Big)\\
&= p_2 \big((w^\epsi-w^0)*|\varphi|^2 \big )(r_1), \numberthis \label{equ:pwp}
\end{align*}
where we used the fact that $(w^0 * |\varphi|^2)(x_1)$ is constant in the $y_1$-direction to rewrite the term as $(w^0 * |\varphi|^2|\chi|^2 )(r_1)$.
If we enter \eqref{equ:pwp} in the term $\mathrm{I}$ we obtain
\begin{align*}
 {\mathrm{I}}&=2|\llangle \psi , p_1 p_2 \tilde W_{12} q_1 p_2 \psi \rrangle|=2 |\llangle \psi , p_1 p_2  \big((w^\epsi-w^0)*|\varphi|^2 \big )(r_1)  q_1  \psi \rrangle|\\
& \leq 2 \norm{ q_1 \psi} \norm{ \big((w^\epsi-w^0)*|\varphi|^2 \big )(r_1)  p_1 p_2 \psi} \\
&\leq   2 \| \big((w^\epsi-w^0)*|\varphi|^2 \big )(r_1)  \|_{\LiO} \numberthis \label{equ:p2.I.1}.
%&=(\norm{\varphi}^2_{4} +1) (f(\epsi)+\frac{C}{N}).
\end{align*}
%The part $w_\infty^\epsi-w_\infty^0 $ of $w^\epsi-w^0 $ can be estimated directly with the young inequality and is bounded by $f(\epsi,N)$.
This operator norm can be estimated with the help of Young's inequality, where we use $\sup w^\epsi = \tilde \Omega$ and $\sup \varphi=\Omega$,
\begin{align*}
 \norm{(w^\epsi-w^0)*|\varphi|^2 }_{\LiO}&\leq   \norm{(w_\infty^\epsi-w_\infty^0)*|\varphi|^2}_\LiO +   \norm{(w_s^\epsi-w_s^0)*|\varphi|^2}_\LiO \\
&\leq \norm{w_\infty^\epsi-w_\infty^0}_\LiOt + \norm{(w_s^\epsi-w_s^0)}_{\LeOt}\norm{\varphi}^2_{\LiO} \\
&\weq{A1 }{\leq} f(\epsi)(1+ \norm{ \varphi}_\LiO^2) \numberthis \label{equ:p2.I.2}.
\end{align*}
%
% \begin{align*}
%  \norm{\big((w_1^\epsi-w_1^0)*|\varphi|^2 \big )(r_1)  p_1}_{\rm{Op}}&= \sup_{\norm{\theta}=1} \norm{(w_1^\epsi-w_1^0)*|\varphi|^2 |\varphi  \rangle  \langle \varphi | \theta } \\
% &= \sup_{\norm{\theta}=1} \big (\langle \theta,   |\varphi  \rangle  \langle \varphi |  \big((w_1^\epsi-w_1^0)*|\varphi|^2 \big )^2 |\varphi  \rangle  \langle \varphi | \theta \rangle \big)^\frac{1}{2}  \\
% &=  \norm{ \big((w_1^\epsi-w_1^0)*|\varphi|^2 \big )^2 |\varphi|^2  }_{L^1{(\R^3)}}^\frac{1}{2} \\
% &\leq \norm{(w_1^\epsi-w_1^0)*|\varphi|^2} \norm{\varphi}_{L^\infty(\R^3)} \\
% &\leq  \norm{(w_1^\epsi-w_1^0)}_{1} \norm{\varphi}_{4}^2 \norm{\varphi}_{\infty} \\
% &\leq f(\epsi) \norm{ \varphi}_\infty^2 \numberthis \label{equ:p2.I.2}
% \end{align*}
%where we used $L^p$ interpolation in the last line to find $\norm{\varphi}_{4}^2 \leq   \norm{ \varphi}_\infty $.
Putting \eqref{equ:p2.I.1} and \eqref{equ:p2.I.2} together yields
\begin{align*}
  {\mathrm{I}} \leq 2 f(\epsi) (1+\norm{ \varphi}_\LiO^2).
\end{align*}

\end{proof}

\begin{proof}[Proof of Lemma \ref{lem:aabschaetzen}.\ref{lem:aabschaetzen2} ]
This term can be bounded by $\alpha$ since with the help of the symmetry we can figuratively swap a $q$ with a $p$ at a cost of a term which is of order $N^{-1}$.
Before we swap we have to rewrite the term and then use Lemma \ref{trick2q} to swap. First notice that the mean field interaction vanishes since
it only acts on the first coordinate which results in $p_2q_2=0$.
\begin{align*}
    \llangle \psi , p_1 p_2 W^\epsi_{12} q_1 q_2 \psi \rrangle&=  \llangle \psi , p_1 p_2 w^\epsi_{12} q_1 q_2 \psi \rrangle\\
	% \scriptstyle{[\mathrm{by \, sym.} ]} 
    & \stackrel{\mathclap{\mathrm{sym.}}}{=}  \frac{1}{(N-1)}  \llangle \psi, \sum_{j=2}^N p_1 p_j w^\epsi_{1j} q_1 q_j \psi \rrangle \\
    &\leq \frac{1}{(N-1)} \norm{q_1 \psi} \norm{\sum_{j=2}^N  q_j w^\epsi_{1j}p_1 p_2 \psi}\\
    &\stackrel{\mathclap{\ref{trick2q}}}{\leq}
	%\overset{\mathclap{......}}{\leq}
	% \scriptstyle{[\mathrm{by\,Lem. \ref{trick2q} }]}
	% & \leq
    \alpha^\frac{1}{2}  \norm{ w^\epsi_{12}p_1}_\mathrm{Op} \Big( \alpha  +\frac{1}{N}  \Big)^\frac{1}{2} \\
    &\leq \norm{ w^\epsi_{12}p_1}_\mathrm{Op} (\alpha+\frac{1}{N})\\
      %&\stackrel{\ref{abpot}}{ \leq}
      % \scriptscriptstyle{[\mathrm{by\,Lem. \ref{abpot} }]} & \leq 
    &\weq{\ref{abpot}}{\leq}
 \norm{w^\epsi}_{\LzOt+\LiOt}(1+\norm{\varphi}_\LiO)(\alpha+\frac{1}{N})
%\\&\leq C(1+\frac{1}{N}+\epsi^2+\frac{\epsi^2}{N})(\alpha+\frac{1}{N}+\frac{2}{N^2})
\end{align*}
%The last steps follows with the short Lemma \ref{abpot}. 
\end{proof}
\begin{proof}[Proof of Lemma \ref{lem:aabschaetzen}.\ref{lem:aabschaetzen3} ]
In this term we have enough $q$s to get an $\alpha$ and the norm of the interaction which remains can be bounded with Lemma \ref{abpot}.
\begin{align*}
 \llangle \psi , p_1 q_2 W^\epsi q_1 q_2 \psi \rrangle & \leq  \norm{ W^\epsi p_1  }_\mrm{Op} \norm{q_2 \psi} \norm{q_1 q_2 \psi}\\
& \stackrel{\mathclap{\eqref{equ:Wepsi12}}}{ \leq}
\big(\norm{w^0 * |\Phi|^2}_\mrm{Op}+\norm{w_{12}^\epsi p_1}_\mrm{Op} \big)  \alpha  \\
& \stackrel{\mathclap{\ref{abpot}}}{\leq} \Big(\norm{w^0}_{\LeOf+\LiOf}(1+ \norm{\Phi}^2_\LiOf) \\
&\qquad +\norm{w^\epsi}_{\LzOt+\LiOt}(1+ \norm{\varphi}_\LiO \Big) \alpha\\
&\leq \big (\norm{w^0}_{\LeOf+\LiOf} +\norm{w^\epsi}_{\LzOt+\LiOt} \big)\\
&\qquad\big(1+\norm{\varphi}_\LiO +\norm{\Phi}^2_\LiOf \big)\alpha
\end{align*}
\end{proof}

\begin{lem}\label{trick2q}
\begin{align*}
 \bigg \| \sum_{j=2}^N  q_j  w^\epsi_{1j}p_1 p_2 \psi \bigg \| \leq (N-1) \norm{ w^\epsi_{12}p_1}_\mathrm{Op} \Big( \alpha  +\frac{1}{N}  \Big)^\frac{1}{2}  
\end{align*}
\end{lem}
\begin{proof}
 \begin{align*}
&\norm{\sum_{j=2}^N  q_j  w^\epsi_{1j} p_1 p_2 \psi}^2=  \sum_{l,j=2}^N  \llangle \psi, p_1 p_j  w^\epsi_{1j} q_j q_l  w^\epsi_{1l} p_1 p_l \psi \rrangle \\
&=  \sum_{l\neq j}^N  \llangle \psi,q_l p_1 p_j  w^\epsi_{1j}    w^\epsi_{1l} p_1 p_l q_j \psi \rrangle+ \sum_{j=2}^N  \llangle \psi, p_1 p_j  w^\epsi_{1j} q_j  w^\epsi_{1j} p_1 p_j \psi \rrangle  \\
&\leq (N-1)(N-2)\norm{q_2\psi}\norm{q_3\psi}\norm{ w^\epsi_{13}p_1}_\mathrm{Op} \norm{ w^\epsi_{12} p_1}_\mathrm{Op} +(N-1)\norm{ w^\epsi_{12}p_1}_\mathrm{Op}^2\\
&\leq (N-1) \norm{ w^\epsi_{12}p_1}_\mathrm{Op}^2  \Big( (N-2)\alpha  +1 \Big)  \\
&\leq (N-1)^2 \norm{ w^\epsi_{12}p_1}_\mathrm{Op}^2 \Big( \alpha  +\frac{1}{N}  \Big)
\end{align*}
%Where we used Lemma \ref{Hilflem}, \ref{abpot} and the constants defined there. \\
%This is just a rough estimate of this term with a closer look one can get better estimates but for simplicity we refrain from that. 

% \textcolor{red}{ist das richtig?}
% With a little more work one can get $C = 2 \max(C_1,C_2)$ :
% For the first term we can use $B_{12}B_{13}=B_{13}B_{12} $ and thus write 
% 
% \begin{align*}
% \norm{p_1 B_{12}B_{13} p_1}= \norm{p_1 \frac{1}{2}(B_{12}B_{13}+B_{13}B_{12}) p_1} \leq \norm{p_1(B_{12}^2+B_{13}^2)p_1}= 2 \norm{p_1 B_{12}^2 p_1}  
% \end{align*}
% and for the second summand
% \begin{align*}
%  \norm{p_1 p_2 B_{12} q_2 B_{12} p_2 p_1}= \norm{p_1 p_2 B_{12} (1-p_2)  B_{12}  p_1 p_2}\leq \norm{p_1 B_{12}^2 p_1 }+ \norm{p_2 B_{12} p_2}^2 \\
% \end{align*}

\end{proof}

% \begin{lem}\label{Hilflem}
%  Let A be a orthogonal projection $A:\mathcal{H} \rightarrow \mathcal{H} \ $ and  B a selfadjoint operator in $\mathcal{L} (D, \mathcal{H}) $ then
% \begin{align*}
%  \norm{B A}_{\rm{op}}^2 \leq \norm{AB^2A}_{\rm{op}}
% \end{align*}
% holds where the right hand side is well defined.
% \end{lem}
% \begin{proof}
% 
% \begin{align*}
% \norm{B A}_{\rm{op}}^2& = \sup_{\norm{\psi}=1, \psi \in D^2 } \langle B A \psi, BA \psi \rangle=\sup_{\norm{\psi}=1, \psi \in D^2 } \langle \psi, AB^2A \psi \rangle\\
%  &\leq \sup_{\norm{\psi}=1, \psi \in D^2 } \norm{\psi}^2_\mathcal{H} \norm{AB^2A}_{\rm{op}}=\norm{AB^2A}_{\rm{op}}
% \end{align*}

% \end{proof}

\begin{lem}\label{abpot}

\begin{enumerate}
 \item 
	\begin{align}
        \norm{ w^0* |\Phi|^2}_{\LiOf} \leq \norm{ w^0}_{\LeOf+ \LiOf}(1+\norm{\Phi}_{\LiOf}^2)
       \end{align}
 \item 
	\begin{align}
	 \norm{ w^\epsi_{12} p_1 }_{\mathrm{Op}} &\leq %\big (\norm{w^{\epsi,2}}_2+\norm{w^{\epsi,\infty}}_\infty \big)^2 \big(\norm{\varphi}_\infty+ \norm{\varphi}_4 \big)^2 \notag \\
		\norm{w^\epsi}_{\LzOt+ \LiOt} \big(1+\norm{\varphi}_\LiO \big)
	\end{align}

\end{enumerate}

\end{lem}

\begin{proof}
For the proof we use the assumptions on $w^\epsi$, $w^0$ and $\varphi$ and Young's inequality.
The first estimate is obtained by
% Since $ w^0$ does not depend on $y$ and  $\norm{\chi}_2=1$ we find
% $\norm{ w^0 * |\varphi|^2}_{L^\infty(\R^3)}= \norm{ w^0 * |\Phi|^2}_{L^\infty(\R^2)} $ and hence

\begin{align*}
\norm{ w^0 * |\Phi|^2}_\LiOf &\leq  \norm{ w^0_\infty * |\Phi|^2 }_\LiOf+  \norm{ w^0_s * |\Phi|^2 }_\LiOf\\
&  \leq \norm{w^0_\infty}_\LiOf +\norm{ w^0_s}_{\LeOf} \norm{ \Phi}^2_\LiOf\\
&= \norm{ w^0}_{\LeOf+ \LiOf}(1+\norm{\Phi}_\LiOf^2).
\end{align*}
The second statement follows with
\begin{align*}
 \norm{ w^\epsi_{12} p_1 }_{\mathrm{Op}} \leq \norm{ w^{\epsi,\infty}_{12} p_1 }_{\mathrm{Op}}+\norm{ w^{\epsi,s}_{12} p_1 }_{\mathrm{Op}} 
\leq \norm{w^\epsi_\infty}_\LiOt + \norm{ w^{\epsi,s}_{12} p_1 }_{\mathrm{Op}}
\end{align*}
together with 
\begin{align*}
 \norm{ w^{\epsi,s}_{12} p_1 }_{\mathrm{Op}}&=\sup_{\norm{\rho}=1} \norm{ w^{\epsi,s}_{12} |\varphi  \rangle  \langle \varphi | \rho }_{L^2(\Omega^2)} \\
&= \sup_{\norm{\rho}=1} \big ( \big \langle \rho,   |\varphi  \rangle  \langle \varphi |  (w^{\epsi,s}_{12})^2 |\varphi  \rangle  \langle \varphi | \rho \big \rangle_{L^2(\Omega^2)}  \big)^\frac{1}{2}  \\
&\leq  \norm{  (w^\epsi_{s})^2* |\varphi|^2  }_{L^1{(\Omega)}}^\frac{1}{2} \\
&\leq \norm{w_s^\epsi}_\LzOt \norm{\varphi}_{\LiO}. \\
\end{align*}
% As in the proof of Lemma \ref{aabschaetzen} equation \eqref{aabschaetzen2} we have
% \begin{align*}
% \norm{p_1 (w^\epsi_{12})^2 p_1 }_{\mathrm{Op}} \leq  \norm {p_1 \big((w^\epsi)^2 * |\varphi|^2\big)(r_2)}_{\mathrm{Op}} \leq \norm{(w^\epsi)^2 * ( |\chi|^2  |\Phi|^2)}_\infty 
% \end{align*}
% %
% With $(w^\epsi)^2 = (w^{\epsi,2}+w^{\epsi,\infty})^2$ and $\norm{w^{\epsi,2}w^{\epsi,\infty}}_2 \leq   \norm{w^{\epsi,2}}_2  \norm{w^{\epsi,\infty}}_\infty $
% \begin{align*}
%  &\norm{(w^\epsi)^2 * ( |\chi|^2  |\Phi|^2)}_\infty \leq \norm{\big((w^{\epsi,2})^2+2 w^{\epsi,2} w^{\epsi,\infty} +(w^{\epsi, \infty} )^2 \big)*( |\chi|^2  |\Phi|^2) }_\infty\\
% &\leq \norm{\big( w^{\epsi,2 }\big)^2}_1\norm{|\Phi|^2}_\infty \norm{|\chi|^2}_\infty + 2\norm{ w^{\epsi,2}w^{\epsi,\infty} }_2\norm{|\Phi|^2}_2 \norm{|\chi|^2}_2+\norm{(w^{\epsi,\infty})^2 }_\infty \norm{|\Phi|^2}_1 \norm{|\chi|^2}_1\\
% &\leq \norm{ w^{\epsi,2}}_2^2 \norm{\Phi}^2_\infty \norm{\chi}^2_\infty+ 2\norm{ w^{\epsi,2}}_2 \norm{w^{\epsi,\infty} }_\infty \norm{\Phi}_4^2 \norm{\chi}_4^2+\norm{w^{\epsi,\infty}}^2_\infty
% \end{align*}
% To be able to forget the square roots later we estimate this by
% \begin{align*}
%  \big (\norm{w^{(2)}}_2+\norm{w^\infty}_\infty \big)^2 \big(\norm{\varphi}_\infty+ \norm{\varphi}_4 \big)^2,
% \end{align*}
% which is true for $\norm{\varphi}_\infty+ \norm{\varphi}_4 \geq 1$
\end{proof}

\begin{proof}[Proof of Lemma \ref{lem:gron}]
 Let $g: \R \rightarrow \R$ be a continuous function in $[0,T]$ and differentiable in  $(0,T)$ with
\begin{align*}
 \dot g(t) \leq C(t) g(t).
\end{align*}
Define G(t) as
\begin{align*}
 G(t):=\E^{\int_0^t C(s) \D s  }.
\end{align*}
Note that $\dot G(t)= C(t) G(t) $ and $\frac{g(0)}{G(0)}=g(0)$. 
\begin{align*}
 \partial_t \Big( \frac{g(t)}{G(t)} \Big)= \frac{\dot g(t) G(t)- g(t)\dot G(t)}{G(t)^2} \leq \frac{C(t)g(t) G(t)- C(t) g(t) G(t)}{G(t)^2}=0
\end{align*}
Thus $\frac{g(t)}{G(t)}\leq g(0)$ which implies 
\begin{align*}
 g(t)\leq \E^{\int_0^t C(s) \D s  } g(0).
\end{align*}
Now let $g(t)=f(t)+\delta$ with $\dot g(t)=\dot f(t) \leq C(t)(f(t)+\delta)= C(t)g(t) $. Hence
\begin{align*}
 f(t)+\delta \leq \E^{\int_0^t C(s) \D s  }(f(0)+\delta)
\end{align*}
and consequently
\begin{align*}
 f(t)\leq \E^{\int_0^t C(s) \D s  }f(0)+ \big(\E^{\int_0^t C(s) \D s  }-1\big) \delta.
\end{align*}

\end{proof}

% 
% \subsection{Hilfssätze:}
% \begin{enumerate}
% \item $N\geq 2$ $a \geq 0$
% \begin{align*}
% a^\frac{1}{2}(a+\frac{1}{N-1})^\frac{1}{2} \leq a+\frac{1}{N}       
% \end{align*}
% 
%  \item Lemma \ref{vp1} für $\tilde W$:
% \begin{align*}
%  \norm{p_1\tilde W }_\mrm{Op}=\norm{\tilde W p_1}_\mrm{Op} &\leq  C_1 \big( \frac{N-1}{N}(C_2+\epsi^2 C_3)+C_2)\\
% &=C_1(C_2(2-\frac{1}{N})+\epsi^2 \frac{N-1}{N}C_3 )\\
% &\leq C_1\max(C_2,C_3)(2-\frac{1}{N}+\epsi^2(1-\frac{1}{N}))\\
% &\leq C + C\epsi^2 \\
% &=C(1+\epsi^2)
% \end{align*}
% mit $C= 2C_1\max(C_2,C_3)$
% \item 
% Der Fehler von $ \dot \alpha$ hat folgende Form
% \begin{align*}
%  2 &C \bigg(\frac{1}{N}+\epsi^2+ (1+\epsi^2)(\alpha+\frac{1}{N})+\alpha (1+\epsi^2)  \bigg)\\
% & \leq 2C\bigg(2 \alpha +\frac{2}{N}+\epsi^2 +2 \alpha \epsi^2 + \frac{\epsi^2}{N}   \bigg)\\
% & \leq 4C\bigg( \alpha +\frac{1}{N}+\epsi^2 + \alpha \epsi^2 + \frac{\epsi^2}{N}     \bigg)\\
% & = 4C\bigg( \alpha(1+\epsi^2) +\frac{1}{N}+\epsi^2  + \frac{\epsi^2}{N}     \bigg)
% \end{align*}
%\end{enumerate}

%\include{measursofconvergence}

\chapter{Measures of Convergence: $\alpha$ and $\beta$} \label{chap:meaofconv}

In this section we discuss the properties of the functionals $\alpha$ and $\beta$ and how they relate to $ \Tr \big| \gamma^{\psi} - |\varphi \rangle \langle \varphi | \big|$.
The functional $\alpha$ and $\beta$ were first introduced by Pickl in \cite{Pic08,KnoPic09,Pic10} and the fermionic counterpart to $\alpha$ was recently 
used by Petrat and Pickl to derive the mean field for fermions \cite{PetPic14}. 
In these papers the properties of the functionals were developed and discussed in detail. Here we represent the parts needed for a basic understanding and which are necessary for our further calculations.
For a complete presentation we also restate the proofs given by Pickl.

We first state the functional $\alpha$ in the way we defined it in equation \eqref{equ:alpha1} 
\begin{align*}
 \alpha := 1- \langle  \varphi, \gamma^\psi  \varphi \rangle_{L^2(\Omega)}, 
\end{align*}
where $\varphi \in L^2(\Omega) $ and $\gamma^{\psi}$ is the one-particle density matrix of $\psi \in L^2(\Omega^N)$. The one-particle density matrix is a positive trace class operator which is defined by
its kernel
\begin{align*}
 \gamma^{\psi}(x'_1,x_1):= \int \psi(x_1', \dots, x_N) \overline{ \psi(x_1, \dots, x_N)} \D x_2 \cdots \D  x_N.
\end{align*}
As seen in the last chapter it is helpful to work with a different representation of $\alpha$.   
% it is best to rewrite $\alpha$ before calculating its time derivative.
% \textcolor{red}{ recast} such can easily calculate its time derivative and  its physical interpretation can easily be seen and is. 
To this end we define the following projections.
\begin{defn}\label{def:pP}
Let $\varphi \in L^2(\Omega) $ with $\norm{\varphi}_{L^2(\Omega)}=1$.
\begin{enumerate}[(a)] 
 \item For all $i \in \{1, \dots ,N\} $ we define
    \begin{align*}
      p_i:= \underbrace{ \id \otimes \cdots \otimes \id}_{i-1 \; \mathrm{times}}  \otimes \varphi(r_i) \langle \varphi(r_i), \cdot \rangle_{L^2(\Omega,\D r_i)} \otimes \underbrace{ \id \otimes \cdots  \otimes \id}_{N-i \; \mathrm{times}}
    \end{align*}
    and 
    \begin{align*}
      q_i:=\id - p_i.
    \end{align*}

  \item For any $0 \leq k \leq N$ we define
  \begin{align}\label{equ:defP}
   P_{k,N}:=  \Big( q_1\cdots q_k p_{k+1} \cdots p_N \Big)_{\mrm{sym}}=  \sum_{ \substack{ a_i\in \{0,1\}:\\ \sum_{i=1}^N a_i = k} }  \prod_{i=1}^N q_i^{a_i} p_i^{1-a_i},
  \end{align}
  where for $k<0 $ and $k>N$ we set $P_{k,N}=0$.
\end{enumerate}
\end{defn}
Part (a) of this definition allows to rewrite $\alpha$ for a symmetric $\psi$ with $\norm{\psi}_{L^2(\Omega^N)}=1$ as
\begin{align}\label{def:alpha}
 \alpha= 1- \langle  \varphi, \gamma^\psi  \varphi \rangle_{L^2(\Omega)} = 1- \frac{1}{N} \sum_{i=1}^N \llangle \psi, p_i \psi \rrangle 
 = 1-  \llangle \psi, p_1 \psi \rrangle  = \llangle \psi, q_1 \psi \rrangle.
\end{align}
The last representation of $\alpha$ is, as seen in the proof of Theorem \ref{thm:thm1}, the most useful one for calculating the derivative and applying the Grönwall Lemma.
With part (b) of the definition we can rewrite $\alpha$ further which will offer a way to generalize this functional to apply the approximation scheme 
of Chapter \ref{chap:thm1} to stronger singularities and to the derivation of the Gross-Pitaevskii equation. 

\begin{lem}\label{lem:qP}
 \begin{enumerate}[(a)]
  \item 
  \begin{align*}
   \sum_{k=0}^N P_{k,N}= \id
  \end{align*}
  \item 
  \begin{align*}
    \sum_{i=1}^N q_i P_{k,N}= k P_{k,N}
  \end{align*}
 \end{enumerate}
\end{lem}
The proofs are deferred to the end of this section.
If we apply this Lemma to $\alpha$ for a symmetric $\psi$ with $L^2$-norm one
\begin{align*}
 \alpha=  \llangle \psi, q_1 \psi \rrangle =  \llangle \psi, \frac{1}{N} \sum_{i=1}^N q_i \sum_{k=0}^N P_{k,N} \psi \rrangle = \llangle \psi,   \sum_{k=0}^N \frac{k}{N} P_{k,N} \psi \rrangle.
\end{align*}
Now we can interpret $\alpha$ as a counting functional which counts with the weight $\frac{k}{N}$ the wave function's norm in the image of the projections $P_{k,N}$. 
For a symmetric product state one can read off the counting functional's value:
Let $\varphi_j^\perp \in \mathop{Span}{\varphi}^\perp $ and $\psi =\big( \varphi^{\otimes (N-k)}\otimes \bigotimes_{j=1}^k {\varphi_j^\perp} \big)_\mathrm{sym}$ for a $ k$ with $0 \leq k \leq N$ then
\begin{align*}
 \alpha= \llangle \psi,  \sum_{k=0}^N \frac{k}{N} P_{k,N} \psi \rrangle=  \frac{k}{N}.
\end{align*}
The following aspect is far more important: We can generalize the functional if we use any positive function $f(k)$ as a counting measure 
\begin{align*}
\alpha_{f}= \llangle \psi,  \sum_{k=0}^N f(k) P_{k,N} \psi \rrangle .
\end{align*}
It turns out that the function $\sqrt{\frac{k}{N}}$ is in a sense explained at the end of this section the optimal weight. Thus we define 

\begin{align*}
 \beta := \llangle \psi,  \sum_{k=0}^N \sqrt{\frac{k}{N}} P_{k,N} \psi \rrangle .
\end{align*}
Since $\frac{k}{N} \leq \sqrt{\frac{k}{N}}$ for $k\in \{0,\dots,N\}$  we have
\begin{align}\label{equ:relalphabeta}
 \alpha \leq \beta.
\end{align}
Before we collect some facts for the use of $\alpha$ and $\beta$ we discuss the relationship of these functionals with
\begin{align*}
 \Tr \big| \gamma^{\psi} - |\varphi \rangle \langle \varphi | \big|.
\end{align*}

\section{The Relationship between $\alpha $ and Density Matrices}\label{sec:relalphatr}

It turns out that convergence to zero of the functional $\alpha$ is equivalent to convergence to zero of
\begin{align}\label{equ:tracegp}
\Tr \big| \gamma^{\psi} - |\varphi \rangle \langle \varphi | \big|.
\end{align}
This is encapsulated in the following lemma.

\begin{lem*}
 Let $\gamma^\psi$ be a density matrix and $\varphi \in L^2$ satisfy $\norm{\varphi}=1$. Then
\begin{align}\label{equ:traceproof}
 \alpha \leq \Tr \big| \gamma^{\psi} - |\varphi \rangle \langle \varphi | \big| \leq \sqrt{8 \alpha}.
\end{align}

\end{lem*}
 \begin{proof}
We restate the proof given in \cite{PetPic14} for fermions since it offers a nice interpretation of the origin of the different rates of convergence. 
A proof for the statement above which covers also a generalization can be found in \cite{KnoPic09}.
% We begin with the first "$\leq $" 
For the proof it is convenient to define  $ p:=|\varphi \rangle \langle \varphi |$ and $q:=1-p$.
\begin{align*}
 \alpha = 1- \langle \varphi, \gamma^\psi \varphi \rangle = \Tr \big(p-  p \gamma^\psi \big) 
 \leq   \norm{p}_\mathrm{Op} \Tr \big|p - \gamma^\psi   \big|=  \Tr \big|  \gamma^\psi  -|\varphi \rangle \langle \varphi | \big|
\end{align*}
 For the second "$\leq$" we notice that $q \gamma q $ and $p-p\gamma p$ are positive operators; the latter since $ \gamma \leq \id  $.
Now we find 
\begin{align*}
  \Tr \big|p - \gamma^\psi   \big|&=\Tr \big| p - p \gamma^\psi p -q \gamma^\psi q  - q \gamma^\psi p - p \gamma^\psi q \big |\\
&\leq \Tr \big| p - p \gamma^\psi p \big| + \Tr \big| q \gamma^\psi q \big |  + \Tr \big| q \gamma^\psi p \big |+  \Tr \big| p \gamma^\psi q \big |\\
&= \Tr ( p - p \gamma^\psi p ) +\Tr ( q \gamma^\psi q) +\Tr \big| q \gamma^\psi p \big |+  \Tr \big| p \gamma^\psi q \big |\\
&= 2 \alpha + \Tr \big| q \sqrt{ \gamma^\psi} \sqrt{ \gamma^\psi}  p \big| +\Tr \big|  p \sqrt{ \gamma^\psi} \sqrt{ \gamma^\psi}  q \big| \\
&\leq 2 \alpha + 2 \norm{  \sqrt{ \gamma^\psi} q }_\mathrm{HS}  \norm{ \sqrt{ \gamma^\psi}  p}_\mathrm{HS}  \\
&= 2 \alpha + 2 \sqrt{ \Tr{ (  q \gamma^\psi q )} \Tr{( p \gamma^\psi p)} }\\
&=2 \alpha + 2 \sqrt{\alpha (1-\alpha)} \leq \sqrt{8 \alpha},
\end{align*}
where the last inequality holds since $0 \leq \alpha \leq 1$ and the fact that the function $2x+2 \sqrt{x(1-x)}-\sqrt{8 x}$ is not positive for $x \in [0,1]$.
 \end{proof}

Although convergence to zero in one measure implies convergence to zero in the other measure the rates of convergence differ in general. 
%The reason for this is the weight of the
%cross terms in the trace norm, whereas the functional $\alpha$ controls only the diagonal entries of $\gamma^\psi$ with respect to $p$ and $q$. 
The reason for this is the different treatment of $p\gamma^\psi q$ in \eqref{equ:tracegp} and $\alpha$. 
Since $\alpha$ controls only the diagonal entries of $\gamma^\psi$ with respect to $p$ and $q$ the cross terms
have to be controlled by the diagonal terms which is only possible at a the cost of taking the square root.

%The following operators and identities will be useful their introduction goes back to Pickl in \cite{...}. Also the proofs can be found there but for readability and completeness sake we repeat them here. 
% \begin{defn}
%  \begin{align*}
%  P_{k,N}:=  \Big( q_1\cdots q_k p_{k+1} \cdots p_N \Big)_{\mrm{sym}}
% \end{align*}
% \end{defn} 
% \begin{rem}
% 
% Here ${\mrm{sym}}$ means symmetrization over the indices:
%  \begin{align*}
%   P_{k,N}= \sum_\sigma q_{j_1}\cdots q_{j_k} p_{j_{k+1}} \cdots p_{j_N}, 
% \end{align*}
% where $\sigma$ runs over are all orderd subsets $\{j_1,\dots, j_k\}$ of $\{1,\dots,N\}$ with $k$ elements and $ \{1,...,N\}\setminus \{j_1,\dots, j_k \}=:\{j_{k+1}, \dots, j_{N}\} $.\\
% In some situations a equivalent way of writing $P_{k,N}$ is more useful
% \begin{align*}
% P_{k,N}= \sum_{ \substack{ a_i\in \{0,1\}:\\ \sum_{i=1}^N a_i = k} }  \prod_{i=1}^N p_i^{1-a_i}q_i^{a_i}.
% \end{align*}
% \end{rem}
% 

\section{Elementary Properties for Working with $\beta$}\label{sec:beta}
In this section we introduce some notation % for the work 
to be able to estimate expressions containing the projections $P_{k,N}$. We also state some estimates which recur often in the proof of the theorems
 and we explain why we use the weight $\sqrt{\frac{k}{N}}$. % ${\frac{\sqrt k}{\sqrt N}}$ 
%is special.
%We begin with
\begin{defn}
 \begin{enumerate}[(a)]
    \item For any function $f:\{0, \dots , N \}\rightarrow \C$ we define the operator
      \begin{align*}
	\widehat f := \sum_{k=0}^N f(k) P_{k,N}.
      \end{align*}
    \item For any $j \in \Z $ we define the shift operator on a function by
      \begin{align*}
	(\tau_j f)(k) = f(k+j),
      \end{align*}
where we set $(\tau_j f)(k)=0 $ for $k+j \notin \{0, \dots , N \} $.
 \end{enumerate}
\end{defn}
The function $\sqrt{ \frac{k}{N}}$ will be used quite often in the proofs thus we define 
\begin{align*}
 n:&\{0,\dots,N\} \rightarrow [0,1] \\ & k  \mapsto \sqrt{\frac{k}{N}}.
\end{align*}

Now we collect some properties of the operator $\widehat f$.
\begin{lem}\label{lem:weights}
\begin{enumerate}[(a)] 
 \item \label{a} For all functions $f,g: \{0,\cdots, N  \} \rightarrow \C $ 
\begin{align*}
 \widehat{f}\widehat{g}=\widehat{fg}=\widehat{g}\widehat{f} \qquad \widehat f p_j = p_j \widehat f \qquad  \widehat f P_{k,N}= P_{k,N} \widehat f.
\end{align*}

%  \item \label{b} Let $n:\{0,1,\cdots,N  \}\rightarrow \R^+_0 $ be $n(k)=\sqrt{\frac{k}{N}}$ then
% \begin{align*}
% \widehat n^2 = \frac{1}{N} \sum_{j=1}^N q_j 
% \end{align*}
% hence by symmetry 
% \begin{align*}
%  \alpha=\langle \psi, q_1 \psi \rangle= \langle \psi, \widehat n^2 \psi \rangle.
% \end{align*}

\item \label{c} Let $f$ be a nonnegative function $\{0, \dots , N  \}\rightarrow [0,\infty) $ and $\psi \in L^2(\R^{3N}) $ a symmetric function, then for $j \in \{1,\dots, N\}$ 
\begin{align*}
\langle \psi, \widehat f q_j \psi \rangle &=\langle \psi, \widehat f \widehat n^2 \psi \rangle 
\end{align*}
and for $i\in \{1,\dots N\}$, $i \neq j$
\begin{align*}
\langle \psi, \widehat f q_i q_j \psi \rangle &\leq \frac{N}{N-1} \langle \psi, \widehat f \widehat n^4 \psi \rangle .
\end{align*}

 \item \label{lem:weightsc} For any function $f:\{0,1,\cdots,N  \}  \rightarrow  \C $ and any operator $T$ acting on two coordinates $r_i,r_j$  of  $\mathcal H^N $
\begin{align}
\widehat f Q_j T Q_k = Q_j T  Q_k \widehat{\tau_{j-k} f} \label{vertausch1} \\
 Q_j T  Q_k \widehat f=\widehat { \tau_{k-j} f} Q_j T Q_k  \label{vertausch2}
\end{align}
for $Q_0 := p_i p_j$, $Q_1 \in \{ p_i q_j, q_i p_j\}$, $ Q_2:= q_i q_j $.

\end{enumerate}

\end{lem}
The second statement illustrates how the $q$s fit in the framework of the hatted operators and the third statement is crucial for the use of general weights.
The reason for this is that the fact $[H_N,q_1]= \mathcal{O}(1) $ used in equation \eqref{comHNq1} seems at first untrue for arbitrary operators $\hat f$. However, with (c) one can show that for suitable $f$ for example $f=\sqrt{k/N}$ the commutator $ [H_N,\hat f]$
is still of order one.

To simplify the notation in the proofs we formally write $n^{-1}$ for
\begin{align*}
\sum_{k=0}^N \Big(\frac{k}{N}\Big)^{-1/2} P_{k,N}.
\end{align*}
We will use this to estimate terms of the form $\norm{ \widehat n^{-1}  q_1  \psi} $,
where the $q_1$ ensures that we do not divide by $0$.
% We will use this to write
% \begin{align*}
% \llangle \psi, f(x_1-x_2) q_1  \psi \rangle= \llangle \psi, f(x_1-x_2) \widehat n \widehat n^{-1}  q_1  \psi \rrangle \leq \norm {   \widehat n  f(x_1-x_2) \psi} \norm{ \widehat n^{-1}  q_1  \psi},
% \end{align*}
% where we would have to introduce $p$ and $q$ to make the last $"\leq"$ rigorous but in any case the $q_1$ ensures that we do not divide by $0$. 

To be able to compute the time derivative of $\langle \psi, \hat f \psi \rangle$ we note:

\begin{lem}\label{hat.}
Let $\varphi \in C^1(\R,L^2(\Omega))$, then
\begin{enumerate}[(a)]
 \item $\forall k\in \{0, \dots ,N\}$ 
\begin{align*}
 P_{k,N}(t) \in C^1(\R, \mathcal{L}(\mathcal{H}^N   )).
\end{align*}
\end{enumerate}
Let $\varphi= \Phi \chi $, where $\chi $ is an eigenfunction of $-\epsi^2 \Delta_y$ on $\Omega_\mathrm{c}$, then
\begin{enumerate}[(b)]
\item 
\begin{align*}
 [-\Delta_y,\widehat f]=0
\end{align*}

\item[(c)] 
\begin{align*}
  \im \partial_t \widehat f =[H^\Phi, \widehat f],
 \end{align*}
where $H^\Phi:= \sum_{i}^N h_i^\Phi $ and $h^\Phi$ is the Hamiltonian associated with $\Phi$.
\end{enumerate}

\end{lem}

% 
% \begin{rem}\textcolor{red}{für mich}\\
% For $p= |\varphi \rangle \langle \varphi |$ it is true that $\im \partial_t p = [h^\Phi- \frac{1}{\epsi^2} \Delta_y, p  ]=:  [h^\varphi, p  ] $  
% \end{rem}

The next estimates are needed for the control of the terms emerging from the derivation of $\alpha$ and $\beta$.

\begin{lem}\label{lem:young}
 \begin{enumerate}[(a)]
Let $h\in L^2(\R^3)$ and $p= |\varphi \rangle \langle \varphi |$.
  \item 
    \begin{align*} 
	\norm{h(r)p}_{Op} \leq \norm{h}_{L^2(\R^3)}\norm{\varphi}_{L^\infty(\R^3)}
     \end{align*}

\item  

	\begin{align*}
        \norm{h(r_1-r_2)p_1}_{Op} \leq \norm{f}_{L^2(\R^3)} \norm{\varphi}_{L^\infty(\R^3)}
       \end{align*}

\item Let $g \in L^1(\R^3)$.
\begin{align*}
\norm{ p_1 g(r_1-r_2)p_1}_{Op} \leq \norm{g}_{L^1(\R^3)}\norm{\varphi}^2_{L^\infty(\R^3)}
\end{align*}

% \item Let $f=f_1+f_2$
% \begin{align*}
%  \norm{f*|\varphi|^2}_{L^\infty(\R^3)} \leq \norm{f_1}_{L^1(\R^3)} \norm{\varphi}_{L^\infty(\R^3)} + \norm{f_2}_{L^\infty(\R^3)} \norm{\varphi}^2_{L^2(\R^3)}
% \end{align*}

\end{enumerate}

\end{lem}

\begin{cor}\label{cor:op}
 Let $A1'$ hold for $w^0$ and $ w^\epsi$.
\begin{enumerate}[(a)]
 \item  
	\begin{align*}
          \norm{w^\epsi*|\varphi|^2}_{\mathrm{Op}} \lesssim (1+ \norm{\varphi}_{L^\infty(\Omega)})^2
        \end{align*}
% \item
%       \begin{align*}
%           \norm{w^0*|\varphi|^2 }_{\mathrm{Op}} \lesssim (1+ \norm{\varphi}_{L^\infty(\Omega)})^2
%       \end{align*}

\item
     \begin{align*}
      \norm{p_2 w_{12}^\epsi p_2 }_{\mathrm{Op}} \lesssim (1+ \norm{\varphi}_{L^\infty(\Omega)})^2
     \end{align*}

\end{enumerate}

\end{cor}

\begin{lem}\label{lem:qs&N}
For all $l \in \N$ the expression 
\begin{align*}
 \norm { \big(\widehat{m}-\widehat{\tau_l m}\big) q_1 \psi}
\end{align*}

  can be estimated,
\begin{enumerate}[(a)] 
 \item if $m(k)= \frac{k}{N}$ by
\begin{align*}
 \norm { \big(\widehat{m}-\widehat{\tau_l m}\big) q_1 \psi} \leq \frac{l}{N},
\end{align*}

\item  if $m(k)= \sqrt{\frac{k}{N}}$ by
\begin{align*}
 \norm { \big(\widehat{m}-\widehat{\tau_l m}\big) q_1 \psi} \leq \frac{l}{N}.
\end{align*}
% \item and if $m(k)= 
% \begin{cases}
%               \frac{k}{N^\eta}  &\mathrm{for}\;  k  \leq N^\eta \\
% 	      1  &\mathrm{else}
% \end{cases}$ 
% by
% \begin{align*}
% \norm { \big(\widehat{m}-\widehat{\tau_l m}\big) q_1 \psi} \leq l N^{\frac{-1-\eta}{2}} 
% \end{align*}
\end{enumerate}
 
\end{lem}
 Now we can explain why the weight $ \sqrt \frac{k}{N}$ is special. % We formulate this in a lemma. 
On the one hand we will have to find bounds of the form 
\begin{align}\label{equ:beta1q}
\llangle \psi, p_1 p_2 g(r_1-r_2) q_1 p_2 \psi  \rrangle \leq C \llangle \psi, \widehat f \psi \rrangle+ \mathcal{O}(N^{-1})
\end{align}
for suitable functions $g$. With the tools now developed we can
estimate the left hand side of \eqref{equ:beta1q} by
\begin{align*}%\label{equ:beta1q}
\llangle \psi, p_1 p_2 g(r_1-r_2) p_2 q_1 \psi  \rrangle \weq{\ref{lem:weights}}{ =} \llangle \psi, \widehat {\tau_1 h}   p_1 p_2 g(r_1-r_2)  \widehat h^{-1} q_1 p_2 \psi  \rrangle\\
\lesssim \norm{\widehat{ \tau_1 h} \psi } \norm{ \widehat h^{-1} q_1 \psi } \stackrel{!}{\leq} C \llangle \psi, \widehat f \psi \rrangle+ \mathcal{O}(N^{-1}),
\end{align*}
 where $h$ is a suitable function.
By the scaling behavior this implies
\begin{align}\label{equ:un2}
  \norm{ \widehat h^{-1} q_1 \psi },\norm{\widehat{ \tau_1 h} \psi }  \leq \norm{\widehat f^{1/2} \psi} + \mathcal{O}(N^{-1}).
\end{align}
On the other hand we will need a bound of the form
\begin{align*}
\norm { \big(\widehat{f}-\widehat{\tau_1 f}\big) q_1 \psi}=\mathcal{O}(N^{-1}). 
\end{align*}
If both conditions hold we indeed find that the function $f$ is up to a positive constant determined and given by
\begin{align*}
 f= \sqrt{ \frac{k}{N}}.
\end{align*}
We formulate this in the following lemma.

% It is up to a constant the only monotone function $f:\{0,\dots, N\}\rightarrow \R$ with $f(0)=0$ for which 
% 
% \begin{align*}
% \norm { \big(\widehat{f}-\widehat{\tau_1 f}\big) q_1 \psi}=\mathcal{O}(N^{-1}) 
% \end{align*}
% and at the same time the statement 
% %
% \begin{align}\label{equ:beta1q}
% \llangle \psi, g(r_1-r_2) q_1 \psi  \rrangle \leq C \llangle \psi, \widehat f \psi \rrangle+ \mathcal{O}(N^{-1})
% \end{align}
% %can be bounded by $C \langle \psi, \widehat f \psi \rangle+ \mathcal{O}(N^{-1})$
% holds for any function $g\in L^\infty$.

%We give a prove of this uniqueness of $\beta$ in the Appendix \ref{sec:uniquness}.

\begin{lem}\label{lem:uni}
 If for a  monotone function $f:\{0,\dots, N\}\rightarrow \R$ with $f(0)=0$
\begin{align}\label{equ:un1}
\norm { \big(\widehat{f}-\widehat{\tau_1 f}\big) q_1 \psi}=\mathcal{O}(N^{-1}) 
\end{align}
% and for any $g\in L^\infty$
% \begin{align}\label{equ:beta1q}
% \llangle \psi, g(r_1-r_2) q_1 \psi  \rrangle \leq C \llangle \psi, \widehat f \psi \rrangle+ \mathcal{O}(N^{-1})
% \end{align}
holds and $ \exists h:\{0,\dots, N\}\rightarrow \R$ such that
\begin{align}\label{equ:un3}
  \norm{ \widehat h^{-1} q_1 \psi },\norm{\widehat{ \tau_1 h} \psi }  \leq \norm{\widehat f^{1/2} \psi}
\end{align}
holds, then up to a positive constant
\begin{align*}
 f= \sqrt \frac{k}{N}.
\end{align*}
\end{lem}
 The two properties \eqref{equ:beta1q} and \eqref{equ:un2} will be
crucial in the proofs of Theorem \ref{thm:thm2} and Theorem \ref{thm:thm3} thus we have to use the counting functional $\beta$ to proof them with the used method.

\section{Remaining Proofs of this Chapter}
\begin{proof}[Proof of Lemma\,\ref{lem:qP}]
 \begin{enumerate}[(a)]
  \item  This follows from the fact that $q_i+p_i=\id$.
  
%   \item 
%     \begin{align*}
%      \sum_{j=1}^N q_j &=  \sum_{j=1}^N q_j \sum_{k=0}^N P_{k,N} = \sum_{k=0}^N  \sum_{j=1}^N  q_j P_{k,N}\\
% &= \sum_{k=0}^N  \sum_{j=1}^N \sum_\sigma \underbrace{  q_j q_{j_1}\cdots q_{j_k} p_{j_{k+1}} \cdots p_{j_N}}_{\delta_{j \in \{j_1,\dots, j_k\} }q_{j_1}\cdots q_{j_k} p_{j_{k+1}} \cdots p_{j_N} }\\
% &= \sum_{k=0}^N \sum_\sigma k q_{j_1}\cdots q_{j_k} p_{j_{k+1}} \cdots p_{j_N} = \sum_{k=0}^N k P_{k,N} 
%     \end{align*}

\item 
    \begin{align*}
     \sum_{j=1}^N q_j &=  \sum_{j=1}^N q_j \sum_{k=0}^N P_{k,N} = \sum_{k=0}^N  \sum_{j=1}^N  q_j P_{k,N}\\
&=  \sum_{k=0}^N  \sum_{j=1}^N   \sum_{ \substack{ a_i\in \{0,1\}:\\ \sum_{i=1}^N a_i = k} }   \prod_{i=1}^N  \underbrace{ q_j q_i^{a_i} p_i^{1-a_i}}_{\delta_{a_j,1} q_i^{a_i} p_i^{1-a_i}}\\
&=  \sum_{k=0}^N  k   \sum_{ \substack{ a_i\in \{0,1\}:\\ \sum_{i=1}^N a_i = k} }   \prod_{i=1}^N  q_i^{a_i} p_i^{1-a_i} \\
&=  \sum_{k=0}^N  k  P_{k,N}
%&= \sum_{k=0}^N  \sum_{j=1}^N \sum_\sigma \underbrace{  q_j q_{j_1}\cdots q_{j_k} p_{j_{k+1}} \cdots p_{j_N}}_{\delta_{j \in \{j_1,\dots, j_k\} }q_{j_1}\cdots q_{j_k} p_{j_{k+1}} \cdots p_{j_N} }\\
%&= \sum_{k=0}^N \sum_\sigma k q_{j_1}\cdots q_{j_k} p_{j_{k+1}} \cdots p_{j_N} = \sum_{k=0}^N k P_{k,N} 
    \end{align*}

 \end{enumerate}

\end{proof}

\begin{proof}[Proof of Lemma\,\ref{lem:weights}] 
 \begin{enumerate}[(a)]
  \item  Using the definitions
    \begin{align*}
    \widehat f \widehat g= \sum_k f(k) P_{k,N} \sum_l g(l) P_{l,N} = \sum_{k,l} f(k)g(l) \underbrace{ P_{k,N}P_{l,N}}_{\delta_{k,l}P_{k,N}}= \widehat{fg}= \widehat g \widehat f.
    \end{align*}

\item
The equality follows from symmetry of $\hat f \psi $ and Lemma \ref{lem:qP}(b).\\
For the proof of the inequality let without loss of generality $i=1,j=2$:
\begin{align*}
\langle \psi, \widehat f q_1 q_2 \psi \rangle &= \frac{1}{N(N-1)} \langle \psi, \widehat f \sum_{i \neq j} q_i q_j \psi \rangle \\
&\leq \frac{1}{N(N-1)} \langle \psi, \widehat f \sum_{i,j} q_i q_j \psi \rangle \\
&=  \frac{N}{(N-1)} \langle \psi, \widehat f \widehat n^4 \psi \rangle .
\end{align*}

  \item 
\begin{align*}
\widehat f Q_j T Q_k&= \sum_{l=0}^N f(l) P_{l,N} Q_j T Q_k = \sum_{l=0}^N f(l) \underbrace{P_{l-j,N-2}}_{\mrm{only \, on \,  }x_3..x_N}Q_j T Q_k\\
&= \sum_{l=0}^N Q_j T Q_k f(l) P_{l-j,N-2}\\
&= \sum_{l=k-j}^{N+k-j} Q_j T Q_k f(l+j-k) P_{l-k,N-2}\\
&= \sum_{l=k-j}^{N+k-j} Q_j T Q_k (\tau_{j-k} f)(l) P_{l,N}\\
&=\sum_{l=0}^{N} Q_j T Q_k (\tau_{j-k} f)(l) P_{l,N} 
=Q_j T  Q_k \widehat{\tau_{j-k}  f}  
\end{align*}
The converse direction follows in the same way.
\end{enumerate}
\end{proof}

\begin{proof}[Proof of Lemma\,\ref{hat.}]
\begin{enumerate}[(a)]
 \item 
This follows from the fact that for $\varphi \in C^1(\R, L^2(\Omega))$ the operator 
\begin{align*}
 \varphi(t)\langle \varphi(t), \cdot \rangle_{ L^2(\Omega)}  
\end{align*}
is an element of $C^1(\R,\mathcal{L}(\mathcal {H}) )$.

\item This is the fact that an eigenfunction of an operator computes with this operator.
 
\item
Using $\im \partial_t p_i(t)= [h^\Phi_i,p_i(t) ]$, $\im \partial_t q_i(t)= [h^\Phi_i,q_i(t) ]$  and the product rule we get
\begin{align*}
  \im \partial_t \widehat f &= \im \partial_t \sum_{k=0}^N f(k) P_{k,N} =  \sum_{k=0}^N f(k)  \im \partial_t \sum_{ \substack{ a_i\in \{0,1\}:\\ \sum_{i=1}^N a_i = k} }   \prod_{i=1}^N  q_i^{a_i} p_i^{1-a_i}\\
&=  \sum_{k=0}^N f(k)  [\sum_{l=1}^N h^\Phi_l, \sum_{ \substack{ a_i\in \{0,1\}:\\ \sum_{i=1}^N a_i = k} }   \prod_{i=1}^N  q_i^{a_i} p_i^{1-a_i} ] \\
&=  [\sum_{l=1}^N h^\Phi_l, \sum_{k=0}^N f(k) \sum_{ \substack{ a_i\in \{0,1\}:\\ \sum_{i=1}^N a_i = k} }   \prod_{i=1}^N  q_i^{a_i} p_i^{1-a_i}] \\
&=: [H^\Phi, \widehat f].
\end{align*}

\end{enumerate}

% 
%  \begin{align*}
%   \im \partial_t \widehat f &= \im \partial_t \sum_{k=0}^N f(k) P_{k,N} =  \sum_{k=0}^N f(k)  \im \partial_t \sum_\sigma q_{j_1}\cdots q_{j_k} p_{j_{k+1}} \cdots p_{j_N}\\
% &=  \sum_{k=0}^N f(k)  \sum_\sigma [\sum_{l=1}^N h^\varphi_l, q_{j_1}\cdots q_{j_k} p_{j_{k+1}} \cdots p_{j_N}]  \\
% &=  [\sum_{l=1}^N h^\varphi_l,] \\
% &=: [H^\varphi, \widehat f]
% \end{align*}
\end{proof}

\begin{proof}[Proof of Lemma\,\ref{lem:young}]
 \begin{enumerate}[(a)]
  \item For any $f \in L^2(\R^3)$
\begin{align*}
 \norm{f(r_1)p_1}^2_\mathrm{Op}&=\sup_{\norm{\psi}=1} \langle \psi, p_1 f(r_1)^2 p_1 \psi  \rangle_{L^2(\R^3)} \\
&=    \sup_{\norm{\psi}=1} \Big \langle \psi, |\varphi(r_1) \rangle \langle \varphi(r_1) | f(r_1)^2 |\varphi(r_1) \rangle \langle \varphi(r_1) | \psi  \Big \rangle_{L^2(\R^3)}\\
&=   \langle \varphi(r_1) | f(r_1)^2 |\varphi(r_1) \rangle_{L^2(\R^3)}  \sup_{\norm{\psi}=1} \Big \langle \psi, |\varphi(r_1) \rangle\langle \varphi(r_1) | \psi \Big \rangle_{L^2(\R^3)}\\
&=   \langle \varphi(r_1) | f(r_1)^2 |\varphi(r_1) \rangle_{L^2(\R^3)}   \sup_{\norm{\psi}=1} \langle \psi, p_1  \psi  \rangle_{L^2(\R^3)}.
\end{align*}
Using the Hölder inequality for the first term and the fact that $p_1$ is a projection we find
\begin{align*}
 \norm{f(r_1)p_1}^2_\mathrm{Op}\leq \norm{ \varphi}_{L^\infty(\R^3)}^2 \norm{f}^2_{L^2(\R^3)} \sup_{\norm{\psi}=1} \norm{ \psi}^2_{L^2(\R^3)}=\norm{ \varphi}_{L^\infty(\R^3)}^2 \norm{f}^2_{L^2(\R^3)}.
\end{align*}

\item
\begin{align*}
\norm{f(r_1-r_2)p_1}^2_\mathrm{Op}&=\sup_{\norm{\psi}=1} \Big \langle \psi, p_1 f(r_1-r_2)^2 p_1 \psi  \Big \rangle_{L^2(\R^6)} \\
&=    \sup_{\norm{\psi}=1} \Big \langle \psi, |\varphi(r_1) \rangle \langle \varphi(r_1) | f(r_1-r_2)^2 |\varphi(r_1) \rangle \langle \varphi(r_1) | \psi \Big  \rangle_{L^2(\R^6)}\\
&= \sup_{\norm{\psi}=1} \Big \langle \psi, |\varphi(r_1) \rangle  \langle \varphi(r_1) |  \langle \varphi(r_1) | f(r_1-r_2)^2 |\varphi(r_1) \rangle  \psi  \Big \rangle_{L^2(\R^6)}\\
&\leq  \sup_{\norm{\psi}=1} \norm{\psi}^2_{L^2(\R^6)} \norm{ \langle \varphi(r_1) | f(r_1-r_2)^2 |\varphi(r_1) \rangle}_{L^\infty(\R^3)}\\
&= \norm{|\varphi|^2*f^2}_{L^\infty(\R^3)} \leq \norm{\varphi}_{L^\infty(\R^3)}^2\norm{f}^2_{L^2(\R^3)}
\end{align*}

\item
For any $g\in L^2(\R^3)$
\begin{align*}
 \norm{p_1 g(r_1-r_2)p_1}_\mathrm{Op}&=\sup_{\norm{\psi}=1} \norm{ p_1 g(r_1-r_2) p_1 \Psi }_{L^2(\R^6)} \\
&=\sup_{\norm{\psi}=1} \norm{ |\varphi(r_1) \rangle  \langle \varphi(r_1) | g(r_1-r_2) |\varphi(r_1) \rangle  \langle \varphi(r_1) | \Psi }_{L^2(\R^6)} \\
&=\sup_{\norm{\psi}=1} \norm{   \langle \varphi(r_1) | g(r_1-r_2) |\varphi(r_1) \rangle  |\varphi(r_1) \rangle \langle \varphi(r_1) | \Psi }_{L^2(\R^6)} \\
&\leq \norm{   \langle \varphi(r_1) | g(r_1-r_2) |\varphi(r_1) \rangle }_{L^\infty(\R^3)}  \sup_{\norm{\psi}=1} \norm{p_1 \psi}_{L^2(\R^6)}\\
&= \norm{|\varphi|^2*g }_{L^\infty(\R^3)} \leq \norm{\varphi}^2_{L^\infty(\R^3)} \norm{g}_{L^1(\R^3)}.
\end{align*}

 \end{enumerate}

\end{proof}

\begin{proof}[Proof of Corollary\,\ref{cor:op} ]
\begin{enumerate}[(a)]
\item  
With Young's inequality and $L^p$ interpolation for $\varphi$ we get
 \begin{align*}
  \norm{w^\epsi*|\varphi|^2}_{\mathrm{Op}} &\leq \norm{w^\epsi_s*|\varphi|^2}_{L^\infty(\R^3)} +\norm{w^\epsi_\infty*|\varphi|^2}_{L^\infty(\R^3)}\\
 &\leq \norm{w^\epsi_s}_{L^s(\tilde \Omega)} \norm{\varphi}^2_{L^{2s/(s-1)}(\Omega)}+  \norm{w^\epsi_\infty}_{L^\infty(\tilde \Omega)} \\
 &\stackrel{\mathclap{A1'}}{ \lesssim} (1+ \norm{\varphi}_\LiO )^2.
 \end{align*}

% \item 
% \begin{align*}
%   \norm{w^0*|\varphi|^2}_{\mathrm{Op}} & \norm{w^\epsi_s*|\varphi|^2}_{L^\infty(\R^3)} +\norm{w^\epsi_\infty*|\varphi|^2}_{L^\infty(\R^3)}
% \end{align*}

\item 
  \begin{align*}
      \norm{p_2 w_{12}^\epsi p_2 }_{\mathrm{Op}} &\leq \norm{ w^\epsi * |\varphi^2|}_{L^\infty(\R^3)} %\leq  \norm{p_2 w_{12}^{\epsi,s} p_2 }_{\mathrm{Op}} + \norm{p_2 w_{12}^{\epsi,\infty} p_2 }_{\mathrm{Op}}\\
    %\lesssim  \norm 
   \lesssim (1+ \norm{\varphi}_{L^\infty(\Omega)})^2
     \end{align*}
The first inequality follows from the proof of Lemma\,\ref{lem:young}(c) and the second inequality is part (a) of this corollary. 
    
\end{enumerate}

\end{proof}

\begin{proof}[Proof of Lemma\,\ref{lem:qs&N}]
 We calculate
\begin{align}\label{equ:helpuni}
\norm { \big(\widehat{m}-\widehat{\tau_l m}\big) q_1 \psi}^2=\llangle \Psi, \sum_{k=1}^N \Big(m(k)-m(k+l)\Big)^2\frac{k}{N} P_{k,N} \psi \rrangle .
\end{align}
\begin{enumerate}[(a)]
 \item  The difference of the weights squared is $\frac{l^2}{N^2}$ hence the result follows.

 \item  The difference of the weights squared is 
\begin{align*}
 \Bigg(\frac{\sqrt{k}-\sqrt{k+l}}{\sqrt{N}}\Bigg)^2= \frac{l^2}{(\sqrt{k}+\sqrt{k+l})^2N} \leq \frac{l^2}{k N}.
\end{align*}
Multiplying this with the remaining term $\frac{k}{N}$ gives the desired result.

% \item 
% The difference of the weights squared is
% \begin{align*}
% \begin{cases}
%               \frac{l^2}{N^{2 \eta}}  &\mathrm{for}\;  k  \leq N^\eta \\
% 	      0  &\mathrm{else}
% \end{cases}
% \end{align*}
% hence for $k\leq   N^\eta$
% \begin{align*}
%  \frac{l^2}{N^{2 \eta}}\frac{k}{N} \leq \frac{l^2}{N N^\eta}\frac{k}{N^\eta} \leq  \frac{l^2}{N N^\eta}
% \end{align*}

\end{enumerate}

\end{proof}

\begin{proof}[Proof of Lemma\,\ref{lem:uni} ]
% \section{Uniquness of  $\sqrt{\frac{k}{N}}$}\label{sec:uniquness}
Equation \eqref{equ:un1} implies with \eqref{equ:helpuni} the condition
%The two conditions assumed in the lemma for the weight function f lead to with  the two conditions
\begin{align}\label{equ:uniq2}
 (f(k)-f(k-1))^2 \lesssim (N k)^{-1}
\end{align}
and equation \eqref{equ:un3} with the Lemma\,\ref{lem:weights} the condition
\begin{align*}
 \exists h \geq 0:  h(k)^2 \lesssim  f(k) \; , \; h(k)^{-2} \frac{k}{N}  \lesssim f(k) .
\end{align*}
The first condition implies $f(k)\lesssim  \sqrt{ \frac{k}{N}}$ and the second one $ \sqrt{ \frac{k}{N}} \lesssim f(k)  $ thus the claim follows.
The second implication follows by contradiction. For the first implication we rewrite \eqref{equ:uniq2} as 
\begin{align*}
 f(k)\lesssim (N k)^{-1/2}+ f(k-1).
\end{align*}
This leads to 
\begin{align*}
 f(k)\lesssim \sum_{l=1}^k (N l)^{-1/2} \leq 2 \sqrt{\frac{k}{N}}  \
\end{align*}
since $f(0)=0$ and by estimating the sum by the integral of $l^{-1/2}$ on the interval $[0,k]$.
 %can be seen by squaring, adding $f(k-1)$ to the right, inserting the condition iteratively for all $f(k-l)$ and estimation by integration of $\frac{1}{\sqrt{k}} $the second by contradiction. 

\end{proof}

\chapter{Proof of Theorem\,\ref{thm:thm2} }\label{chp:proofthm2}

For stronger singularities it is clear that the method used in the proof of Theorem\,\ref{thm:thm1} has to be adopted since there we control
\begin{align*}
 \norm{w^\epsi_{12}p_1 }_{\mathrm{Op}}
\end{align*}
by Lemma \ref{lem:young}(b). This is only possible for $w^\epsi \in L^2(\R^3) $.
The main idea how to treat the stronger singularities will be the introduction of a vector field $\xi$ which is chosen such that $\nabla \xi=w $. This vector field will have 
higher $L^p$ regularity then $w$ and hence we will be able to control 
\begin{align*}
  \norm{\xi_{12} p_1 }_{\mathrm{Op}} 
\end{align*}
with Lemma \ref{lem:young}(b). However, we can only make use of such an estimate after partial integration which in turn means that we need to control $\nabla p $ and $\nabla q$. For the first term
this is no problem since we have enough regularity since $p$ is a solution of a one-particle equation. For the second term we have to invest some effort but with the help of 
energy conservation we are able to bound this term as well.
Other than this the proof uses the same ideas as in Chapter\,\ref{chap:thm1}. We organize the proof by showing the smallness of $ \norm{\nabla q_1 \psi}^2$ first in a separate section. Afterwards we
bound the derivative of $\beta$ in the following section.

\section{The Energy Lemma}

This section is devoted to finding a bound for
$ \norm{\nabla q_1 \psi}^2$. 
%The reason for this is that for stronger singularities in the interaction one needs a new way to estimate the terms
%with one and two $p$'s. At the end of this modified estimates one needs to control $ \norm{\nabla q_1 \psi}^2$, hence we illustrate how this can be achieved.
The main ingredients for the proof are energy conservation, refining the weights and writing the interaction as a divergence of a vector field.

We first recall the definitions and the assumptions which are necessary for the formulation and prove of the Energy Lemma. Thereafter we state the lemma and give a motivation and an outline of the proof.
The last part of this section are some auxiliary lemmas which prove the Energy Lemma and will be used again to prove the smallness of $\dot \beta$ in the next section.

\subsection{Assumptions, Definitions and Preliminaries}
For convenience we restate the assumptions of Theorem\,\ref{thm:thm2}.

\begin{description}
 \item[A1'] Let $w=w_s+w_\infty$ such that 
 %  \begin{itemize}
 %   \item 
	    for all $\epsi \in (0,1]$ there exists a $C \in \R$ such that
	    \begin{align*}
	      \norm{w_s^\epsi}_{L^s(\tilde \Omega)} \leq C  \quad \norm{w_\infty^\epsi}_{L^\infty(\tilde \Omega) }\leq C
	    \end{align*}
	    for a $s\in (6/5,2)$. And there exist $w^0_s,w^0_\infty :\Omega_f \rightarrow \R $ and a function $f(\epsi):(0,1]\rightarrow \R^+ $ with $f(\epsi) \stackrel{\epsi \rightarrow 0}{\rightarrow} 0$ such that
	    \begin{align*}
	      \norm{w_s^\epsi- w_s^0}_{L^1(\tilde \Omega)} \leq f(\epsi) \quad \norm{w^\epsi_\infty-w^0_\infty}_{L^\infty(\tilde \Omega) } \leq f(\epsi)
	    \end{align*}
	    and $w^0_s \in  L^1(\Omega_{\mathrm{f}}) $, $w^0_\infty \in  L^\infty(\Omega_{\mathrm f})$. For short notation we define 
	    \begin{align}
	     w^0:=w^0_s+w^0_\infty.
	    \end{align}

% 
%   \item[A1']
%   Let the same conditions hold as in A1 except 
% %    
%     \begin{align*}
%       \norm{w_s^\epsi}_{L^2(\Omega)} \leq C 
%     \end{align*}
% %
%   replaced by
% %
%     \begin{align*}
%       \norm{w_s^\epsi}_{L^r(\Omega)} \leq C 
%     \end{align*}
%    for a $r\in (r_0,2)$ with $r_0=\frac{6}{5}$
%
 \item[A2']
  Let $H_N^\epsi$ be self-adjoint with $D(H_N^\epsi) \subset D(\sum_{i=1}^N h_i^\epsi)$.

  \item[A3']
  Let the two-particle interaction $w$ be nonnegative.
\end{description}

\begin{rem}\label{rem:negpot}
 The condition $A3'$  can be replaced by a weaker condition. Let the one-particle Hamiltonian $h$ be such that the potential energy can be bounded by a part of the kinetic energy:\
There exists a constant $\kappa \in (0,1) $ such that
\begin{align*}
 0 \leq (1-\kappa)( h_1+ h_2)+w^\epsi_{12}.
\end{align*}
% \textcolor{red}{for which $\epsi$}\\    
% \textcolor{red}{stimmt das? benutze approx mit $w^0$ ? } 
\end{rem}
%
%
% 
% \begin{defn}
% 
% \begin{align*}
% \beta(t):= \langle \psi_N, \widehat n \psi_N \rangle|_t
% \end{align*}
%  
% \end{defn}
%  
% \begin{rem}
%  Since $\frac{k}{N} \leq  \sqrt{\frac{k}{N}}$ for $k\in \{0,\dots N\}$ we have
% \begin{align*}
% \alpha(t) \leq \beta(t) \quad \forall t 
% \end{align*}
% 
% 
% \end{rem}
As defined in equation \ref{equ:engpsi} and \ref{equ:enghart} the energy per particle $ E^\psi(t)$ of $\psi$ is
 \begin{align*}
  E^\psi(t):= \frac{1}{N}\llangle \psi^\epsi_N(t),H^\epsi_N \psi^\epsi_N(t) \rrangle
 \end{align*}
and the energy $ E^\varphi(t)$ of the function $\varphi$ is 
 \begin{align*}
  E^\varphi(t):&= \langle \varphi(t) ,\big(-\Delta_x-\frac{1}{\epsi^2} \Delta_y+ \frac{1}{2}(w^0*|\Phi(t)|^2) \big)  \varphi(t) \rangle_{L^2(\Omega)} \\
&= \langle \Phi(t), \big(-\Delta_x+ \frac{1}{2}(w^0*|\Phi(t)|^2) \big) \Phi(t) \rangle_\LzOf + \langle \chi, -\frac{1}{\epsi^2} \Delta_y \chi \rangle_\LzOc.
 \end{align*}
If we use symmetry of $\psi$ we can rewrite $E^\psi$ as
 \begin{align*}
  E^\psi=\frac{1}{N}\llangle \Psi_N,H^\epsi_N \Psi_N \rrangle  &=\frac{1}{N} \llangle \Psi_N, (- \sum_{j=1}\Delta_{x_j}-\frac{1}{\epsi^2} \Delta_{y_j}+ \frac{1}{N} \sum_{i < j} w^\epsi_{ij}) \Psi_N \rrangle \\
&=\llangle \psi_N, h_1^\epsi \psi_N\rrangle + \frac{1}{N^2} \frac{N}{2}(N-1) \llangle \psi_N, w^\epsi_{12} \psi_N\rrangle \\
&= \llangle \psi_N, h_1^\epsi \psi_N\rrangle +  \frac{N-1}{2N} \llangle \psi_N, w^\epsi_{12} \psi_N\rrangle.
 \end{align*}

\begin{lem}\label{lem:constE}
Both $E^\psi(t)$ and $ E^\varphi(t)$ are constant in time: 
\begin{align*}
 \frac{d}{dt} E^\psi(t)=0 \qquad   \frac{d}{dt} E^\varphi(t)=0.
\end{align*}

\end{lem}
\begin{proof}
This is proven by the calculation
\begin{align*}
 \frac{d}{dt}E^\psi(t)= \frac{1}{N} \big(\llangle - \im H^\epsi_N \psi_N^\epsi(t),H_N^\epsi \psi^\epsi_N(t) \rrangle+ \llangle \psi_N^\epsi(t), -\im H_N^\epsi H_N^\epsi \psi^\epsi_N(t) \rrangle\big )=0 
\end{align*}
and 
\begin{align*}
  \frac{d}{dt}E^\phi(t)&= \frac{d}{dt} \langle \Phi(t), \big(-\Delta_x+ \frac{1}{2}(w^0*|\Phi(t)|^2) \big) \Phi(t) \rangle_\LzOf 
\\&\qquad+\frac{d}{dt} \langle \chi^\epsi(t), -\frac{1}{\epsi^2} \Delta_y \chi^\epsi(t) \rangle_\LzOc \\
&= \im  \langle \Phi(t),[h^\Phi,-\Delta_x+ \frac{1}{2}(w^0*|\Phi(t)|^2) ] \Phi \rangle_\LzOf \\
&\qquad+ \frac{ \im}{2} \langle \Phi(t),[h^\Phi,(w^0*|\Phi(t)|^2)] \Phi(t) \rangle_\LzOf \\
&= \im  \langle \Phi(t),[h^\Phi,-\Delta_x+(w^0*|\Phi(t)|^2)] \Phi(t) \rangle_\LzOf\\
&=\im  \langle \Phi(t),[h^\Phi,h^\Phi] \Phi(t) \rangle_\LzOf=0,
\end{align*}
where we introduced $h^\Phi:= -\Delta_x+ w^0*|\Phi|^2$ for short notation.
\end{proof}
As a reminder the ground state energy of $- \Delta_y$ on $\Omega_\mathrm{c}$ was denoted by $E_0$.
\begin{lem}\label{lem:Energyh}
The operator
\begin{align*}
 \tilde h:= { -\Delta_x-\frac{1}{\epsi^2}\Delta_y-\frac{E_0}{\epsi^2}}
\end{align*}
 is a positive self-adjoint operator with
 \begin{align*}
  -\Delta \leq \tilde h + E_0
 \end{align*}
and 
\begin{align}\label{equ:htildep}
 \norm{\tilde h p}_\mathrm{Op} \leq \norm{\Delta \Phi}_\LzOf.
\end{align}

\end{lem}
\begin{proof}
The first statement follows from $\Delta -E_0 \leq \tilde h$. The second one is derived as Lemma\,\ref{lem:young} together with 
\begin{align*}
\langle  \varphi, \tilde  h^2 \varphi \rangle_{L^2(\Omega) } = \norm{\Delta \Phi}_\LzOf, 
\end{align*}
where we used $(\frac{1}{\epsi^2}\Delta_y-\frac{E_0}{\epsi^2})\chi(y)=0$.
\end{proof}

% 
% \begin{rem}
%  A sufficient condition for $E^\psi-E^\varphi \rightarrow 0 $ to hold for  $\epsi \rightarrow 0$, is  $\psi(0)= \varphi^{\otimes N}$.
% The condition that the Energy per particle away from the scaled direction, $y$-direction stays finite, which is physically sensible to assume is a necessary condition for $E^\psi-E^\varphi \rightarrow 0 $ to hold
% an at the same time guarantees that  $E^\psi-E^\varphi \rightarrow 0 $ only poses a condition on the wave function in $x$-directions   \\
% \textcolor{red}{Ist das richtig, Gibt es auch schwächer Voraussetzungen unter denen obiges gilt?! }\\
% \textcolor{red}{for Gross Pitaiveski $E^\psi-E^\varphi \rightarrow 0 $ was shown by \cite{??} }
% \end{rem}

\subsection{The Energy Estimate and its Proof}\label{sec:pen}

\begin{lem}[Energy Lemma]\label{lem:energy}
 Let the assumptions A1'-A3' hold, then
\begin{align*}
 \norm{\nabla q_1 \Psi}^2 \lesssim  (E^\psi- E^\phi)+ \norm{\varphi}_{H^2(\Omega)\cap L^\infty(\Omega)}^2(\beta+\frac{1}{\sqrt{N}}+f(\epsi)).
\end{align*}
 
\end{lem}

The way we prove this is to first use Lemma\,\ref{lem:Energyh} which implies 
\begin{align}\label{equ:nablaqpsi}
 \norm{\nabla_1 q_1 \psi }^2 \leq  \norm{\sqrt{\tilde h_1} q_1 \psi }^2 + E_0\, \beta .
\end{align}
We find with $q_1=\id-p_1(p_2+q_2)$ and equation \eqref{equ:htildep} that
\begin{align}\label{equ:esqrth2}
 \norm{\sqrt{ \tilde h_1} q_1 \psi } \leq  \norm{\sqrt {\tilde h_1}(1-p_1 p_2 )\psi} + \norm{\Delta \Phi}_{\LzOf} \sqrt {\beta}
\end{align}
%This is only true for $\tilde h_1$ since it is bounded by $ \norm{\nabla \Phi}^2$ on $p_1$ where as $h_1$ does diverge on $p_1$.
this implies
\begin{align}\label{equ:h_1q_1}
 \norm{\sqrt{ \tilde h_1} q_1 \psi }^2 \leq  \norm{\sqrt {\tilde h_1}(1-p_1 p_2 )\psi}^2 + \norm{\Delta \Phi}_\LzOf^2  {\beta}
\end{align}
hence we try to find a bound for  $\norm{\sqrt {\tilde h_1}(1-p_1 p_2 )\psi}^2$ to bound $\norm{\nabla_1 q_1 \psi }^2$.
The necessary estimate is given in the next lemma.

\begin{lem}\label{lem:h1}
 \begin{align*}
  \llangle \psi, (1-p_1p_2 ) \tilde h_1 (1-p_1p_2) \psi \rrangle %\lesssim \big( E^\psi- E^\phi \big)  +\norm{\Phi}_{H^1} \beta+ \norm{ \Phi}_{H^2}(\beta + \frac{1}{\sqrt N})\\
% &+\frac{3}{2} \norm{w^\epsi* |\varphi|^2}_\infty (\beta+ \frac{1}{N}) + \norm{(w^0*|\varphi|^2-w^\epsi* |\varphi|^2)}_\infty\\
% &+ \norm{ (w^\epsi *|\varphi|^2)(x_1)}_\infty (\beta+\frac{1}{\sqrt N})\\
% &+C (\alpha+\frac{2}{N} )+ C  \sqrt{\alpha+\frac{2}{N}} \norm{\nabla q_1 \psi}\\
&\lesssim (E^\psi- E^\phi)+ \norm{\varphi}_{H^2(\Omega)\cap L^\infty(\Omega)}^2(\beta+\frac{1}{\sqrt{N}}+f(\epsi))\\
 &\quad +\norm{\varphi}_{L^2(\Omega) \cap L^\infty(\Omega) } \sqrt{\beta+\frac{1}{N}} \norm{\nabla_1 q_1 \psi}
 \end{align*}

\end{lem}
\begin{proof}[Proof of the Energy Lemma]

After rewriting the left-hand side of Lemma\,\ref{lem:h1} we find
\begin{align*}
 \norm{\sqrt {\tilde h_1}(1-p_1 p_2 )\psi}^2 &\lesssim (E^\psi- E^\phi)+ \norm{\varphi}_{H^2(\Omega)\cap L(\Omega)^\infty}^2(\beta+\frac{1}{\sqrt{N}}+f(\epsi))\\
&\quad+ \norm{\varphi}_{L^2(\Omega) \cap L^\infty(\Omega) } \sqrt{\beta+\frac{1}{N}} \norm{\nabla_1 q_1 \psi} \numberthis \label{equ:esqrth}.
\end{align*}
% Plugging this in the right hand side of \eqref{equ:h_1q_1} and then using the equation 
% \begin{align*}
%  \norm{\nabla_1 q_1 \psi } \stackrel{ \eqref{equ:nablaqpsi}}{ \lesssim}  \norm{\sqrt{\tilde h_1} q_1 \psi }+\sqrt{\beta},
% \end{align*} 
%which is a result of equation \eqref{equ:nablaqpsi},
This leads to 

\begin{align*}
\norm{\sqrt{ \tilde h_1} q_1 \psi }^2 & \stackrel{ \eqref{equ:h_1q_1}}{ \lesssim}  \norm{\sqrt {\tilde h_1}(1-p_1 p_2 )\psi}^2+\norm{\Phi}^2_{H^2(\Omega_\mathrm{f})}  {\beta} \\
&\stackrel{\eqref{equ:esqrth}}{\lesssim}  (E^\psi- E^\phi)+ \norm{\varphi}_{H^2(\Omega)\cap L^\infty(\Omega)}^2(\beta+\frac{1}{\sqrt{N}}+f(\epsi))\\
&\qquad+ \norm{\varphi}_{L^2(\Omega) \cap L^\infty(\Omega) } \sqrt{\beta+\frac{1}{N}} \norm{\nabla_1 q_1 \psi}\\
&\stackrel{\eqref{equ:nablaqpsi}}{\lesssim}  (E^\psi- E^\phi)+ \norm{\varphi}_{H^2(\Omega)\cap L^\infty(\Omega)}^2(\beta+\frac{1}{\sqrt{N}}+f(\epsi))\\
&\qquad+ \norm{\varphi}_{L^2(\Omega) \cap L^\infty(\Omega) } \sqrt{\beta+\frac{1}{N}}\norm{\sqrt{\tilde h_1} q_1 \psi } . \numberthis \label{equ:penlem} 
\end{align*}
Now we have an inequality of the form $x^2 \leq C(R+ax) $ from which follows that $ x^2 \leq 2 C R+ C^2 a^2$ since $Cax \leq \frac{1}{2}C^2 a^2+ \frac{1}{2}x^2$. Applying this estimate to \eqref{equ:penlem} we find
\begin{align*}
 \norm{\sqrt{ \tilde h_1} q_1 \psi }^2 \lesssim   (E^\psi- E^\phi)+ \norm{\varphi}_{H^2(\Omega)\cap L^\infty(\Omega)}^2(\beta+\frac{1}{\sqrt{N}}+f(\epsi))
\end{align*}
and equation \eqref{equ:nablaqpsi} yields
\begin{align*}
 \norm{\nabla_1 q_1 \psi }^2  \lesssim (E^\psi- E^\phi)+ \norm{\varphi}_{H^2(\Omega)\cap L^\infty(\Omega)}^2(\beta+\frac{1}{\sqrt{N}}+f(\epsi))
\end{align*}
which is exactly the claim of Lemma\,{\ref{lem:energy}}.
\end{proof}

\subsubsection{Proof of Lemma\,\ref{lem:h1} }
The remaining part of this chapter is devoted to proving Lemma\,\ref{lem:h1}.  
To keep the notation to a minimum we do not write, whenever it does not lead to confusion, the underlying sets of the function spaces and write $\norm{\cdot}$ and $\langle \cdot, \cdot \rangle $ 
for the $L^2$-norm and scalar product on the appropriate set. As an example
\begin{align*}
 \norm{\Phi}= \norm{\Phi}_\LzOf \qquad \norm{\varphi}_{H^2 \cap L^\infty}= \norm{\varphi}_{H^2(\Omega) \cap L^\infty(\Omega)} \qquad \langle \chi,\chi \rangle=\langle \chi,\chi \rangle_{\LzOc}.
\end{align*}
 
%Also we will use the results of Lemma\, \ref{lem:weights} numerous times for the upcoming proofs without explicitly referring to them. 
\begin{proof}[Proof of Lemma\,\ref{lem:h1} ]
The estimate of $ \llangle \psi, (1-p_1p_2 )\tilde h_1 (1-p_1p_2) \psi \rrangle $ is obtained by rewriting the expression in terms of the energy difference $E^\psi-E^\varphi$ 
% \begin{align*}
%  \underbrace{E^\psi-\frac{1}{N} \langle \psi, N\frac{E}{\epsi^2} \psi   \rangle}_{=: \tilde E^\psi}  \underbrace {-E^\varphi +\langle \varphi, \frac{E}{\epsi^2} \varphi \rangle}_{=:\tilde E^\varphi}=
%E^\psi-E^\varphi
% \end{align*}
 and the remaining parts. Since
\begin{align*}
 E^\psi-E^\varphi&= \llangle  \psi ,(p_1p_2+1-p_1p_2) \tilde h_1 (p_1 p_2 +1 -p_1p_2) \psi \rrangle \\
&\quad+\frac{N-1}{2N}\llangle  \psi, (p_1 p_2 +1 -p_1p_2) w^\epsi_{12} (p_1 p_2 +1 -p_1p_2)\psi \rrangle\\
&\quad- \langle \varphi, -\Delta_x- \frac{1}{\epsi^2} (\Delta_y+E) \varphi \rangle - \langle \Phi ,\frac{1}{2} (w^0*|\Phi|^2) \Phi \rangle
\end{align*}
After expanding the terms in the first row, isolating the term $ \llangle \psi, (1-p_1p_2 ) h_1 (1-p_1p_2) \psi \rrangle$
% \begin{align*}
%  \llangle \psi, (1-p_1p_2 ) h_1 (1-p_1p_2) \psi \rrangle
% \end{align*}
 and subsequently arranging the terms in a convenient we find %way for the following estimation process:
\begin{align}\label{equ:hp}
& \llangle \psi, (1-p_1p_2 ) \tilde h_1 (1-p_1p_2) \psi \rrangle \notag \\
&= E^\psi- E^\phi \notag \\
&\quad- \llangle \psi, p_1p_2 \tilde h_1 p_1p_2 \psi \rrangle+ \langle \varphi, - \Delta_x - \frac{1}{\epsi^2} (\Delta_y+E) \varphi \rangle \notag \\
&\quad-\llangle \psi, (1-p_1p_2 )\tilde h_1 p_1p_2 \psi \rrangle-\llangle \psi, p_1p_2 \tilde  h_1 (1-p_1p_2) \psi \rrangle \notag\\
&\quad-\frac{N-1}{2N} \llangle  \psi, p_1 p_2 w^\epsi_{12} p_1 p_2 \psi \rrangle + \langle \Phi, \frac{1}{2} (w^0*|\Phi|^2) \Phi \rangle \notag\\
&\quad-\frac{N-1}{2N} \llangle  \psi,(1- p_1 p_2) w^\epsi_{12} p_1 p_2 \psi \rrangle- \frac{N-1}{2N} \llangle  \psi, p_1 p_2 w^\epsi_{12}(1- p_1 p_2) \psi \rrangle \notag \\
&\quad -\frac{N-1}{2N} \llangle  \psi,(1- p_1 p_2) w^\epsi_{12}(1- p_1 p_2) \psi \rrangle .
\end{align}
After estimating the terms line by line we will obtain
 \begin{align*}\label{equ:hpest}
&  \llangle \psi, (1-p_1p_2 ) \tilde h_1 (1-p_1p_2) \psi \rrangle \\
& \lesssim \big( E^\psi- E^\phi \big)\\
&\quad +\norm{\Phi}_{H^1}^2 \beta\\
&\quad+ \norm{ \Phi}_{H^2}(\beta + \frac{1}{\sqrt N})\\
&\quad+(1+\norm{\varphi}_{ L^\infty})^2(\beta+ \frac{1}{N}+f(\epsi))\\
&\quad+ \norm{ \varphi}_{H^1 \cap L^\infty}^2 (\beta+\frac{1}{N})+(1+ \norm{\varphi}_{L^\infty} ) \sqrt{\beta+\frac{1}{N}} \norm{\nabla_1 q_1 \psi} \numberthis.
\end{align*}
A finale simplification leads to the claimed result
\begin{align*}
&\llangle \psi, (1-p_1p_2 ) \tilde h_1 (1-p_1p_2) \psi \rrangle \\
% &\lesssim (E^\psi- E^\phi)+ (1+\norm{\varphi}_\infty+\norm{\varphi}_{H^2})^2(\beta+\frac{1}{\sqrt{N}}+f(\epsi))\\
% &\quad+\norm{\varphi}_{L^2 \cap L^\infty } \sqrt{\beta+\frac{1}{N}} \norm{\nabla_1 q_1 \psi}\\
&\lesssim(E^\psi- E^\phi)+ \norm{\varphi}_{H^2\cap L^\infty}^2(\beta+\frac{1}{\sqrt{N}}+f(\epsi))\\
&\quad+ \norm{\varphi}_{L^2 \cap L^\infty } \sqrt{\beta+\frac{1}{N}} \norm{\nabla_1 q_1 \psi}.
 \end{align*}
We prove the estimates that lead from equation \eqref{equ:hp} to \eqref{equ:hpest} line by line. We do not have to estimate the first line. %There is nothing to do for the first line.

\textit{Line 2}. 
\begin{align*}
 | \langle \varphi, \tilde h_1 \varphi \rangle- \llangle \psi, p_1p_2  \tilde h_1 p_1p_2 \psi \rrangle  |&=
| \langle \varphi, \tilde h_1 \varphi \rangle-  \langle \varphi, \tilde h_1 \varphi \rangle \llangle \psi, p_1p_2  \psi \rrangle  | \\
&=\langle \varphi, \tilde h_1 \varphi \rangle |\llangle \psi, (1-p_1p_2 )\psi \rrangle|\\
&=\langle \Phi,-\Delta_x \Phi \rangle |\langle \psi, (p_1q_2 +q_1p_2+ q_1q_2) \psi \rangle|\\
&\weq{\eqref{def:alpha}}{\leq}  3 \norm{\Phi}^2_{H^1} \alpha  \stackrel{\eqref{equ:relalphabeta}}{\lesssim} \norm{\Phi}^2_{H^1} \beta 
\end{align*}

\textit{Line 3}.
The term 
\begin{align*} 
-\llangle \psi, (1-p_1p_2 ) \tilde h_1 p_1p_2 \psi \rrangle-\llangle \psi, p_1p_2  \tilde h_1 (1-p_1p_2) \psi \rrangle\
\end{align*}
is bounded in absolute value by 

\begin{align*}
2  |\llangle \psi, (1-p_1p_2 ) \tilde h_1 p_1p_2 \psi \rrangle| &= 2 |\llangle \psi, (q_1+p_1q_2)  \tilde h_1 p_1p_2 \psi \rrangle|\\
&= 2 |\llangle \psi, q_1  \tilde h_1 p_1p_2 \psi \rrangle|\\
&= 2 |\llangle \psi, q_1 \widehat n^{-\frac{1}{2}} \widehat n^\frac{1}{2}   \tilde h_1 p_1p_2 \psi \rrangle|\\
&\stackrel{\ref{lem:weights}}{=} 2 |\llangle \psi, q_1 \widehat n^{-\frac{1}{2}}    \tilde h_1 \widehat {\tau_1 n}^\frac{1}{2} p_1p_2 \psi \rrangle|\\
%&\stackrel{\mathrm{CS}}{\leq} 
&\leq 2 \sqrt{\llangle \psi, \widehat n^{-1} q_1  \psi \rrangle} \sqrt {\llangle \psi,   p_1 p_2 \widehat {\tau_1 n}^\frac{1}{2} \tilde h_1^2 \widehat {\tau_1 n}^\frac{1}{2} p_1 p_2 \psi \rrangle}\\
&\stackrel{\mathclap{ \ref{lem:weights}}}{=} 2\sqrt{\llangle \psi, \widehat n \psi \rrangle} \sqrt{\langle \varphi, \tilde h^2 \varphi \rangle} \sqrt {\llangle \psi,  \widehat {\tau_1 n} \psi \rrangle}\\
&\stackrel{\mathclap{ \eqref{equ:htildep}}}{\leq } 2 \sqrt{\beta} \norm{ \Phi}_{H^2}   \sqrt{\llangle \psi,  \widehat {\tau_1 n}    \psi \rrangle}\\ %+\widehat{\frac{1}{\sqrt{N}}}
&\stackrel{\ref{lem:weights}}{\leq} 2 \sqrt{\beta} \norm{ \Phi}_{H^2}   \sqrt{\beta + \frac{1}{\sqrt N} } \\
&\lesssim  \norm{ \Phi}_{H^2}(\beta + \frac{1}{\sqrt N}) \numberthis \label{equ:line3}.
\end{align*}

\textit{Line 4}.
\begin{align*}
 |\langle &\Phi, \frac{1}{2} (w^0*|\Phi|^2) \Phi \rangle- \frac{N-1}{2N} \llangle  \psi, p_1 p_2 w^\epsi_{12} p_1 p_2 \psi \rrangle  |\\
& \stackrel{\mathclap{ \eqref{equ:pwp}}}{ =} \,  \frac{1}{2} |\langle  \varphi, (w^0*|\varphi|^2) \varphi  \rangle - (1+\frac{1}{N})
 \langle \varphi, (w^\epsi* |\varphi|^2)  \varphi  \rangle \llangle \psi, p_1 p_2 \psi \rrangle|\\
&\leq  \frac{1}{2} |\langle  \varphi, (w^0*|\varphi|^2-w^\epsi* |\varphi|^2)\varphi  \rangle|
 + \frac{1}{2} |\langle \varphi, (w^\epsi* |\varphi|^2)  \varphi  \rangle \llangle \psi,(1- p_1 p_2) \psi \rrangle| \\
&\quad+\frac{1}{2N} \langle \varphi, (w^\epsi* |\varphi|^2)  \varphi  \rangle \llangle \psi, p_1 p_2 \psi \rrangle|\\
& \leq \norm{(w^0*|\varphi|^2-w^\epsi* |\varphi|^2)}_\infty+ \frac{3}{2} \norm{w^\epsi* |\varphi|^2}_\infty (\beta+ \frac{1}{N}) \\
&\stackrel{\mathclap{\ref{cor:op},A1'}}{\lesssim} \norm{\varphi}^2_{L^2 \cap L^\infty}(\beta+ \frac{1}{N}+f(\epsi)) =(1+\norm{\varphi}_{ L^\infty})^2(\beta+ \frac{1}{N}+f(\epsi))
\end{align*}

\textit{Line 5}.
This line is bounded in absolute value by
\begin{align*}\label{equ:energy1q2q}
|\llangle  \psi, p_1 p_2 w^\epsi_{12}(1- p_1 p_2) \psi \rrangle| &= |\llangle  \psi, p_1 p_2 w^\epsi_{12}(q_1p_2+ p_1q_2+ q_1q_2) \psi \rrangle| \\
&\leq 2| \llangle  \psi, p_1 p_2 w^\epsi_{12} q_1p_2 \psi \rrangle|+| \llangle  \psi, p_1 p_2 w^\epsi_{12} q_1q_2 \psi \rrangle |. \numberthis
\end{align*}
The first term is bounded by 
\begin{align*}
 | \llangle  \psi, p_1 p_2 w^\epsi_{12} q_1p_2 \psi \rangle|
%\, &\stackrel{\mathclap{ \eqref{equ:pwp}}} {=}\, | \llangle  \psi, p_1 p_2 (w^\epsi *|\varphi|^2)(r_1) \widehat n^{-\frac{1}{2}} \widehat n^\frac{1}{2}  q_1  \psi \rrangle|\\
&=| \llangle  \psi, p_1 p_2 w^\epsi_{12} \widehat n^{-\frac{1}{2}} \widehat n^\frac{1}{2}   q_1 p_2  \psi \rrangle|\\
%&\stackrel{\mathclap{ \ref{lem:weights}}}{=}\, | \llangle  \psi, p_1 p_2 \widehat {\tau_1 n}^\frac{1}{2} (w^\epsi *|\varphi|^2)(r_1) \widehat n^{-\frac{1}{2}}   q_1  \psi \rrangle|\\
&\stackrel{\mathclap{ \ref{lem:weights}}}{=}\, | \llangle  \psi, p_1 p_2 \widehat {\tau_1 n}^\frac{1}{2} w^\epsi_{12} \widehat n^{-\frac{1}{2}}   q_1 p_2 \psi \rrangle|\\
%&\stackrel{\mathrm{CS}}{ \leq}
&\leq \norm{ p_2 w^\epsi_{12} p_2}_{\mathrm{Op}} \norm{\widehat {\tau_1 n}^\frac{1}{2} \psi} \norm{\widehat n^{-\frac{1}{2}}   q_1 \psi}\\
&\stackrel{\mathclap{ \ref{lem:weights}}}{\leq}  \norm{ p_2 w^\epsi_{12} p_2}_{\mathrm{Op}}\sqrt{\langle \psi, \widehat {\tau_1 n} \psi \rangle} \sqrt{\psi, \widehat n \psi \rangle }\\
& \leq    \norm{ p_2 w^\epsi_{12} p_2}_{\mathrm{Op}}  (\beta+\frac{1}{\sqrt N})\\
& \stackrel{\mathclap{\ref{cor:op}} }{ \lesssim}  (1+\norm{\varphi}_\infty)^2   (\beta+\frac{1}{\sqrt N}).
%&\stackrel{ \mathclap{\ref{lem:young}, A1'} }{\lesssim}  (1+\norm{\varphi}_\infty)^2 (\beta+\frac{1}{\sqrt N}).
\end{align*}
\begin{rem}
 As already mentioned at the end of Section \ref{sec:beta} the estimation of this term leads to the condition \eqref{equ:beta1q} and is thus the main reason why we need to use $\beta$ for stronger singularities.  
\end{rem}
The second term of equation \eqref{equ:energy1q2q} demands a more elaborate proof and is thus treated separately in Lemma\,\ref{lem:intbyparts}.\\
\textit{Line 6}.
If assumption A3' holds the interaction is nonnegative and we obtain
\begin{align*}
 -\frac{N-1}{2N} \llangle  \psi,(1- p_1 p_2) w^\epsi_{12}(1- p_1 p_2) \psi \rrangle \ \leq 0.
\end{align*}
In the case Remark \ref{rem:negpot} holds we can use the appropriate fraction of the kinetic energy from the left-hand side of equation \eqref{equ:hp} to control this term
\begin{align*}
-\frac{N-1}{2N} \llangle  \psi,(1- p_1 p_2) w^\epsi_{12}(1- p_1 p_2) \psi \rrangle -(1-\kappa)\llangle \psi, (1-p_1p_2 ) \tilde h_1 (1-p_1p_2) \psi \rrangle \leq 0.
\end{align*}
The only thing changed by this is the addition of the negligible constant $\kappa^{-1}$ in front of all terms of the right-hand side of equation \eqref{equ:hp}.
\end{proof}
%\textcolor{red}{formulierung}
The following lemmas are necessary to provide the final bound for 
\begin{align*}
 | \llangle  \psi, p_1 p_2 w^\epsi_{12} q_1q_2 \psi \rrangle |.
\end{align*}
%Under the assumption $\frac{w^\epsi}{\sqrt{N}} \in L_2$ $\forall \epsi,N$ the estimate would be easier as the following presentation but
%since we will need the method presented here for varies terms in the Grönwall estimate for $\beta$ we choose to use this estimation as a detailed example
Since we need similar arguments in the estimations of the derivatives of $\beta$ we give a detailed account of the used techniques.  

\begin{lem}[Writing a $L^s$ function as divergence of a vector field]\label{lem:div}

Let $D$ be a domain of $\R^d$ with $d \geq 3$ and smooth boundary, $f \in L^s(D)$ and %$\Gamma(x):= \big((d-2) |\field{S}^{d-1}| \big)^{-1} |x|^{2-d} $ then
\begin{align*}
 \Gamma(x):= \big((d-2) |\field{S}^{d-1}| \big)^{-1} |x|^{2-d},
\end{align*}
then
\begin{align*}
 \xi (x):=  \int_D  - \nabla  \Gamma(x-y)  f(y) \D y 
\end{align*}
     is a well-defined function on $D$ and solves
\begin{align}\label{div} 
\nabla \xi =  f.
\end{align}
Furthermore $ \xi \in W^{1,s}(D)$ and
\begin{align}\label{divab}
\norm{ |\xi| }_{L^q(D)} \leq C(d,p) \norm{f}_{L^s(D)}, 
\end{align}
where $\frac{1}{q}=\frac{1}{s}-\frac{1}{d}$.
 
\end{lem}
\begin{proof}
The fact that $\xi$ is well-defined and equation \eqref{div} follows directly from Poisson's equation for distributions e.g. Theorem 6.21 in \cite{LieLos01}.
The fact that $  \xi \in W^{1,s}(D)$ follows e.g. from Theorem 9.9 and the remark at the end of its proof in \cite{GilTru01}.
Equation \eqref{divab} is a result of the Generalized Young inequality and 
\begin{align*}
 \norm {| \nabla \Gamma |}_{L^r_\mathrm{w}} \leq C \norm{ \frac{1}{|x|^{d-1}} }_{L^r_\mathrm{w}}  \leq C(d) 
\end{align*}
with $r=\frac{d}{d-1}$. Since
\begin{align*}
 \norm{ |\xi| }_{L^q(D)} \leq \norm{C \frac{1}{|x|^{d-1}}*|f|}_{L^q(D)}\leq C \norm{\frac{1}{|x|^{d-1}}}_{L^r_\mathrm{w}(\R^3)} \norm{f}_{L^s(D)} \leq C(d)\norm{f}_{L^s(D)} 
\end{align*}
with $\frac{1}{q}=\frac{1}{s}+\frac{1}{r}-1=\frac{1}{s}-\frac{1}{d}$ for  $1<q,s <\infty $.
\end{proof}

\begin{cor}\label{cor:xi}
 Let $A1'$ hold for $w^{\epsi,s}$  and define $\xi $ as the vector field from Lemma\,\ref{lem:div}
\begin{align*}
 \xi(r) :=  \int_{\tilde \Omega}  - \nabla  \Gamma(r-r_1)  w^{\epsi,s}(r_1) \D r_1,
\end{align*}
then
\begin{align}\label{equ:corxi}
 \norm{|\varphi|^2* \xi^2 }_{L^\infty(\R^3)} \lesssim  (1+ \norm{\varphi}_{ L^\infty(\Omega) })^2.
\end{align}

\end{cor}
\begin{proof}
 
\begin{align*}
 \norm{|\varphi|^2* \xi^2 }_{L^\infty(\Omega)} & \stackrel{\mathclap{\ref{cor:op}}}{ \lesssim} \norm{\xi^2}_{L^r(\tilde \Omega)} (1+ \norm{\varphi}_{\LiO })^2 \leq \norm{ |\xi| }^2_{L^q(\Omega)}(1+ \norm{\varphi}_{ \LiO })^2 \\
&\weq{\eqref{divab}}{ \lesssim} \norm{w^{\epsi,s}}_{L^p(\Omega)}(1+ \norm{\varphi}_{\LiO })^2\\
& \weq{A1'}{\lesssim} (1+ \norm{\varphi}^2_{ \LiO })
\end{align*}
with $1/(2r)=1/q=1/s-1/3 $.%, where we used interpolation and $L^2$-norm conservation for $\varphi$.
\end{proof}

Now we can estimate the second term of \eqref{equ:energy1q2q}. This is done in the next lemma.

\begin{lem}\label{lem:intbyparts}
 \begin{align*}
 | \langle  \psi, p_1 p_2 w^\epsi_{12} q_1q_2 \psi \rangle | &%\leq 2\sqrt{2} \norm{ |\xi| }_q \norm{\varphi}_{L^2 \cap L^\infty } \sqrt{\alpha+\frac{2}{N}} \norm{\nabla_1 q_1 \psi}\\
% &+\sqrt{2} \norm{ w^{\epsi,\infty}_{12} }_\infty (\alpha+ \frac{2}{N} )+ \sqrt{ 2 }\norm{ |\xi| }_q \norm{\varphi}_{L^2 \cap L^\infty }\norm{\nabla \varphi}_{L^2} (\alpha+\frac{2}{N} )\\
% &\leq C (\alpha+\frac{2}{N} )+ C  \sqrt{\alpha+\frac{2}{N}} \norm{\nabla_1 q_1 \psi}
\lesssim \norm{ \varphi}_{H^1 \cap L^\infty} (\beta+\frac{1}{N})+(1+ \norm{\varphi}_{L^\infty} ) \sqrt{\beta+\frac{1}{N}} \norm{\nabla_1 q_1 \psi}
 \end{align*}

\end{lem}

\begin{proof}
 First we write $w^\epsi= w^{\epsi,\infty}+w^{\epsi,s} $. This splitting gives two terms
\begin{align}\label{equ:enppwqq}
 | \llangle  \psi, p_1 p_2 w^\epsi_{12} q_1q_2 \psi \rrangle | \leq | \llangle  \psi, p_1 p_2 w^{\epsi,\infty}_{12} q_1q_2 \psi \rrangle |+ | \llangle  \psi, p_1 p_2 w^{\epsi,s}_{12} q_1q_2 \psi \rrangle |,
 \end{align}
 where the first can be estimated directly
\begin{align*}
  | \llangle  \psi, p_1 p_2 w^{\epsi,\infty}_{12} q_1q_2 \psi \rrangle | &= | \llangle  \psi, p_1 p_2 w^{\epsi,\infty}_{12} \widehat n \widehat n^{-1} q_1q_2 \psi \rrangle |\\
& \stackrel{\mathclap{ \ref{lem:weights}}}{=}| \llangle  \psi, p_1 p_2 \widehat {\tau_2 n} w^{\epsi,\infty}_{12}  \widehat n^{-1} q_1q_2 \psi \rrangle |\\
&\stackrel{}{\leq} \norm{ w^{\epsi,\infty}_{12} }_\infty \sqrt{\llangle \psi, p_1 p_2 \widehat {\tau_2 n}^2 \psi \rrangle  } \sqrt{\llangle \psi, \widehat n^{-2}  q_1 q_2 \psi \rrangle}\\
&\stackrel{\mathclap{\ref{lem:weights}}}{\leq}  \norm{ w^{\epsi,\infty}_{12} }_\infty \sqrt{\alpha+ \frac{2}{N} }\sqrt{\frac{N}{N-1} \alpha }\\
&\lesssim %\norm{ w^{\epsi,\infty}_{12} }_\infty
 \beta+ \frac{1}{N} .
\end{align*}
For the second term of \eqref{equ:enppwqq} we use Lemma\,\ref{lem:div} and write $w^{\epsi}_s $ as the divergence of a vector field $\xi^\epsi$ and estimate this with the help of Corollary\,\ref{cor:xi}.
%As expected the singularities of $\xi^\epsi$ have higher regularity as the ones of $w^{\epsi,p} $ and we will use this regularity in the estimation of
%
% This is however solved by \textcolor{red}{Konstanten!}
% \begin{align*}
%  \xi^\epsi= \frac{d \, \omega_d}{2} \frac{x}{|x|^d} *  w^{\epsi,p}
% \end{align*}
% since for Dimension $2$ or bigger with $\Gamma$ the fundamental solution of the Laplace equation 
% \begin{align*}
%  \nabla \Gamma (x)= \frac{2}{d \, \omega_d} \frac{x}{|x|^d}   
% \end{align*}
% and hence
% \begin{align*}
%  \nabla  \xi^\epsi= \Delta  \Gamma (x) *  w^{\epsi,p}= w^{\epsi,p} 
% \end{align*}
% This is a helpfull trick since with the generaliesd Young in equality
% \begin{align*}
% 1 <p,q,r < \infty  \norm{f*g}_q \leq C_{p,r} \norm{f}_r \norm{g}_{p,w}  
% \end{align*}
% and 
% \begin{align*}
%  \norm{|x|^{-d+1}}_{\frac{d}{d-1},w}= C(d)
% \end{align*}
% we have 
% \begin{align*}
%  \norm{ |\xi| }_q \leq C(d,p) \norm{w}_p
% \end{align*}
% with $\frac{1}{q}=\frac{1}{p}-\frac{1}{d}$.
%Later we will use this for
% \begin{align*}
%  \norm{|\varphi|^2* \xi^2 }_\infty &\leq \norm{\xi^2}_r \norm{\varphi}_{L^2 \cap L^\infty } \leq \norm{ |\xi| }^2_q(1+ \norm{\varphi}^2_{ L^\infty }) 
% \lesssim \norm{w^{\epsi,p}}_{L^p(\Omega)}(1+ \norm{\varphi}^2_{ L^\infty })\\
% & \lesssim (1+ \norm{\varphi}^2_{ L^\infty })
% \end{align*}
% with $1\leq r \leq \infty$ and $ 2 \leq  q \leq  \infty $, where we used interpolation and $L^2$-norm conservation for $\varphi$.
%
%
%Now we can estimate the final terms, where we suppress the $\epsi$-dependents of $\xi$ for better readability,
For the following estimates we  suppress the $\epsi$-dependents % of $\xi$ 
for better readability.
\begin{align*}
  | \llangle  \psi, p_1 p_2 w^{s}_{12} q_1q_2 \psi \rrangle | &= | \llangle  \psi, p_1 p_2 w^{s}_{12}  \widehat n \widehat n^{-1}  q_1q_2 \psi \rrangle |\\
  &\stackrel{\ref{lem:weights}}{=}| \llangle  \psi, p_1 p_2 \widehat {\tau_2 n} w^{s}_{12}   \widehat n^{-1}  q_1q_2 \psi \rrangle |\\
&\stackrel{\ref{lem:div}}{=}| \llangle  \psi, p_1 p_2 \widehat {\tau_2 n} (\nabla_1^\nu \xi^\nu )_{12}  \widehat n^{-1}  q_1q_2 \psi \rrangle |,
\end{align*}
where we sum over $\nu =1,2,3$.
Now we integrate by parts which is possible since $\xi^\epsi \in W^{1,s}(\tilde \Omega)$, 
 $ p_1 p_2 \widehat {\tau_2 n} \psi  \in H^1_0(\Omega) $ and $ \widehat n^{-1}  q_1q_2 \psi  \in H^1_0(\Omega)$. This also implies that there are no boundary terms. 
\begin{align}\label{equ:ibp}
  \llangle  \psi, p_1 p_2 \widehat {\tau_2 n} (\nabla_1^\nu \xi^\nu )_{12}  \widehat n^{-1}  q_1q_2 \psi \rrangle |&\leq 
|\llangle \nabla_1^\nu   p_1 p_2 \widehat {\tau_2 n} \psi,   \xi^\nu_{12}  \widehat n^{-1}  q_1q_2 \psi \rrangle |\nonumber \\
&\quad+ |  \llangle  p_1 p_2 \widehat {\tau_2 n} \psi,   \xi^\nu_{12} \nabla_1^\nu \widehat n^{-1}  q_1q_2 \psi \rrangle |
\end{align}
%
%
%
%If we remember that $p_1= | \varphi \rangle \langle \varphi  |$ the first term can be estimated by
The first term can be estimated by
\begin{align}\label{equ:2p2qpart1}
 |\llangle  \xi^\nu_{12} (\nabla_1^\nu   p_1) p_2 \widehat {\tau_2 n} \psi,    \widehat n^{-1}  q_1q_2 \psi \rrangle | 
&\leq {\llangle \underbrace{ (\nabla_1^\nu   p_1) \widehat {\tau_2 n} \psi}_{\eta}, \underbrace{p_2   \xi^\nu_{12} \xi^{\mu}_{12}p_2}_{A} \underbrace{ (\nabla_1^\mu   p_1)  
 \widehat {\tau_2 n} \psi}_{\eta} \rrangle}^\frac{1}{2} \nonumber \\
& \quad \times \norm{ \widehat n^{-1}  q_1q_2 \psi }.
% & \leq \norm{(\nabla_1^\nu   p_1) \widehat {\tau_2 n} \psi}\norm{(\nabla_1^\mu   p_1) \widehat {\tau_2 n} \psi} \norm{ \xi^\nu_{12} \xi^{\mu}_{12}* |\varphi|^2 }\\
% &\leq \sqrt{\langle  \psi   \rangle } \norm{\nabla \varphi} \sqrt{\langle   \psi, \widehat {\tau_2 n} \psi \rangle} 
\end{align}

A formal way to deal with this is to define $\mathscr{F}:=L^2(\R^{3N}) \oplus L^2(\R^{3N}) \oplus L^2(\R^{3N})$. So the first part is 
\begin{align*}
 \langle \eta, A \eta  \rangle_{\mathscr{F}}^\frac{1}{2}\leq \norm{ \eta }_\mathscr{F} \norm{A}_{\mathcal{L}(\mathscr{F})}^\frac{1}{2}.
\end{align*}
% Since A is a semidefinit matrix and only acts on one $L^2(\R^3)$ of $L^2(\R^{3N})$ we have
%  \begin{align*}
% \norm{A}_{\mathcal{L}(\mathscr{H})} \leq \norm{ \Tr A}_{\mathcal{L}(L^2(\R^3))}= \norm{ |\varphi|^2 * \xi^2}_\infty 
%  \end{align*}
%
Here  
\begin{align}\label{equ:normA}
 \norm{A}_{\mathcal{L}(\mathscr{F})}= \norm{( |\varphi|^2*\xi^2)(r_1)}_\infty
\end{align}
since an operator of the form $v v^t$, where $v$ is a vector, has the operator norm $v^2$ and the entries of $A$ are $ |\varphi|^2*\xi_i \xi_j$.
The vector $\eta$ has the norm 
\begin{align*}
 \norm{\eta}_\mathscr{F}^2&= \sum_{\mu=1}^3 \llangle \eta_\mu, \eta_\mu \rrangle\\
& = \sum_{\mu=1}^3 \int_\Omega \overline{ \nabla^\mu \varphi(r)} \nabla^\mu \varphi(r)\, \D r \, \llangle \widehat {\tau_2 n} \psi, p_1 \widehat {\tau_2 n} \psi  \rrangle \\ %_{L^2(\R^{3N})}\\
&=\norm{\nabla \varphi}^2 \norm{\widehat{\tau_2 n}\psi}^2  \numberthis \label{equ:normeta}.
\end{align*}
The right-hand side of \eqref{equ:2p2qpart1} can be bounded with the help of equation \eqref{equ:normA} and  \eqref{equ:normeta} by
\begin{align*}
 \langle \eta, A \eta  \rangle_{\mathscr{F}}^\frac{1}{2} \norm{ \widehat n^{-1}  q_1q_2 \psi} &\leq
 \norm{ |\varphi|^2 * \xi^2}^\frac{1}{2}_\infty  \norm{\nabla \varphi}_{L^2} \norm{\widehat{\tau_2 n}\psi}  \norm{ \widehat n^{-1}  q_1q_2 \psi }\\
&\stackrel{\ref{lem:weights} }{\lesssim} \norm{ |\varphi|^2 * \xi^2}^\frac{1}{2}_\infty  \norm{\nabla \varphi}_{L^2}  \sqrt{\alpha+\frac{2}{N}} \sqrt{\alpha}\\
&\stackrel{\mathclap{ \eqref{equ:corxi}}}{\lesssim} \norm{ \varphi}_{H^1 \cap L^\infty} (\beta+\frac{1}{N}) .
\end{align*}
Now there is only the second term of \eqref{equ:ibp} left to be estimated. We claim that
\begin{align} \label{equ:2p2qpart2}
 |  \llangle  p_1 p_2 \widehat {\tau_2 n} \psi,   \xi^\nu_{12} \nabla_1^\nu \widehat n^{-1}  q_1q_2 \psi \rrangle |  
&= |   \llangle  \underbrace{ \xi^\nu_{12}  p_1 p_2 \widehat {\tau_2 n} \psi}_{\eta},  \underbrace{  \nabla_1^\nu \widehat n^{-1}  q_1q_2 \psi }_{\kappa}\rrangle | \notag\\
&\leq \norm{\eta} \norm{\kappa} \notag \\
&\lesssim   (1+ \norm{\varphi}_{L^\infty} ) \sqrt{\alpha+\frac{1}{N}} \norm{\nabla_1 q_1 \psi}.
\end{align}
This holds since 
\begin{align*}
 \norm{\eta}^2 &= \llangle  \widehat {\tau_2 n} \psi , p_1 p_2 \xi^2_{12} p_1 p_2   \widehat {\tau_2 n} \psi \rrangle \\
&\stackrel{\mathclap{ \eqref{equ:pwp}}}{=} \llangle  \widehat {\tau_2 n} \psi , p_1 ( |\varphi|^2* \xi^2_{})(x_1) p_1  p_2   \widehat {\tau_2 n} \psi \rrangle \\
&\leq  \norm{ |\varphi|^2 * \xi^2}_\infty \norm{\widehat {\tau_2 n} \psi}^2 \\
&\stackrel{\mathclap{ \eqref{equ:corxi}}}{\lesssim} (1+ \norm{\varphi}_{L^\infty} )^2(\alpha+\frac{1}{N} ) \label{equ:eta}\numberthis
\end{align*}
and $\kappa$ is estimated by introducing $1=p_1+q_1$ to use Lemma\,\ref{lem:weights}. We only present the calculation for $p_1 \kappa$; $q_1 \kappa$ follows in the same manner. 
\begin{align*}
\norm{p_1 \kappa}^2 &= \norm{ p_1 \nabla_1 \widehat n^{-1}  q_1q_2 \psi}^2 \stackrel{\ref{lem:weights}}{=} \norm{ p_1 q_2 \widehat{ \tau_1 n^{-1}} \nabla_1 q_1 \psi}^2\\
&\leq \llangle \nabla^\nu_1 q_1 \psi,  q_2 \widehat{ \tau_1 n^{-2}} \nabla^\nu_1 q_1 \psi \rrangle \\
&= \llangle \nabla^\nu_1 q_1 \psi,  \frac{1}{N-1}\sum_{i=2}^N q_i \widehat{ \tau_1 n^{-2}} \nabla^\nu_1 q_1 \psi \rrangle\\
&\leq \llangle \nabla^\nu_1 q_1 \psi,  \frac{1}{N-1}\sum_{i=1}^N q_i \widehat{ \tau_1 n^{-2}} \nabla^\nu_1 q_1 \psi \rrangle\\
&\stackrel{\mathclap{ \ref{lem:weights}}}{=}\llangle \nabla^\nu_1 q_1 \psi,  \frac{N}{N-1} \widehat {n^2} \widehat{ \tau_1 n^{-2}} \nabla^\nu_1 q_1 \psi \rrangle\\
& \lesssim  \llangle \nabla^\nu_1 q_1 \psi,   \nabla^\nu_1 q_1 \psi \rrangle\\
&\lesssim  \norm{\nabla_1 q_1 \psi}^2 \numberthis \label{equ:kappa}
\end{align*}

\end{proof}

\section{Controlling the Derivative of $\beta$}

%The idea stays the same as in the proof of Theorem\, \ref{thm:thm1} thus we begin with the same steps for the proof.
We use the same basic idea as for the  proof of Theorem\,\ref{thm:thm1}. We start again by calculating the
derivative of the functional.
\begin{lem}\label{beta.}
 \begin{align*}
  |\frac{d}{dt} \beta| \leq 2 |\mathrm{I}| + 2 |\mathrm{II}| + |\mathrm{III}|,
 \end{align*}
where 
\begin{align*}
 \mathrm{I}&:=\llangle \psi, p_1 p_2  [ (N-1)  w_{12}^\epsi-N w_1^\varphi- N w_2^\varphi , \widehat n] p_1q_2 \psi \rrangle\\
\mathrm{II}&:=\llangle \psi, p_1 q_2  [ (N-1)  w_{12}^\epsi-N w_1^\varphi- N w_2^\varphi , \widehat n] q_1q_2 \psi \rrangle\\
\mathrm{III}&:=  \llangle \psi, p_1 p_2  [ (N-1)  w_{12}^\epsi-N w_1^\varphi- N w_2^\varphi , \widehat n] q_1q_2 \psi \rrangle
\end{align*}
with $w^\varphi_i:=( w^0*|\varphi|^2)(r_i) = (w^0*(|\Phi|^2 |\chi|^2))(r_i)= ( w^0*|\Phi|^2)(x_i)  $.
\end{lem}
Since we now use the weight $\sqrt{\frac{k}{N}}$ in contrast to $\frac{k}{N}$ the terms $\mathrm{I}-\mathrm{III}$ now look a bit different than in Theorem\,\ref{thm:thm1} but are essentially the same.
It is more important we can estimate them in a similar way.

\begin{lem}\label{lem:estimating3}

\begin{enumerate}
 \item \label{lem:estimating3I}

 \begin{align} \label{equ:estimating3I}
  |\mathrm{I}| \lesssim (f(\epsi)+\frac{1}{N})\norm{\varphi}_{L^2(\Omega) \cap L^\infty(\Omega)}^2
 \end{align}

\item \label{lem:estimating3II}
\begin{align}\label{equ:estimating3II}
 |\mathrm{II}|& \lesssim \norm{\varphi}_{L^2(\Omega) \cap L^\infty(\Omega) } (\beta+\norm{ \nabla_1 q_1 \psi}^2)
\end{align}
 
\item  \label{lem:estimating3III}
\begin{align} \label{equ:estimating3III}
  |\mathrm{III}|&\lesssim\norm{ \varphi}_{H^1(\Omega) \cap L^\infty (\Omega)}^3(\beta+N^\eta)+ \norm{\varphi}_{\infty }\norm{\nabla_1 q_1 \psi}^2,
%\leq \sqrt{2} \norm{w_{12}^\infty}_\infty  (\beta+ \sqrt{\frac{1}{N}})  \\
% &C  \norm{\varphi}_{\infty }\big(\norm{\nabla_1 \varphi} \beta+ \norm{\nabla_1 \varphi}a^{2-\frac{2p}{p_0}}+\sqrt{2}\norm{\nabla_1 q_1 \psi}^2\big)\\
% &+ N^{-\frac{\delta}{2}}\norm{\varphi}^2_{L^\infty \cap  L^2} \norm{w^p}_p +  a^{1-\frac{p}{2}} N^{-\frac{1}{2}} \norm{\varphi}_\infty \norm{w^p}_p^{p}\\
% &+N^{-1+\frac{\delta}{2}} a^{2-p} \norm{w^p}^p_p \norm{\varphi}^2_\infty + 2 \norm{\varphi}^2_{L^\infty \cap  L^2} \norm{w^p}_p \beta
\end{align}
where $\eta=-\frac{s/s_0-1}{2s/s_0-s}=-\frac{5s-6}{4s}  $.
\end{enumerate}
\end{lem}

% 
% \begin{thm}
% Under the assumtions of Theorem \ref{thm:thm1} and let $f(\epsi,N): (0,1] \times N \rightarrow \R^+ $ with $\lim_{N\rightarrow \infty\\ \epsi \rightarrow 0} f(\epsi,N) =0  $ such that
%    $\frac{ 1}{N} \norm{w^\epsi}_{L^v(\Omega} \leq \tilde f(\epsi,N)$ for a $v \in [2,\infty]$.
%  holds then
% \begin{align*}
% \beta(t) \leq \exp(C  h(t))  (\beta_0 + E^\psi_N-E^\varphi+f(\epsi)+N^{-1/2} \tilde  f(\epsi)^2+\frac{1}{\sqrt{N}})
% \end{align*}
% % and thus 
% % \begin{align*}
% %  E_N^{(k)}\leq k  {\scriptstyle \mathcal{O}}(N,\epsi) {(\E^{a(t)}-1)} \qquad R_N^{(k)}\leq \sqrt{k  {\scriptstyle \mathcal{O}}(N,\epsi) }\; (\E^{a(t)}-1)^\frac{1}{2}
% % \end{align*}
% 
% \end{thm}

The estimation of the term $\mathrm{III}$ is the most laborious out of the three terms. It can be shortened substantially if an additional assumption on $w^\epsi$ holds.   

\begin{lem} \label{lem:estimating3alt}
Let 
\begin{align}\label{equ:ass+w}
{ \norm{w^\epsi}_{L^v(\tilde \Omega)}} \leq \tilde f(\epsi) 
\end{align}
 for a $v \in (2,\infty) $, then
\begin{align}\label{equ:estimating3alt}
 |\mathrm{III}| \lesssim (1+\norm{\varphi}_{L^\infty(\Omega)})^2(\beta+ N^{-1/ 2} \tilde f(\epsi)^2+N^{-1/2}).
\end{align}
\begin{rem}
Lemma\,\ref{lem:estimating3alt} is only meaningful if $\epsi$ can be chosen to depend on $N$ such that $N^{-1/ 2} \tilde f(\epsi(N))^2 \stackrel{N \rightarrow \infty}{\longrightarrow }0$.
Although the rate of convergence of the estimate \eqref{equ:estimating3alt} is always equal or slower than the rate in \eqref{equ:estimating3III} we state its proof since it is a byproduct of the proof
of Lemma\,\ref{lem:estimating3}.\ref{lem:estimating3III} and it illustrates the used techniques nicely.
 
\end{rem}

\end{lem}

\begin{exmp}\label{exp:coul2}
For the Coulomb interaction with confinement in one direction the additional assumption required for Lemma\,\ref{equ:estimating3alt} holds if $\epsi$ is chosen to depend on $N$ as any negative power 
since 
\begin{align*}
 \norm{\frac{1}{\sqrt{x^2+\epsi^2 y^2}}}_{L^2(\tilde \Omega)+L^\infty(\tilde \Omega) } \lesssim \log{\epsi^{-1}}.
\end{align*}
For the calculation of this rate see Appendix\,\ref{app:coul}.
\end{exmp}
\begin{proof}[Proof of Theorem\,\ref{thm:thm2}]
 If we combine Lemma\,\ref{lem:estimating3} with the Energy Lemma we can bound $\beta$ by 
\begin{align*}
 \dot \beta \lesssim  \norm{\varphi}_{H^2(\Omega)\cap L^\infty(\Omega)}^3 \big( E^\psi- E^\phi+ \beta+ f(\epsi)+N^\eta \big).
\end{align*}
Now the Grönwall Lemma\,\ref{lem:gron} yields the claimed result.
\end{proof}
For the rest of this chapter we do not write the underlying sets of the function spaces and write $\norm{\cdot}$ and $\langle \cdot, \cdot \rangle $ 
for the $L^2$-norm and scalar product on the appropriate set.
\begin{proof}[Proof of Lemma\,\ref{beta.}]
Because of Lemma\,\ref{hat.} $\beta \in C^1(\R,\R)$. Thus we can calculate
 \begin{align*}
  \partial_t \beta&= \partial_t \llangle \psi, \widehat n \psi \rrangle \stackrel{\ref{hat.}}{ =} \im \llangle \psi, [H^\epsi_N-H^\Phi, \widehat n] \psi \rrangle\\
&\stackrel{\mathclap{\ref{hat.}}}{ =}  \im \llangle \psi,  [ \frac{1}{N}\sum_{i< j} w^\epsi_{ij}-\sum_{i} w_i^\varphi, \widehat n] \psi \rrangle\\
& = \frac{\im}{2} \llangle \psi,  [ (N-1)  w^\epsi_{12}-N w_1^\varphi- N w_2^\varphi , \widehat n] \psi \rrangle\\
&= \frac{\im}{2} \llangle \psi,(p_1+q_1)(p_2+q_2)  [ (N-1)  w^\epsi_{12}-N w_1^\varphi- N w_2^\varphi , \widehat n](p_1+q_1)(p_2+q_2) \psi \rrangle.
\end{align*}
As a result of the Lemma\,\ref{lem:weights}(\ref{lem:weightsc}) all terms with the same number of $p$ and $q$ on each side of the commutator vanish. Therefore we find %, where we used symmetry from the second to the third block. 
\begin{align*}
 &\frac{\im}{2} \llangle \psi,(p_1+q_1)(p_2+q_2)  [ (N-1)  w^\epsi_{12}-N w_1^\varphi- N w_2^\varphi , \widehat n](p_1+q_1)(p_2+q_2) \psi \rrangle \\
\\
&=\frac{\im}{2} \llangle \psi, p_1 p_2  [ (N-1)  w^\epsi_{12}-N w_1^\varphi- N w_2^\varphi , \widehat n] (p_1q_2+q_1p_2 +q_1q_2) \psi \rrangle\\
&\quad+\frac{\im}{2} \llangle \psi, (p_1 q_2+ q_1 p_2)  [ (N-1)  w^\epsi_{12}-N w_1^\varphi- N w_2^\varphi , \widehat n](p_1p_2+q_1q_2) \psi \rrangle\\
&\quad+\frac{\im}{2} \llangle \psi, q_1 q_2  [ (N-1)  w^\epsi_{12}-N w_1^\varphi- N w_2^\varphi , \widehat n](p_1p_2+ q_1p_2 + p_1q_2) \psi \rrangle\\
\\
&\weq{\mathrm{sym.}}{=}\,\im \llangle \psi, p_1 p_2  [ (N-1)  w^\epsi_{12}-N w_1^\varphi- N w_2^\varphi , \widehat n] p_1q_2 \psi \rrangle + c.c.\\
&\quad+\im \llangle \psi, p_1 q_2  [ (N-1)  w^\epsi_{12}-N w_1^\varphi- N w_2^\varphi , \widehat n] q_1q_2 \psi \rrangle + c.c.\\
&\quad+\frac{\im}{2}  \llangle \psi, p_1 p_2  [ (N-1)  w^\epsi_{12}-N w_1^\varphi- N w_2^\varphi , \widehat n] q_1q_2 \psi \rrangle +c.c. \\
\\
&= \im \mathrm I+c.c.+ \im \mathrm {II}+c.c.+ \frac{\im}{2} \mathrm {III}+c.c.\\
\\
&= - 2 \Im \mathrm I- 2 \Im  \mathrm {II} -  \Im  \mathrm {III}.
\end{align*}

\end{proof}

%{As before we will use Lemma\,\ref{lem:weights} and Lemma\,\ref{lem:young} numerous times with out referring to them explicitly in the following proofs.  }
\begin{proof}[Proof of Lemma\,\ref{lem:estimating3}.\ref{lem:estimating3I}]
In this term the mean filed cancels the full interaction
\begin{align*}
 |\mathrm I|&= |\llangle \psi, p_1 p_2  [ (N-1)  w^\epsi_{12}- N w_2^\varphi , \widehat n] p_1q_2 \psi \rrangle|\\
&\stackrel{\mathclap{ \eqref{equ:pwp}}}{= } \, |\llangle  \psi, p_1 p_2 [(N-1)(w^\epsi*|\varphi|^2)(r_2)-N (w^0*|\varphi|^2)(r_2), \widehat  n ]q_2   \psi \rrangle  |\\
&\weq{\ref{lem:weights}} {=} |\llangle  \psi, p_1 p_2 (N-1)\Big((w^\epsi*|\varphi|^2)(r_2)-N (w^0*|\varphi|^2)(r_2)\Big) (\widehat  n- \widehat  {\tau_{-1} n}  ) q_2   \psi \rrangle  |.
\end{align*}
If we define
\begin{align} \label{equ:mu}
 \mu:= N(  n-  {\tau_{-1} n} ) = \sqrt{N}(\sqrt{k}-\sqrt{k-1})= \frac{\sqrt{N}}{\sqrt{k}+\sqrt{k-1}}\leq \frac{\sqrt{N}}{\sqrt{k}}= n^{-1} %\quad \forall k \geq 1
\end{align}
we can write $|\mathrm I|$ as
\begin{align*}
 \frac{1}{N} &|\llangle  \psi, p_1 p_2 (N-1)\Big( (w^\epsi*|\varphi|^2)(r_2)-N (w^0*|\varphi|^2)(r_2) \Big) \widehat \mu  q_2   \psi \rrangle  |\\
&\leq \big (\norm{(w^\epsi*|\varphi|^2 -w^0*|\varphi|^2)}_\infty+ \frac{1}{N} \norm{w^\epsi*|\varphi|^2}_\infty  \big) \sqrt{\llangle \psi, \widehat \mu^2 q_2 \psi \rrangle}\\
&\leq \big (\norm{(w^\epsi*|\varphi|^2 -w^0*|\varphi|^2)}_\infty+ \frac{1}{N} \norm{w^\epsi*|\varphi|^2}_\infty  \big) \sqrt{\llangle \psi, n^{-2} q_2 \psi \rrangle}\\
&\stackrel{\mathclap{\ref{cor:op} ,A1'}}{ \lesssim }(f(\epsi)+\frac{1}{N})\norm{\varphi}_{L^2 \cap L^\infty}^2.
\end{align*}
%
%\textcolor{red}{vielleicht eine Bemerkung wie die Abschätzung hier geht! vielleicht schon vorher }
%
\end{proof}

\begin{proof}[Proof of Lemma\,\ref{lem:estimating3}.\ref{lem:estimating3II}]
 \begin{align*}\label{equ:p1II.1}
|\mathrm{II}|&=  |\llangle \psi, p_1 q_2  [ (N-1)  w^\epsi_{12}-N w_1^\varphi- N w_2^\varphi , \widehat n] q_1q_2 \psi \rrangle| \\
&=|\llangle \psi, p_1 q_2  [ (N-1)  w^\epsi_{12}-N w_1^\varphi, \widehat n] q_1q_2 \psi \rrangle|\\
&\leq |\llangle \psi, p_1 q_2  [ (N-1)  w^\epsi_{12}, \widehat n] q_1q_2 \psi \rrangle|+|\llangle \psi, p_1 q_2 [ N w_1^\varphi, \widehat n] q_1q_2 \psi \rrangle| \numberthis
 \end{align*}
The second term of \eqref{equ:p1II.1} can be estimated by
\begin{align*}\label{equ:p1II.4}
 |\llangle \psi, p_1 q_2 [ N w_1^\varphi, \widehat n] q_1q_2 \psi \rrangle|&\weq{\ref{lem:weights}}{ =}||\llangle \psi, p_1   w_1^\varphi N(\widehat n- \widehat{\tau_{-1}n} )  q_1 q_2 \psi \rrangle| \\
&\weq{\ref{equ:mu}}{ =}|\llangle q_2 \psi, p_1     w_1^\varphi  \widehat \mu  q_1 q_2 \psi \rrangle|\\
&\leq \norm{q_2 \psi} \norm{w^\varphi}_\infty \sqrt{\llangle \psi, \mu^2 q_1 q_2 \psi \rrangle}\\
&\weq{\ref{lem:weights},\ref{cor:op}}{ \lesssim}  \norm{\varphi}_{L^2 \cap L^\infty} \alpha \numberthis .
\end{align*}
The first term of \eqref{equ:p1II.1}  is controlled by 
\begin{align*}\label{equ:p1II.2}
 |\llangle \psi, p_1 q_2  [ (N-1)  w^\epsi_{12}, \widehat n] q_1q_2 \psi \rrangle| &\leq  |\llangle \psi, p_1 q_2   w^\epsi_{12} \widehat{ \mu}  q_1q_2 \psi \rrangle|\\
&=|\llangle \psi, p_1 q_2  ( w^{\epsi,p}_{12}+ w^{\epsi, \infty}_{12})  \widehat{ \mu}  q_1q_2 \psi \rrangle|\\
&\leq |\llangle \psi, p_1 q_2  w^{\epsi,p}_{12}   \widehat{ \mu}  q_1q_2 \psi \rrangle|+|\llangle \psi, p_1 q_2  w^{\epsi, \infty}_{12}  \widehat{ \mu}  q_1q_2 \psi \rrangle| \numberthis.
\end{align*}
The second summand of \eqref{equ:p1II.2} can be estimated by 
\begin{align}\label{equ:p1II.5}
|\llangle \psi, p_1 q_2  w^{\epsi,\infty}_{12}  \widehat{ \mu}  q_1q_2 \psi \rrangle| \weq{\ref{lem:weights}}{ \leq} 2 \norm{ w^{\epsi,\infty}}_\infty \alpha \weq{A1'}{\lesssim} \alpha.
\end{align}
For the first summand of \eqref{equ:p1II.2} we use the idea of writing $w$ as a divergence of a vector field as introduced in Lemma\,\ref{lem:div} and estimate the terms as in Lemma\,\ref{lem:intbyparts}. 
To be in the exact same setting as in Lemma\,\ref{lem:div} we changed the labeling of particle $1$ and $2$.
\begin{align*}\label{equ:p1II.3}
 |\llangle \psi, p_2 q_1  w^{\epsi,s}_{12}   \widehat{ \mu}  q_1q_2 \psi \rrangle|  & \weq{\ref{lem:div}}{=} |\llangle \psi, p_2 q_1  (\nabla_1^\nu \xi^\nu)_{12}   \widehat{ \mu}  q_1q_2 \psi \rrangle|\\
&\weq{\ref{equ:ibp}}{ \leq }|\llangle  \nabla_1^\nu q_1 \psi, p_2    \xi^\nu_{12}   \widehat{ \mu}  q_1q_2 \psi \rrangle| +|\llangle \psi, p_2 q_1   \xi^\nu_{12}   \nabla_1^\nu \widehat{ \mu}  q_1q_2 \psi \rrangle| \numberthis
\end{align*}
The first term of the sum \eqref{equ:p1II.3} is smaller than 
\begin{align*}\label{equ:p1II.6}
 |\llangle     \xi^\nu_{12} p_2 \nabla_1^\nu q_1 \psi,    \widehat{ \mu}  q_1q_2 \psi \rrangle| &\leq  \sqrt{ \llangle \nabla_1^\nu q_1 \psi, p_2\xi^\nu_{12} \xi^\iota_{12}  p_2  \nabla_1^\iota q_1 \psi \rrangle} 
\sqrt{\llangle \psi, \mu^2 q_1q_2  \psi \rrangle}\\
&\weq{\eqref{equ:normA}} {\leq}    \norm{\nabla_1^\nu q_1 \psi} \norm{|\varphi|^2* \xi^2 }^\frac{1}{2}_\infty \sqrt{ \alpha} \\
&\weq{\ref{cor:xi}}{ \lesssim} \norm{\varphi}_{L^2 \cap L^\infty }(\alpha+\norm{\nabla_1 q_1 \psi}^2 ).\numberthis
%& \lesssim \norm{\varphi}_{L^2 \cap L^\infty }  (\alpha +\norm{\nabla_1 q_1 \psi}^2 ) 
\end{align*}
%
%\textcolor{red}{stimmt hier $\norm{\varphi}_{L^2 \cap L^\infty }$ genau nachrechnen, stimmt! vielleicht immer wieder klar machen? :) }
%
%
We deal with the second term of \eqref{equ:p1II.3} as before
\begin{align*}\label{equ:p1II.7}
 |\llangle \underbrace{  \xi^\nu_{12} p_2 q_1  \psi}_{\eta},     \underbrace{  \nabla_1^\nu \widehat{ \mu}  q_1q_2 \psi}_{\kappa} \rrangle| &\leq \norm{\eta } \norm{\kappa}\\
& \lesssim \norm{\varphi}_{L^2 \cap L^\infty } \sqrt{\alpha} \norm{ \nabla_1 q_1 \psi}\\
%&\leq   \norm{ |\xi| }_q \norm{\varphi}_{L^2 \cap L^\infty } (\alpha+\norm{ \nabla_1 q_1 \psi}^2) \\
&\lesssim \norm{\varphi}_{L^2 \cap L^\infty } (\alpha+\norm{ \nabla_1 q_1 \psi}^2) \numberthis.
\end{align*}
Since similar to equation \eqref{equ:eta} we have
\begin{align*}
  \norm{\eta} \lesssim \norm{\varphi}_{L^2 \cap L^\infty } \sqrt{\alpha}
\end{align*}
and similar to equation \eqref{equ:kappa} we have
\begin{align*}
 \norm{ \kappa} \lesssim  \norm{ \nabla_1 q_1 \psi}.
\end{align*}
Now the bound for $|\mathrm{II}|$ follows from collecting all the different bounds from equations \eqref{equ:p1II.4},\eqref{equ:p1II.5},\eqref{equ:p1II.6} and \eqref{equ:p1II.7}.

\end{proof}

We are left with proving the estimates of term $\mathrm{III}$. We start with the part Lemma\,\ref{lem:estimating3}.\ref{lem:estimating3III} and\, \ref{lem:estimating3alt} have in common and 
continue with the easier proof for Lemma\, \ref{lem:estimating3alt} which will give an blueprint for the following proof of Lemma\,\ref{lem:estimating3}.\ref{lem:estimating3III}.
\begin{proof}[Proof of Lemma\,\ref{lem:estimating3alt} and Lemma\,\ref{lem:estimating3}.\ref{lem:estimating3III} ]
 Both mean field terms in term III do not contribute since for both of them a $p$ acts on a $q$ in the same coordinate. 
\begin{align*}\label{equ:p1III.1}
 \mathrm{III}&= |\llangle \psi, p_1 p_2  [ (N-1)  w^\epsi_{12}-N w_1^\varphi- N w_2^\varphi , \widehat n] q_1q_2 \psi \rrangle|\\
 &= | \llangle \psi, p_1 p_2  [ (N-1)  w_{12} , \widehat n] q_1q_2 \psi \rrangle|\\
 &\weq{\ref{lem:weights}} {=}|\llangle \psi, p_1 p_2   w_{12} \underbrace{ N( \widehat n - \widehat \tau_{-2} n) }_{=:\mu_1}  q_1q_2 \psi \rrangle|\\
&\leq |\llangle \psi, p_1 p_2   w_{12}^{\epsi,\infty} \mu_1  q_1q_2 \psi \rrangle|+|\llangle \psi, p_1 p_2   w_{12}^{\epsi,s} \mu_1  q_1q_2 \psi \rrangle|,\numberthis
\end{align*}
where $\mu_1$
\begin{align}\label{equ:mu1}
 \mu_1= N \big (\frac{\sqrt{k}}{\sqrt{N}}-\frac{\sqrt{k-2}}{\sqrt{N}} \big ) = \frac{2 \sqrt{N}}{\sqrt{k}+\sqrt{k-2}} \leq 2 \frac{\sqrt{N}}{\sqrt{k}} = 2n^{-1} \quad \forall k \geq 2.
\end{align}
The $w^{\epsi,\infty}$ part of \eqref{equ:p1III.1} does not pose any problems and can be estimated by
\begin{align*}
 |\llangle \psi, p_1 p_2   w_{12}^\infty \mu_1  q_1q_2 \psi \rrangle|&=|\llangle \psi, p_1 p_2   w_{12}^\infty  \widehat n^{\frac{1}{2}} \widehat n^{-\frac{1}{2}} \widehat \mu_1   q_1q_2 \psi \rrangle|\\
&\weq{\ref{lem:weights}}{ =}| \llangle \widehat {\tau_2 n}^{\frac{1}{2}} \psi, p_1 p_2   w_{12}^\infty   \widehat n^{-\frac{1}{2}} \widehat \mu_1   q_1q_2 \psi \rrangle|\\
&\leq  \norm{w_{12}^\infty}_\infty  \sqrt{\llangle \psi, \widehat{ \tau_2 n } \psi \rrangle}\sqrt {\llangle \psi , \widehat n^{-1} \widehat \mu_1^2 q_1 q_2 \psi \rrangle   } \\
&\weq{\ref{lem:weights} }{ \lesssim }  \norm{w_{12}^\infty}_\infty  (\beta+ \frac{1}{\sqrt{N}})^{\frac{1}{2}} \sqrt{\beta}  \\
&\weq{A1'}{ \lesssim}  (\beta+ \frac{1}{\sqrt{N}}) .
\end{align*}
However, the second summand of \eqref{equ:p1III.1} is more complicated to handle. The leading part of it could be dealt with the same methods as in the proof of Theorem\,\ref{thm:thm1}. 
The problem which occurs is that the resulting subleading term which is of order $N^{-1}$ can only be bound by these methods if we have control of $\norm{w^\epsi}_v$ for a $v \geq 2$. If this condition holds
we can use the same idea as in Lemma\,\ref{lem:aabschaetzen}.\ref{lem:aabschaetzen2}. The different presentation of the proof here only arises from the different weight and from the intention 
to reuse the calculation for the proof of Lemma\,\ref{lem:estimating3}.\ref{lem:estimating3III}.

\begin{proof}[Proof of Lemma\, \ref{lem:estimating3alt}]
We split $\mu_1=\mu_1^\frac{1}{2}\mu_1^\frac{1}{2} $ and rewrite the second summand of \eqref{equ:p1III.1} as
\begin{align*}\label{equ:p1III.2}
 |\llangle \psi, p_1 p_2   w_{12}^{\epsi,s}  \widehat \mu_1  q_1q_2 \psi \rrangle|
 &= \frac{1}{N-1}|\llangle \psi, \sum_{j=2}^N p_1 p_j   w_{1j}^{\epsi,s}  \widehat \mu_1^\frac{1}{2} \widehat \mu_1^\frac{1}{2}   q_1q_j \psi \rrangle|\\
&\leq \frac{1}{N-1} \norm{\widehat \mu_1^\frac{1}{2}   q_1 \psi} \sqrt{\sum_{j,i=2}^N \llangle \psi, p_1 p_j   w_{1j}^{\epsi,s} \widehat \mu_1    q_j q_i  w_{1i}^{\epsi,s} p_1 p_i \psi \rrangle }. \numberthis
\end{align*}
Since 
\begin{align*}
 \norm{\widehat \mu_1^\frac{1}{2}   q_1 \psi}^2 \stackrel{\ref{lem:weights},\eqref{equ:mu1}}{ \leq} \llangle \psi, \widehat n^{-1} \widehat n^2 \psi \rrangle  = \beta
\end{align*}
we can estimate
\begin{align}\label{equ:2p2qohney}
  |\llangle \psi, p_1 p_2   w_{12}^{\epsi,s}  \widehat \mu_1  q_1q_2 \psi \rrangle| \leq \frac{\sqrt{\beta}}{N-1} \sqrt{A+B},
\end{align}
where $A$ is the "off-diagonal" term of the sum
\begin{align*}
 A&:=\sum_{2\leq j\neq i \leq N} | \llangle \psi, p_1 p_j   w_{1j}^{\epsi,s} \widehat \mu_1    q_j q_i q_1  w_{1i}^{\epsi,s} p_1 p_i \psi \rrangle|\\
\end{align*}
and $B$ the "diagonal" term
\begin{align*}
B&:=\sum_{i=2}^N | \llangle \psi, p_1 p_i   w_{1i}^{\epsi,s}  \widehat \mu_1    q_i  q_1 w_{1i}^{\epsi,s} p_1 p_i \psi \rrangle|.
\end{align*}
We continue by estimating $B$ 
\begin{align}\label{equ:Bohne}
 B &\leq \norm{\widehat \mu_1  q_1}_\Op \sum_{i=2}^N \norm{w_{1i}^{\epsi,s} p_1 p_i \psi }^2 \notag\\
&\leq N^\frac{1}{2} \sum_{i=2}^N \norm{w_{1i}^{\epsi,s} p_1 p_i \psi }^2 \notag\\
 &= N^\frac{1}{2} \sum_{i=2}^N \langle \psi, p_1 p_i   (w_{1i}^{\epsi,s})^2 p_1 p_i \psi \rangle \notag\\
&\leq  N^{\frac{3}{2}} \norm{p_1  (w_{1i}^{\epsi,s})^2 p_1} \notag\\
&= N^{\frac{3}{2}} \norm{ (w^{\epsi,s})^2*|\varphi|^2 }_\infty \notag\\
&\leq N^{\frac{3}{2}} \norm{w^{\epsi,s}}^2_v (1+\norm{\varphi}_\infty)^2\notag \\
&\weq{\eqref{equ:ass+w}}{\leq} N^\frac{3}{2} \tilde f(\epsi)^2 (1+\norm{\varphi}_\infty)^2.
\end{align}
For A we find 
\begin{align}\label{equ:Aohne}
A&= \sum_{2\leq j\neq i \leq N} |\llangle \psi, p_1 p_j   w_{1j}^{\epsi,s} \widehat \mu_1    q_j q_i q_1  w_{1i}^{\epsi,s} p_1 p_i \psi \rrangle |\notag \\
&=\sum_{2\leq j\neq i \leq N} |\llangle \psi, p_1 p_j  q_i \widehat {\tau_2 \mu}_1^\frac{1}{2}  w_{1j}^{\epsi,s}      q_1  w_{1i}^{\epsi,s} \widehat {\tau_2 \mu}_1^\frac{1}{2}  q_jp_1 p_i \psi \rrangle|.
\end{align}
In the last equation we write $q_1 = 1-p_1$  
%with the triangle equality
and after using the triangle inequality for the emerged sum, $A$ can be estimated by two terms called $A_1$ and  $A_2$. In the next steps we use for negative $w$ any branch of the complex square root and symmetry to find
\begin{align*}
 |A_1| &= \sum_{2\leq j\neq i \leq N}| \llangle \psi, p_1 p_j  q_i  \widehat {\tau_2 \mu}_1^\frac{1}{2}  w_{1j}^{\epsi,s}  
    w_{1i}^{\epsi,s} \widehat {\tau_2 \mu}_1^\frac{1}{2}  q_jp_1 p_i \psi \rrangle|\\
%&\leq  \sum_{2\leq j\neq i \leq N} \llangle \psi, p_1 p_j  q_i \widehat {\tau_2 \mu}_1^\frac{1}{2} | w_{1j}^{p,2}|  
 %   |w_{1i}^{\epsi,s}| \widehat {\tau_2 \mu}_1^\frac{1}{2}  q_ip_1 p_j \psi \rrangle \\
&\leq  \sum_{2\leq j\neq i \leq N}| \llangle \psi, p_1 p_j  q_i \widehat {\tau_2 \mu}_1^\frac{1}{2} \sqrt{ w_{1j}^{\epsi,s}}  
    \sqrt{w_{1i}^{\epsi,s}} \sqrt{ w_{1j}^{\epsi,s}}   \sqrt{w_{1i}^{\epsi,s}} \widehat {\tau_2 \mu}_1^\frac{1}{2}  q_jp_1 p_i \psi \rrangle| \\
& \leq  \sum_{2\leq j\neq i \leq N} \norm{  \sqrt{w_{1j}^{\epsi,s}} \sqrt{w_{1i}^{\epsi,s}} \widehat {\tau_2 \mu}_1^\frac{1}{2}  q_jp_1 p_i \psi}^2\\
& = \sum_{2\leq j\neq i \leq N}\norm{ \sqrt{w_{1i}^{\epsi,s}}p_i  \sqrt{w_{1j}^{\epsi,s}} p_1  \widehat {\tau_2 \mu}_1^\frac{1}{2}  q_j  \psi}^2\\
&\leq N^2 \norm{p_1 |w^{\epsi,s}_{12}|p_1}^2 \llangle \psi, \widehat {\tau_2 \mu}_1 q_1 \psi \rrangle \\
& \weq{\ref{lem:weights}}{ \leq} \norm{p_1 |w^{\epsi,s}_{12}|p_1}^2 \beta\\
& \weq{\ref{cor:op}}{ \lesssim} N^2 (1+\norm{\varphi}_\infty)^4 \beta.
\end{align*}
We estimate $A_2$
\begin{align*}
|A_2| &=   \sum_{2\leq j\neq i \leq N} \llangle \psi, p_1 p_j  q_i  \widehat {\tau_2 \mu}_1^\frac{1}{2}  w_{1j}^{\epsi,s}  p_1
    w_{1i}^{\epsi,s}  \widehat {\tau_2 \mu}_1^\frac{1}{2}  q_jp_1 p_i \psi \rrangle\\
&\leq N^2 \norm{p_1 w^{\epsi,p} p_1  }_\mathrm{Op}^2 \beta \\
%& \leq N^2 \norm{w^{\epsi,p}}_p^2 (1+\norm{\varphi}_\infty)^4 \beta.
& \weq{\ref{cor:op}}{ \lesssim} N^2 (1+\norm{\varphi}_\infty)^4 \beta.
\end{align*}
Collecting the estimates for $A$ and $B$ we have

\begin{align*}
 |\llangle \psi, p_1 p_2   w_{12}^{\epsi,s}   \widehat \mu_1  q_1q_2 \psi \rrangle|&\lesssim
 \frac{\sqrt{\beta}}{N-1}\sqrt{N^\frac{3}{2} \tilde f(\epsi)^2 (1+\norm{\varphi}_\infty)^2  +N^2 (1+\norm{\varphi}_{L^\infty})^4 \beta}\\
&\lesssim\sqrt{\beta}\sqrt{N^{-\frac{1}{2}} \tilde f(\epsi)^2 (1+\norm{\varphi}_\infty)^2  + (1+\norm{\varphi}_{L^\infty})^4 \beta}\\
&\lesssim  N^{-1/2} \tilde f(\epsi)^2(1+\norm{\varphi}_\infty)^2+(1+\norm{\varphi}_{L^\infty})^2\beta. \numberthis \label{equ:ppwqq+a}
\end{align*}
This ends the proof of Lemma\,\ref{lem:estimating3alt}.
\end{proof}

\textit{Proof of the remaining part of Lemma\,\ref{lem:estimating3}.\ref{lem:estimating3III}}
Without the possibility of the estimate in \eqref{equ:Bohne} the idea is to use an $N$-dependent splitting of the potential. This separates the singularities from the rest in a suitable way
to exploit the fact that only the subleading term poses problems in the calculation and combine this with the different scaling behaviors of $L^p$-norms for different $p$.
The splitting of $w_{12}^{\epsi,s}$ which does the trick is 
\begin{align*}
 w^{\epsi,s}=w^{s,1}+w^{s,2}:= w^{\epsi,s} \id_{\{|w^s|>c\}}+w^{\epsi,s} \id_{\{|w^s| \leq c\}},
\end{align*}
where $c$ is a positive $N$-dependent constant which we fix later by optimization of the convergence rates. 
In the following we will neglect the dependence of $w^s$ on $\epsi$.
Now we have for $s_0 < s < 2 $
\begin{align*}
 \norm{w^{s,1}}^{s_0}_{s_0}= \int |w^{s,1}|^{s_0} \D x = \int |w^{s}|^s |w^{s}|^{s_0-s} \id_{\{|w^s|>c\}} \D x \leq c^{s_0-s}\int |w^{s}|^s \id_{\{|w^s|>c\}} \D x \\
\leq c^{s_0-s} \int |w^{s}|^s \D x = c^{s_0-s} \norm{w^s}_s^s
\end{align*}
and  
\begin{align*}
 \norm{w^{s,2}}^{2}_{2}= \int |w^{s,2}|^{2} \D x = \int |w^{s}|^s |w^{s}|^{2-s} \id_{\{|w^s|<c\}} \D x \leq c^{2-s}\int |w^{s}|^s \id_{\{|w^s|<c\}} \D x \\
\leq c^{2-s} \int |w^{s}|^s \D x = c^{2-s} \norm{w^s}_s^s.
\end{align*}
Thus
\begin{align}
 \norm{w^{s,1}}_{s_0} & \leq c^{1-\frac{s}{s_0}} \norm{w^s}_s^\frac{s}{s_0} \label{equ:wp1} \\ 
 \norm{w^{s,2}}_2  &\leq c^{1-\frac{s}{2}} \norm{w^s}_s^{\frac{s}{2}}\label{equ:wp2}.
\end{align}
Now the idea becomes more obvious since if we set $c=N^\vartheta$ the $L^{s_0}$-norm of $ w^{s,1}$ becomes small for large $N$ because $1-s/s_0 < 0$. On the other hand the $L^2$-norm of $w^{s,2}$ will
diverge with some power of $N$ but sine we only need the $L^2$-norm of $w^{s,2}$ in the subleading part we can control this as long as $N^{-1/2} c^{2-s}= {\scriptstyle{ \mathcal{O}}}(1)$.       
This enables us to treat the part with $w^{s,1}$ by writing it as a divergence and then use integration by parts as done before. We define $\nabla \xi^j=w^{s,j}$.
Now we are in the same setting as in Lemma\,\ref{lem:intbyparts} and go through the same estimation process.
\begin{align*}
 |\llangle \psi, p_1 p_2   w_{12}^{s,1} \widehat \mu_1  q_1q_2 \psi \rrangle| &= |\llangle \psi,  p_1 p_2  \nabla^\nu_1 \xi_{12}^{1,\nu}\widehat \mu_1  q_1q_2 \psi \rrangle| \\
&\leq  |\llangle  \xi_{12}^{1,\nu} p_2 \nabla^\nu_1  p_1   \psi,  \widehat  \mu_1  q_1q_2 \psi \rrangle| +
|\llangle   p_1 p_2  \psi,   \xi_{12}^{1,\nu}   \nabla^\nu_1 \widehat \mu_1  q_1q_2 \psi \rrangle| \numberthis \label{equ:ws1}
\end{align*}
The first term is estimated by
\begin{align*}
 |\llangle  \xi_{12}^{1,\nu} p_2 \nabla^\nu_1  p_1   \psi,   \widehat \mu_1  q_1q_2 \psi \rrangle| &\leq \sqrt{\llangle \nabla^\nu_1 p_1 \psi, p_2 \xi_{12}^{1,\nu} \xi_{12}^{1,\iota} \  p_2 \nabla_1^\iota p_1 \psi  \rrangle} \norm{ \widehat \mu_1  q_1q_2 \psi}\\
&\lesssim \norm{ |\varphi|^2 * (\xi^1)^2}^\frac{1}{2}_\infty  \norm{\nabla \varphi}  \norm{ \widehat n^{-1}  q_1q_2 \psi }\\
& \lesssim \norm{ |\xi^1| }_2  \norm{ \varphi}_{\infty} \norm{\nabla \varphi} \sqrt{\alpha} \\
&\weq{\ref{lem:div}}{ \lesssim}  \norm{ w^{s,1} }_{s_0}  \norm{\nabla \varphi} \norm{\varphi}_\infty \sqrt{\alpha}\\
&\weq{\eqref{equ:wp1}}{ \lesssim}   \norm{\nabla \varphi} \norm{\varphi}_\infty \norm{w^s}_{s}^{\frac{s}{s_0}}c^{1-\frac{1s}{s_0}} \sqrt{\alpha}\\
&\weq{A1'}{ \lesssim }  \norm{\nabla \varphi} \norm{\varphi}_\infty (c^{2-\frac{2s}{s_0}}+\alpha ),
%&\lesssim \norm{\nabla \varphi} \norm{\varphi}_\infty (a^{2-\frac{2p}{p_0}}+\beta )
\end{align*}
where we refer to the proof  Lemma\,\ref{lem:intbyparts} for the step from the first to the second line. %of  we used \eqref{equ:wp1}.
The second term is estimated similar to equation \eqref{equ:2p2qpart2}%, where to estimate $\kappa$ we have to insert $p_1+q_1$ which we do not elaborate.   
\begin{align*}
 |\llangle    \xi_{12}^{1,\nu}   p_1 p_2  \psi,   \nabla^\nu_1 \widehat \mu_1  q_1q_2 \psi\rrangle| 
&\leq \sqrt{\llangle  \psi, p_1 p_2 (\xi_{12}^{1})^2 p_1 p_2 \psi \rrangle } \norm{\nabla_1^\nu \widehat \mu_1 q_1 q_2 \psi}\\
&\lesssim  \norm{ |\xi^1| }_{2} \norm{\varphi}_{\infty }  \norm{\nabla_1 q_1 \psi}\\
 &\weq{\ref{lem:div}}{ \lesssim}  \norm{w^{s,1}}_{s_0} \norm{\varphi}_{\infty }  \norm{\nabla_1 q_1 \psi}\\
 &\weq{\eqref{equ:wp1}}{ \lesssim}    \norm{\varphi}_{\infty } \norm{w^s}_{s}^{\frac{s}{s_0}} c^{1-\frac{s}{s_0}} \norm{\nabla_1 q_1 \psi}\\
&\weq{A1'}{ \lesssim }  \norm{\varphi}_{\infty } (c^{2-\frac{2s}{s_0}}+ \norm{\nabla_1 q_1 \psi}^2).
\end{align*}
Collecting both estimates we find for the right-hand side of \eqref{equ:ws1}
\begin{align}\label{equ:2p2qw1}
 |\llangle \psi, p_1 p_2   w_{12}^{s,1} \widehat \mu_1  q_1q_2 \psi \rrangle| &\lesssim  \norm{\nabla \varphi} \norm{\varphi}_\infty (c^{2-\frac{2p}{p_0}}+\beta ) 
+\norm{\varphi}_{L^\infty } (c^{2-\frac{2p}{p_0}}+ \norm{\nabla_1 q_1 \psi}^2)\notag \\
&\leq  \norm{\varphi}_{\infty }\big(\norm{ \nabla \varphi} \beta+ \norm{ \varphi}_{H^1}c^{2-\frac{2s}{s_0}}+\norm{\nabla_1 q_1 \psi}^2\big).
\end{align}
Now we come to the term $|\llangle \psi, p_1 p_2   w_{12}^{s,2} \widehat \mu_1  q_1q_2 \psi \rrangle|$. This term can be dealt with the help of Lemma\,\ref{lem:estimating3alt}.
The only difference is that $ \norm{w^{s,2}}_2$ is bounded by $ c^{1-\frac{s}{2}} \norm{w^s}_s^{\frac{s}{2}}$ instead of $\norm{w^{\epsi,s}}_v$ being bounded by  $\tilde f(\epsi)$. Thus we find

\begin{align*}
 |\llangle \psi, p_1 p_2   w_{12}^{s,2}   \widehat \mu_1  q_1q_2 \psi \rrangle|& \weq{\eqref{equ:ppwqq+a}}{ \lesssim}
 \frac{\sqrt{\beta}}{N-1}\sqrt{N^\frac{3}{2}c^{2-s} (1+\norm{\varphi}_\infty)^2  +N^2 (1+\norm{\varphi}_{L^\infty})^4 \beta}\\
&\lesssim  \sqrt{\beta} \sqrt{N^{-\frac{1}{2}}c^{2-s} (1+\norm{\varphi}_\infty)^2  + (1+\norm{\varphi}_{L^\infty})^4 \beta} \\
&\lesssim  N^{-1/2} c^{2-s}(1+\norm{\varphi}_\infty)^2+(1+\norm{\varphi}_{L^\infty})^2\beta \numberthis \label{equ:ppws2qq}.
\end{align*} 
Putting this together with \eqref{equ:2p2qw1} we can optimize in $\vartheta$ when setting $c=N^\vartheta$ 
\begin{align*}
|\llangle \psi, p_1 p_2   w_{12}^{\epsi,s}  \widehat \mu_1  q_1q_2 \psi \rrangle|& \lesssim
  \norm{\varphi}_{\infty }\big(\norm{ \nabla \varphi} \beta+ \norm{ \varphi}_{H^1}c^{2-\frac{2s}{s_0}}+\norm{\nabla_1 q_1 \psi}^2\big)\\
\quad&+N^{-1/2} c^{2-s}(1+\norm{\varphi}_\infty)^2+(1+\norm{\varphi}_{L^\infty})^2\beta\\
&\lesssim \norm{ \varphi}_{H^1\cap L^\infty }^3(\beta+N^\eta)+ \norm{\varphi}_{\infty }\norm{\nabla_1 q_1 \psi}^2
\end{align*}
with %\textcolor{red}{Nachrechnen!}
\begin{align}\label{equ:etagr}
 \eta=-\frac{s/s_0-1}{2s/s_0-s}=-\frac{5s-6}{4s} %-  \frac{5}{4}+\frac{3}{2s} .
\end{align}
This finishes the proof of Lemma\, \ref{lem:estimating3alt}.
\end{proof}
 By introducing yet another splitting the estimate of \eqref{equ:ppws2qq} can be improved slightly. We defer this to Appendix\,\ref{app:ratebes}.

\include{GP_Energie_epsi}

\section{Controlling the Derivative of $\tilde \beta$}
Now we come to the main part of the proof which has the well-known structure. The additional term IV stems from the introduction of the external potential V.
For the ease of representation and to be in the same setting as in the proof of Theorem\,\ref{thm:thm2} we define 
\begin{align}\label{equ:potgps}
w^{\epsi, \theta,N} := N W^{\epsi, \theta,N}= (N^{-1} \epsi^2)^{-3\theta} \epsi^2 w \Big ( (N^{-1} \epsi^2)^{-\theta}  \big( (x_i-x_j),\epsi(y_i-y_j) \big  ) \Big).
\end{align}

\begin{lem}\label{lem:beta.g}
 \begin{align*}
  |\frac{d}{dt} \tilde \beta| \leq  2|\mathrm{I}| +  |\mathrm{II}| +  2|\mathrm{III}|+ |\mathrm{IV}|,
 \end{align*}
where 
\begin{align*}
 \mathrm{I}&:=\llangle \psi, p_1 p_2  [ (N-1)  w^{\epsi, \theta,N}_{12}- N b|\Phi|^2(x_2) ,\widehat n ] p_1q_2 \psi \rrangle\\
\mathrm{II}&:= \llangle \psi, p_1 p_2  [ (N-1)  w^{\epsi, \theta,N}_{12}, \widehat n] q_1q_2 \psi \rrangle\\ 
\mathrm{III}&:= \llangle \psi, p_1 q_2  [ (N-1)  w^{\epsi, \theta,N}_{12}-N  b|\Phi|^2(x_1) , \widehat n] q_1q_2 \psi \rrangle\\
\mathrm{IV}&:= |\llangle \psi, \dot V(x_1,\epsi y_1) \psi \rrangle - \langle \Phi, \dot V(x_1,0) \Phi \rangle_{\LzOf}| \\
	     &\quad  + 2 \llangle \psi, p_1  N [   V(x_1,\epsi y_1)-V(x_1,0) , \widehat n] q_1 \psi \rrangle.
\end{align*}

\end{lem}

\begin{lem}\label{lem:3termeg}
\begin{enumerate}
 \item \label{lem:3.1g}
\begin{align*}
 |\mathrm{I}|&\lesssim N^{-2 \theta}\epsi^{4\theta-2}\norm{\Delta |\varphi|^2}_\LzO \norm{\varphi}_\LiO+ N^\frac{1}{2} \epsi \norm{\varphi}_\LiO^2 
\end{align*}

 \item  \label{lem:3.2g}
For $\delta> 0$
\begin{align*}
 |\mathrm{II}|&\lesssim \norm{\varphi}_{\LiO}^2 \tilde \beta + N^{-\frac{1}{2}} N^{\frac{3\theta}{2}}N^{\frac{\delta}{4}}\epsi^{-3\theta+1}  \norm{\varphi}_{\LiO} +
 N^{-\frac{\delta}{2}} \norm{\varphi}_{\LiO}^2
\end{align*}

 \item \label{lem:3.3g}
\begin{align*}
 |\mathrm{III}|&\lesssim \norm{\varphi}_{H^2(\Omega)\cap L^\infty(\Omega)} \norm{\chi}_{\LiOc}^2 \big(\tilde \beta +\epsi^4 (N\epsi^{-2})^{3\theta}+  f(N,\epsi) \big)\\&\qquad + 
 \norm{\chi}_{\LiOc}^2   \norm{V}_\LiO^{1/2} \beta
\end{align*}

\item \label{lem:3.4g}
\begin{align*}
 |\mathrm{IV}| &\lesssim  \norm{\dot V}_\LiO \tilde \beta+ \epsi
\end{align*}

\end{enumerate}

\end{lem}

\begin{proof}[Proof of Theorem\,\ref{thm:thm3}]
 The Lemmas\,\ref{lem:beta.g} and\,\ref{lem:3termeg} together with the Grönwall argument %Lemma\,\ref{lem:gron} 
prove Theorem\,\ref{thm:thm3}.
If $\theta\in (\frac{1}{4},\frac{1}{3}) $ and $\epsi=N^{-\nu}$ for $\nu \in (\frac{1}{2}, \frac{\theta}{1-2\theta} )$, then all error terms converge to zero.
The optimal rate is $N^{-\eta(\theta)}$ with
\begin{align}\label{equ:rate}
 \eta(\theta) = \begin{cases}
       \frac{4\theta-1}{3-4\theta} \qquad &\mathrm{for}\; \theta \in (\frac{1}{4}, \frac{7}{24}]\\
	\frac{1-3\theta}{4-9\theta}  \qquad &\mathrm{for}\; \theta \in (\frac{7}{24} ,\frac{1}{3})
      \end{cases}
\end{align}
which follows by optimization of $\delta$ and $\nu$. 
 
\end{proof}
% and the convergence 
%rate is determined by the maximum of the error terms. 

% We define
% \begin{align*}
%   \max({N^{(-2 \theta)}\epsi^{4\theta-2},N^\frac{1}{2} \epsi, \epsi^4 (N\epsi^{-2})^{3\theta} ,
%  \min_{\delta > 0}( {N^{-\frac{1}{2}} N^{\frac{3\theta}{2}}N^{\frac{\delta}{4}}\epsi^{-3\theta+1},
%  N^{-\frac{\delta}{2}}}}))
% \end{align*}
% is negative for  $\theta\in (\frac{1}{4},\frac{1}{3}) $ and $\epsi=N^{\nu}$ for $\nu \in (\frac{1}{2}, \frac{\theta}{1-2\theta} ) $.
% For given $\theta $ and $\mu$ it is defined by
\begin{rem}
 For $ \theta \in (\frac{1}{4}, \frac{7}{24}]$ the optimal rate in \eqref{equ:rate} is determined by the terms $N^{-2 \theta}\epsi^{4\theta-2}$ and $ N^\frac{1}{2} \epsi$.
For $\theta \in (\frac{7}{24} ,\frac{1}{3})$ the optimal rate is determined by $N^{-2 \theta}\epsi^{4\theta-2}$, $N^{-\frac{1}{2}} N^{\frac{3\theta}{2}}N^{\frac{\delta}{4}}\epsi^{-3\theta+1}$ and $N^{\frac{-\delta}{2}}$.
For $\theta= \frac{7}{24}$ all four terms have the value $\frac{-1}{11}$ if we choose $\delta=\frac{2}{11}$ and $\nu=\frac{13}{22}$. 
\end{rem}

\section{Proofs of the Lemmas}

In order to keep the notation in these proofs to a minimum we do not write, whenever it does not lead to confusion, the underlying sets of the function spaces and write $\norm{\cdot}$ and $\langle \cdot, \cdot \rangle $ 
for the $L^2$-norm and scalar product on the appropriate set.

\begin{proof}[Proof of Lemma \ref{lem:beta.g}]

Compare Lemma\,\ref{beta.} for the terms I-III.
The term IV stems from the change of $\beta$. The first summand of IV is the time derivative of $|E^\Psi-E^\varphi|$ and the second
summand arises from the different external potentials in the Hamiltonians of $\psi$ and $\varphi$.  
\end{proof}

%In the following proofs of the different parts of Lemma\;\ref{lem:3termeg} we will again use Lemma\;\ref{lem:weights}, Lemma\;\ref{lem:young} and Lemma\;\ref{lem:qs&N} without referring to them explicitly.  

\begin{proof}[Proof of Lemma\;\ref{lem:3termeg}.\ref{lem:3.1g}  ]

The term I is small due to the cancellation of $b|\Phi|^2 $ and the full interaction. Before one can see this cancellation we have to separate this term into a part which stays in the ground state of the 
confined direction and the orthogonal complement. 
To this end we use the projections 
% have to rewrite $\delta(x_1-x_2)$ in $\delta(r_1-r_2)$, so we start by rewrite the projections $q$ with 
\begin{align*}
p_j^\chi&:=1 \otimes |\chi(y_j) \rangle \langle \chi(y_j) | &  q^\chi_j&:=1-p^\chi_j \\
p_j^\Phi&:= |\Phi(x_j) \rangle \langle \Phi(x_j) | \otimes 1 &  q^\Phi_j&:=1-p^\Phi_j  \numberthis \label{equ:defpqx}.
\end{align*}
With this projections we can rewrite %$q_j$
\begin{align}\label{equ:qinqx}
q_j= 1- p_j = 1 - p^\Phi_j p^\chi_j = (1-p_j^\chi)+ (1-p_j^\Phi)p_j^\chi=q_j^\chi+q_j^\phi p_j^\chi.
\end{align}
For later use we note that for any function $f: \Omega_{\mathrm{f}} \rightarrow \C $
\begin{align} \label{equ:pfqchi}
 p_2 f(x_2) q_2^\chi=0.
\end{align}
Now with \eqref{equ:qinqx} and Lemma\,\ref{lem:weights} %and Lemma\,\ref{lem:com2hat}
\begin{align*}
 |I|&= |\llangle \psi, p_1 p_2  [ (N-1)  w^{\epsi,\theta,N}_{12}- N b|\Phi|^2(x_2) , \widehat {n}] p_1q_2 \psi \rrangle|\\
& \weq{\ref{lem:weights}}{=} |\llangle  \psi, p_1 p_2 \Big((N-1)w^{\epsi,\theta,N}_{12}-N b|\Phi|^2(x_2)\Big) ({\widehat { {n}}}- { \widehat  {\tau_{-1} n}}  ) p_1 q_2   \psi \rrangle  |\\
&\weq{\eqref{equ:qinqx}}{=} |\llangle  \psi, p_1 p_2 \Big((N-1)w^{\epsi,\theta,N}_{12}-N b|\Phi|^2(x_2)\Big) ({\widehat { {n}}}- { \widehat  {\tau_{-1} n}}  ) p_1 ( p_2^\chi q_2^\phi+ q_2^\chi )    \psi \rrangle  |\\
&= |\llangle  \psi, p_1 p_2 \Big((N-1)w^{\epsi,\theta,N}_{12}-N b \delta(x_1-x_2) \Big) (\widehat{n}- \widehat{\tau_{-1} n}  ) p_1 ( p_2^\chi q_2^\phi+ q_2^\chi )    \psi \rrangle  | \numberthis \label{equ:I.1},
\end{align*}
where we use the idea of $\eqref{equ:pwp}$ to write $|\Phi|^2(x_2)$ as $\delta(x_1-x_2)$.
The cancellations can be obtained by viewing the difference of both interactions as a right-hand side
of Poisson's equation. %The solution of this equation can be chosen such that it has compact support, since $a$ is equal to the $L^1$ of the interaction. Hence we can use the scaling behavior
 %of the solution of the poisson equation to prove that this term is small. To make this idea work we have to rewrite $\delta(x_1-x_2)$ as $\delta(r_1-r_2)$ 
To this end we define 
\begin{align*}
 \tilde b := \frac{b}{  \int_{\Omega_{\mathrm{c}}} |\chi|^4(y) \, \D y} = \int_{\R^3} w \, \D r .
\end{align*}
 Now we can rewrite the $\delta$ distribution
\begin{align*}
p_1 p_2  b &\delta(x_1-x_2) p_1 q_2^\Phi p_2^\chi  %&= p_1 p_2 a \delta(x_1-x_2) p_1 (q_2^\chi+ q_2^\Phi p_2^\chi) \\
\\& \, \weq{\eqref{equ:defpqx}}{=} \, p_1 p_2 b \delta(x_1-x_2) \langle \chi(y_1) \chi(y_2),  \chi(y_1) \chi(y_2) \rangle_{L^2(\Omega_{\mathrm{c},y_1} \times \Omega_{\mathrm{c},y_2})} p_1 q_2^\Phi p_2^\chi \\
&=p_1 p_2  b \delta(x_1-x_2) \langle \chi(y_1) \chi(y_2), \frac{ \delta(y_1-y_2)}{\norm{\chi^4}_{L^1(\Omega_\mathrm{c})}}  \chi(y_1) \chi(y_2) \rangle_{L^2(\Omega_{\mathrm{c},y_1} 
\times \Omega_{\mathrm{c},y_2})} p_1 q_2^\Phi p_2^\chi \\
&\weq{\eqref{equ:defpqx}}{=} p_1 p_2 \tilde b \delta(r_1-r_2) p_1 q_2^\Phi p_2^\chi \numberthis \label{equ:delta}.
%&=  p_1 p_2 \tilde a \delta(r_1-r_2) p_1 q_2 - p_1 p_2 \tilde a \delta(r_1-r_2) p_1 q_2^\chi \\
%&=p_1 p_2 \tilde a \delta(r_1-r_2) p_1 q_2 - p_1 p_2 \tilde a |\varphi|^2(r_2) p_1 q_2^\chi \\
\end{align*}
This term together with the full interaction will turn out to be small.
Entering the above calculation in I we get

\begin{align*}
|\mathrm I|&= |\llangle  \psi, p_1 p_2 \Big((N-1)w^{\epsi,\theta,N}_{12}-N b \delta(x_1-x_2) \Big) (\widehat  {n}- \widehat  {\tau_{-1} n}  ) p_1 q_2   \psi \rrangle  |\\
\\
&\weq{\eqref{equ:qinqx}}{ \leq}|\llangle  \psi, p_1 p_2 \Big(Nw^{\epsi,\theta,N}_{12}-N b \delta(x_1-x_2) \Big) (\widehat  {n}- \widehat  {\tau_{-1} n}  ) p_1 p_2^\chi q_2^\Phi  \psi \rrangle  |\\
&\quad+|\llangle  \psi, p_1 p_2 \Big(N w^{\epsi,\theta,N}_{12}-N b \delta(x_1-x_2) \Big) (\widehat  {n}- \widehat  {\tau_{-1} n}  ) p_1 q_2^\chi  \psi \rrangle  |\\
&\quad+|\llangle  \psi, p_1 p_2 w^{\epsi,\theta,N}_{12} (\widehat  {n}- \widehat  {\tau_{-1} n}  ) p_1 q_2   \psi \rrangle  |\\
\\%&\weq{\substack{\eqref{equ:pfqchi}\\ \eqref{equ:delta}}}{ \leq}
&\weq{\eqref{equ:pfqchi}}{ \leq} |\llangle  \psi, p_1 p_2 \Big(N w^{\epsi,\theta,N}_{12}-N b \delta(x_1-x_2) \Big) (\widehat  {n}- \widehat  {\tau_{-1} n}  ) p_1 p_2^\chi q_2^\Phi  \psi \rrangle  |\\
&\quad+|\llangle  \psi, p_1 p_2  N w^{\epsi,\theta,N}_{12}  (\widehat  {n}- \widehat  {\tau_{-1} n}  ) p_1  q_2^\chi   \psi \rrangle  |\\
&\quad+|\llangle  \psi, p_1 p_2 w^{\epsi,\theta,N}_{12} (\widehat  {n}- \widehat  {\tau_1 n}  ) p_1 q_2   \psi \rrangle  \\
\\
\begin{split}
&\weq{\eqref{equ:delta}}{ \leq} |\llangle  \psi, p_1 p_2 N \Big(w^{\epsi,\theta,N}_{12}- \tilde b \delta(r_1-r_2) \Big) (\widehat  {n}- \widehat  {\tau_1 n}  )  p_1 p_2^\chi q_2^\Phi   \psi \rrangle  |\\
&\quad+|\llangle  \psi, p_1 p_2 N w^{\epsi,\theta,N}_{12} (\widehat  {n}- \widehat  {\tau_{-1} n}  ) p_1 q_2^\chi   \psi \rrangle  |\\
& \quad +|\llangle  \psi, p_1 p_2 w^{\epsi,\theta,N}_{12} (\widehat  {n}- \widehat  {\tau_1 n}  ) p_1 q_2   \psi \rrangle  |.
\end{split} \numberthis \label{equ:I.2g}
\end{align*}
To estimate the first summand we first collect some properties of the difference 
\begin{align*}
  w^{\epsi,\theta,N}(r)- \tilde b \delta(r) \weq{\eqref{equ:potgps}}{=} ({N}\epsi^{-2})^{3 \theta} \epsi^2 w\Big( (N \epsi^{-2})^\theta(x, \epsi y)\Big)- \tilde  b \delta(x,y).
\end{align*}
This illustrates the scaling of the first line of \eqref{equ:I.2g}. 
We regard the above expression as a right-hand side of Poisson's equation for a function $f$. The idea is to use Newton's theorem to deduce that $f$ has compact support. 
However, to use Newton's theorem we need rotational symmetry. 
Because of that we define $\tilde  f^{\theta, \epsi} :\R^3 \rightarrow \R $ in the unscaled coordinates $ y'=\epsi y $ by 
\begin{align}\label{equ:defftilde}
\Delta \tilde f^{\theta, \epsi}(x,y') = ({N}\epsi^{-2})^{3 \theta} \epsi^2  w( (N \epsi^{-2})^\theta(x,y'))- \epsi^2 \tilde b \delta((x,y'))
\end{align}
and the same function in the scaled coordinates by
\begin{align}\label{equ:fthetay}
 f^{\theta, \epsi}(x, y): = \tilde  f^{\theta, \epsi}(x,y').
\end{align}
Since $w$ has compact support and $\tilde b = \int w \D r$ we find after scaling $\tilde  x= (N \epsi^{-2})^\theta x $ and $\tilde y= (N \epsi^{-2})^\theta y'$
\begin{align*}
 &\int_{\R^3} ({N}\epsi^{-2})^{3 \theta}   w\Big( (N \epsi^{-2})^\theta(x,y')\Big)- \tilde b \delta(x,y') \D x \D y'\\
&=\int_{\supp w}  w(\tilde x,\tilde y)- \tilde b \delta(\tilde  x, \tilde y) \D \tilde x \D \tilde y=0.
\end{align*}
Thus, we can indeed choose $\tilde f$ such that it has compact support. Using the definition \eqref{equ:defftilde} we find the following scaling behavior %for $\tilde f^{\theta,\epsi}$
\begin{align*}
  \tilde f^{\theta,\epsi}=& ({N}\epsi^{-2})^{ \theta} \epsi^2 \tilde f\Big( (N \epsi^{-2})^\theta(x,y')\Big).
\end{align*}
Since $\tilde f$ is solution of Poisson's equation $ \tilde f \in L^1_\mathrm{loc}(\R^3)$. This implies together with the compact support $\tilde f \in L^1(\R^3)$. So the $L^1(\R^3)$-norm of $ \tilde f^{\theta,\epsi}$ scales like
\begin{align}\label{equ:scaleftilde}
\norm{\tilde f^{\theta,\epsi}}_{L^1(\R^3)}=& \epsi^2 (N \epsi^{-2})^{-2 \theta} \norm{ \tilde f}_{L^1(\R^3)}= \epsi^2 (N \epsi^{-2})^{-2 \theta} \norm{\tilde f}_{L^1(\R^3)}.
\end{align}
It follows that the scaling of $f^{\theta,\epsi}$ is such that
\begin{align}
 \norm{f^{\theta,\epsi}}_{L^1(\R^3)} \weq{\eqref{equ:fthetay}}{=} \,\frac{1}{\epsi^2} \norm{\tilde f^{\theta,\epsi}}_{L^1(\R^3)} \weq{\eqref{equ:scaleftilde}}{=} \,
(N \epsi^{-2})^{-2 \theta} \norm{\tilde f}_{L^1(\R^3)} \; \, \weq{|\supp{\tilde f}| \leq C}{ \lesssim} \, \quad
(N \epsi^{-2})^{-2 \theta}.\label{equ:estfus}
\end{align}
This is the ingredient with which we can estimate the first summand of \eqref{equ:I.2g}. Let $\Delta^\epsi:= \Delta_x +\frac{1}{\epsi^2} \Delta_y $, then 
\begin{align*}
 |&\llangle  \psi, p_1 p_2 N \Big(w^{\epsi,\theta,N}_{12}- \tilde b \delta(r_1-r_2) \Big) (\widehat  {n}- \widehat  {\tau_{-1} n}  ) p_1 q_2^\Phi p_2^\chi   \psi \rrangle  | \\
&\weq{\eqref{equ:defftilde}}{=}\, N |\llangle  \psi, p_1 p_2 \Delta^\epsi_1 \tilde f^{\theta,\epsi}((x_1-x_2),\epsi (y_1-y_2)) (\widehat  {n}- \widehat  {\tau_{-1} n}  ) p_1q_2^\Phi p_2^\chi   \psi \rrangle  |\\
&\weq{\eqref{equ:fthetay}}{ =}\, N |\llangle  \psi, p_1 p_2 (\Delta^\epsi  f^{\theta,\epsi}* |\varphi|^2)(r_2) (\widehat  {n}- \widehat  {\tau_{-1} n}  ) p_1 q_2^\Phi p_2^\chi   \psi \rrangle  |\\
&= N  |\llangle  \psi, p_1 p_2 ( f^{\theta,\epsi}* \Delta^\epsi |\varphi|^2)(r_2) (\widehat  {n}- \widehat  {\tau_{-1} n}  ) p_1q_2^\Phi p_2^\chi  \psi \rrangle  |\\
&\leq  \norm{ ( f^{\theta,\epsi}* \Delta^\epsi |\varphi|^2)(r_2) p_2 }_{\mathrm{Op}}  \norm{ N (\widehat  {n}- \widehat  {\tau_{-1} n}  ) q_2^\Phi p_2^\chi   \psi} \\
&\weq{\ref{lem:weights}}{\leq}  \norm{ ( f^{\theta,\epsi}* \Delta^\epsi |\varphi|^2)(r_2) p_2 }_{\mathrm{Op}}  
\weq{\ref{lem:young}}{ \leq} \norm{f^{\theta,\epsi}* \Delta^\epsi |\varphi|^2} \norm{\varphi}_\infty\\
& \weq{\ref{lem:young}} {\leq}  \norm{f^{\theta,\epsi}}_1 \norm{ \Delta^\epsi |\varphi|^2} \norm{\varphi}_\infty
\weq{\eqref{equ:estfus}}{\leq}\,  (N \epsi^{-2})^{-2 \theta} \epsi^{-2} \norm{\Delta |\varphi|^2} \norm{\varphi}_\infty   \\
&\leq N^{-2\theta }\epsi^{4\theta-2}\norm{\Delta |\varphi|^2} \norm{\varphi}_\infty, \numberthis \label{equ:3p1qgp}
\end{align*}
where Lemma\,\ref{lem:weights} holds for $q^\Phi q^\chi $ since $q^\Phi q^\chi \leq q  $ in the sense of operators.  

The second summand of \eqref{equ:I.2g} %the term $|\llangle \psi, p_1 p_2 N w^{\epsi,\theta,N}_{12} (\widehat  {n}- \widehat  {\tau_{-1} n}  ) p_1 q_2^\chi \psi \rrangle  |$ 
is estimated by
%we use 
%the Corollary\,\ref{cor:energyepsi}.
%Using Cauchy Schwarz and the Lemma\ref{??}
\begin{align*}
 |\llangle \psi, p_1 p_2& N w^{\epsi,\theta,N}_{12} (\widehat  {n}- \widehat  {\tau_{-1} n}  ) p_1 q_2^\chi \psi \rrangle  |
\leq \norm{p_1 w^{\epsi,\theta,N}_{12} p_1 }_\mathrm{Op} \llangle \psi,   (N(\widehat  {n}- \widehat  {\tau_{-1} n}  ))^2 q_2^\chi  \psi \rrangle^\frac{1}{2} \\
&\weq{\ref{lem:young}}{\lesssim} \norm{\varphi}^2_\infty \llangle \psi,   \sum_{k=0}^{N} N^2 \frac{(\sqrt k-\sqrt{k-1})^2 }{N}  P_{k,N} \sum_{j=1}^N \frac{j}{N} P^{\chi^\perp}_{j,N}   \psi \rrangle^\frac{1}{2} \\
&\weq{ \substack{ \eqref{equ:defP} \\ \eqref{equ:defpchi}}}{ =}\, \norm{\varphi}^2_\infty \llangle \psi,   \sum_{k=1}^{N}  \sum_{j=1}^k  N^2 \frac{(\sqrt k-\sqrt{k-1})^2 }{N} \frac{j}{N}  P_{k,N} P^{\chi^\perp}_{j,N}   \psi \rrangle^\frac{1}{2} \\
&\lesssim \norm{\varphi}^2_\infty \llangle \psi,   \sum_{k=1}^{N}  \sum_{j=1}^k   \frac{j}{k}  P_{k,N}  P^{\chi^\perp}_{j,N}   \psi \rrangle^\frac{1}{2} \\
&\lesssim \norm{\varphi}^2_\infty \llangle \psi,   \sum_{k=1}^{N}  \sum_{j=1}^N   P_{k,N}  P^{\chi^\perp}_{j,N}   \psi \rrangle^\frac{1}{2} \\
&= \norm{\varphi}^2_\infty \llangle \psi,   \sum_{j=1}^N    P^{\chi^\perp}_{j,N}   \psi \rrangle^\frac{1}{2} \weq{\ref{cor:energyepsi}}{ \lesssim} \norm{\varphi}^2_\infty N^{1/2}\epsi.
\end{align*}
%
% 
% \begin{align*}
%  |\langle \psi, p_1 p_2 N w^{\beta,\epsi}_{12} (\widehat  {m}- \widehat  {\tau_{-1} m}  ) p_1 q_2^\chi \psi \rangle  |
% &\lesssim  \norm{\varphi}^2_\infty \langle \psi,   (N(\widehat  {m}- \widehat  {\tau_{-1} m}  ))^2 q_2^\chi  \psi \rangle^\frac{1}{2} \\
% &= \norm{\varphi}^2_\infty \langle \psi,   \sum_{k=0}^{N} N^2 \frac{(\sqrt k-\sqrt{k-1})^2 }{N}  P_{k,N} \sum_{j=1}^N \frac{j}{N} P^{\chi^\perp}_{j,N}   \psi \rangle^\frac{1}{2} \\
% &= \norm{\varphi}^2_\infty \langle \psi,   \sum_{k=1}^{N}  \sum_{j=1}^k  N^2 \frac{(\sqrt k-\sqrt{k-1})^2 }{N} \frac{j}{N}  P_{k,N} \tilde P^{\chi^\perp}_{k,j,N}   \psi \rangle^\frac{1}{2} \\
% &\lesssim \norm{\varphi}^2_\infty \langle \psi,   \sum_{k=1}^{N}  \sum_{j=1}^k   \frac{j}{k}  P_{k,N} \tilde P^{\chi^\perp}_{k,j,N}   \psi \rangle^\frac{1}{2} \\
% &\lesssim \norm{\varphi}^2_\infty \langle \psi,   \sum_{k=1}^{N}  \sum_{j=1}^k   \tilde P^{\chi^\perp}_{k,j,N}   \psi \rangle^\frac{1}{2} \\
% &= \norm{\varphi}^2_\infty \langle \psi,   \sum_{j=1}^N    P^{\chi^\perp}_{j,N}   \psi \rangle^\frac{1}{2} \\
% &\lesssim \norm{\varphi}^2_\infty N^{1/2}\epsi
% \end{align*}
%
%
%
For the third summand $|\langle  \psi, p_1 p_2 w^{\epsi,\theta,N}_{12} (\widehat  n- \widehat  {\tau_{-1} n}  ) p_1 q_2   \psi \rangle  |$ of \eqref{equ:I.2g} we again use Lemma\,\ref{lem:young} and \ref{lem:qs&N}
and the fact that the $L^1$-norm of $w^{\epsi,\theta,N}$ is bounded to find
\begin{align*}
 |\llangle  \psi, p_1 p_2 w^{\epsi,\theta,N}_{12} (\widehat  {n}- \widehat  {\tau_{-1} n}  ) p_1 q_2   \psi \rrangle  | \leq N^{-1} \norm{\varphi}^2_\infty.
\end{align*}

\end{proof}

\begin{proof}[Proof of Lemma\;\ref{lem:3termeg}.\ref{lem:3.2g}  ]

We first note that for any function $f$

\begin{align}
 \norm{\sum_{j=2}^N q_j w^{\epsi,\theta,N}_{12}  \widehat f  p_1 p_j \psi }^2 
\lesssim N^2 \norm{\varphi}_{\infty}^4 \norm{ \widehat f \widehat n \psi }^2 + N N^{3\theta} \epsi^{-6\theta+2} \norm{\varphi}_{\infty}^2 \sup_{1 \leq k \leq N} |f(k,N)|^2.
 \label{equ:II.6gp}
\end{align}
To prove this we split the right-hand side of \eqref{equ:II.6gp} into the "diagonal" and the "off-diagonal" term and find

\begin{align*}
 \norm{\sum_{j=2}^N q_j w^{\epsi,\theta,N}_{12}  \widehat f  p_1 p_j \psi }^2
&= \sum_{j,k=2}^N \llangle \psi, p_1 p_l \widehat f  w^{\epsi,\theta,N}_{1l}  q_l q_j w^{\epsi,\theta,N}_{1j} 
 \widehat f  p_1 p_j \psi  \rrangle\\
\begin{split}
&\leq \sum_{2 \leq  j < k \leq N} \llangle \psi, q_j p_1 p_l \widehat f  w^{\epsi,\theta,N}_{1l}    w^{\epsi,\theta,N}_{1j} q_l \widehat f  p_1 p_j \psi  \rrangle\\
&\quad+(N-1) \norm{w^{\epsi,\theta,N}_{12}  \widehat f  p_1 p_2 \psi   }^2 .
 \end{split} \numberthis \label{equ:II.3gp}
\end{align*}
The first summand of \eqref{equ:II.3gp} is bounded by
\begin{align*}
 &(N-1)(N-2)\llangle \psi, q_2 p_1 p_3 \widehat f  w^{\epsi,\theta,N}_{13}    w^{\epsi,\theta,N}_{12} q_3 \widehat f  p_1 p_2 \psi  \rrangle\\
&\leq N^2 \norm { \sqrt{ w^{\epsi,\theta,N}_{13}}    \sqrt{ w^{\epsi,\theta,N}_{12}} q_3 \widehat f  p_1 p_2 \psi}^2\\
&\leq N^2 \norm{   \sqrt{ w^{\epsi,\theta,N}_{12}} p_2 \sqrt{ w^{\epsi,\theta,N}_{13}} p_1\widehat f q_3 \psi } ^2 \\
&\leq N^2 \norm{  \sqrt{ w^{\epsi,\theta,N}_{12}} p_2  }^4_\mathrm{Op} \norm{ \widehat f q_3 \psi }^2\\
&\weq{\ref{lem:weights}}{ \leq} N^2 \norm{p_1 w^{\epsi,\theta,N}_{12} p_1  }_\mathrm{Op}^2 \norm{ \widehat f \widehat n \psi }^2\\
&\weq{\ref{lem:young}}{\lesssim} N^2 \norm{\varphi}_{\infty}^4 \norm{ \widehat f \widehat n \psi }^2. \numberthis \label{equ:II.4gp}
\end{align*}
The second summand of \eqref{equ:II.3gp} is bounded by
\begin{align*}
&N \llangle \psi p_1 p_2 \widehat f (w^{\epsi,\theta,N}_{12})^2 \widehat f  p_1 p_2 \psi  \rrangle \\
&\leq N   \norm{p_1 (w^{\epsi,\theta,N}_{12})^2 p_1 }_\mathrm{Op} \norm{\widehat f }^2_\mathrm{Op}\\
& \weq{\ref{lem:young}}{ \lesssim } N N^{3\theta} \epsi^{-6\theta+2} \norm{\varphi}_{\infty}^2 \sup_{1 \leq k \leq N} |f(k,N)|^2   \numberthis \label{equ:II.5gp}
\end{align*}
since %$\norm{w^{\epsi,\theta,N}}_2^2 \lesssim (\frac{N}{\epsi^2})^{3 \beta} \epsi^2 $.
\begin{align}\label{equ:wt2norm}
 \norm{w^{\epsi,\theta,N}}_2^2 \lesssim (\frac{N}{\epsi^2})^{3 \theta} \epsi^2.
\end{align}
Putting \eqref{equ:II.4gp} and \eqref{equ:II.5gp} together proves \eqref{equ:II.6gp}.
%
% \begin{align}
%  \norm{\sum_{j=2}^N q_j w^{\epsi,\theta,N}_{12}  \widehat r  p_1 p_j \psi }^2 
% \lesssim N^2 \norm{\varphi}_{\infty}^4 \norm{ \widehat r \widehat n \psi }^2 + N N^{3\beta} \epsi^{-6\beta+2} \norm{\varphi}_{\infty}^2 \sup_{1 \leq k \leq N} |r(k,N)|^2
%  \label{equ:II.6gp}
% \end{align}
%
To apply $\eqref{equ:II.6gp}$ to II we define for any function $f: \{0,\dots, N \} \rightarrow \R^+$ and $\delta> 0$
\begin{align*}
 f^a(k):= \begin{cases}
           f(k)\quad &\mathrm{for} \quad k < N^{1-\delta} \\
	   0  &\mathrm{for} \quad  k\geq N^{1-\delta} 
          \end{cases}         \numberthis \label{equ:splitting}                                                                                                        
\end{align*}
and $f^b:=f-f^a $. Furthermore we define 
\begin{align}\label{equ:mugp}
 \mu:= (N-1)(  n-  {\tau_{-2} n} )% \leq  \sqrt{N}(\sqrt{k}-\sqrt{k-2}) 
\leq \frac{ 2\sqrt{N}}{\sqrt{k}+\sqrt{k-2}}\leq \frac{\sqrt{N}}{\sqrt{k}}= n^{-1} \quad \forall k \geq 2
\end{align}
and estimate II by 

\begin{align*}
 |\mathrm{II}|&= N \big| \llangle \psi, p_1 p_2  [ (N-1)  w^{\epsi,\theta,N}_{12},  \widehat{n}] q_1q_2 \psi \rrangle\big|\\
&\weq{\ref{lem:weights}}{ =}  N \big |\llangle  \psi, p_1 p_2 (N-1)w^{\epsi,\theta,N}_{12} ( { \widehat  {n}}- { \widehat  {\tau_{-2} n}}  ) q_1 q_2   \psi \rrangle \big|   \\
&=  \big |\llangle  \psi, p_1 p_2 w^{\epsi,\theta,N}_{12}  \widehat \mu   q_1 q_2   \psi \rrangle  \big| \\
\begin{split}
&\weq{\eqref{equ:splitting}}{ \leq} \, \big |\llangle  \psi, p_1 p_2 w^{\epsi,\theta,N}_{12}  \widehat \mu^a   q_1 q_2   \psi \rrangle \big|  \\
&\quad +  \big |\llangle  \psi, p_1 p_2 w^{\epsi,\theta,N}_{12}  \widehat \mu^b   q_1 q_2   \psi \rrangle  \big|. 
\end{split} \numberthis \label{equ:II.-1}
\end{align*}
We define the constant function  $g:\{0,\dots,N\} \rightarrow 1$ hence $\mu^a=\mu^a g^a$. Inserting this in the first factor of \eqref{equ:II.-1} we get
\begin{align*}
 \big |\llangle  \psi,  p_1 p_2 w^{\epsi,\theta,N}_{12}  \widehat \mu^a \widehat g^a   q_1 q_2   \psi \rrangle \big|&\weq{\ref{lem:weights}}{ =}\big |\llangle  \psi,  \widehat {\tau_2 g^a}   p_1 p_2 w^{\epsi,\theta,N}_{12}  \widehat \mu^a   q_1 q_2   \psi \rrangle \big| \\
& = \frac{1}{N}  \big |\llangle  \psi, \sum_{j=2}^N \widehat {\tau_2 g^a}  p_1 p_j w^{\epsi,\theta,N}_{1j}  \widehat \mu^a   q_1 q_j   \psi \rrangle \big|\\
&\leq \frac{1}{N}  \norm{\sum_{j=2}^N q_j w^{\epsi,\theta,N}_{12}  \widehat {\tau_2 g^a}  p_1 p_j \psi } \norm{ \widehat \mu^a   q_1 \psi}.
\end{align*}
Since $ \norm{ \widehat \mu^a   q_1 \psi} \weq{\ref{lem:weights}}{ \leq} 1$ and in view of \eqref{equ:II.6gp} this can be estimated by

\begin{align*}
 &\frac{1}{N} \big( N \norm{\varphi}_{\infty}^2 \norm{ \widehat g_2^a \widehat n \psi } + N^\frac{1}{2} N^{\frac{3\theta}{2}} \epsi^{-3\theta+1} \norm{\varphi}_{\infty} \sup_{1 \leq k \leq N} |g^a_2(k)| \big)\\
&\lesssim N^{-\frac{\delta}{2}} \norm{\varphi}_{\infty}^2+ N^{\frac{-1}{2}}N^{\frac{3\theta}{2}} \epsi^{-3\theta+1} \norm{\varphi}_{\infty}. \numberthis \label{equ:II.-2}
\end{align*}
The second summand of \eqref{equ:II.-1} can be estimated in the following way

\begin{align*}
\big |\llangle  \psi, p_1 p_2 w^{\epsi,\theta,N}_{12}  \widehat \mu^b   q_1 q_2   \psi \rrangle  \big|  &\weq{\ref{lem:weights}}{=} |\llangle  \psi, p_1 p_2  (\widehat {\tau_2 \mu}^b  )^\frac{1}{2} w^{\epsi,\theta,N}_{12} (\widehat \mu^b  )^\frac{1}{2}  q_1 q_2  \psi \rrangle  |\\
&= \frac{1}{N} |\sum_{j=2}^N  \llangle  \psi,   (\widehat{ \tau_2 \mu}^b )^\frac{1}{2} p_1 p_j w^{\epsi,\theta,N}_{1j}  q_1 q_j (\widehat \mu^b  )^\frac{1}{2}   \psi \rrangle |\\
&\weq{\ref{lem:weights}}{ \lesssim}  \frac{1}{N} \norm{ (\widehat \mu^b  )^\frac{1}{2} q_1 \psi}\norm{\sum_{j=2}^N  q_j  w^{\epsi,\theta,N}_{1j} ( \widehat { \tau_2 \mu}^b  )^\frac{1}{2}  p_1 p_j \psi } .\numberthis \label{equ:II.1}
\end{align*}
The first factor of \eqref{equ:II.1} is estimated by
\begin{align*}
 \norm{ (\widehat \mu^b  )^\frac{1}{2} q_1 \psi}^2=\llangle \psi, \widehat \mu^b   q_1 \psi \rrangle  \weq{\ref{lem:weights}}{ \leq} \beta. \numberthis \label{equ:II.2}
%\leq \llangle \psi, \sum_{j=1}^{N}  \sqrt{ \frac{N}{j}} \frac{j}{N} P_{j,N} \psi \rrangle \leq \alpha \numberthis \label{equ:II.2}
\end{align*}
For the second factor we use \eqref{equ:II.6gp}. Since %$\sup_{1\leq k \leq N } (\mu(k)^b)^{1/2} \weq{\eqref{equ:splitting},\eqref{equ:mugp}}{ \leq} N^{\frac{\delta}{4}} $ and
\begin{align}\label{equ:muab}
 \sup_{1\leq k \leq N } (\mu(k)^b)^{1/2} \; \weq{\substack{\eqref{equ:splitting}\\ \eqref{equ:mugp}}}{ \leq} \; N^{\frac{\delta}{4}}
\end{align}
and
 
\begin{align*}
 \norm{(\widehat { \tau_2 \mu}^b  )^\frac{1}{2} \widehat n \psi}^2%=  \langle \psi, \sum_{k=N^{1-\eta}}^N \sqrt{ \frac{N}{k}} \frac{k}{N} P_{k,N} \psi  \rangle 
\lesssim \beta
\end{align*}
we get
\begin{align*}
 \big |\llangle  \psi, p_1 p_2 & w^{\epsi,\theta,N}_{12}  \widehat \mu^b   q_1 q_2   \psi \rrangle  \big|\\
 &\weq{ \eqref{equ:II.6gp}} { \lesssim} \sqrt \beta  \Big ( \norm{\varphi}_{\infty}^2 \norm{ (\widehat \mu^b )^\frac{1}{2} \widehat n \psi } + N^{-\frac{1}{2}} N^{\frac{3\theta}{2}} 
\epsi^{-3\theta+1} \norm{\varphi}_{\infty} \sup_{1 \leq k \leq N} |(\mu(k)^b)^{1/2}|\Big )\\
&\weq{\eqref{equ:muab}}{ \leq}  \norm{\varphi}_{\infty}^2 \beta + N^{-\frac{1}{2}} N^{\frac{3\theta}{2}}N^{\frac{\delta}{4}}\epsi^{-3\theta+1}  \norm{\varphi}_{\infty}, \numberthis  \label{equ:werwei}
\end{align*}
where we refrain from taking the square of the second term which results in slower convergence rates but simplifies the next calculation. 
Combining \eqref{equ:werwei} with the estimate \eqref{equ:II.-2} and inserting them in \eqref{equ:II.-1} yields the claimed result
\begin{align*}
 |\mathrm{II}| &\lesssim \norm{\varphi}_{\infty}^2 \alpha + N^{-\frac{1}{2}} N^{\frac{3\theta}{2}}N^{\frac{\delta}{4}}\epsi^{-3\theta+1}  \norm{\varphi}_{\infty} +
 N^{-\frac{\delta}{2}} \norm{\varphi}_{\infty}^2+ N^{\frac{-1}{2}}N^{\frac{3\theta}{2}} \epsi^{-3\theta+1} \norm{\varphi}_{\infty}\\
&\lesssim \norm{\varphi}_{\infty}^2 \alpha + N^{-\frac{1}{2}} N^{\frac{3\theta}{2}}N^{\frac{\delta}{4}}\epsi^{-3\theta+1}  \norm{\varphi}_{\infty} +
 N^{-\frac{\delta}{2}} \norm{\varphi}_{\infty}^2.
\end{align*}
The optimal $\delta$ and therefore the optimal convergence rate of this term depends on $\theta$ and $\nu$.
For fixed $\theta$ and $\nu$ the optimal $\delta$ can be found by setting $N^{-\frac{1}{2}} N^{\frac{3\theta}{2}}N^{\frac{\delta}{4}}\epsi^{-3\theta+1} \sim  N^{-\frac{\delta}{2}}$ 
 under the constraint $0 < \delta$. Such a $\delta$ exists for $\theta \in (\frac{1}{4},\frac{1}{3})$.
% For give $0 < \delta(\theta,\nu) $ such that $N^{-\frac{1}{2}} N^{\frac{3\theta}{2}}N^{\frac{\delta}{4}}\epsi^{-3\theta+1} \sim  N^{-\frac{\delta}{2}} $ 
%which is possible for $\theta \in (\frac{1}{4},\frac{1}{3})$.
 \end{proof}

\begin{proof}[Proof of Lemma\;\ref{lem:3termeg}.\ref{lem:3.3g}  ]
For this term we can use the abundance of $q$s to extract terms with 
enough negative power of $N$ to get convergence. We will use the function 
\begin{align}\label{equ:mu3}
 \mu:= N(  n-  {\tau_{-1} n} ) = \sqrt{N}(\sqrt{k}-\sqrt{k-1})= \frac{\sqrt{N}}{\sqrt{k}+\sqrt{k-1}}\leq \frac{\sqrt{N}}{\sqrt{k}}= n^{-1} \quad \forall k \geq 1.
\end{align}
We begin with the usual simplifications
\begin{align*}
 |\mathrm{III}|&= |\llangle \psi, p_1 q_2  [ (N-1)  w^{\epsi,\theta,N}_{12}- N b|\Phi|^2(x_1) ,  \widehat  n] q_1q_2 \psi \rrangle|\\
&\weq{\ref{lem:weights}}{=} |\llangle  \psi, p_1 q_2 \Big((N-1)w^{\epsi,\theta,N}_{12}-N b|\Phi|^2(x_1)\Big) (\widehat  n- \widehat  {\tau_{-1} n}  ) q_1 q_2   \psi \rrangle  |\\
&\weq{\eqref{equ:mu3}}{ =}\; |\llangle  \psi, p_1  q_2  \Big( \frac{N}{N-1} w^{\epsi,\theta,N}_{12}- b|\Phi|^2(x_1)\Big) \widehat \mu q_1 q_2   \psi \rrangle  |\\
\begin{split}
&\lesssim |\llangle  \psi, p_1  q_2  w^{\epsi,\theta,N}_{12} \widehat \mu q_1 q_2   \psi \rrangle  |\\
&\quad +|\llangle  \psi, p_1  q_2 b|\Phi|^2(x_1) \widehat \mu q_1 q_2   \psi \rrangle  |.\\
\end{split} \numberthis \label{equ:III.1gp}
\end{align*}
The second term of \eqref{equ:III.1gp} can be estimated by
\begin{align*}
 |\llangle  \psi, p_1  q_2 b|\Phi|^2(x_1) \widehat \mu q_1 q_2   \psi \rrangle  | \weq{\ref{lem:young}} {\lesssim}  \norm {q_2 \psi} \norm{\widehat \mu q_1 q_2 \psi} \weq{\ref{lem:weights}}{\leq} \beta.
\end{align*}
For the first term of \eqref{equ:III.1gp} we use $q=q^\chi + p^\chi q^\Phi$ to obtain four terms
\begin{align*}
|\llangle  \psi, p_1  q_2  w^{\epsi,\theta,N}_{12} \widehat \mu q_1 q_2   \psi \rrangle  | &\leq
|\llangle  \psi, p_1  p_2^\chi q_2^\Phi  w^{\epsi,\theta,N}_{12} \widehat \mu p_1^\chi q_1^\Phi p_2^\chi q_2^\Phi   \psi \rrangle  |\\
&\quad+|\llangle  \psi, p_1  q_2^\chi  w^{\epsi,\theta,N}_{12} \widehat \mu q_1 q_2   \psi \rrangle  |\\
&\quad+|\llangle  \psi, p_1  q_2  w^{\epsi,\theta,N}_{12} \widehat \mu  q_1^\chi q_2  \psi \rrangle  |\\
&\quad +|\llangle  \psi, p_1  q_2  w^{\epsi,\theta,N}_{12} \widehat \mu q_1 q_2^\chi   \psi \rrangle  |.
 \numberthis \label{equ:III.2gp}
\end{align*}
All terms but the first are easy to handle.
The second term of \eqref{equ:III.2gp} can be estimated by
\begin{align*}
 |\llangle  \psi, p_1  q_2^\chi  w^{\epsi,\theta,N}_{12} \widehat \mu q_1 q_2   \psi \rrangle  | \leq \norm{ q_2^\chi \psi} \norm{w^{\epsi,\theta,N}_{12}p_1}_\mathrm{Op} \norm{\widehat \mu q_1 q_2 \Psi}\\
%\stackrel{\ref{lem:weights}\ref{lem:young}}{ \leq} 
\lesssim \epsi (N\epsi^{-2})^\frac{3\theta}{2} \epsi \sqrt{\beta} \norm{\varphi}_\infty \leq  \norm{\varphi}_\infty \big( \beta + \epsi^4 (N\epsi^{-2})^{3\theta}\big), \numberthis \label{equ:III.8gp}
\end{align*}
where we used Lemmas\,\ref{lem:weights},\,\ref{lem:young} and \ref{lem:energyepsi} and equation \eqref{equ:wt2norm} in the second step.    
The third and the fourth term of \eqref{equ:III.2gp} can be estimated in the same way if we use $q^\chi \leq q$. Hence we find 
\begin{align*}
 |\llangle  \psi, p_1  q_2  w^{\epsi,\theta,N}_{12} \widehat \mu  q_1^\chi q_2  \psi \rrangle  |,|\llangle  \psi, p_1  q_2  w^{\epsi,\theta,N}_{12} \widehat \mu  q_1 q_2^\chi  \psi \rrangle  | %& \leq \norm{ q_2 \psi} \norm{w^{\epsi,\theta,N}_{12}p_1}_\mathrm{Op} \norm{\widehat \mu q^\chi_1 q_2 \Psi}\\
%&\leq  \sqrt{\alpha} (N\epsi^{-2})^\frac{3\beta}{2} \epsi \llangle \psi, \sum_k^N \frac{N}{k}  \frac{k}{N} P_{k,N} \sum_{j=1}^N \frac{j}{N} P^{\chi^\perp}_{j,N} \psi \rrangle^{\frac{1}{2}}\\
%&\leq \sqrt{\alpha} (N\epsi^{-2})^\frac{3\beta}{2} \epsi^2 \leq 
\lesssim \norm{\varphi}_\infty \big(  \beta + \epsi^4 (N\epsi^{-2})^{3\theta} \big) \numberthis \label{equ:III.7gp}.
%|\llangle  \psi, p_1  q_2  w^{\epsi,\theta,N}_{12} \widehat \mu  q_1 q_2^\chi  \psi \rrangle  | \lesssim \norm{\varphi}_\infty \big(  \beta + \epsi^4 (N\epsi^{-2})^{3\theta} \big) \numberthis \label{equ:III.7gp}
\end{align*}
%We can estimate the fourth term of \eqref{equ:III.2gp} in the same way as the third term.
%
For the first term of \eqref{equ:III.2gp} we have to use a different approach. Here we know that the potential only acts on the function $\chi$ in the confined direction. Thus,
 we can integrate the potential explicitly in this direction
\begin{align*}
 &|\llangle  \psi, p_1  p_2^\chi q_2^\Phi  w^{\epsi,\theta,N}_{12} \widehat \mu p_1^\chi q_1^\Phi p_2^\chi q_2^\Phi   \psi \rrangle  |\\
\begin{split}
&=  |\llangle  \psi, p_1  p_2^\chi q_2^\Phi \int_{\Omega_c} \int_{\Omega_c} ( N\epsi^{-2})^{3 \theta} \epsi^2  
w\Big ( (N\epsi^{-2})^\theta \big(x_1-x_2,\epsi(y_1-y_2)\big) \Big) \\
&\quad \times |\chi(y_1)|^2 |\chi(y_2)|^2 \D y_1 \D y_2  p_1^\chi q_1^\Phi p_2^\chi q_2^\Phi  \widehat \mu   \psi \rrangle  |.\\
\end{split}\numberthis \label{equ:wxxxx}
\end{align*}
 For short notation we define the function 
\begin{align*}
 \tilde w^{\epsi,\theta,N}(x_1-x_2):&= \int_{\Omega_c} \int_{\Omega_c} ( N\epsi^{-2})^{3 \beta} \epsi^2  w\Big ((N\epsi^{-2})^\beta \big(x_1-x_2,\epsi(y_1-y_2)\big) \Big)\\
&\quad \times |\chi(y_1)|^2 |\chi(y_2)|^2 \D y_1 \D y_2 
\end{align*}
and since it lives in one dimension we can explicitly define its anti-derivative   
\begin{align*}
\tilde W^{\epsi,\theta,N}(x_1-x_2):= \int_{-\infty}^{x_1-x_2} \tilde  w^{\epsi,\theta,N}(x) \D x.
\end{align*}
The next step is to estimate the operator norm of the multiplication operator $\tilde W^{\epsi,\theta,N}$ by scaling arguments.
 Set $\tilde x = (N\epsi^{-2})^\theta x $, $\tilde y = \epsi (N\epsi^{-2})^\theta y $ and $\tilde \Omega_c=  \epsi (N\epsi^{-2})^\theta \Omega_c $, so

\begin{align*}
 &\norm{\tilde W^{\epsi,\theta,N}(x_1-x_2)}_{\infty}= \sup_{x_1,x_2 \in \R}   \int_{-\infty}^{x_1-x_2}  \tilde w^{\epsi,\theta,N}(x) \D x \\
&= \sup_{x_1,x_2 \in \R} \int_{-\infty}^{x_1-x_2}  \int_{\Omega_c} \int_{\Omega_c} ( N\epsi^{-2})^{3 \theta} \epsi^2  w\Big ((N\epsi^{-2})^\theta \big(x,\epsi(y_1-y_2)\big) \Big)\\
&\phantom{ \sup_{x_1,x_2 \in \R} \int_{-\infty}^{x_1-x_2}  \int_{\Omega_c} \int_{\Omega_c} ( N\epsi^{-2})^{3 \theta} \epsi^2  w\Big ((N\epsi^{-2})^\theta  \qquad} \times |\chi(y_1)|^2 
|\chi(y_2)|^2 \D y_1 \D y_2 \D x \\
&= \sup_{x_1,x_2 \in \R} \int_{-\infty}^{(N\epsi^{-2})^\theta(x_1-x_2)}  \int_{\Omega_c} \int_{\Omega_c} ( N\epsi^{-2})^{2 \theta} \epsi^2  w\Big (\tilde x,(N\epsi^{-2})^\theta \epsi(y_1-y_2) \Big) \\
&\phantom{ \sup_{x_1,x_2 \in \R} \int_{-\infty}^{x_1-x_2}  \int_{\Omega_c} \int_{\Omega_c} ( N\epsi^{-2})^{3 \theta} \epsi^2  w\Big ((N\epsi^{-2})^\theta  \qquad}  \times
|\chi(y_1)|^2 |\chi(y_2)|^2 \D y_1 \D y_2 \D  \tilde x \\
&\leq  \int_{-\infty}^{\infty} \int_{ \Omega_c} \int_{ \Omega_c} ( N\epsi^{-2})^{2 \theta} \epsi^2  w\Big (\tilde x,(N\epsi^{-2})^\theta \epsi(y_1-y_2) \Big) 
|\chi(y_1)|^2 |\chi(y_2)|^2 \D y_1 \D y_2 \D  \tilde x \\
&= \int_{-\infty}^{\infty} \int_{\tilde \Omega_c} \int_{\tilde \Omega_c} ( N\epsi^{-2})^{- 2 \theta} \epsi^{-2}  w\Big ( \tilde x, \tilde y_1- \tilde y_2 \Big)
 |\chi( \frac{\tilde y_1}{\epsi (N\epsi^{-2})^{\theta}})|^2  \\
&\phantom{ \sup_{x_1,x_2 \in \R} \int_{-\infty}^{x_1-x_2}  \int_{\Omega_c} \int_{\Omega_c} ( N\epsi^{-2})^{3 \theta} \epsi^2  w\Big ((N\epsi^{-2})^\theta  \qquad}  \times |\chi( \frac{\tilde y_2}{\epsi (N\epsi^{-2})^\theta})|^2 \D \tilde y_1 \D \tilde y_2  \D \tilde x\\
&\leq   \sup_{\tilde y_2} |\chi( \frac{\tilde y_2}{\epsi (N\epsi^{-2})^\theta})|^2 \int_{-\infty}^{\infty} \int_{ \R^2} \int_{ \R^2} 
( N\epsi^{-2})^{- 2 \theta} \epsi^{-2}  w\Big ( \tilde x,\tilde  y_1-\tilde y_2 \Big) \\
&\phantom{ \sup_{x_1,x_2 \in \R} \int_{-\infty}^{x_1-x_2}  \int_{\Omega_c} \int_{\Omega_c} ( N\epsi^{-2})^{3 \theta} \epsi^2  w\Big ((N\epsi^{-2})^\theta  \qquad}\times  |\chi( \frac{\tilde y_1}{\epsi (N\epsi^{-2})^{\theta}})|^2 \D \tilde y_1 \D \tilde y_2  \D \tilde x\\
& \leq \norm{ \chi}_\infty^2 \int_{-\infty}^{\infty} \int_{ \R^2} w\Big ( \tilde x, \tilde y \Big)
  \D \tilde y         \int_{\R^2} ( N\epsi^{-2})^{- 2 \theta} \epsi^{-2}   |\chi( \frac{\tilde y_1}{\epsi (N\epsi^{-2})^{\theta}})|^2   \D \tilde y   \D \tilde x \\
&\leq \norm{ \chi}_\infty^2  \norm{w}_1 \numberthis \label{equ:III.3gp},
\end{align*}
where the last step holds since $\chi $ is normed.
%
% \begin{align*}
% =\int_{\Omega_2} \int_{\Omega_2} ( N\epsi^{-2})^{- \beta} \epsi^{-2}  w\Big ((N\epsi^{-2})^\beta(x_1-x_2),y_1-y_2 \Big) 
% |\chi( \frac{\tilde y_1}{\epsi (N\epsi^{-2})^{\beta}})|^2 |\chi( \frac{\tilde y_2}{\epsi (N\epsi^{-2})^\beta})|^2 \D \tilde y_1 \D \tilde y_2 \\
% \leq  \sup_{\tilde y_2} |\chi( \frac{\tilde y_2}{\epsi (N\epsi^{-2})^\beta})|^2 \int_{\Omega_2} \int_{\Omega_2} ( N\epsi^{-2})^{ - \beta} \epsi^{-2}  w\Big ((N\epsi^{-2})^\beta(x_1-x_2),y_1-y_2 \Big) 
% |\chi( \frac{\tilde y_1}{\epsi (N\epsi^{-2})^\beta})|^2 \D \tilde y_1 \D \tilde y_2 \\
% \leq \sup_{\tilde y_2} |\chi( \frac{\tilde y_2}{\epsi (N \epsi^{-2})^\beta }|^2 \int_{\Omega_2} ( N\epsi^{-2})^{ \beta}  w\Big ((N\epsi^{-2})^\beta(x_1-x_2),y \Big)  \D  y	
%  \int_{\Omega_2} ( N\epsi^{-2})^{-2 \beta} \epsi^{-2}   |\chi( \frac{ y}{\epsi (N\epsi^{-2})^\beta})|^2 \D  y	\\
% \lesssim  \int_{\Omega_2} ( N\epsi^{-2})^{ \beta}  w\Big ((N\epsi^{-2})^\beta(x_1-x_2),y \Big)  \D  y	
% \end{align*}
% 
% \textcolor{red}{Für stetige Funtkionen würde gehen $.. \leq C \sup_{y \in \R^2 } ( N\epsi^{-2})^{ \beta}  w\Big ((N\epsi^{-2})^\beta(x_1-x_2),y \Big) $ }
% Definiere Stammfunktion $h^{\epsi,\beta}(x_1-x_2)$ mit
% \begin{align*}
%  \frac{d}{dx } h^{\epsi,\beta}(x_1-x_2) = w^{\beta,\epsi}(x_1-x_2)
% \end{align*}
% aus skalierungs verhalten von $w^{\epsi,\beta} $ folgt dass
% \begin{align*}
%  \sup_{x \in \R } h^{\beta,\epsi}(x) \leq C
% \end{align*}
% damit ergibt sich
To use this estimate for $\tilde W$ we rewrite term III by integrating by parts
\begin{align*}
 |\llangle  \psi, p_1  p_2^\chi q_2^\Phi  w^{\epsi,\theta,N}_{12} \widehat \mu p_1^\chi q_1^\Phi p_2^\chi q_2^\Phi   \psi \rrangle  |\, & \weq{\eqref{equ:wxxxx}}{=}\, 
 %|\llangle  \psi, p_1  p_2^\chi q_2^\Phi  w^{\theta,\epsi}(x_1-x_2) \widehat \mu p_1^\chi q_1^\Phi p_2^\chi q_2^\Phi   \psi \rrangle  |\\
 |\llangle  \psi, p_1  p_2^\chi q_2^\Phi   \frac{d}{dx_1} \tilde W^{\epsi,\theta,N}(x_1-x_2) \widehat \mu p_1^\chi q_1^\Phi p_2^\chi q_2^\Phi   \psi \rrangle  |\\
&\leq |\llangle  \psi,  p_2^\chi q_2^\Phi  ( \frac{d}{dx_1} p_1) \tilde  W^{\theta,\epsi}(x_1-x_2) \widehat \mu p_1^\chi q_1^\Phi p_2^\chi q_2^\Phi   \psi \rrangle  |\\
&\, + |\llangle  \psi, p_1  p_2^\chi q_2^\Phi \tilde  W^{\epsi,\theta,N}(x_1-x_2) \frac{d}{dx_1} \widehat \mu p_1^\chi q_1^\Phi p_2^\chi q_2^\Phi   \psi \rrangle  |. \numberthis \label{equ:III.4gp}
\end{align*}
The first term of equation \eqref{equ:III.4gp} is bounded by
\begin{align*}
& |\llangle  \psi,  p_2^\chi q_2^\Phi  ( \frac{d}{dx_1} p_1) \tilde  W^{\epsi,\epsi,N}(x_1-x_2) \widehat \mu p_1^\chi q_1^\Phi p_2^\chi q_2^\Phi   \psi \rrangle  | \\
&\leq \norm{p_2^\chi q_2^\Phi \psi}\norm{\frac{d}{dx}p}_\mathrm{Op} \norm{ \tilde W^{\epsi,\theta,N}}_\infty \norm{\widehat \mu p_1^\chi q_1^\Phi p_2^\chi q_2^\Phi   \psi}\\
&\weq{\ref{lem:weights},\eqref{equ:III.3gp}}{ \lesssim} \; \sqrt{ \beta} \norm{\frac{d}{dx} \Phi } \norm{ \chi}_\infty^2 \sqrt{ \beta}
 \leq \norm{ \Phi }_{H^1}  \norm{ \chi}_\infty^2  \beta  \numberthis \label{equ:III.6gp}.
\end{align*}
The second term of \eqref{equ:III.4gp} we estimate by
\begin{align*}
 |\llangle  \psi, p_1  p_2^\chi q_2^\Phi \tilde   W^{\epsi,\theta,N}(x_1-x_2) \frac{d}{dx_1} \widehat \mu p_1^\chi q_1^\Phi p_2^\chi q_2^\Phi   \psi \rrangle  |\\
\weq{ \ref{lem:weights},\eqref{equ:III.3gp} }{ \lesssim }\; \sqrt{\beta}  \norm{ \chi}_\infty^2  \norm{\frac{d}{dx_1} \widehat \mu p_1^\chi q_1^\Phi p_2^\chi q_2^\Phi   \psi } \numberthis \label{equ:III.5gp}.
\end{align*}
To bound the last term we note 
\begin{align*}
 \norm{\frac{d}{dx_1} \widehat \mu p_1^\chi q_1^\Phi p_2^\chi q_2^\Phi   \psi }^2  \lesssim  \norm{\frac{d}{dx_1}p_1^\chi q_1^\Phi \psi}^2,
\end{align*}
where the proof follows exactly the same pattern as the one for $\kappa$ in equation \eqref{equ:kappa}.
We continue by bounding the right-hand side
\begin{align*}
 \norm{\frac{d}{dx_1}p_1^\chi q_1^\Phi \psi}^2 &\leq \norm{\frac{d}{dx_1}p_1^\chi q_1^\Phi \psi}^2+ \norm{\frac{d}{dx_1} q_1^\chi \psi}^2 \\
&= \llangle \psi, p_1^\chi q_1^\Phi \frac{-d^2}{dx_1^2}p_1^\chi q_1^\Phi \psi  \rrangle + \llangle \psi, q_1^\chi \frac{-d^2}{dx_1^2} q_1^\chi \psi  \rrangle\\
&= \llangle \psi,( p_1^\chi q_1^\Phi +q_1^\chi )\frac{-d^2}{dx_1^2}(p_1^\chi q_1^\Phi +q_1^\chi) \psi  \rrangle\\
&=\llangle \psi,q_1 \frac{-d^2}{dx_1^2}q_1 \psi  \rrangle =  \norm{\frac{d}{dx_1}q_1 \psi}^2 \leq \norm{\nabla q_1 \psi}^2.
\end{align*}
Finally this estimate together with the Energy Lemma 
\begin{align*}
 \norm{\nabla q_1 \psi}^2 \leq  \norm{\varphi}_{H^2\cap L^\infty}^2( \tilde \beta+\frac{1}{\sqrt{N}}+f(N,\epsi))+\norm{V}_\LiO \beta
\end{align*}
 leads to 
\begin{align}\label{equ:estlast}
 \norm{\frac{d}{dx_1} \widehat \mu p_1^\chi q_1^\Phi p_2^\chi q_2^\Phi   \psi }^2 \leq  \norm{\varphi}_{H^2\cap L^\infty}^2(\tilde \beta+\frac{1}{\sqrt{N}}+f(N,\epsi))+\norm{V}_\LiO \beta.
\end{align}
%
%yields an estimate for the second factor of equation \eqref{equ:III.5gp} and thus we can bound the second term of \eqref{equ:III.4gp} by
Inserting \eqref{equ:estlast} into \eqref{equ:III.5gp} results in 
\begin{align*}
 |\llangle  \psi, p_1  p_2^\chi q_2^\Phi & \tilde  W^{\epsi,\theta,N}(x_1-x_2) \frac{d}{dx_1} \widehat \mu p_1^\chi q_1^\Phi p_2^\chi q_2^\Phi   \psi \rrangle  |\\
&\lesssim \sqrt{ \beta}  \norm{ \chi}_\infty^2  \Big( \norm{\varphi}^2_{H^2\cap L^\infty}(\tilde \beta+\frac{1}{\sqrt{N}}+f(N,\epsi)) +\norm{V}_\LiO \beta \Big)^{\frac{1}{2}}\\
%= \norm{\varphi}_{H^2\cap L^\infty} \sqrt{ \alpha} \Big( \alpha + \norm{\varphi}_{H^2\cap L^\infty}^{-2} (E^\psi- E^\phi)+\frac{1}{\sqrt{N}}+f(N,\epsi)  \Big )^{\frac{1}{2}}\\
&\leq  \norm{\varphi}_{H^2\cap L^\infty}  \norm{ \chi}_\infty^2 (\tilde \beta +\frac{1}{\sqrt{N}}+f(N,\epsi) ) + \norm{ \chi}_\infty^2 \norm{V}_\LiO^{1/2} \beta .
\end{align*}
Combining this estimate with \eqref{equ:III.8gp},\eqref{equ:III.7gp} and  \eqref{equ:III.6gp} finishes this part of the lemma.

\end{proof}

\begin{proof}[Proof of Lemma\;\ref{lem:3termeg}.\ref{lem:3.3g}  ]

For both summands in IV we expand the potential around $y_1=0$. The assumption B2 guarantees that in both cases the error is a bounded operator.
Therefore, we can write
\begin{align*}
 \dot V(x_1,\epsi y_1)= \dot V(x_1,0)+ \epsi R \qquad  V(x_1,\epsi y_1)=  V(x_1,0)+ \epsi \tilde R
\end{align*}
with $\norm{R}_\mathrm{Op}, \|\tilde R \| _\mathrm{Op} \leq C$. Thus we find for the second part of IV 
\begin{align*}
 2 |\llangle \psi, p_1  N& [   V(x_1,\epsi y_1)-V(x_1,0) , \widehat n] q_1 \psi \rrangle| \\&= 2| \llangle \psi, p_1  N [   V(x_1,0)+\epsi  R -V(x_1,0) , \widehat n] q_1 \psi \rrangle| \\
&= 2 |\llangle \psi, p_1  N  \epsi R  (\widehat{n}-\widehat{\tau_{-1} n}) q_1 \psi \rrangle| \\
&\lesssim  \epsi  \norm{ N (\widehat{n}-\widehat{\tau_{-1} n}) q_1 \psi } \weq{\ref{lem:qs&N} }{\leq} \epsi.
\end{align*}
For the first part of IV we note that for $f \in L^\infty(\Omega_\mathrm{f})$
\begin{align}\label{equ:einteilchenop}
 |\llangle \psi, f(x_1) \psi \rrangle - \langle \Phi, f(x) \Phi \rangle| \lesssim \norm{f}_\infty \beta.
\end{align}
Thus we can estimate 
\begin{align*}
 |\llangle \psi, \dot V(x_1,\epsi y_1) \psi \rrangle - \langle \Phi, \dot V(x_1,0) \Phi \rangle| &=  |\llangle  \psi, (\dot V(x_1,0)+ \epsi R) \psi \rrangle  - \langle \Phi, \dot V(x_1,0) \Phi \rangle|\\
&\lesssim |\llangle  \psi, (\dot V(x_1,0)\psi \rrangle - \langle \Phi, \dot V(x_1,0) \Phi \rangle| +\epsi\\
& \weq{ \eqref{equ:einteilchenop}}{ \lesssim}\; \norm{\dot V(\cdot,0)}_\infty \beta +\epsi.
\end{align*}
Equation \eqref{equ:einteilchenop} holds since
\begin{align*}
  |\llangle  \psi, f(x_1) \psi \rrangle  - \langle \Phi, f(x) \Phi \rangle|&= |\llangle  \psi, p_1 f(x_1) p_1 \psi \rrangle  - \langle \Phi, f(x) \Phi \rangle
+ \llangle  \psi, q_1 f(x_1) p_1 \psi \rrangle \\ &\quad+\llangle  \psi, p_1 f(x_1) q_1 \psi \rrangle + \llangle  \psi, q_1 f(x_1) q_1 \psi \rrangle |  \\
&\leq (1-\norm{p_1 \psi}^2)\langle \Phi, f(x) \Phi \rangle\\ &\quad+ 2 |\llangle  \psi, \widehat n^{1/2} p_1 f(x_1) \widehat n^{-1/2} q_1 \psi \rrangle | 
+ \norm{f}_\infty \beta\\
&\weq{\ref{lem:weights}}{ \lesssim} \norm{f}_\infty \beta.
\end{align*}

\end{proof}

\section{Proof of Lemma\,\ref{lem:energygp}}

% \begin{lem}[Energy Lemma]\label{lem:energy}
%  Let \textcolor{red}{assumptions} hold then
% \begin{align*}
%  \norm{\nabla q_1 \Psi}^2 \lesssim  (\tilde E^\psi- E^\Phi)+ \norm{\varphi}_{H^2\cap L^\infty}^2(\beta+\frac{1}{\sqrt{N}}+f(N,\epsi))
% \end{align*}
%  
% \end{lem}
As the ideas in this proof are the same as in Lemma\,\ref{lem:energy} we stay very brief here and give little extra explanation.
Let $\tilde h$ be defined as in Lemma\,\ref{lem:Energyh}.
From Section\,\ref{sec:pen} we know
\begin{align}\label{equ:nablaqpsig}
 \norm{\nabla_1 q_1 \psi }^2 \leq  \norm{\sqrt{\tilde h_1} q_1 \psi }^2 + E_0\, \beta 
\end{align}
%
% Also we have with $q_1=\id-p_1(p_2+q_2)$
% \begin{align*}
%  \norm{\sqrt{ \tilde h_1} q_1 \psi } \leq  \norm{\sqrt {\tilde h_1}(1-p_1 p_2 )\psi} + \norm{\nabla \Phi} \sqrt {\beta}
% \end{align*}
and
\begin{align}\label{equ:h_1q_1g}
 \norm{\sqrt{ \tilde h_1} q_1 \psi }^2 \leq  \norm{\sqrt {\tilde h_1}(1-p_1 p_2 )\psi}^2 + \norm{\nabla \Phi}^2  {\beta}
\end{align}
hence we bound  $\norm{\sqrt {\tilde h_1}(1-p_1 p_2 )\psi}^2$ to prove Lemma\,\ref{lem:energygp}.

\begin{lem}\label{lem:h1g}
 \begin{align*}
  \llangle \psi, (1-p_1p_2 ) \tilde h_1 (1-p_1p_2) \psi \rrangle %\lesssim \big( E^\psi- E^\phi \big)  +\norm{\Phi}_{H^1} \beta+ \norm{ \Phi}_{H^2}(\beta + \frac{1}{\sqrt N})\\
% &+\frac{3}{2} \norm{w^\epsi* |\varphi|^2}_\infty (\beta+ \frac{1}{N}) + \norm{(w^0*|\varphi|^2-w^\epsi* |\varphi|^2)}_\infty\\
% &+ \norm{ (w^\epsi *|\varphi|^2)(x_1)}_\infty (\beta+\frac{1}{\sqrt N})\\
% &+C (\alpha+\frac{2}{N} )+ C  \sqrt{\alpha+\frac{2}{N}} \norm{\nabla q_1 \psi}\\
&\lesssim % (E^\psi- E^\phi)+ 
\norm{\varphi}_{H^2\cap L^\infty}^2(\tilde \beta+\frac{1}{\sqrt{N}}+f(N,\epsi)) + \norm{V}_\LiO \beta
 \end{align*}
with
\begin{align*}
 f(N,\epsi)= \max(N^{-2\theta} \epsi^{4\theta-2}, N^{-1+3 \theta} \epsi^{-6\theta+2}).
\end{align*}
\end{lem}
% \begin{proof}[Proof of the Energy Lemma]
% 
% With Lemma\,\ref{lem:h1} we find
% \begin{align*}
%  \norm{\sqrt {\tilde h_1}(1-p_1 p_2 )\psi}^2 \lesssim (E^\psi- E^\phi)+ \norm{\varphi}_{H^2\cap L^\infty}^2(\beta+\frac{1}{\sqrt{N}}+f(N,\epsi))
% \end{align*}
% together with \eqref{equ:h_1q_1} 
%
%
% \begin{align*}
% \norm{\sqrt{ \tilde h_1} q_1 \psi }^2 & \lesssim  \norm{\sqrt {\tilde h_1}(1-p_1 p_2 )\psi}^2+\norm{\Phi}^2_{H^2}  {\beta} \\
% &\lesssim  (\tilde E^\psi- E^\Phi)+ \norm{\varphi}_{H^2\cap L^\infty}^2(\beta+\frac{1}{\sqrt{N}}+f(N,\epsi))\\
% \end{align*}
With \eqref{equ:nablaqpsig} and \eqref{equ:h_1q_1g} Lemma\,\ref{lem:h1g} proves
\begin{align*}
 \norm{\nabla_1 q_1 \psi }^2  \lesssim %(E^\psi- E^\Phi)+
 \norm{\varphi}_{H^2\cap L^\infty}^2(\tilde \beta+\frac{1}{\sqrt{N}}+f(\epsi)) + \norm{V}_\LiO \beta
\end{align*}
which is  Lemma\,\ref{lem:energygp}.
All that is left to do is to show the bound of Lemma\,\ref{lem:h1g}.

\begin{proof}[Proof of Lemma\,\ref{lem:h1g} ]
%The estimate of $ \langle \psi, (1-p_1p_2 )\tilde h_1 (1-p_1p_2) \psi \rangle $ is obtained by rewriting it in terms of the Energy difference $E^\psi-E^\varphi$ 
% \begin{align*}
%  \underbrace{E^\psi-\frac{1}{N} \langle \psi, N\frac{E}{\epsi^2} \psi   \rangle}_{=: \tilde E^\psi}  \underbrace {-E^\varphi +\langle \varphi, \frac{E}{\epsi^2} \varphi \rangle}_{=:\tilde E^\varphi}=
%E^\psi-E^\varphi
% \end{align*}
 %and the remaining parts. Since
% \begin{align*}
%  \tilde E^\psi-E^\Phi&= \langle  \psi ,(p_1p_2+1-p_1p_2) \tilde h_1 (p_1 p_2 +1 -p_1p_2) \psi \rangle \\
% &\quad+\frac{N-1}{2N}\langle  \psi, (p_1 p_2 +1 -p_1p_2) w^{\epsi,\theta,N}_{12} (p_1 p_2 +1 -p_1p_2)\psi \rangle\\
% &\quad- \langle \varphi, -\Delta_x- \frac{1}{\epsi^2} (\Delta_y+E) \varphi \rangle - \langle \Phi ,\frac{1}{2} (a |\Phi|^2) \Phi \rangle
% \end{align*}
% 
% we find, after multiplying out the terms in the first row, separating $\langle \psi, (1-p_1p_2 ) h_1 (1-p_1p_2) \psi \rangle$ and then ordering in a convenient way for the following estimation process:

After rearranging the energy difference $E^\psi-E^\Phi$  we arrive at the same lengthy equation as in \eqref{equ:hp} with an additional term from the time dependent external potential V.
\begin{align}\label{hpg}
\llangle \psi , (1-p_1p_2 ) &\tilde h_1 (1-p_1p_2) \psi  \rrangle\\
&= E^\psi- E^\phi \notag \\
&\quad- \llangle \psi , p_1p_2 \tilde h_1 p_1p_2 \psi  \rrangle+ \langle \varphi, - \Delta - \frac{1}{\epsi^2} (\Delta_y+E_0) \varphi \rangle \notag \\
&\quad-\llangle \psi , (1-p_1p_2 )\tilde h_1 p_1p_2 \psi  \rrangle-\llangle \psi , p_1p_2 \tilde  h_1 (1-p_1p_2) \psi  \rrangle \notag\\
&\quad-\frac{N-1}{2N} \llangle  \psi, p_1 p_2 w^{\epsi,\theta,N}_{12} p_1 p_2 \psi  \rrangle + \langle \Phi, \frac{1}{2} (b*|\Phi|^2) \Phi \rangle \notag\\
&\quad-\frac{N-1}{2N} \Big( \llangle  \psi,(1- p_1 p_2) w^{\epsi,\theta,N}_{12} p_1 p_2 \psi  \rrangle+ \llangle  \psi, p_1 p_2 w^{\epsi,\theta,N}_{12}(1- p_1 p_2) \psi  \rrangle \Big) \notag \\
&\quad -\frac{N-1}{2N} \llangle  \psi,(1- p_1 p_2) w^{\epsi,\theta,N}_{12}(1- p_1 p_2) \psi  \rrangle \notag \\
&\quad -\llangle \psi , V(x_1,\epsi y_1) \psi  \rrangle - \langle \Phi, V(x,0) \Phi \rangle . 
\end{align}
After estimating the terms line by line we obtain the claimed estimate
 \begin{align*}
  \llangle \psi , (1-p_1p_2 )& \tilde h_1 (1-p_1p_2) \psi  \rrangle\\
 &\lesssim \big( E^\psi- E^\phi \big)\\
& +\norm{\Phi}_{H^1}^2 \beta\\
&+ \norm{ \Phi}_{H^2}(\beta + \frac{1}{\sqrt N})\\
&+N^{-2\beta} \epsi^{4\beta-2} \norm{ \Delta |\varphi|^2} \norm{\varphi}_\infty+ N^{-1}\norm{\varphi}_\infty^2+ \norm{\Phi}_\infty^2 \alpha \\
&+\norm{\varphi}^2_\infty \beta+ N^{-1+3 \beta} \epsi^{-6\beta+2}\\
&+\norm{V(\cdot,0)}_{L^\infty}%(\Omega_{\mathrm{f}})}
 \beta +\epsi\\
%&\lesssim (E^\psi- E^\phi)+ (\norm{\varphi}_\infty+\norm{\varphi}_{H^2})^2(\beta+\frac{1}{\sqrt{N}}+f(N,\epsi))\\
&\lesssim \norm{\varphi}_{H^2\cap L^\infty}^2(\tilde \beta+\frac{1}{\sqrt{N}}+f(N,\epsi)) +\norm{V(\cdot,0)}_{L^\infty} \beta .  \numberthis \label{equ:enfertig}
 \end{align*}
The line-by-line approximation turns out to be a little bit simpler than before % from \eqref{hpg} to \eqref{equ:enfertig} as  from \eqref{equ:hp} to \eqref{equ:hpest} 
but some estimates have to be adjusted. 
We do not have to estimate the first line.

\textit{Line 2}. 
\begin{align*}
 | \langle \varphi, \tilde h_1 \varphi \rangle- \llangle \psi , p_1p_2  \tilde h_1 p_1p_2 \psi  \rrangle  |&=| \langle \varphi, \tilde h_1 \varphi \rangle- \langle \varphi, \tilde h_1 \varphi \rangle \llangle \psi , p_1p_2  \psi  \rrangle  | \\
&= \langle \varphi, \tilde h_1 \varphi \rangle |\llangle \psi , (1-p_1p_2 )\psi  \rrangle|\\
& \weq{\eqref{equ:htildep} }{ =}\langle \Phi,-\Delta \Phi \rangle |\llangle \psi , (p_1q_2 +q_1p_2+ q_1q_2) \psi  \rrangle|\\
%&\leq 3 \norm{\Phi}^2_{H^1} \alpha
& \weq{\ref{lem:weights}}{  \lesssim }\norm{\Phi}^2_{H^1} \beta 
\end{align*}
\textit{Line 3}. 
\begin{align*} 
-\langle \psi, (1-p_1p_2 ) \tilde h_1 p_1p_2 \psi \rangle-\langle \psi, p_1p_2  \tilde h_1 (1-p_1p_2) \psi \rangle\
\end{align*}
is bounded in absolute value by 
\begin{align*}
 2  |\llangle \psi, (1-p_1p_2 ) \tilde h_1 p_1p_2 \psi \rrangle| %&= 2 |\langle \psi, (q_1+p_1q_2)  \tilde h_1 p_1p_2 \psi \rangle|\\
% &= 2 |\langle \psi, q_1  \tilde h_1 p_1p_2 \psi \rangle|\\
% &= 2 |\langle \psi, q_1 \widehat n^{-\frac{1}{2}} \widehat n^\frac{1}{2}   \tilde h_1 p_1p_2 \psi \rangle|\\
% &= 2 |\langle \psi, q_1 \widehat n^{-\frac{1}{2}}    \tilde h_1 \widehat {\tau_1 n}^\frac{1}{2} p_1p_2 \psi \rangle|\\
% &\leq 2 \sqrt{\langle \psi, \widehat n^{-1} q_1  \psi \rangle} \sqrt {\langle \psi,   p_1 p_2 \widehat {\tau_1 n}^\frac{1}{2} \tilde h_1^2 \widehat {\tau_1 n}^\frac{1}{2} p_1 p_2 \psi \rangle}\\
% &= 2\sqrt{\langle \psi, \widehat n \psi \rangle} \sqrt{\langle \varphi, \tilde h^2 \varphi \rangle} \sqrt {\langle \psi,  \widehat {\tau_1 n} p_1 p_2 \psi \rangle}\\
% &\leq  2 \sqrt{\beta} \norm{ \Phi}_{H^2}   \sqrt{\langle \psi,  (\widehat n +\widehat{\frac{1}{\sqrt{N}}} )  \psi \rangle}\\
% &=2 \sqrt{\beta} \norm{ \Phi}_{H^2}   \sqrt{\beta + \frac{1}{\sqrt N} } \\
&\weq{\eqref{equ:line3}}{ \lesssim}  \norm{ \Phi}_{H^2}(\beta + \frac{1}{\sqrt N}).
\end{align*}
\textit{Line 4}.
We first note that 
\begin{align}\label{equ:phib}
 |\langle &\Phi, \frac{1}{2} (b|\Phi|^2) \Phi \rangle- \langle \psi,  p_1 p_2 \frac{1}{2} (b|\Phi|^2) p_1 p_2 \psi \rangle | \lesssim \norm{\Phi}_\infty^2 \beta
\end{align}
since 
\begin{align*}
 |\langle \Phi, \frac{1}{2} &(b|\Phi|^2) \Phi \rangle- \langle \psi,  p_1 p_2 \frac{1}{2} (a|\Phi|^2) p_1 p_2 \psi \rangle |\\
&= |\langle \Phi, \frac{1}{2} (b|\Phi|^2) \Phi \rangle- \langle \varphi, \frac{1}{2} (a|\Phi|^2)  \varphi \rangle  \rangle   \langle \psi,  p_1 p_2 \psi \rangle |\\
&= |\langle \Phi, \frac{1}{2} (b|\Phi|^2) \Phi \rangle- \langle \Phi, \frac{1}{2} (a|\Phi|^2)  \Phi \rangle  \rangle   \langle \psi,  p_1 p_2 \psi \rangle |\\
&=  |\langle \Phi, \frac{1}{2} (b|\Phi|^2) \Phi \rangle||\langle \psi, (1-p_1p_2) \psi \rangle| \\
&\lesssim \norm{\Phi}_\infty^2 \beta.
\end{align*}
Hence,
\begin{align*}
 |\langle &\Phi, \frac{1}{2} (b|\Phi|^2) \Phi \rangle- \frac{N-1}{2N} \llangle  \psi, p_1 p_2 w^{\epsi,\theta,N}_{12} p_1 p_2 \psi \rrangle  |\\
 &\weq{\eqref{equ:phib}}{ \lesssim } \, \frac{1}{2} |\llangle \psi,  p_1 p_2 (b|\Phi|^2) p_1 p_2 \psi \rrangle  - (1+\frac{1}{N}) \llangle  \psi, p_1 p_2 w^{\epsi,\theta,N}_{12} p_1 p_2 \psi \rrangle| + \norm{\Phi}_\infty^2 \beta\\ 
&\leq \frac{1}{2}  |\llangle \psi,  p_1 p_2 (b|\Phi|^2-  w^{\epsi,\theta,N}_{12})  p_1 p_2 \psi \rrangle|  +\frac{1}{N} |\llangle  \psi, p_1 p_2 w^{\epsi,\theta,N}_{12} p_1 p_2 \psi \rrangle|  + \norm{\Phi}_\infty^2 \beta\\
%&=\frac{1}{2}  |\langle \psi,  p_1 p_2 (\tilde b \delta(r_1-r_2)-  w^{\epsi,\theta,N}_{12})  p_1 p_2 \psi \rangle|  +\frac{1}{N} |\langle  \psi, p_1 p_2 w^{\epsi,\theta,N}_{12} p_1 p_2 \psi \rangle|  + \norm{\Phi}_\infty^2 \beta
&\lesssim  N^{-2\theta} \epsi^{4\theta-2} \norm{ \Delta |\varphi|^2} \norm{\varphi}_\infty+ N^{-1}\norm{\varphi}_\infty^2+ \norm{\Phi}_\infty^2 \beta,
\end{align*}
where we used the estimate from equation \eqref{equ:3p1qgp} for the first summand and Lemma\,\ref{lem:young} for the second summand. \\
% \begin{align*}
% \lesssim  N^{-2\beta} \epsi^{4\beta-2} \norm{ \Delta |\varphi|^2} \norm{\varphi}_\infty+ N^{-1}\norm{\varphi}_\infty^2+ \norm{\Phi}_\infty^2 \alpha 
% \end{align*}
%
\textit{Line 5}.
Is bounded in absolute value by 
\begin{align*}
|\llangle  \psi, p_1 p_2 w^{\epsi,\theta,N}_{12}(1- p_1 p_2) \psi \rrangle| &= |\llangle  \psi, p_1 p_2 w^{\epsi,\theta,N}_{12}(q_1p_2+ p_1q_2+ q_1q_2) \psi \rrangle| \\
&\leq 2| \llangle  \psi, p_1 p_2 w^{\epsi,\theta,N}_{12} q_1p_2 \psi \rrangle|+| \llangle  \psi, p_1 p_2 w^{\epsi,\theta,N}_{12} q_1q_2 \psi \rrangle |.
\end{align*}
The first term is bounded by 
\begin{align*}
 | \llangle  \psi, p_1 p_2 w^{\epsi,\theta,N}_{12} q_1p_2 \psi \rrangle|&= | \llangle  \psi, p_1 p_2 w^{\epsi,\theta,N }_{12} \widehat n^{-\frac{1}{2}} \widehat n^\frac{1}{2}   q_1 p_2  \psi \rrangle|\\
&\weq{\ref{lem:weights}}{ =}| \llangle  \psi, p_1 p_2 \widehat {\tau_1 n}^\frac{1}{2} w^{\epsi,\theta,N }_{12}  \widehat n^{-\frac{1}{2}}   q_1 p_2 \psi \rrangle|\\
&\leq \norm{p_2 w^{\epsi,\theta,N }_{12} p_2}_{\mathrm{Op}} \norm{\widehat {\tau_1 n}^\frac{1}{2} \psi} \norm{\widehat n^{-\frac{1}{2}}   q_1 \psi}\\
%&\leq  \norm{ (w^\epsi *|\varphi|^2)(x_1)}_\infty \sqrt{\llangle \psi, \widehat {\tau_1 n} \psi \rrangle} \sqrt{\llangle \psi, \widehat n \psi \rrangle }\\
%&\leq  \norm{ (w^\epsi *|\varphi|^2)(x_1)}_\infty (\beta+\frac{1}{\sqrt N})\\
&\weq{\ref{lem:weights},\ref{lem:young} }{ \lesssim}  \norm{\varphi}_\infty^2 (\beta+\frac{1}{\sqrt N}) .
\end{align*}
For the second term we use a slightly altered version of Lemma\,\ref{lem:estimating3alt}.
So in the first step we use symmetry to write
\begin{align*}
 |\llangle  \psi, p_1 p_2   w^{\epsi,\theta,N}_{12}  q_1 q_2   \psi \rrangle  |
%&=|\llangle  \psi, p_1 p_2 \widehat m^{\frac{1}{2}}   w^{\epsi,\theta,N}_{12}  \widehat m^{\frac{-1}{2}}  q_1 q_2   \psi \rrangle  |\\
&= \frac{1}{N-1} |\sum_{j=2}^N  \llangle  \psi,   p_1 p_j w^{\epsi,\theta,N}_{1j}  q_1 q_j    \psi \rrangle |\\
&\leq \frac{1}{N-1} \norm{    q_1 \psi}\norm{\sum_{j=2}^N  q_j  w^{\epsi,\theta,N}_{1j}   p_1 p_j \psi } \\
&\leq \frac{1}{N-1} \sqrt{\beta }\norm{\sum_{j=2}^N  q_j  w^{\epsi,\theta,N}_{1j}   p_1 p_j \psi } \numberthis \label{equ:II.1g}.
\end{align*}
Now the second factor of \eqref{equ:II.1g} is split in the "diagonal" term and "off-diagonal" term 
\begin{align*}
  \norm{\sum_{j=2}^N  q_j w^{\epsi,\theta,N}_{12}    p_1 p_j  \psi }^2
&=  \sum_{j,k=2}^N \llangle \psi, p_1 p_l   w^{\epsi,\theta,N}_{1l}  q_l q_j w^{\epsi,\theta,N}_{1j} 
   p_1 p_j \psi  \rrangle \\
\begin{split}
& \leq \sum_{2 \leq  j < k \leq N} \llangle \psi, p_j^\chi q_j^\Phi p_1 p_l   w^{\epsi,\theta,N}_{1l}    w^{\epsi,\theta,N}_{1j} p_l^\chi q_l^\Phi  p_1 p_j \psi  \rrangle\\
&\quad+ (N-1) \norm{w^{\epsi,\theta,N}_{12}  p_1 p_2 \psi   }^2. 
 \end{split} \numberthis \label{equ:II.3g}
\end{align*}
The first summand of \eqref{equ:II.3g} is bounded by
\begin{align*}
 &(N-1)(N-2)\langle \psi,q_2 p_1 p_3    w^{\epsi,\theta,N}_{13}    w^{\epsi,\theta,N}_{12} q_3  p_1 p_2 \psi  \rangle\\
&\leq N^2 \norm { \sqrt{ w^{\epsi,\theta,N}_{13}}    \sqrt{ w^{\epsi,\theta,N}_{12}} q_3 p_1 p_2 \psi}^2\\
&\leq N^2 \norm{   \sqrt{ w^{\epsi,\theta,N}_{12}} p_2 \sqrt{ w^{\epsi,\theta,N}_{13}} p_1 q_3 \psi } ^2 \\
&\leq N^2 \norm{  \sqrt{ w^{\epsi,\theta,N}_{12}} p_2  }^4_\mathrm{Op} \norm{ q_3 \psi }^2\\
&\leq N^2 \norm{p_2 w^{\epsi,\theta,N}_{12} p_2  }_\mathrm{Op}^2  \beta\\
&\weq{\ref{lem:young}}{ \lesssim} N^2 \norm{\varphi}_{\infty}^4 \beta \numberthis \label{equ:II.4g}.
\end{align*}
The second summand of \eqref{equ:II.3g} is bounded by
\begin{align*}
&N \langle \psi, p_1 p_2   (w^{\epsi,\theta,N}_{12})^2  p_1 p_2 \psi  \rangle \\
&\leq N   \norm{p_1 (w^{\epsi,\theta,N}_{12})^2 p_1 }_\mathrm{Op}\\
& \weq{\ref{lem:young}}{ \leq} N \norm{\varphi}^2_\infty \norm{w^{\epsi,\theta,N}}_2^2 \\
& =  N^{1+3\theta} \epsi^{-6\theta+2}  \norm{\varphi}^2_\infty \numberthis \label{equ:II.5g}
\end{align*}
since $\norm{w^{\epsi,\theta,N}}_2^2 \lesssim (\frac{N}{\epsi^2})^{3 \theta} \epsi^2 $.
Now putting \eqref{equ:II.4g} and \eqref{equ:II.5g} together we find
\begin{align*}
 \norm{\sum_{j=2}^N q_j w^{\epsi,\theta,N}_{12}   p_1 p_j \psi }^2 \lesssim  N^2 \norm{\varphi}^4_\infty \beta +  N^{1+3\theta} \epsi^{-6\theta+2}\norm{\varphi}^2_\infty.
\end{align*}
Inserting this in \eqref{equ:II.1g} yields the claimed result
\begin{align*}
|\langle  \psi, p_1 p_2   w^{\epsi,\theta,N}_{12}  q_1 q_2   \psi \rangle  | &\lesssim \frac{1}{N} \sqrt{ \beta } \sqrt{ N^2 \norm{\varphi}^4_\infty \beta + \norm{\varphi}^2_\infty N^{1+3\theta} \epsi^{-6\theta+2} }\\
&\lesssim \norm{\varphi}^2_\infty  \sqrt{\beta}  \sqrt{ \beta+\norm{\varphi}_\infty^{-2} N^{-1+3 \theta} \epsi^{-6\theta+2}} \\
%&\leq \norm{\varphi}^2_\infty (\beta+ \norm{\varphi}^{-2}_\infty N^{-1+3 \theta} \epsi^{-6\theta+2})\\
&\leq \norm{\varphi}^2_\infty \beta+ N^{-1+3 \theta} \epsi^{-6\theta+2}.
\end{align*}

\textit{Line 6}.
The interaction is nonnegative so we have 
\begin{align*}
 -\frac{N-1}{2N} \langle  \psi,(1- p_1 p_2) w^{\epsi,\theta,N}_{12}(1- p_1 p_2) \psi \rangle \ \leq 0.
\end{align*}

\textit{Line 6}.
With the methods used in the proof of Lemma\,\ref{lem:3termeg}.\ref{lem:3.4g} we find
\begin{align*}
 |\llangle \psi , V(x_1,\epsi y_1) \psi  \rrangle - \langle \Phi, V(x,0) \Phi \rangle | \lesssim  \norm{V(\cdot,0)}_{L^\infty}%(\Omega_{\mathrm{f}})}
 \beta +\epsi.
\end{align*}

\end{proof}

\appendix

\chapter{Properties of the Solutions to the Considered Equations }\label{app:regsol}

In this section we summarize the well-known results for the regularity of solutions to the considered equations. These results ensure that
the estimates of Theorems\,\ref{thm:thm1}-\ref{thm:thm3} are meaningful.

\section{Properties of the Solution to the N-particle Equation}

The assumptions on the $N$-particle Hamiltonian $H_N$ are for all cases, even with time dependent external potential, such that $H_N$ generates a unitary time evolution on $D(H_N)$.
Thus for solutions $\psi$ of the Schrödinger equation we have global existence and conservation of the $L^2$-norm and without a time depending external potential conservation of energy.  

\section{Properties of the Solutions to the One-particle Equations} 

The questions of well-posedness, global existence and conservation laws for the Hartree and NLS/Gross-Pitaevskii equation in our setting are well understood.
The standard way of deriving the claimed results follows in two steps. The first step is to prove local existence of solutions by approximating by the free evolution for example with the help of 
variation of constants formula. %, which is often called Duhamel's principle in this context. 
The second step is extending the local solutions with the
help of conservation laws to global solutions. We only state the results of the properties we use. For an overview on this topic see for example the book of Tao \cite{Tao06} and literature therein.

\subsection{The Hartree Equation}

%\textcolor{red}{ist wirklich die schwache Lösunge gemeint? ja!}
\begin{lem}
For $\Phi(x,t): \R^n \times \R \rightarrow \C $ and $n \in {1,2}$ consider the Cauchy-Problem for the Hartree equation  
\begin{align}\label{hatreecauchy}
      \begin{cases}
      \im \partial_t \Phi(x,t)= -\Delta \Phi(x,t)+ (w*|\Phi|^2)(x,t) \Phi(x,t)\\
      \Phi(x,0)=\Phi_0,\\
    \end{cases} 
\end{align}
where $w $ is spherically symmetric and $w=w_1+w_2$ with $ w_1 \in L^{p_1} $ and $w_2 \in L^\infty $, where $p_1 > 1$.

\begin{enumerate}
  \item For $\Phi_0 \in H^1(\R^n)$ the Cauchy-Problem has a unique weak solution $\Phi(x,t) \in C_b(\R,H^1(\R^n))  $ with $\norm{\Phi_0}_2=\norm{\Phi_t}_2=1 $ 
and $\norm{\Phi_0}_{H^1}=\norm{\Phi_t}_{H^1}$ for all $t\in \R^+$.
\item If $\Phi_0 \in H^k(\R^n)$ for $k \in \N$, $k>2$ then the solution of \eqref{hatreecauchy} is in $ C_b(\R,H^1) \cap C(\R,H^k) \cap C^1(\R,H^{k-2})$.
 \end{enumerate}
\end{lem}
\begin{proof}
 1. and 2. are Proposition 2.2 and Theorem 3.1 in \cite{GinVel80}.
\end{proof}

\subsection{The Gross-Pitaevskii/NLS  Equation}

\begin{lem}
For $\Phi(x,t): \R^n \times \R \rightarrow \C $ and $n \in {1,2}$ consider the Cauchy problem for the Gross-Pitaevskii equation  
\begin{align}\label{gpcauchy}
      \begin{cases}
      \im \partial_t \Phi(x,t)= -\Delta \Phi(x,t)+ |\Phi|^2 \Phi(x,t)\\
      \Phi(x,0)=\Phi_0.\\
    \end{cases} 
\end{align}

\begin{enumerate}
  \item For  $\Phi_0 \in H^1({\R^n})$ the Cauchy-Problem has a unique weak solution $\Phi(x,t) \in C_b(\R,H^1(\R^n))  $ with $\norm{\Phi_0}_2=\norm{\Phi_t}=1 $ for all $t\in \R^+$.
\item If $\Phi_0 \in H^2({\R^n})$ the solution of \eqref{gpcauchy} is in $ C_b(\R,H^1) \cap C(\R,H^2) \cap C^1(\R,L^2)$.
 \end{enumerate}
\end{lem}

These results are summarized in Proposition 3.1 of \cite{BenOliSch12} for the more complicated case $n=3$.
In the case $n=1$ there are even stronger results. For $k\in \N$ let $\Phi_0 \in H^k({\R})$ then $\norm{\Phi(t)}_{H^k} \leq \norm{\Phi(0)}_{H^k} \;\forall t$. This follows from 
exercise 3.36 in  \cite{Tao06}.

\subsection{Eigenfunctions of the Laplacian on a Bounded Domain }
% It is well known cf. \cite{Eva10} that the solutions $w_k$ 
% of the boundary-value problem
% \begin{align*}
%  \begin{cases}
%   L w = \lambda w \qquad &\mathrm{in}\, U\\
% w=0  &\mathrm{ on}\, \partial U,
%  \end{cases}
% \end{align*}
% where $U$ is open and bounded, and $L$ is a uniform elliptic, symmetric operator whose coefficients are smooth, are elements of $  C^\infty(U) $ and
% for $\partial U $ smooth $w_k \in  C^\infty(\bar U)$.    

Last we summarize the well-known results for the boundary-value problem 
\begin{align*}
 \begin{cases}
  L w = \lambda w \qquad &\mathrm{in}\, U\\
w=0  &\mathrm{ on}\, \partial U,
 \end{cases}
\end{align*}
 where $U$ is open and bounded, $L$ is a uniform elliptic, symmetric operator with smooth coefficients which are elements of $ C^\infty(U)$.
See for example  \cite{Eva10} for the following facts.
\begin{enumerate} 
 \item  The eigenvalues $\{ \lambda_k \}_{k=1}^\infty$ of $L$ can be ordered such that
\begin{align*}
 0 < \lambda_1 < \lambda_2 \leq \lambda_3 \leq \dots
\end{align*}
\item
There exists an orthonormal basis $\{ w_k\}_{k=1}^\infty$ of $L^2(U)$, where $w_k \in C^\infty(U)$ is an eigenfunction with eigenvalue $\lambda_k$ for each $k$.
Furthermore, for smooth $\partial U $ we have $w_k \in C^\infty(\bar U)$.
  
\end{enumerate}

\chapter{Estimates for the Coulomb potential}\label{app:coul}

%  Interpolation
% $2 \leq p \leq \infty$
% \begin{align*}
%  \norm{\varphi}_p \leq \norm{\varphi}_{L^2 \cap L^\infty}=(1+\norm{\varphi}_\infty)
% \end{align*}
% since there exist a $\theta \in (0,1)$, $\norm{\varphi}_2=1$
% \begin{align*}
%  \norm{\varphi}_p \leq \norm{\varphi}_2^{1-\theta} \norm{\varphi}_\infty^\theta= \norm{\varphi}_\infty^\theta \leq (1+\norm{\varphi}_\infty)=\norm{\varphi}_{L^2 \cap L^\infty}.
% \end{align*}

%\begin{section}{Coulomb Abschätzungen}

In this section we show that the assumptions of Theorem\,\ref{thm:thm2} hold for 
\begin{align*}
 w=\frac{1}{|r|} \qquad w^0=\frac{1}{|x|},
\end{align*}
if we have confinement in one direction.
For the ease of the calculation we set $\tilde  \Omega_{\mathrm{c}}=[-1,1]$. However, the following calculation holds for arbitrary intervals allowed by the assumptions.
We decompose the potentials in a part with the singularity and a bounded part

\begin{align*}
 w_s=\frac{1}{|r|} \chi_{\{B_1(0)\times [-1,1]\}} \qquad  w_\infty=\frac{1}{|r|} \chi_{\{B_1(0)^C\times [-1,1]\}},
\end{align*}
where $\chi$ denotes, only in this section, the indicator function.
The function $w^0$ is understood as the constant function $1$ in the $y$-direction
\begin{align*}
 w^0_s=\frac{1}{|x|} \chi_{\{B_1(0)\times [-1,1] \}} \qquad  w_\infty^0=\frac{1}{|x|} \chi_{\{B_1(0)^C \times [-1,1] \}}.
\end{align*}

\section{Approximation for Example\,\ref{exp:coul} }
\subsection{Convergence of $|r^\epsi|^{-1}$ to $|x|^{-1}$}
We first show that in the sense of assumption A1' $\frac{1}{|r^\epsi|}$ is approximated by $\frac{1}{|x|}$.
With the definition of %${L^1(\Omega_{\mathrm{f}}\times 2 \Omega_{\mathrm{c} })+L^\infty(\Omega_{\mathrm{f}}\times 2 \Omega_{\mathrm{c} }})$
$L^1(\tilde \Omega)+L^\infty(\tilde \Omega)$
 we have
\begin{align*}
& \norm{w^\epsi-w^0}_{L^1(\Omega_{\mathrm{f}}\times \tilde \Omega_{\mathrm{c} })+L^\infty(\Omega_{\mathrm{f}}\times \tilde \Omega_{\mathrm{c} })} 
=  \norm{w^\epsi_s-w^0_s}_{L^1(\Omega_{\mathrm{f}}\times \tilde \Omega_{\mathrm{c} })}+\norm{w^\epsi_\infty-w^0_\infty}_{L^\infty(\Omega_{\mathrm{f}}\times \tilde \Omega_{\mathrm{c} })} .
\end{align*}
% \begin{align*}
% & \norm{\frac{1}{\sqrt{x^2+\epsi^2 y^2}}-\frac{1}{|x|}}_{L^1(\R^3)+L^\infty(\R^3)}\\%{L^1(\R^2 \times [-1,1]) \cup L^\infty (\R^2 \times [-1,1])}\\
% &=  \norm{(\frac{1}{\sqrt{x^2+\epsi^2 y^2}}-\frac{1}{|x|})\chi_{\{B_1(0)\times [-1,1]\}}+\chi_{\{B^c_1(0)\times [-1,1]\}}}_{L^1(\R^2 \times [-1,1]) \cup L^\infty (\R^2 \times [-1,1])}\\
% &\leq   \norm{\frac{1}{\sqrt{x^2+\epsi^2 y^2}}-\frac{1}{|x|}}_{L^1(B_1(0) \times [-1,1]) }+\norm{\frac{1}{\sqrt{x^2+\epsi^2 y^2}}-\frac{1}{|x|}}_{L^\infty(B^c_1(0) \times [-1,1]) }
% \end{align*}
We first approximate the $L^\infty$ part
\begin{align*}
 \norm{\frac{1}{\sqrt{x^2+\epsi^2 y^2}}-\frac{1}{|x|}}_{L^\infty(B_1(0)^C \times [-1,1]) } =  \norm{\frac{1}{\sqrt{r^2+\epsi^2 y^2}}-\frac{1}{r}}_{L^\infty((1,\infty) \times [-1,1]) }\\
=\norm{\frac{r-\sqrt{r^2+\epsi^2 y^2}}{r\sqrt{r^2+\epsi^2 y^2}}}_{L^\infty((1,\infty) \times [-1,1]) }.
\end{align*}
After a Tayler expansion we find $r\sqrt{1+\epsi^2 \frac{y^2}{r^2}} = r(1+\theta \frac{\epsi^2 y^2}{r^2} )$ for a $\theta \in [0,1]$. Thus we obtain
% since $Bild \frac{1}{2 \sqrt{[1,1+ \frac{\epsi^2 y^2}{r^2}]}} \subset [0,1] $ $\forall \epsi, r,y$

\begin{align*}
 \norm{\frac{1}{\sqrt{x^2+\epsi^2 y^2}}-\frac{1}{|x|}}_{L^\infty(B^C_1(0) \times [-1,1]) }=\norm{\frac{\theta \epsi^2 y^2}{r^2 \sqrt{r^2+\epsi^2 y^2}}}_{L^\infty((1,\infty) \times [-1,1]) }\leq \epsi^2. \\
\end{align*}
For the $L^1$-part we can solve the integral directly
\begin{align*}
  &\norm{\frac{1}{\sqrt{x^2+\epsi^2 y^2}}-\frac{1}{|x|}}_{L^1(B_1(0) \times [-1,1]) }= \int_{B_1(0)}\int_{-1}^1 |\frac{1}{\sqrt{x^2+\epsi^2 y^2}}-\frac{1}{|x|}| \D x \D y\\
&\qquad=2 \pi \int_0^1\int_{-1}^1 |\frac{1}{\sqrt{r^2+\epsi^2 y^2}}-\frac{1}{r}| r \D r \D y =2 \pi \int_0^1\int_{-1}^1 1-\frac{r}{\sqrt{r^2+\epsi^2 y^2}}  \D r \D y \\
&\qquad= 4 \pi \int_0^1\int_{0}^1 1-\frac{r}{\sqrt{r^2+\epsi^2 y^2}}  \D r \D y=4 \pi(1 +\int_{0}^1 (\epsi y - \sqrt{1+\epsi^2 y^2} ) \D y\\
&\qquad=1+\frac{\epsi}{2}-\bigg[\frac{1}{2}y\sqrt{\epsi^2 y^2+1}+ \frac{\sinh^{-1}(\epsi y)}{2 \epsi}\bigg]_0^1=1+\frac{\epsi}{2}-\frac{1}{2}\sqrt{\epsi^2+1}-\frac{\sinh^{-1}(\epsi )}{2 \epsi}\\
&\qquad=1+\frac{\epsi}{2}-{(\frac{1}{2}+\frac{1}{4}\epsi^2+\dots)}-\frac{1}{2\epsi}(\epsi-\frac{1}{6}\epsi^3+\dots)=\frac{\epsi}{2}+\mathcal{O}(\epsi^2).
\end{align*}
Putting both estimates together we have
\begin{align*}
 \norm{\frac{1}{|r^\epsi|}-\frac{1}{|x|}}_{L^1(\Omega_{\mathrm{f}}\times \tilde \Omega_{\mathrm{c} })+L^\infty(\Omega_{\mathrm{f}}\times \tilde \Omega_{\mathrm{c} })} \lesssim \epsi.
\end{align*}

\subsection{Uniform Bound for $|r^\epsi|^{-p}$ for $p < 2$}%{$L_p$ Norm of $\frac{1}{\sqrt{x^2+\epsi^2 y^2}}$ in $B_1(0)$ for $p < 2$}
We  consider $\frac{1}{|r^\epsi|}$ on ${L^p(\Omega_{\mathrm{f}}\times \tilde \Omega_{\mathrm{c} })+L^\infty(\Omega_{\mathrm{f}}\times \tilde \Omega_{\mathrm{c} }})$.
The $L^\infty$-part does not pose any problems. The singularity can be estimated for $p <2$ by
\begin{align*}
\int_{B_1(0)} \int_{-1}^1  \frac{1}{(x^2+\epsi^2  y^2)^{\frac{p}{2}}} \D x \D y \leq  \int_{B_1(0)} \int_{-1}^1  \frac{1}{(x^2)^{\frac{p}{2}}} \D x \D y\\
= 4 \pi \int_{0}^1 \frac{1}{r^p} r \D r = 4 \pi \int_{0}^1 r^{1-p} \D r = 4 \pi [r^{2-p}]_0^1 = C.
\end{align*}
This estimate is sharp in $p$ in the sense that for $p=2$ it does not work since
\begin{align*}
\int_{B_1(0)} \int_{-1}^1  \frac{1}{(x^2+\epsi^2  y^2)} \D x \D y \leq 2 \int_{0}^1 r^{-1} \D r
\end{align*}
 does diverge. %On way to bound this term would be
%      
% \begin{align*}
% \int_{B_1(0)} \int_{-1}^1  \frac{1}{x^2+\epsi^2 \tilde y^2} \D x \D \tilde y = \frac{1}{\epsi} \int_{B_1(0)} \int_{-\epsi}^\epsi  \frac{1}{x^2+ y^2} \D x \D y \leq C \frac{1}{\epsi} \int_{B_2(0)} \frac{1}{r^2} r^2 \D r
% \leq  \frac{C}{\epsi}
% \end{align*}
% However this is not uniform in $\epsi$.

\section{Bound for Example\,\ref{exp:coul2} }

The logarithmic divergence of $ |r^\epsi|^{-2} $ follows from estimating

\begin{align*}
 \int_{B_1(0)} \int_{-1}^1  \frac{1}{x^2+\epsi^2  y^2} \D x \D y &\leq  \int_{B^\epsi_1(0)} \int_{-1}^1   \frac{1}{x^2+\epsi^2  y^2}  \D x \D y\\
 & \qquad+  \int_{B_1(0)\setminus  B^\epsi_1(0)} \int_{-1}^1   \frac{1}{x^2+\epsi^2  y^2}  \D x \D y\\
&\leq \frac{1}{\epsi} \int_{B^\epsi_1(0)} \int_{-\epsi}^\epsi   \frac{1}{x^2+  y^2}  \D x \D y+
  \int_{B_1(0)\setminus  B^\epsi_1(0)} \int_{-1}^1   \frac{1}{x^2}  \D x \D y\\
&\lesssim \frac{1}{\epsi} \int_{B^\epsi_1(0)}   \frac{1}{r^2}  \D (r, \theta,\varphi) + \int_\epsi^1 \frac{1}{r } \D r \\
&\lesssim  \frac{1}{\epsi} \int_0^\epsi \frac{1}{r^2} r^2 \D r - \log \epsi \lesssim 1+ \log \epsi^{-1}.
%&= 4 \pi \int_{0}^1 \frac{1}{r^p} r \D r = 4 \pi \int_{0}^1 r^{1-p} \D r = 4 \pi [r^{2-p}]_0^1 = C
\end{align*}

% 
%  Interpolation
% $2 \leq p \leq \infty$
% \begin{align*}
%  \norm{\varphi}_p \leq \norm{\varphi}_{L^2 \cap L^\infty}=(1+\norm{\varphi}_\infty)
% \end{align*}
% since there exist a $\theta \in (0,1)$, $\norm{\varphi}_2=1$
% \begin{align*}
%  \norm{\varphi}_p \leq \norm{\varphi}_2^{1-\theta} \norm{\varphi}_\infty^\theta= \norm{\varphi}_\infty^\theta \leq (1+\norm{\varphi}_\infty)=\norm{\varphi}_{L^2 \cap L^\infty}.
% \end{align*}

%\end{section}

% \section{Uniquness of  $\sqrt{\frac{k}{N}}$}\label{sec:uniquness}
% The two conditions given in section \ref{sec:beta} for the weight function f lead for $l=1$ to the two conditions
% \begin{align*}
%  (f(k)-f(k-1))^2 \lesssim (N k)^{-1}
% \end{align*}
% and 
% \begin{align*}
%  \exists g \geq 0:  g(k)^2 \lesssim  f(k) \; , \; g(k)^{-2} \frac{k}{N}  \lesssim f(k) 
% \end{align*}
% The first condition implies $f(k)\lesssim  \sqrt{ \frac{k}{N}}$ and the second one $ \sqrt{ \frac{k}{N}} \lesssim f(k)  $ thus the claimed result follows.
% The second implication follows by a contradiction argument the first implication follows from rewriting the condition as 
% \begin{align*}
%  f(k)\lesssim (N k)^{-2}+ f(k-1)
% \end{align*}
% This implies 
% \begin{align*}
%  f(k)\lesssim (N k)^{-2}+ f(k-)
% \end{align*}
% 
% 
%  can be seen by squaring, adding $f(k-1)$ to the right, inserting the condition iteratively for all $f(k-l)$ and estimation by integrat
% ion of $\frac{1}{\sqrt{k}} $the second by contradiction. 

\chapter{Improvement of the Convergence of Theorem\,\ref{thm:thm2}}\label{app:ratebes}

We can slightly improve the rate of convergence of equation \eqref{equ:estimating3III} by improving the estimate of \eqref{equ:ppws2qq}. We use the same idea as in the proof of Lemma\,\ref{lem:3termeg}.\ref{lem:3.2g}.
Therefore, we split this term in a part, where at least a "few particles" of $ \psi$ are in the state $p$ and the complement.
This helps since the diagonal term and the off diagonal term arising in the estimation can be treated differently.
With the split we can distinguish the behavior of the terms beforehand and estimate them accordingly. Hence we gain a tiny bit of convergence speed in the estimation process.  
  
We define the same splitting as in \eqref{equ:splitting}. However, to use the estimates from the proof of Lemma\,\ref{lem:estimating3alt} we implement the splitting in a different way.
Define $\varUpsilon^1(k)= \id_{\{k \leq N^{1-\delta}\}} $ and $\varUpsilon^2(k):=\id - \varUpsilon^1(k) $. We rewrite the term on the left-hand side of  \eqref{equ:ppws2qq} 
\begin{align}\label{equ:2p2qmity}
 |\llangle \psi, p_1 p_2   w_{12}^{s,2} \widehat \mu_1  q_1q_2 \psi \rrangle| &= |\llangle \psi, p_1 p_2   w_{12}^{s,2}(\widehat \varUpsilon_1 + \widehat \varUpsilon_2)  \widehat \mu_1  q_1q_2 \psi \rrangle| \notag \\
&\leq  |\llangle \psi, p_1 p_2   w_{12}^{s,2} \widehat \varUpsilon_1  \widehat \mu_1  q_1q_2 \psi \rrangle|+ |\llangle \psi, p_1 p_2   w_{12}^{s,2} \widehat \varUpsilon_2  \widehat \mu_1  q_1q_2 \psi \rrangle|.
\end{align}
We start with estimating  $ |\llangle \psi, p_1 p_2   w_{12}^{s,2} \widehat \varUpsilon_1  \widehat \mu_1  q_1q_2 \psi \rrangle|$. Here we have cut the parts with too many bad particles 
so we can squeeze out an $N$ to some power of ${-\delta}$, hence we do not have to try to get a $\beta$. Except of writing $\mu^1=\mu^\frac{1}{2} \mu^\frac{1}{2}$ and bringing one of them on the other side of the
interaction the calculation stays exactly the same as in Lemma\,\ref{lem:estimating3alt}, so we only give a rough sketch of the proof here. 
\begin{align*}
  |\llangle \psi, p_1 p_2   w_{12}^{s,2} \widehat \varUpsilon_1  \widehat \mu_1  q_1q_2 \psi \rrangle| &= \frac{1}{N-1} |\llangle \psi, \sum_{j=2}^N p_1 p_j   w_{1j}^{s,2} \widehat \varUpsilon_1  \widehat \mu_1  q_1q_j \psi \rrangle|\\
& =\frac{1}{N-1} |\llangle \psi, \sum_{j=2}^N  p_1 p_j \widehat{\tau_2 \varUpsilon_1}  w_{1j}^{s,2}   \widehat \mu_1  q_1q_j \psi \rrangle|\\
&\leq \frac{1}{N-1} \norm{\widehat \mu_1 q_1 \psi} \sqrt{\sum_{i,j=2}^N\llangle \psi,  p_1 p_j  \widehat{\tau_2 \varUpsilon_1}   w_{1j}^{s,2} q_j q_i w_{1i}^{s,2}   p_1 p_i  \widehat{\tau_2 \varUpsilon_1} \psi  \rrangle} \\
\end{align*}
Since $ \norm{\widehat \mu_1 q_1 \psi} \leq 1$ similarly to \eqref{equ:2p2qohney}
\begin{align*}
 |\llangle \psi, p_1 p_2   w_{12}^{s,2} \widehat \varUpsilon_1  \widehat \mu_1  q_1q_2 \psi \rrangle| \leq \frac{1}{N-1}\sqrt{A+B},
\end{align*}
where 
\begin{align*}
 A&:= \sum_{2\leq i \neq j \leq N} \llangle \psi,  p_1 p_j  \widehat{\tau_2 \varUpsilon_1}   w_{1j}^{s,2} q_j q_i w_{1i}^{s,2}   p_1 p_i  \widehat{\tau_2 \varUpsilon_1} \psi  \rrangle\\
B&:= \sum_{i=2}^N \llangle \psi,  p_1 p_i  \widehat{\tau_2 \varUpsilon_1}   w_{1i}^{s,2} q_i q_i w_{1i}^{s,2}   p_1 p_i  \widehat{\tau_2 \varUpsilon_1} \psi  \rrangle.
\end{align*}
We do not use the cutoff here. %so with \eqref{equ:wp2} and
 With \eqref{equ:wp2} and similarly to \eqref{equ:Bohne} % without the $\widehat \mu^1q_1$ with \eqref{equ:wp2} and %which means we do not get an extra $N^{1/2}$, 
we get
\begin{align}\label{equ:Bmit}
 B &\lesssim N c^{2-s }\norm{\varphi}_\infty^2.
\end{align}
Since there is no $q_1$ in the middle of the term $A$ as in \eqref{equ:Aohne} we can estimate it directly and get as before 
\begin{align*}
 A&\lesssim N^2 \norm{w^{\epsi,s}}_s^2 (1+\norm{\varphi}_\infty)^4  \llangle \psi, q_1 \widehat{\tau_2 \varUpsilon_1} \psi \rrangle \\
&\lesssim N^2 (1+\norm{\varphi}_\infty)^4  \llangle \psi, \widehat{\tau_2 \varUpsilon_1} \widehat n^2 \psi \rrangle,
\end{align*}
where we have $\varUpsilon_1$ still left in the expression.
Since $\varUpsilon_1= \id_{\{k \leq N^{1-\delta}\}}$ we get
\begin{align*}
 {\tau_2 \varUpsilon_1} n^2 \leq N^{-\delta}
\end{align*}
and obtain
\begin{align}\label{equ:Amit}
 |A|\lesssim N^{2-\delta}\norm{\varphi}^4_{L^\infty \cap  L^2}.
\end{align}
Collecting the estimates \eqref{equ:Bmit} and \eqref{equ:Amit}
\begin{align}\label{equ:2p2qmity1}
  |\llangle \psi, p_1 p_2   w_{12}^{s,2} \widehat \varUpsilon_1  \widehat \mu_1  q_1q_2 \psi \rrangle| &\lesssim \frac{1}{N} \sqrt{A+B} \notag \\
&\leq  N^{-\frac{\delta}{2}}\norm{\varphi}^2_{L^\infty \cap  L^2} +  c^{1-\frac{s}{2}} N^{-\frac{1}{2}}\norm{\varphi}_\infty.
\end{align}
\\
The second part of \eqref{equ:2p2qmity}
\begin{align*}
 |\llangle \psi, p_1 p_2   w_{12}^{s,2} \widehat \varUpsilon_2  \widehat \mu_1  q_1q_2 \psi \rrangle|
\end{align*}
is dealt with splitting $\mu_1=\mu_1^\frac{1}{2}\mu_1^\frac{1}{2} $ to be able to get a $\beta$.
As in \eqref{equ:2p2qohney}
% \begin{align*}
%  |\langle \psi, p_1 p_2   w_{12}^{p,2} \widehat \varUpsilon_2  \widehat \mu_1  q_1q_2 \psi \rangle|
%  &\leq \frac{1}{N-1}|\langle \psi, \sum_{j=2}^N p_1 p_j   w_{1j}^{p,2} \widehat \varUpsilon_2  \widehat \mu_1^\frac{1}{2} \widehat \mu_1^\frac{1}{2}   q_1q_j \psi \rangle|\\
% &\leq \norm{\widehat \mu_1^\frac{1}{2}   q_1 \psi} \sqrt{\sum_{j,i=2}^N \langle \psi, p_1 p_j   w_{1j}^{p,2} \widehat \varUpsilon_2 \widehat \mu_1    q_j q_i  w_{1i}^{p,2} p_1 p_i \psi \rangle }
% \end{align*}

% Since 
% \begin{align*}
%  \norm{\widehat \mu_1^\frac{1}{2}   q_1 \psi}^2 \leq \langle \psi, \widehat n^{-1} \widehat n^2 \psi \rangle  = \beta
% \end{align*}
% we have
\begin{align*}
  |\llangle \psi, p_1 p_2   w_{12}^{s,2} \widehat \varUpsilon_2  \widehat \mu_1  q_1q_2 \psi \rrangle| \leq \frac{\sqrt{\beta}}{N-1} \sqrt{A+B}
\end{align*}
with the same splitting as before
\begin{align*}
 A&:=\sum_{2\leq j\neq i \leq N} \llangle \psi, p_1 p_j   w_{1j}^{s,2} \widehat \varUpsilon_2 \widehat \mu_1    q_j q_i q_1  w_{1i}^{s,2} p_1 p_i \psi \rrangle\\
B&:=\sum_{i=2}^N \llangle \psi, p_1 p_i   w_{1i}^{s,2} \widehat \varUpsilon_2 \widehat \mu_1    q_i  q_1 w_{1i}^{s,2} p_1 p_i \psi \rrangle.
\end{align*}
With
\begin{align*}
 \varUpsilon_2(k)= \id_{\{ k > N^{1-\delta} \}}
\end{align*}
we find $\varUpsilon_2 \mu_1 \leq \varUpsilon_2 n^{-1} \leq N^\frac{\delta}{2}  $. Hence
\begin{align*}
 \norm{\widehat \varUpsilon_2 \widehat \mu_1    q_i q_1}_\Op \leq N^\frac{\delta}{2}
\end{align*}
and $B$ can be estimated similar to \eqref{equ:Bohne} by 
\begin{align*}
 B &\lesssim N^{1+\delta/2} c^{2-s }\norm{\varphi}_\infty^2,
%  N^\frac{\delta}{2} \sum_{i=2}^N \norm{w_{1i}^{p,2} p_1 p_i \psi }^2\\
%  &= N^\frac{\delta}{2} \sum_{i=2}^N \langle \psi, p_1 p_i   {w_{1i}^{p,2}}^2 p_1 p_i \psi \rangle\\
% &\leq  N^{1+\frac{\delta}{2}} \norm{p_1  {w_{1i}^{p,2}}^2 p_1}\\
% &= N^{1+\frac{\delta}{2}} \norm{ {w^{p,2}}^2*|\varphi|^2 }_\infty\\
% &\leq N^{1+\frac{\delta}{2}} \norm{w^{p,2}}^2_2 \norm{\varphi}^2_\infty\\
% &\leq N^{1+\frac{\delta}{2}} a^{2-p} \norm{w^p}^p_p \norm{\varphi}^2_\infty
\end{align*}
whereas there only appears an $N^{\delta/2}$ and not $N^\frac{1}{2}$.
The term A can be estimated exactly like the A of \eqref{equ:Aohne}, the $\varUpsilon_2$ does not help here and can be neglected 
\begin{align*}
A \lesssim  N^2 (1+\norm{\varphi}^4_\infty) \beta.
\end{align*}
Putting the estimates  $ |\llangle \psi, p_1 p_2   w_{12}^{s,2} \widehat \varUpsilon_2  \widehat \mu_1  q_1q_2 \psi \rrangle| $ together
\begin{align*}
 |\llangle \psi, p_1 p_2   w_{12}^{s,2} \widehat \varUpsilon_2  \widehat \mu_1  q_1q_2 \psi \rrangle|&\lesssim
 \frac{\sqrt{\beta}}{N-1}\sqrt{ N^{1+\frac{\delta}{2}} c^{2-s} \norm{\varphi}^2_\infty+ N^2 \norm{\varphi}^4_{L^\infty \cap  L^2} \beta}\\
 &\lesssim  N^{-1+\frac{\delta}{2}} c^{2-s}  \norm{\varphi}^2_\infty + \norm{\varphi}^2_{L^\infty \cap  L^2} \beta.
\end{align*}
Hence, this implies with \eqref{equ:2p2qmity1} for the equation \eqref{equ:2p2qmity} 
\begin{align*}
  |\llangle \psi, p_1 p_2   w_{12}^{s,2} \widehat \mu_1  q_1q_2 \psi \rrangle| \lesssim
 N^{-\frac{\delta}{2}}\norm{\varphi}^2_{L^\infty \cap  L^2} +  c^{1-\frac{s}{2}} N^{-\frac{1}{2}}\norm{\varphi}_\infty\\
+N^{-1+\frac{\delta}{2}} c^{2-s}  \norm{\varphi}^2_\infty + \norm{\varphi}^2_{L^\infty \cap  L^2} \beta \numberthis \label{equ:letzt} .
\end{align*}
Finally we use \eqref{equ:letzt} and \eqref{equ:2p2qw1} to obtain the improved estimate of \eqref{equ:estimating3III}
\begin{align*}
 \mathrm{III}& \lesssim %C  \norm{\varphi}_{\infty }\big(\norm{\nabla_1 \varphi} \beta+ \norm{\nabla_1 \varphi}a^{2-\frac{2p}{p_0}}+\sqrt{2}\norm{\nabla_1 q_1 \psi}^2\big)\\
% + N^{-\frac{\delta}{2}}\norm{\varphi}^2_{L^\infty \cap  L^2} \norm{w^p}_p +  a^{1-\frac{p}{2}} N^{-\frac{1}{2}} \norm{\varphi}_\infty \norm{w^p}_p^{p}\\
% +N^{-1+\frac{\delta}{2}} a^{2-p} \norm{w^p}^p_p \norm{\varphi}^2_\infty + 2 \norm{\varphi}^2_{L^\infty \cap  L^2} \norm{w^p}_p \beta\\
 N^{-\frac{\delta}{2}}\norm{\varphi}^2_{L^\infty \cap  L^2} +  c^{1-\frac{s}{2}} N^{-\frac{1}{2}}\norm{\varphi}_\infty+ 
N^{-1+\frac{\delta}{2}} c^{2-s}  \norm{\varphi}^2_\infty + \norm{\varphi}^2_{L^\infty \cap  L^2} \beta\\
&\quad+\norm{\varphi}_{\infty }\big(\norm{ \nabla \varphi} \beta+ \norm{ \varphi}_{H^1}c^{2-\frac{2s}{s_0}}+\norm{\nabla_1 q_1 \psi}^2\big).
\end{align*}
After setting $c=N^\vartheta$ and optimizing $\delta$ and $\vartheta$ we find 
\begin{align*}
  \mathrm{III} \lesssim \norm{ \varphi}_{H^1\cap L^\infty }^3(\beta+N^\eta)+ \norm{\varphi}_{\infty }\norm{\nabla_1 q_1 \psi}^2
\end{align*}
with
\begin{align*}
 \eta= -\frac{s/s_0-1}{2s/s_0-s/2-1}
\end{align*}
which is slightly better than the $\eta$ given in equation \eqref{equ:etagr}.

%\end{chapter}

%\setglossarystyle{longragged3colheader}
\printglossary[type=is]

\bibliography{Doklit}
\bibliographystyle{alphanum}

\end{document}